\documentclass[pra, aps,superscriptaddress,nofootinbib,10pt]{revtex4-2}

\usepackage[a4paper, left=1in, right=1in, top=1in, bottom=1in]{geometry}

\usepackage{cmap} 
\usepackage[utf8]{inputenc}
\usepackage[english]{babel}
\usepackage[T1]{fontenc}
\usepackage{appendix}
\usepackage{braket}
\usepackage{url, amsfonts, tikz, physics, listings, xcolor, graphicx, float}
\usepackage[shortlabels]{enumitem}
\usepackage{dsfont}
\usepackage{amsmath, amssymb, amstext, amscd, amsthm, makeidx, graphicx, url, mathrsfs, mathtools, longdivision, polynom, bbm, complexity}
\usepackage{nicefrac}
\usepackage[linesnumbered,ruled,vlined]{algorithm2e}
\usepackage{xcolor}
\usepackage[normalem]{ulem}
\usepackage{booktabs}
\usepackage{algorithmicx}
\usepackage{natbib}

\definecolor{blueviolet}{rgb}{0.2, 0.2, 0.6}
\definecolor{webgreen}{rgb}{0,.5,0}
\definecolor{webbrown}{rgb}{.6,0,0}
\usepackage[pdftex,
  bookmarks=false,
  colorlinks=true, 
  urlcolor=webbrown,
  linkcolor=blueviolet, 
  citecolor=webgreen,
  pdfstartpage=1,
  pdfstartview={FitH},  
  bookmarksopen=false
  ]{hyperref}
\usepackage{cleveref}

\makeatletter
\newcommand\RedeclareMathOperator{%
  \@ifstar{\def\rmo@s{m}\rmo@redeclare}{\def\rmo@s{o}\rmo@redeclare}%
}

\newcommand\rmo@redeclare[2]{%
  \begingroup \escapechar\m@ne\xdef\@gtempa{{\string#1}}\endgroup
  \expandafter\@ifundefined\@gtempa
     {\@latex@error{\noexpand#1undefined}\@ehc}%
     \relax
  \expandafter\rmo@declmathop\rmo@s{#1}{#2}}
\newcommand\rmo@declmathop[3]{%
  \DeclareRobustCommand{#2}{\qopname\newmcodes@#1{#3}}%
}
\@onlypreamble\RedeclareMathOperator
\makeatother

\allowdisplaybreaks

\RedeclareMathOperator*{\E}{{\mathbb{E}}}

\makeatletter
\DeclareFontFamily{OMX}{MnSymbolE}{}
\DeclareSymbolFont{MnLargeSymbols}{OMX}{MnSymbolE}{m}{n}
\SetSymbolFont{MnLargeSymbols}{bold}{OMX}{MnSymbolE}{b}{n}
\DeclareFontShape{OMX}{MnSymbolE}{m}{n}{
    <-6>  MnSymbolE5
   <6-7>  MnSymbolE6
   <7-8>  MnSymbolE7
   <8-9>  MnSymbolE8
   <9-10> MnSymbolE9
  <10-12> MnSymbolE10
  <12->   MnSymbolE12
}{}
\DeclareFontShape{OMX}{MnSymbolE}{b}{n}{
    <-6>  MnSymbolE-Bold5
   <6-7>  MnSymbolE-Bold6
   <7-8>  MnSymbolE-Bold7
   <8-9>  MnSymbolE-Bold8
   <9-10> MnSymbolE-Bold9
  <10-12> MnSymbolE-Bold10
  <12->   MnSymbolE-Bold12
}{}

\let\llangle\@undefined
\let\rrangle\@undefined
\DeclareMathDelimiter{\llangle}{\mathopen}%
                     {MnLargeSymbols}{'164}{MnLargeSymbols}{'164}
\DeclareMathDelimiter{\rrangle}{\mathclose}%
                     {MnLargeSymbols}{'171}{MnLargeSymbols}{'171}
\makeatother

\newtheorem{theorem}{Theorem}
\newtheorem{prop}{Proposition}
\newtheorem{lemma}{Lemma}
\newtheorem{corollary}{Corollary}
\newtheorem{definition}{Definition}

\newtheorem{claim}{Claim}

\newcommand{\indicator}{\mathds{1}}

\newcommand{\nocontentsline}[3]{}
\let\origcontentsline\addcontentsline
\newcommand\stoptoc{\let\addcontentsline\nocontentsline}
\newcommand\resumetoc{\let\addcontentsline\origcontentsline}

\begin{document}

\title{Quantum advantage for learning shallow neural networks with natural data distributions}

\author{Laura Lewis}
\email{llewis@alumni.caltech.edu}
\affiliation{Google Quantum AI, Venice, CA, USA}
\affiliation{University of Cambridge, Cambridge, UK}
\affiliation{University of Edinburgh, Edinburgh, UK}

\author{Dar Gilboa}
\affiliation{Google Quantum AI, Venice, CA, USA}

\author{Jarrod R. McClean}
\affiliation{Google Quantum AI, Venice, CA, USA}

\begin{abstract}
Without large quantum computers to empirically evaluate performance, theoretical frameworks such as the quantum statistical query (QSQ) are a primary tool to study quantum algorithms for learning classical functions and search for quantum advantage in machine learning tasks.
However, we only understand quantum advantage in this model at two extremes: either exponential advantages for uniform input distributions or no advantage for arbitrary distributions.
Our work helps close the gap between these two regimes by designing an efficient quantum algorithm for learning periodic neurons in the QSQ model over a variety of non-uniform distributions and the first explicit treatment of real-valued functions.
We prove that this problem is hard not only for classical gradient-based algorithms, which are the workhorses of machine learning, but also for a more general class of SQ algorithms, establishing an exponential quantum advantage.
\end{abstract}

\maketitle

\stoptoc
\section{Introduction}

Machine learning (ML) is currently experiencing explosive success, made possible by an overwhelming growth of compute power, data availability, and improved models~\cite{LeCun2015-ot, Brown2020-uq, wei2022emergent, jumper2021highly}.
In parallel, quantum technology is also witnessing remarkable progress, including breakthroughs in quantum error correction~\cite{lacroix2024scaling,acharya2024quantum,eickbusch2024demonstrating,rodriguez2024experimental,reichardt2024logical,bravyi2024high,reichardt2024demonstration,da2024demonstration,caune2024demonstrating,zhou2024algorithmic,putterman2024hardware,lee2024low,wills2024constant,nguyen2024quantum,gidney2024magic} and demonstrations of computations beyond the known limits of classical computers~\cite{arute2019quantum,zhong2020quantum,wu2021strong,zhu2022quantum,acharya2024quantum,huang2022quantum,zhu2023interactive,lewis2024experimental}.
Given that our universe is inherently quantum, it is natural to consider leveraging powerful quantum computers for ML tasks, in hopes of new scientific advancements~\cite{huang2022quantum,aimeur2006machine,aimeur2013quantum,wiebe2012quantum,wiebe2014quantum,harrow2009quantum,kapoor2016quantum,lloyd2013quantum,lloyd2014quantum,rebentrost2014quantum,lloyd2016quantum,cong2016quantum,kerenidis2016quantum,brandao2019quantum,rebentrost2018quantum,zhao2019quantum}.
However, modern classical ML is mainly driven by empirical success, extending far beyond our theoretical understanding.
In contrast, quantum technologies are still in their infancy, where we cannot yet accurately train and test large quantum ML models.
Thus, we must rely on the rigorous frameworks of learning theory to characterize the performance of quantum learning algorithms and their potential advantage over classical learners.

One possible avenue for quantum advantage is to use quantum algorithms to learn classical objects, e.g., classical functions~\cite{bshouty1995learning,jackson2002quantum,arunachalam2020quantum,atici2007quantum,cross2015quantum,bernstein1993quantum,arunachalam2021two,grilo2019learning,kanade2018learning,caro2020quantum,nadimpalli2024pauli,arunachalam2024learning,montanaro2012quantum,servedio2004equivalences,gavinsky2008quantum,caro2024testing} or distributions~\cite{hinsche2021learnability,hinsche2022single,nietner2023average}.
Such results commonly consider the quantum counterparts of frameworks such as probably approximately correct (PAC)~\cite{valiant1984theory} and statistical query (SQ) learning~\cite{kearns1998efficient}, appropriately called quantum PAC~\cite{bshouty1995learning} and quantum SQ (QSQ)~\cite{arunachalam2020quantum}, respectively.
In particular, some exciting results show that there exist function classes for which quantum PAC/QSQ algorithms can provide exponential sample complexity advantages over classical learners when the input data distribution is uniform~\cite{bshouty1995learning,arunachalam2020quantum,atici2007quantum,cross2015quantum,bernstein1993quantum,arunachalam2021two,grilo2019learning,nadimpalli2024pauli,arunachalam2024learning}.
This is in stark contrast to the seminal result proving there is no quantum advantage for arbitrary distributions~\cite{arunachalam2018optimal,atici2005improved,servedio2004equivalences,zhang2010improved}.
The void between exponential advantages on idealized uniform distributions and no advantage on potentially adversarial distributions leaves a large gap in our understanding of quantum learning advantages.
These results also highlight the challenges in analyzing quantum advantage for empirical data distributions and mirror results in classical ML, where there exists problems that are \textsf{NP}-complete for arbitrary distributions but easy for the distribution-specific case~\cite{brutzkus2017globally, Safran2018-xn, Daniely2020-eu,  Kiani2024-qt}.
Moreover, to our knowledge, all results in quantum learning theory to date focus on Boolean or discrete functions, while the majority of large-scale ML focuses on real-valued functions.
Together, these two points raise our central question:

\begin{center}
    \textit{Are there classes of real-valued functions and non-uniform\\distributions for which quantum data is advantageous?}
\end{center}

These are also stated as two open questions in~\cite{arunachalam2017survey}.
Here, by quantum data, we mean classical functions over distributions encoded into so-called quantum example states, as in quantum PAC and QSQ learning.
We provide a new perspective on when these states might arise naturally later in the work.

While some results consider learning Boolean functions over $c$-bounded product distributions~\cite{kanade2018learning,caro2020quantum}, proving quantum advantages for more general non-uniform distributions still remains open.
Moreover, for other forms of quantum data, such as expectation values of ground states, quantum advantages for learning over non-uniform distributions have been explored~\cite{molteni2024exponential}. However, this is incomparable to the present work, where we focus on classical functions encoded in quantum example states.

In this work, we provide a positive answer to our central question by efficiently learning real-valued functions that are a composition of a periodic function and a linear function in the QSQ model over a broad range of non-uniform distributions, which includes Gaussian, generalized Gaussian~\cite{subbotin1923law}, and logistic distributions.
These distributions are practically relevant with generalized Gaussian and logistic distributions finding applications in, e.g., image processing~\cite{do2002wavelet,mallat1989theory,moulin1999analysis} and population growth~\cite{pearl1920rate,pearl1940logistic,schultz1930standard,plackett1959analysis}, respectively.
Moreover, note that success in the QSQ model automatically implies success in the quantum PAC model, as the QSQ model is strictly weaker because it does not allow entangled measurements~\cite{arunachalam2024role}.

We highlight that the function class we consider is well-studied in the classical ML literature~\cite{shalev2017failures,shamir2018distribution,song2017complexity,song2021cryptographic}.
There, such functions --- called \emph{cosine neurons} or, more generally, \emph{periodic neurons} --- are commonly analyzed, as they form the basic structure of neural networks with periodic activation functions~\cite{sitzmann2020implicit, Meronen2021-he, Chan2021-jl, Faroughi2023-th, Mommert2024-sj, Liu2024-pr} and can be seen as an extension of generalized linear models~\cite{muller2012generalized,nelder1972generalized}.
In particular, Ref.~\cite{shamir2018distribution} proves that any gradient-based classical algorithm cannot learn periodic neurons when the input data distribution has a sufficiently sparse Fourier transform, which is satisfied by many natural distributions, e.g., Gaussians, mixtures of Gaussians, Schwartz functions~\cite{hunter2001applied}, etc.
We strengthen their proof to apply to our specific parameter choices that focus on the regime of quantum advantage.
Furthermore, although gradient methods are perhaps the most popular in classical ML, there is strong evidence for classical hardness beyond gradient methods.
In fact, we extend the classical hardness to hold for a more general class of algorithms performing correlational SQs~\cite{bshouty2002using,yang2005new}.
Additionally, Ref.~\cite{song2017complexity} shows an exponential lower bound for any classical SQ algorithm learning this function class with respect to any log-concave distribution.
Ref.~\cite{song2021cryptographic} extends the hardness to any polynomial time classical algorithm learning under small amounts of noise and over Gaussian distributions, assuming the hardness of solving worst-case lattice problems~\cite{regev2009lattices,micciancio2009lattice}.
These results~\cite{song2017complexity,song2021cryptographic} do not directly apply to our setting due to a difference between the parameter regimes needed for quantum advantage versus classical hardness, but we expect classical hardness to still hold in this regime and leave this generalization open to future work.

Our algorithm uses a polynomial number of QSQs and iterations of gradient descent, resulting in a quantum advantage over any classical gradient-based algorithm for sufficiently Fourier-sparse input distributions~\cite{shamir2018distribution}.
Here, the classical algorithms considered are any algorithms that utilize approximate gradients of an average loss function, which includes, e.g., neural networks.
Concretely, we obtain an exponential quantum advantage against classical gradient methods for Gaussian, generalized Gaussian, and logistic distributions.
For Gaussian distributions, we additionally strengthen classical hardness to hold against a natural restriction of SQ algorithms (namely, correlational SQ algorithms~\cite{bshouty2002using,yang2005new}), which includes gradient methods, dimension reduction, and moment-based methods.
To our knowledge, this is the first result in quantum learning of classical functions that explicitly considers real-valued functions.
\Cref{fig:summary} illustrates a schematic overview of our work.

\begin{figure}
    \centering
    \includegraphics[width=0.9\linewidth]{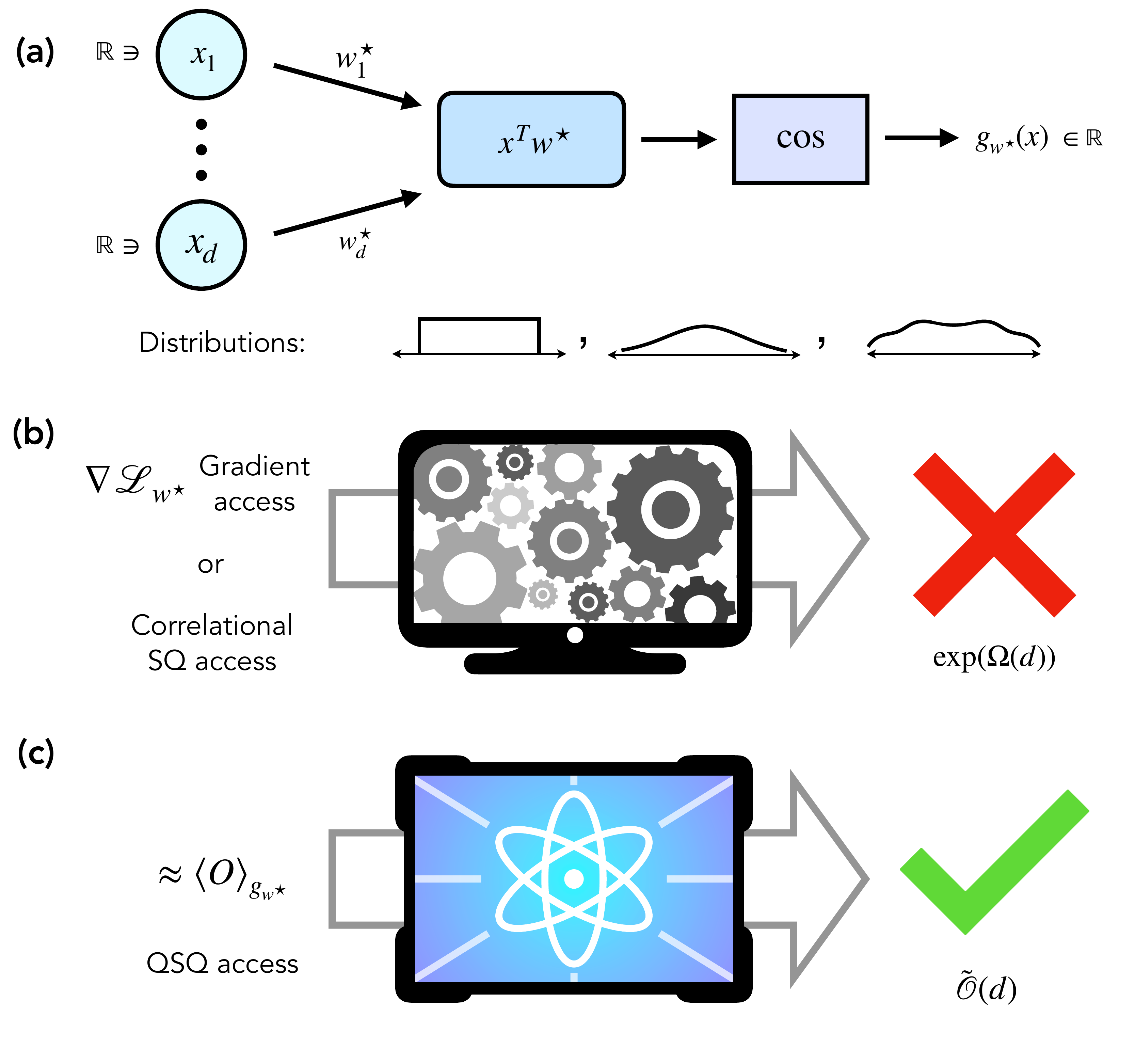}
    \caption{\textbf{Overview of results.} \textbf{(a) Target function and input distributions.} Given an input vector $x \in \mathbb{R}^d$, we consider learning functions of the form $g_{w^\star}(x) = \cos(x^\intercal w^\star)$, where $w^\star \in \mathbb{R}^d$ is an unknown vector.
    Our illustration emphasizes their connection with classical deep learning, where they are called cosine neurons.
    We also consider more general periodic neurons, which one can view as linear combinations of cosine neurons with unknown weights.
    We consider input distributions such as uniform, Gaussians, and more general distributions which are sufficiently flat, as characterized by technical conditions specified in Appendix~\ref{sec:non-unif}.
    \textbf{(b) Classical hardness.} We strengthen the arguments of~\cite{shamir2018distribution} to show that classical gradient methods require an exponential number of iterations (i.e., an exponential number of gradient samples) in the dimension of the problem and the norm $R_w$ of $w^\star$ to learn these functions.
    \textbf{(c) Quantum algorithm.} In contrast, our new quantum algorithm using QSQs is exponentially more efficient with respect to both time and sample complexity.}
    \label{fig:summary}
\end{figure}

\section{Results}

In this section, we introduce the task of learning periodic neurons and show that it is classically hard for a broad class of powerful algorithms. Then, we detail our quantum algorithm that solves this problem efficiently, exhibiting an exponential quantum advantage.

\subsection{Problem definition}
In this section, we define common access models in (quantum) learning theory and describe our learning problem more formally.
We refer to Appendices~\ref{sec:learning} and~\ref{sec:detail-prob} for further details.

We aim to learn a collection of functions $\mathcal{C} \subseteq \{c: \mathcal{X} \to \mathcal{Y}\}$ called the \emph{concept class}, where $\mathcal{X}, \mathcal{Y}$ are the input/output spaces, respectively.
In particular, given some form of access to an unknown concept $c^\star \in \mathcal{C}$, we want to learn an approximation of $c^\star$ with high probability.
Typically in learning theory, one considers Boolean functions with $\mathcal{X} = \{0,1\}^d, \mathcal{Y} = \{0,1\}$.
Importantly, in this work, we instead consider $\mathcal{X} = \mathbb{R}^d, \mathcal{Y} = \mathbb{R}$.

In the classical PAC model~\cite{valiant1984theory}, the learning algorithm is given labeled random examples $(x_i, c^\star(x_i))_{i=1}^N$, where the $x_i$ are sampled from a distribution $\mathcal{D}$ over $\mathcal{X}$ and $c^\star \in \mathcal{C}$ is an unknown target function.
The SQ model~\cite{kearns1998efficient} is weaker than PAC, where, instead of direct access to the examples, the learning algorithm can only obtain noisy expectation values of functions of the data.
This was originally proposed to model learning given noisy examples, and commonly used algorithms such as stochastic gradient descent~\cite{Feldman2015-hj}, Markov chain Monte Carlo methods~\cite{tanner1987calculation,gelfand1990sampling}, and simulated annealing~\cite{kirkpatrick1983optimization,vcerny1985thermodynamical} can be implemented in this model.

We also consider the \emph{correlational} SQ model~\cite{bshouty2002using,yang2005new}.
This is a restriction of general SQs in which queries are only allowed to act on the input space $\mathcal{X}$, not the labeled output space.
We define this more precisely in Appendix~\ref{sec:learning}.
Correlational SQs include gradient methods, dimension reduction, and moment-based methods as special cases.
In the case of Boolean functions, correlational SQs and general SQs are in fact equivalent~\cite{bshouty2002using}, but there exist separations between them for real functions~\cite{andoni2014learning,andoni2019attribute,chen2022learning}.

The quantum PAC and QSQ models are natural generalizations of these settings.
In the quantum PAC model~\cite{bshouty1995learning}, the learning algorithm is given copies of the quantum example state
\begin{equation}
    \label{eq:example-state-main}
    \ket{c^\star} \triangleq \sum_{x \in \mathcal{X}} \sqrt{\mathcal{D}(x)}\ket{x}\ket{c^\star(x)},
\end{equation}
where $\mathcal{D}$ is again some probability distribution.
We note that there are some restrictions on the distributions $\mathcal{D}$ for which one can efficiently prepare this state and discuss this later.
Also notice that upon measuring a quantum example state, one obtains $(x, c^\star(x))$ for $x$ sampled from the distribution $\mathcal{D}$, hence recovering the classical PAC examples.
For QSQ access~\cite{arunachalam2020quantum}, the learner queries an observable $O$ and receives an approximation of the expectation value $\expval{O}{c^\star}$.
This is weaker than the quantum PAC model due to the inability to take entangled measurements across multiple copies of $\ket{c^\star}$~\cite{arunachalam2024role}.
Notice also that because $\mathcal{X}$ is a continuous space in our setting, these definitions require discretization/truncation, which we discuss further in the Methods and Appendix~\ref{sec:learning}.
In all aforementioned cases, the goal is to learn the unknown function $c^\star$ approximately with high probability using as few examples/queries as possible.

We are interested in learning a concept class consisting of functions that are a composition of a periodic function and a linear function.
In other words, these are functions that can be represented as a single-layer neural network with a periodic activation function, hence dubbed periodic neurons.
This ansatz is quite powerful and in some cases is able to achieve universal function approximation~\cite{Cybenko1989-zh, Maiorov1999-qk, Guliyev2018-qt}.
Moreover, the periodic neuron has known relationships to important complexity theoretic problems~\cite{song2021cryptographic,gupte2022continuous}.

Explicitly, let $d \geq 1$ denote the input dimension and let $\mathbb{S}^{d-1}$ denote the $(d-1)$-dimensional unit sphere.
Then, our concept class is defined as
\begin{equation}
    \label{eq:concept-class}
    \mathcal{C} \triangleq \{g_{w^\star}: \mathbb{R}^d \to [-1,1] : g_{w^\star}(x) = \tilde{g}(x^\intercal w^\star), w^\star \in R_w \mathbb{S}^{d-1}\},
\end{equation}
where $R_w > 0$ is the norm of the unknown vector $w^\star$ and $\tilde{g}: \mathbb{R} \to [-1,1]$ is a periodic function of period $1$ that can be written as
\begin{equation}
    \label{eq:g-tilde-main}
    \tilde{g}(y) = \sum_{j=1}^D \beta_j^\star \cos(2\pi j y),\quad \norm{\beta^\star}_1 = 1
\end{equation}
for some constant $D > 0$ and unknown parameters $\beta_j^\star \in \mathbb{R}$.
In other words, our target functions $g_{w^\star}$ are defined as follows.
First, consider an unknown vector $w^\star$ of norm $R_w$, and consider the linear function $x^\intercal w^\star$ defined by this coefficient vector.
Then, compose this linear function with a linear combination of cosines, where the weights $\beta_j^\star$ are unknown.
In our analysis, we have additional constraints on the vector $w^\star$, e.g., restricted to the positive orthant and bounded away from $0$, but for simplicity of presentation, we omit this detail in the main text.
We direct the reader to Appendix~\ref{sec:detail-prob} for more details.

To learn a target concept $g_{w^\star}$ with respect to a distribution $\mathcal{D}$, we want to find a good predictor $f_\theta(x)$ which minimizes the objective function
\begin{equation}
  \label{eq:loss-main}
  \min_{\theta \in \Theta} \mathcal{L}_{w^\star}(\theta)\triangleq \min_{\theta \in \Theta}\mathop{\mathbb{E}}_{x \sim \mathcal{D}}[(f_\theta(x) - g_{w^\star}(x))^2],
\end{equation}
where $\theta$ are some tunable parameters.
Namely, for a given $\epsilon > 0$, we want to find parameters $\hat{\theta}$ such that $\mathcal{L}_{w^\star}(\hat{\theta}) \leq \epsilon$.
Classically, we consider algorithms that have access to gradients of this loss function and can compute it for a given choice of parameters $\theta$.
Our quantum algorithm additionally has QSQ access to the (discretized/truncated) example state $\ket{g_{w^\star}}$.
Here, discretization is necessary to encode the continuous outputs of the target function into a discrete quantum state, and we similarly require truncation to ensure that the superposition is not over an infinite space.

While gradient access is more restrictive than general classical SQ algorithms, SQ algorithms include gradient methods as a special case~\cite{Feldman2015-hj}.
Gradient-based algorithms are also the most widely used methods to train neural networks in practice.
Moreover, for Gaussian distributions, we extend classical hardness to hold against correlational SQ algorithms.
These are more restrictive than general SQ algorithms~\cite{andoni2014learning,andoni2019attribute,chen2022learning}, but we nevertheless view this as an important step towards proving SQ hardness.
As discussed above, there is also strong evidence that the problem remains hard for general SQ algorithms and even all efficient classical algorithms~\cite{song2017complexity,song2021cryptographic}.
In fact, the techniques for proving hardness against gradient methods~\cite{shamir2018distribution} are similar to those for existing SQ hardness results~\cite{kearns1998efficient,blum1994weakly}.

\subsection{Classical hardness}
Previous work from the classical literature~\cite{shalev2017failures,shamir2018distribution} shows that learning periodic neurons as described in the previous section is hard for classical gradient methods, which includes powerful algorithms such as classical neural networks.
This result holds for any input distribution that is sufficiently sparse in Fourier space, defined by the notion of $\epsilon(r)$-Fourier-concentration.
Intuitively, $\epsilon(r)$ is a function which characterizes how quickly the Fourier transform of the density function decays.
We define Fourier concentration formally in \Cref{def:cont_fourier_conc} in Appendix~\ref{sec:classical-hardness}.
In Appendix~\ref{sec:classical-hardness}, we strengthen the proof from~\cite{shamir2018distribution} to show that classical hardness still holds for our additional constraints on the vector $w^\star$.

\begin{theorem}[A variant of Theorem 4 in~\cite{shamir2018distribution}; Informal]
    \label{thm:class-hardness-main}
    Let $g_{w^\star}: \mathbb{R}^d \to [-1,1] \sim \mathrm{Unif}(\mathcal{C})$ be a uniformly sampled target function, where the unknown vector $w^\star \in \mathbb{R}^d$ has norm $R_w$.
    Consider an input distribution whose density $\varphi^2$ can be written as a square of a function $\varphi$ and is $\epsilon(r)$-Fourier-concentrated.
    Let $\epsilon' = \sqrt[3]{c_1(\exp(-c_2 d) + \sum_{n=1}^\infty \epsilon(n R_w/4)}$ for constants $c_1,c_2$.
    Then, any classical gradient-based algorithm requires at least $p/\epsilon'$ gradient samples with $\epsilon'$ precision to learn $g_{w^\star}$ with probability $1-p$ over the choice of $g_{w^\star}$.
\end{theorem}

Note that there are similar classical hardness results which hold for any $1$-Lipschitz loss function~\cite{shalev2017failures}, rather than the squared loss $\mathcal{L}_{w^\star}$ from \Cref{eq:loss-main}.
This theorem tells us that if the function $\epsilon(r)$ decays rapidly with $r$, then unless the number of gradient samples is extremely large or the noise in the problem is unrealistically small, a classical gradient-based algorithm cannot learn the concept class $\mathcal{C}$ from \Cref{eq:concept-class}.
We note that we only obtain a meaningful lower bound when $\epsilon(r)$ decays sufficiently quickly such that the infinite sum in the expression for $\epsilon'$ converges.
This is guaranteed when the Fourier transform of the input distribution has sharply decreasing tails.
For instance, for Gaussian distributions, the number of samples must scale as $\exp(\Omega(\min(d, R_w^2)))$.
Here, the classical hardness stems from the gradient of the loss function concentrating around a fixed value, which, in turn, is due to the Fourier sparsity of the input distribution and target functions.

Furthermore, for the case of Gaussian distributions, we strengthen the classical hardness to hold against any classical algorithm which has access to correlational SQs~\cite{bshouty2002using,bendavid1995learning}.
This model is more general than gradient methods but is still a restriction of general SQs.
We view our proof of classical hardness against such algorithms as an important step towards general SQ hardness.

Additionally, it is interesting to observe that only one type of query made by our quantum algorithm is not a correlational QSQ, i.e., observables of the form $O \otimes I$, where the identity acts on the output register.
Thus, one might argue that considering only correlational SQs for classical hardness is not a significantly unfair comparison.
We prove the following theorem.

\begin{theorem}[Correlational SQ Hardness; Informal]
    \label{thm:csq-hardness-main}
    Consider a Gaussian distribution with a sufficiently large variance.
    Then, any classical algorithm using correlational SQs to $\mathcal{C}$ with respect to this distribution requires at least $2^{\Omega(d)}$ queries to learn $\mathcal{C}$ to error $\epsilon$.
\end{theorem}

The full theorem is stated in \Cref{thm:csq-hardness} in Appendix~\ref{sec:csq-hardness}.
Importantly, we highlight that the condition on the variance of the Gaussian distribution is satisfied by our quantum algorithm presented in the next section.
Thus, classical hardness holds in the same regime as our efficient quantum algorithm.
We also remark that, previously, classical learning theorists have shown similar correlational SQ lower bounds for learning single-layer neural networks~\cite{goel2020superpolynomial,diakonikolas2020algorithms}.
However, these works consider different activation functions, so their results are not immediately applicable.
We prove \Cref{thm:csq-hardness-main} by lower bounding the statistical dimension~\cite{blum1994weakly,yang2005new,feldman2017statistical} of $\mathcal{C}$, which captures the difficulty of learning a concept class, similarly to the more commonly known VC dimension.
The proof is provided in Appendix~\ref{sec:csq-hardness}.

\subsection{Quantum algorithm}
In contrast, the complexity of our quantum algorithm scales only polynomially in $d$ and polylogarithmically in $R_w$, since quantum algorithms can overcome and in fact leverage this Fourier sparsity via the quantum Fourier transform.
We state our guarantee first for the uniform distribution.
We highlight that while classical hardness for gradient methods holds on average over a uniform choice of $g_{w^\star}$ from the concept class, the guarantee for our quantum algorithm applies in the stronger worst-case setting, i.e., it holds for any fixed $g_{w^\star}$. Our correlational SQ hardness result holds in the worst-case setting as well.

\begin{theorem}[Uniform distribution; Informal Version of \Cref{thm:unif} in Appendix~\ref{sec:uniform}]
    \label{thm:unif-main}
    Let $\epsilon > 0$, and let $\varphi^2$ be the uniform distribution.
    Let $g_{w^\star}: \mathbb{R}^d \to [-1,1] \in \mathcal{C}$ be a target function for an unknown vector $w^\star \in \mathbb{R}^d$ with norm $R_w$.
    Then, there exists a quantum algorithm with QSQ access to a suitably discretized quantum example state $\ket{g_{w^\star}}$ that can efficiently find parameters $\hat{\beta} \in \mathbb{R}^D$ such that $\mathcal{L}_{w^\star}(\hat{\beta}) \leq \epsilon$ with high probability using
    \begin{equation}
        N = \mathcal{O}\left(d D \,\mathrm{polylog}(d, D, R_w, 1/\epsilon)\right)
    \end{equation}
    QSQs and $t = \Theta(\log(D/\epsilon))$ iterations of gradient descent.
\end{theorem}

The detailed theorem statement is given in \Cref{thm:unif} in Appendix~\ref{sec:uniform}.
There, we also specify the QSQ noise tolerance needed explicitly. We highlight that it is only required to scale inverse polynomially in all parameters.
We are also able to achieve a similar complexity for learning with respect to ``sufficiently flat'' non-uniform distributions.
The precise technical conditions needed are stated in Appendix~\ref{sec:non-unif}.
In particular, we show in Appendix~\ref{sec:non-unif} that these conditions are satisfied for three practically-relevant classes of distributions: Gaussians, generalized Gaussians~\cite{subbotin1923law}, and logistic distributions.
The wide applicability of Gaussian distributions is clear, and generalized Gaussians and logistic distributions have applications in image processing~\cite{do2002wavelet,mallat1989theory,moulin1999analysis} and population growth~\cite{pearl1920rate,pearl1940logistic,schultz1930standard,plackett1959analysis}, respectively.

The flatness property we require is typically satisfied by taking a distribution's scale parameter large enough.
Nevertheless, the condition still permits distributions that deviate significantly from uniform.
For example, for Gaussian distributions, we need that the variance $\sigma$ is large enough such that $e^{-x^2/\sigma^2}$ is point-wise close to $1$, e.g., $|e^{-x^2/\sigma^2} - 1| \leq 1/10$ over our truncated space, and that the derivative of the density function is not too large.
These conditions are satisfied by $d$-dimensional Gaussians with covariance $\Sigma=4\pi R^{2}I$, where $R$ is the size of the truncated space.
This still leads to a density that decays exponentially with $d$ away from the mean, in contrast to uniform distributions.
There are also some conditions regarding the shape of the distribution, which are detailed in Appendix~\ref{sec:non-unif}.

With this, we obtain the following theorem.
The detailed statement is given in \Cref{thm:non-unif-guarantee} in Appendix~\ref{sec:non-unif}.

\begin{theorem}[Non-uniform distributions; Informal Version of \Cref{thm:non-unif-guarantee} in Appendix~\ref{sec:non-unif}]
    \label{thm:non-unif-main}
    Let $\epsilon > 0$, and let $\varphi^2$ be a sufficiently flat distribution.
    Let $g_{w^\star}: \mathbb{R}^d \to [-1,1] \in \mathcal{C}$ be a target function for an unknown vector $w^\star \in \mathbb{R}^d$ with norm $R_w$.
    Then, there exists a quantum algorithm with QSQ access to a suitably discretized quantum example state $\ket{g_{w^\star}}$ that can efficiently find parameters $\hat{\beta} \in \mathbb{R}^D$ such that $\mathcal{L}_{w^\star}(\hat{\beta}) \leq \epsilon$ with high probability using
    \begin{equation}
        N = \mathcal{O}\left(d D \,\mathrm{polylog}(d, D, R_w, 1/\epsilon)\right)
    \end{equation}
    QSQs and $t = \Theta(\log(D/\epsilon))$ iterations of gradient descent.
\end{theorem}

As a special case, we obtain the same guarantee for the natural distributions of Gaussians, generalized Gaussians, and logistic distributions.
We prove that these distributions are also Fourier-concentrated and hence give us significant quantum advantages, specified in the following corollary.

\begin{corollary}[Informal]
    The guarantee of \Cref{thm:non-unif-main} holds taking $\varphi^2$ as Gaussian, generalized Gaussian, or logistic distributions with large enough scale parameters.
    Meanwhile, any classical gradient-based algorithm requires
    \begin{itemize}
        \item $\exp(\Omega(\min(d, R_w^2)))$ samples for Gaussian distributions.
        \item $\Omega(\min(\exp(d), \mathrm{superpoly}(R_w)))$ samples for generalized Gaussian distributions.
        \item $\exp(\Omega(dR_w))$ samples for logistic distributions.
    \end{itemize}
\end{corollary}

This is a direct implication of the previous theorem combined with Propositions~\ref{prop:gen-gauss-satisfies-assum} and~\ref{prop:logistic-satisfies-assum} and Corollary~\ref{coro:gauss-satisfies-assum} in Appendix~\ref{sec:non-unif}.
Thus, we see that for Gaussian, generalized Gaussian, and logistic distributions, we obtain an exponential quantum advantage over classical gradient methods.
For Gaussian distributions, we retain this exponential advantage over correlational SQ algorithms.

Our key observation is that the classical hardness of~\cite{shamir2018distribution} stems from the objective function $\mathcal{L}_{w^\star}$ being sparse in Fourier space.
This implies that the objective function is difficult to optimize using gradient-based methods.
On the other hand, quantum algorithms can typically take advantage of Fourier-sparsity by leveraging the quantum Fourier transform (QFT).
In fact, we notice that the target functions $g_{w^\star}$ are periodic in each coordinate with period $1/w_j^\star$:
\begin{equation}
    g_{w^\star}\left(x + \frac{e_j}{w_j^\star}\right) = \tilde{g}\left(\left(x + \frac{e_j}{w_j^\star}\right)^\intercal w^\star\right) = \tilde{g}(x^\intercal w^\star + 1) = \tilde{g}(x^\intercal w^\star) = g_{w^\star}(x),
\end{equation}
where we use that $\tilde{g}$ has period $1$ and use $e_j$ to denote the unit vector for coordinate $j \in [d]$.
Thus, information about the unknown vector $w^\star$ is contained in the period of $g_{w^\star}$.
This observation yields a simple quantum algorithm: (1) Perform period finding by encoding the QFT into QSQs to learn the vector $w^\star$ one component at a time, (2) Learn the unknown parameters $\beta_j^\star$ defining the periodic activation function (\Cref{eq:g-tilde-main}) using classical gradient methods.
Note that once we have an approximation of $w^\star$ from Step (1), Step (2) is effectively a regression problem, allowing it to be solved via gradient methods.

Despite the initial simplicity of this algorithm, there are several nontrivial issues that arise, particularly in Step (1).
First, recall that the quantum example state (\Cref{eq:example-state-main}) must be suitably discretized because our target function is real.
However, there exist pathological examples in which discretization eliminates any information about the period of the original function (see, e.g., Section 10 of \cite{jozsa2003notes}).
Thus, it is important to choose the correct discretization such that the period is sufficiently preserved.
Another problem is that the standard period finding algorithm does not apply because the period $1/w_j^\star$ is not necessarily an integer.
Additionally, standard period finding is only analyzed for uniform superpositions, whereas we are primarily interested in non-uniform superpositions.

To resolve these problems, we carefully discretize the target function such that it satisfies \emph{pseudoperiodicity}~\cite{hallgren2007polynomial} with a period proportional to $1/w_j^\star$ in each coordinate.
For a period $S$, instead of requiring that $h(k) = h(k + \ell S)$ for an integer $\ell$, pseudoperiodicity dictates that $h(k) = h(k + [\ell S])$, where $[\ell S]$ denotes rounding $\ell S$ either up or down to the nearest integer.
This ensures that the period of the discretized function still contains useful information, thus excluding pathological discretizations.
Then, for uniform distributions, we can use Hallgren's algorithm~\cite{hallgren2007polynomial}, which finds the (potentially irrational) period of pseudoperiodic functions.
It is still nontrivial to apply Hallgren's algorithm, as it crucially assumes the existence of an efficient verification subroutine to check if a given guess is close to the period of a pseudoperiodic function.
Unlike for periodic functions, such verification is not straightforward for pseudoperiodic functions.
We design a suitable verification procedure which uses $D$ QSQs in \Cref{thm:verification,thm:verification-non-unif} in Appendices~\ref{sec:general-uniform} and~\ref{sec:general-non-unif}, respectively.

Moreover, Hallgren's algorithm does not apply for non-uniform distributions.
To this end, we design a new period finding algorithm that works for sufficiently flat non-uniform distributions, which could be of independent interest.
The sufficiently flat condition on the distributions stems from our generalization of Hallgren's algorithm as well as several integral bounds needed for Step (2) of the algorithm.
We expand on these ideas in the Methods and Appendices~\ref{sec:uniform} and~\ref{sec:non-unif}.

\section{Discussion}

Numerous works have shown exponential quantum advantages for learning Boolean functions when the input data is uniformly distributed.
However, little is known about distributions other than uniform, and settings in classical ML commonly consider real-valued functions, leaving a large gap between known quantum advantages and classical ML in practice.
Our work makes significant progress towards understanding quantum advantage for learning real functions over non-uniform distributions.
Moreover, the function class of periodic neurons that we consider is well-studied in the deep learning theory literature.

One question that has persisted around many quantum learning results, including the present work, is the practical origin of the quantum example state $\ket{c^\star}$ for a target function $c^\star$.
Creating an example state is straightforward when efficient classical descriptions of $c^\star(x)$ and the distribution $\mathcal{D}(x)$ are known, and some conditions on the distribution $\mathcal{D}$ are satisfied~\cite{rattew2022preparingarbitrarycontinuousfunctions,rattew2023nonlineartransformationsquantumamplitudes,rosenkranz2024quantumstatepreparationmultivariate,grover2002creating,izdebski2020improved}.
However, by definition of the problem, $c^\star$ is unknown and is precisely what we wish to learn.
Instead, one may consider coherently loading the data from known classical examples, but this can be costly and eliminate an end-to-end quantum advantage.
For instance, the spacetime volume of the loading circuit is likely to scale exponentially in $d$~\cite{jaques2023qramsurveycritique,aaronson2014quantum}, erasing any practical advantage.

Alternatively, we consider the following perspective on how learning may still be valuable even when a description of $c^\star$ is known.
Suppose we know some complex classical circuit/function that simulates classical physics.
One may instead hope to learn a simpler circuit that can approximately compute the same dynamics more efficiently (e.g.,~\cite{kasim2021building}).
Here, the simpler circuit is unknown to the learner, but the algorithm has access to it through the known, complicated circuit that simulates the same dynamics.
In this case, because $c^\star$ is known, one can construct the quantum example state straightforwardly (albeit with some overhead), which may be helpful in learning a simpler description of the target.

As a practically-relevant example in ML, one can consider the complicated object as a trained neural network.
Such models are highly complex, and while there are heuristic methods for constructing them, there is limited understanding of how neural networks compute their outputs.
As in the subfield of interpretability in ML~\cite{linardatos2020explainable, Bricken2023-lz, Templeton2024-wt}, one may hope to use quantum access to the (known) trained neural network to extract information about this complex model.
By representing it differently or learning a simpler model that performs approximately the same function, this could help us better understand opaque large ML models and, in turn, design better ones using our new knowledge of their inner workings.
One may object that it is not clear if the type of structure that quantum algorithms typically leverage to obtain advantages are present in this setting.
However, we note that there is some evidence of periodic structure in the features of large language models~\cite{Engels2024-ow}.
Regardless, we hope that this perspective provides new insight into scenarios when quantum example states may occur naturally and be efficiently preparable.

Our work also raises many interesting open questions.
First, the classical hardness for our results only holds against classical gradient methods and correlational SQ algorithms (for Gaussian distributions).
While there are results proving hardness for classical SQ algorithms or even general classical algorithms~\cite{song2017complexity,song2021cryptographic}, these results do not directly apply to our parameter regimes.
Can the classical hardness be strengthened for our setting?
We expect the hardness to still hold and leave this generalization to future work.

Second, we assume that the periodic neuron takes a specific form given by $\tilde{g}$.
Could our results be generalized to apply for any periodic function?
In addition, while our results hold for a broad class of non-uniform distributions including Gaussians, generalized Gaussians, and logistic distributions, one may wonder if similar results can be obtained for other natural non-uniform distributions.
We conjecture that our results could be modified to apply to generalized logistic distributions~\cite{balakrishnan1988order} or stable distributions~\cite{levy1925calcul}.
It is also possible that the conditions needed for the non-uniform distributions we consider, i.e., the sufficiently flat condition, could be relaxed, although this would require a significantly different analysis.
More generally, can one obtain a quantum advantage for this task when learning over any Fourier-concentrated distribution?
The main part of the proof that requires modification is the analysis of the non-uniform period finding algorithm.

Finally, while we consider quantum access to real functions via discretized quantum example states, one may consider alternative models for learning classical functions encoded in quantum states, e.g., continuous variable states. Would different models provide new capabilities for quantum learning algorithms?

\section{Methods}

\subsection{Classical hardness}

In this section, we give an overview of the proofs of \Cref{thm:class-hardness-main,thm:csq-hardness-main}.
First, to prove \Cref{thm:class-hardness-main}, we adapt the proof of Theorem 4 from~\cite{shamir2018distribution} to hold for our concept class $\mathcal{C}$, which imposes additional constraints on the vector $w^\star$.
Namely, for our setting, $w^\star$ is restricted to the positive orthant and bounded away from $0$; meanwhile, in~\cite{shamir2018distribution}, $w^\star$ can be any $d$-dimensional vector with norm $R_w$.
We must argue that these restrictions do not make the problem easier for classical algorithms.

The crux of the argument of~\cite{shamir2018distribution} shows that for any function $h$, for a random choice of $w^\star$, the Fourier transform of the target function $g_{w^\star}$ does not correlate well with $h$ (see \Cref{lem:5} in Appendix~\ref{sec:classical-hardness}).
Thus, no matter what our hypothesis function is, obtaining information about $w^\star$ is difficult.
However, their proof crucially uses that $w^\star$ is chosen randomly from an exponentially large set of nearly orthogonal vectors, and their construction of this set does not adhere to our requirements for $w^\star$.
Our main technical contribution for this proof is showing the existence of a large set $\tilde{\mathcal{S}}_w$ of nearly orthogonal vectors which lie in the positive orthant and are bounded away from zero.
We do so by leveraging tools from high-dimensional geometry~\cite{vershynin2018high}.

In \Cref{thm:csq-hardness-main}, we strengthen the classical hardness to hold against any algorithms using correlational SQs when learning with respect to Gaussian distributions.
A standard method for proving correlational SQ lower bounds is via the \emph{statistical dimension}~\cite{blum1994weakly,yang2005new,feldman2017statistical}, which captures the difficulty of learning a concept class, similarly to the more commonly known VC dimension.
Informally, the statistical dimension quantifies the size of the largest subset of the concept class whose elements have ``low correlation'' (see \Cref{def:sd} in Appendix~\ref{sec:csq-hardness} for a formal definition).
Intuitively, functions in this low correlation subset should be hard to distinguish and thus hard to learn.
Several works~\cite{feldman2017statistical,szorenyi2009characterizing,feldman2012complete} have formalized this relationship, proving, roughly, that lower bounds on the statistical dimension imply correlational SQ lower bounds (see \Cref{thm:sda-implies-lower} in Appendix~\ref{sec:csq-hardness}).
Our key contribution is proving an exponential lower bound on the statistical dimension of our concept class $\mathcal{C}$, where we utilize properties of Gaussian integrals and the construction of the set $\tilde{\mathcal{S}}_w$ from the proof of \Cref{thm:class-hardness-main} described above.
Thus, this results in an exponential lower bound for correlational SQ algorithms.

\subsection{Quantum algorithm}

In this section, we describe the ideas behind the proofs of \Cref{thm:unif-main,thm:non-unif-main}.
As discussed above, our key observation is that the target function $g_{w^\star}$ is periodic in each coordinate with period $1/w_j^\star$.
This informs our quantum algorithm, which is as follows: (1) Perform period finding by encoding the QFT into QSQs to learn the vector $w^\star$ one component at a time; (2) Learn the unknown parameters $\beta_j^\star$ defining the periodic activation function (\Cref{eq:g-tilde-main}) using classical gradient methods.
Thus, the proofs are separated into two main parts, each analyzing the sample complexities for these two algorithmic steps.
The proofs for the uniform distribution are in Appendix~\ref{sec:uniform}, and those for the non-uniform distributions are in Appendix~\ref{sec:non-unif}.
In particular, Step (1) is analyzed in detail in Appendices~\ref{sec:linear-uniform} and~\ref{sec:linear-non-unif}, and Step (2) is examined in Appendices~\ref{sec:outer-uniform} and~\ref{sec:outer-non-unif}.
In the following, we give an overview of the proofs.

\subsubsection{Learning the linear function}

We want to apply period finding to our target function $g_{w^\star}: \mathbb{R}^d \to [-1,1]$ to approximate the unknown vector $w^\star$ that defines the inner linear function.
First, because $g_{w^\star}$ has real inputs and outputs, we need to discretize it so that it can be represented by a (discrete) quantum example state (\Cref{eq:example-state-main}).
We require that the chosen discretization is pseudoperiodic, a condition which is weaker than periodicity but still ensures that the discretized function retains information about the period of $g_{w^\star}$.
Specifically, for $d = 1$, a function $h: \mathbb{Z} \to \mathbb{R}$ is pseudoperiodic with period $S \in \mathbb{R}$ if $h(k) = h(k + [\ell S])$ for any integer $\ell$, where $[\ell S]$ denotes rounding $\ell S$ either up or down to the nearest integer.
One should compare this to periodicity, where the necessary condition is instead $h(k) = h(k + \ell S)$.
We choose the following discretization, considering $d = 1$ for simplicity.
We generalize to arbitrary $d \geq 1$ in the discussion surrounding \Cref{eq:general-d}.

\begin{lemma}[Discretization; Informal]
    Let $M_1,M_2$ be suitably chosen discretization parameters with $M_1 > M_2$.
    Consider the discretized function $h_{M_1,M_2}: \mathbb{Z} \to \frac{1}{M_2}\mathbb{Z}$ defined by
    \begin{equation}
        h_{M_1,M_2}(k) \triangleq \left\lfloor g_{w^\star}\left(\frac{k}{M_1}\right)\right\rfloor_{M_2},
    \end{equation}
    where $\lfloor \cdot \rfloor_{M_2}$ denotes rounding down to the nearest multiple of $1/M_2$.
    Then, $h_{M_1,M_2}$ is pseudoperiodic with period $M_1/w^\star$ for a large proportion of the inputs.
\end{lemma}

Requiring $M_1 > M_2$ at an appropriate ratio makes the discretization more coarse on the outputs than the inputs, ensuring that pseudoperiodicity is satisfied.
In our proof, we choose $M_1, M_2$ to scale polynomially in the problem parameters, i.e., $\mathrm{poly}(\epsilon, D, d, R_w)$, where $\epsilon$ is the desired error, $D$ is the number of cosine terms in the periodic activation function (\Cref{eq:g-tilde-main}), $d$ is the input dimension, and $R_w$ the norm of $w^\star$.

Recall that $g_{w^\star}$ has period $1/w_j^\star$ in the $j$th coordinate.
Thus, learning the period of $h_{M_1, M_2}$ also allows us to approximate the period of $g_{w^\star}$.
However, straightforwardly applying standard period finding algorithms to $h_{M_1, M_2}$ fails because $h_{M_1, M_2}$ is only pseudoperiodic rather than periodic and its period $M_1/w_j^\star$ is not necessarily an integer.
Instead, we turn to Hallgren's algorithm~\cite{hallgren2007polynomial}, which determines the period of pseudoperiodic functions and applies to real periods.
Note that Hallgren's algorithm only applies for the uniform distribution, so we consider this case for now.
At a high level, Hallgren's algorithm first quantum Fourier samples twice and computes the continued fraction expansion of the quotient of the results.
Then, it constructs a guess for the period for each convergent of the expansion and iterates through each guess, checking which one approximates the period.
Ref.~\cite{hallgren2007polynomial} shows that one guess is guaranteed to be close to the period.
We discuss Hallgren's algorithm in more detail in Appendix~\ref{sec:hallgren}.

Notice that a crucial subroutine necessary for Hallgren's algorithm is a verification procedure to check if a given guess is close to the period of a pseudoperiodic function.
Unlike for periodic functions, where such verification is straightforward, this is nontrivial for pseudoperiodic functions.
In fact, Ref.~\cite{hallgren2007polynomial} leaves this as an assumption to be instantiated upon applying the guarantee of Hallgren's algorithm.

We design a suitable verification procedure which uses $D$ QSQs (see \Cref{thm:verification,thm:verification-non-unif} in Appendices~\ref{sec:general-uniform} and~\ref{sec:general-non-unif}, respectively).
The main idea is to compute the inner product between $h_{M_1,M_2}$ and $h_{M_1, M_2}(\cdot + T)$, where $T$ is a guess for the period.
Intuitively, this inner product should be large for a guess that approximates the period well.
We define an observable that allows us to compute this inner product using QSQs.
Then, we identify a suitable threshold which the inner product surpasses if and only if the guess is indeed close to the period.
The majority of the technical work for the verification procedure lies in finding such a threshold.
With this, the only remaining quantum part of Hallgren's algorithm is quantum Fourier sampling, which can be accomplished using QSQs by encoding the QFT into the queried observable.
Because we need to repeat this algorithm for each entry in the vector $w^\star \in \mathbb{R}^d$, we use $\tilde{\mathcal{O}}(dD)$ QSQs, where the polylogarithmic factors come from amplifying the success probability of Hallgren's algorithm.

Thus far, we discussed how to utilize Hallgren's algorithm for our problem, which only applies for uniform distributions.
We generalize these ideas to perform period finding for non-uniform input distributions.
This algorithm can be found explicitly in the appendices in \Cref{alg:hallgren-non-unif} in Appendix~\ref{sec:general-non-unif}, and the verification procedure is presented in \Cref{alg:verification-non-unif} in Appendix~\ref{sec:general-non-unif}.
Our algorithm follows the same structure as Hallgren's algorithm but requires a new analysis due to the different input distribution.
Here, we crucially use that the non-uniform distributions we consider are sufficiently flat, e.g., they are pointwise-close to uniform.
As discussed previously, our flatness condition only requires the univariate (unnormalized) marginals to be close to uniform, but the overall density can decay exponentially in $d$.

\subsubsection{Learning the periodic activation function}

In the previous section, we showed how to obtain an approximation $\hat{w}$ of the unknown vector $w^\star$ using quantum period finding.
Using this approximation, we can learn the unknown parameters $\beta_j^\star$, which determines the periodic activation function $\tilde{g}$ given in \Cref{eq:g-tilde-main}.
This step of the algorithm is purely classical.

With the approximation $\hat{w}$, we can consider predictors $f_\beta$ defined by
\begin{equation}
    f_\beta(x) \triangleq \sum_{j=1}^D \beta_j \cos(2\pi j x^\intercal \hat{w}),
\end{equation}
where $\beta \in \mathbb{R}^d$ is a vector of trainable parameters.
These predictors have the same form as the target function $g_{w^\star}$ but replace $w^\star$ and $\beta_j^\star$ with $\hat{w}$ and $\beta_j$, respectively.
Thus, the loss function from~\Cref{eq:loss-main} can be written more explicitly as
\begin{equation}
  \mathcal{L}_{w^\star}(\beta) = \int\limits_{x\sim \varphi^2}\left(\sum_{j=1}^D \beta_j^\star \cos(2\pi j x^\intercal w^\star) - \sum_{j=1}^D \beta_j \cos(2\pi j x^\intercal \hat{w})\right)^2\,dx.
\end{equation}
We use (approximate) gradient access to this loss function to learn parameters $\hat{\beta}$ such that $\mathcal{L}_{w^\star}(\hat{\beta}) \leq \epsilon$.
We acknowledge there may be other approaches to solve for the parameters, but we believe gradient descent is the most straightforward.
First, we show that the gradients are informative, i.e., the derivative of the objective function $\partial \mathcal{L}_{w^\star}/\partial \beta_k$ indeed reflects how far $\beta_k$ is from the true parameter $\beta_k^\star$.
With this, we can simply apply gradient descent (see, e.g.,~\cite{nesterov2018lectures}), where we show that the iterates converge to the true parameters within $t = \Theta(\log(D/\epsilon))$ steps.
Most of the work in this step goes into carefully choosing the hyperparameters (e.g., the number of iterations to run gradient descent, how accurate the approximation $\hat{w}$ is required to be, etc.) to guarantee that the value of the loss function is small.
In this step, the proofs for uniform and non-uniform distributions are very similar.
The full proofs are provided in Appendix~\ref{sec:outer-uniform} for the uniform case and Appendix~\ref{sec:outer-non-unif} for the non-uniform case.

\section*{Data Availability}

No data are generated or analyzed in this theoretical work.

\section*{Acknowledgments}

The authors thank Andrew Childs, András Gilyén, Hsin-Yuan (Robert) Huang, Robbie King, Robin Kothari, and Chirag Wadhwa for helpful discussions. L.L. was supported by a Marshall Scholarship. This work was done (in part) while a subset of the authors were visiting the Simons Institute for the Theory of Computing.

\section*{Author Contributions}

D.G. and J.R.M. conceived the project.
L.L. and D.G. developed the mathematical aspects of the work.
All authors contributed to the writing of the manuscript.

\section*{Competing Interests}

The authors declare no competing interests.

\newpage
\bibliographystyle{unsrt}
\bibliography{refs, paperpile}

\newpage
\resumetoc
\appendix
\appendixpage

\tableofcontents

\section{Preliminaries}
\label{sec:prelim}

\subsection{Quantum learning theory}
\label{sec:learning}

In classical learning theory, the goal is to learn a collection of functions $\mathcal{C} \subseteq \{c : \mathcal{X} \to \mathcal{Y}\}$ with input space $\mathcal{X}$ and output space $\mathcal{Y}$.
Typically, for Boolean functions, $\mathcal{X} = \{0,1\}^d, \mathcal{Y} = \{0,1\}$, where $d$ is the input dimension, but in general, one could have any $\mathcal{X} \subseteq \mathbb{R}^d, \mathcal{Y} \subseteq \mathbb{R}$.
This collection $\mathcal{C}$ is called a \emph{concept class}.
Two common models used in classical learning theory are the \emph{probably approximately correct (PAC) model}~\cite{valiant1984theory} and the \emph{statistical query (SQ) model}~\cite{kearns1998efficient}.
In classical PAC learning, a learning algorithm is given labeled random examples $\{(x_i, c^\star(x_i))\}_{i=1}^N$, where the $x_i$ are sampled i.i.d. according to an unknown distribution $\mathcal{D}$ over the input space $\mathcal{X}$.
The goal is to learn the unknown target function $c^\star$ up to some error with high probability.
More precisely, an \emph{$(\epsilon,\delta)$-PAC learner} for $c^\star$ outputs a hypothesis function $h:\mathcal{X} \to \mathcal{Y}$ such that
\begin{equation}
  \label{eq:pac-err}
  \mathcal{L}(h) \leq \epsilon
\end{equation}
with probability at least $1-\delta$ for some loss function $\mathcal{L}$.
Typically, the loss function is chosen as the squared loss $\mathbb{E}_{x \sim \mathcal{D}}(h(x)-  c^\star(x))^2$ or the misclassification error $\Pr_{x \sim \mathcal{D}}(h(x) \neq c^\star(x))$.
One often wants to minimize the amount of training data $N$, or the \emph{sample complexity}, needed to learn any unknown target function $c^\star$ from the concept class $\mathcal{C}$ for any unknown distribution $\mathcal{D}$.
Meanwhile, in classical SQ learning, rather than having direct access to the examples, a learning algorithm only has access to noisy expectation values of functions of the data.
In particular, an SQ learner has access to a statistical query oracle, which takes as input a tolerance parameter $\tau \geq 0$ and a function $\phi: \mathcal{X} \times \mathcal{Y} \to \mathcal{Y}$ and outputs a number $\alpha$ such that
\begin{equation}
  \left|\alpha - \mathop{\mathbb{E}}_{x \sim \mathcal{D}}[\phi(x, c^\star(x))]\right| \leq \tau.
\end{equation}
Then, an \emph{$(\epsilon, \delta)$-SQ learner} outputs a hypothesis function satisfying \Cref{eq:pac-err} with probability $1-\delta$.
In the some definitions of statistical query learning, the parameter $\delta$ is not present.
Here, we include it to allow for a probability of failure in randomized learning algorithms, as noted in~\cite{kearns1998efficient}.
In this case, the measure of complexity is the number of queries, or the \emph{query complexity}, needed to learn any unknown target function $c^\star$ from the concept class $\mathcal{C}$ for any unknown distribution $\mathcal{D}$.

Both PAC and SQ learning have been extended to the quantum setting in the quantum PAC model~\cite{bshouty1995learning} and quantum statistical query (QSQ) model~\cite{arunachalam2020quantum}, respectively.
Here, the only difference is the access model, in which quantum learning algorithms are given access to quantum data instead.
Specifically, in quantum PAC learning~\cite{bshouty1995learning}, a quantum learner is given copies of the quantum example state
\begin{equation}
  \label{eq:example-state}
  \ket{c^\star} \triangleq \sum_{x \in \mathcal{X}}\sqrt{\mathcal{D}(x)}\ket{x,c^\star(x)}.
\end{equation}
The learning algorithm is allowed to perform (potentially entangled) measurements on the example states, and in this case, one wants to minimize the number of copies of the example states used to learn the concept class.
Finally, in the QSQ model~\cite{arunachalam2020quantum}, a learner has access to a QSQ oracle, which takes as input a tolerance parameter $\tau \geq 0$ and an observable $O$ such that $\norm{O}\leq 1$ and outputs a number $\alpha$ such that
\begin{equation}
  \left|\alpha - \expval{O}{c^\star}\right| \leq \tau.
\end{equation}
The goal is again to minimize the number of queries to the QSQ oracle needed to learn the concept class $\mathcal{C}$.
A key difference between the quantum PAC setting and the QSQ setting is that in the PAC setting, the learner may perform entangled measurements across multiple copies of the quantum example state~\cite{bubeck2020entanglement,arunachalam2024role}.

In this work, we focus on the QSQ access model with noise tolerance $\tau \geq 0$ for learning a particular concept class (defined in Appendix~\ref{sec:detail-prob}) in the distribution-specific setting, where $\mathcal{D}$ is known to be either uniform or a discrete Gaussian with a diagonal covariance matrix.
Moreover, we consider functions with real inputs and outputs, so we redefine QSQ access for real functions.

\begin{definition}[Quantum statistical query access for real functions]
    \label{def:qsq}
    Let $\mathcal{C} \subseteq \{c: \mathbb{R}^d \to \mathbb{R}\}$ be a concept class, where $d \geq 1$ is the input dimension.
    Let $\mathcal{D}$ be a probability distribution over $\mathbb{R}^d$.
    A quantum statistical query oracle for some $c^\star \in \mathcal{C}$ receives as input a tolerance parameter $\tau \geq 0$, discretization/truncation parameters $M, R\geq 1$, respectively, and an observable $O$ such that $\norm{O} \leq 1$, and outputs a number $\alpha$ such that
    \begin{equation}
        |\alpha - \expval{O}{h^*_{M}}| \leq \tau,
    \end{equation}
    where $\ket{h_{M}^*}$ is the quantum example state
    \begin{equation}
        \ket{h_{M}^*} = \sum_{x_1,\dots, x_d = -R}^{R-1} \sqrt{\mathcal{D}(x)} \ket{x}\ket{h_{M}^*(x)}
    \end{equation}
    and $h_M^*$ is a suitable discretization of the target $c^\star$ and $\mathcal{D}$ must be suitably renormalized.
\end{definition}
Without loss of generality, beyond $\tau > 0$, we consider the QSQ model in which the output $\alpha$ is a rational number.
We can do this because the rational numbers are dense in $\mathbb{R}$.
Then, if a QSQ outputs an irrational number, we can find a rational number close to it and consider the error in this approximation as a part of the tolerance of the QSQ.

One may also consider multiple discretization parameters if necessary.
We remark that allowing one to specify the discretization/truncation parameters rather than fixing them throughout should not be too powerful.
Notably, classical SQ access can approximate expectation values of a real target function itself, without needing the intermediary step of discretization at all.

\subsection{Hallgren's irrational period finding algorithm}
\label{sec:hallgren}

In this section, we give an overview of Hallgren's irrational period finding algorithm~\cite{hallgren2007polynomial}.
For more detailed presentations, we refer the reader to~\cite{jozsa2003notes,childs2010quantum}.
This algorithm was originally a subroutine for a quantum algorithm for solving Pell's equation in number theory.
However, we will only focus on this subroutine, which is sufficient for our purposes.

One of the most well-known quantum algorithms is Shor's period finding algorithm~\cite{shor1994algorithms}.
Given access to a function $f: \mathbb{Z}_N \to \mathbb{Z}_M$ which is periodic with period $S \in \mathbb{N}$, this algorithm can identify $S$ up to some precision.
However, the algorithm crucially relies on the fact that the period is an integer.
Namely, recall that Shor's algorithm utilizes the continued fractions algorithm to recover the period from the quantum measurement outcomes.
Without the assumption that $S \in \mathbb{N}$, directly using continued fractions is not guaranteed to recover an approximation of $S$.
Thus, if one hopes to generalize Shor's algorithm to real functions with real periods, one must do something more complicated.
This is exactly what Hallgren's algorithm does.

Consider a function $f: \mathbb{R} \to X$ which is periodic with period $S \in \mathbb{R}$.
Here, $X$ is some output space, which may be continuous-valued.
In order to access $f$ on a quantum computer, we must suitably discretize it.
However, this must be done with some care, as ``bad'' discretizations can cause us to lose all information about the period in the new discretized function.
The notion of pseudoperiodicity defined below excludes this possibility.

\begin{definition}[Pseudoperiodic~\cite{hallgren2007polynomial}]
\label{def:pseudoperiod}
A function $f: \mathbb{Z} \to X$ for some output space $X$ is \emph{pseudoperiodic} with period $S \in \mathbb{R}$ if for each $0 \leq k \leq \lfloor S \rfloor$ and each $\ell \in \mathbb{Z}$, either $f(k + \lfloor \ell S \rfloor)$ or $f(k + \lceil \ell S \rceil)$ equals $f(k)$.
$f$ is \emph{$\eta$-pseudoperiodic} with period $S$ if this condition holds for at least an $\eta$-fraction of inputs $0 \leq k \leq \lfloor S \rfloor$.
\end{definition}

This ensures that the discretization still encodes sufficient information about the period of the original function.
Thus, from here, we consider a pseudoperiodic discretization of the real function we want to learn the period of.
Hallgren's algorithm provides a guarantee for recovering the period of a pseudoperiodic function, which we restate below.
We also present the algorithm in \Cref{alg:hallgren}.

\begin{theorem}[Lemma 3.1 in~\cite{hallgren2007polynomial}]
\label{thm:hallgren}
Let $f$ be an $\eta$-pseudoperiodic function with period $S \in \mathbb{R}$.
Suppose that, given an integer $T$, we can efficiently check (in time $\mathrm{polylog}(S)$) whether or not $|\ell S - T| < 1$ for some $\ell \in \mathbb{Z}$.
Additionally, suppose that we have an upper bound $A$ on $S$.
Then, there exists a quantum algorithm that outputs an integer $a$ such that $|S - a| \leq 1$ with probability $\Omega(\eta^2/(\log A)^4)$.
Moreover, the algorithm runs in time $\mathrm{polylog}(A)$.
\end{theorem}

\begin{algorithm}
   \caption{Hallgren's Algorithm} 
   \label{alg:hallgren}
   \begin{algorithmic}[1]
   \State Choose an integer $q \geq 3S^2$ (this can be satisfied by choosing $q \geq 3A^2$).
   \State Apply quantum Fourier sampling to the function $f$ over $\mathbb{Z}_q$ twice. Let $b, c \in \mathbb{Z}$ be the outputs.
   \State Compute the continued fraction expansion of $b/c$.
   \State For each convergent $b_i/c_i$ in the continued fraction expansion, use the verification procedure to check whether $\lfloor b_i q/b \rfloor$ or $\lceil b_i q/b \rceil$ is an integer multiple of the period $S$.
   \State \Return the smallest value that passed the test from the previous step.
   \end{algorithmic}
\end{algorithm}

We note that there are two key subroutines in Hallgren's algorithm: quantum Fourier sampling (as in the standard period finding algorithm) and the verification procedure to check if a given guess is indeed close to the period.
For a periodic function $f$, checking if a given guess is a multiple of the period is simple with query access to $f$.
However, for $\eta$-pseudoperiodic functions, this is nontrivial.
Hence, in order to apply \Cref{thm:hallgren}, one must ensure that this condition is satisfied.

We give a brief sketch the proof of \Cref{thm:hallgren}, as our proofs in Appendices~\ref{sec:general-uniform} and~\ref{sec:general-non-unif} rely on similar ideas.

\begin{proof}[Proof Sketch of \Cref{thm:hallgren}]
We consider $f$ to be pseudoperiodic on the whole domain for simplicity, as this only affects the success probability, which we will incorporate later.
Querying the pseudoperiodic function $f$ in superposition and measuring the last register, we get
\begin{equation}
    \frac{1}{\sqrt{p}}\sum_{k=0}^{p-1}\ket{x_0 + [kS]},
\end{equation}
where $[kS]$ denotes one of $\lfloor k S \rfloor$ or $\lceil k S \rceil$, $0 \leq x_0 \leq \lfloor S \rfloor$, and $p = \lfloor q/S\rfloor$.
By the shift invariance property of the Fourier transform, we can assume without loss of generality that $x_0 = 0$.
Then, applying the quantum Fourier transform mod $q$, we have
\begin{equation}
    \frac{1}{\sqrt{pq}}\sum_{k=0}^{p-1}\sum_{y=0}^{q-1} e^{2\pi i y [k S]/q}\ket{y}.
\end{equation}
Thus, the probability of measuring some $y$ is $(1/pq)\left|\sum_{k=0}^{p-1} e^{2\pi iy [kS]/q} \right|^2$.
Using this, \cite{hallgren2007polynomial} lower bounds the probability of measuring some $y = \lfloor a q/ S \rceil$ such that $y < q/\log A$, where $a$ is an integer and $\lfloor \cdot \rceil$ denotes rounding to the closest integer.
In particular, they show that one can lower bound this probability by $\Omega(1/S)$.
In total, the probability that quantum Fourier sampling produces two such values (as in Step 2 of \Cref{alg:hallgren}) that are also relatively prime is then $\Omega(\eta^2/\log^4(A))$.

Now, consider obtaining two values $b \triangleq \lfloor k q/S \rceil$ and $c \triangleq \lfloor \ell q / S\rceil$ from this quantum Fourier sampling.
\cite{hallgren2007polynomial} shows that $k/\ell$ is a convergent in the continued fraction expansion of $b/c$.
This is shown by proving that $|b/c - k/\ell| \leq 1/(2\ell^2)$, as this implies the desired result~\cite{schrijver1998theory}.
Finally, the proof concludes by showing that $\lfloor kq/S \rceil$ is close to an integer multiple of the period $S$.
This justifies Steps 3-5 of \Cref{alg:hallgren}, which iterates through all convergents in the continued fractions expansion of $b/c$ and checks which one is close to an integer multiple of the period.
The proof guarantees that at least one such convergent will indeed be close to the period.
\end{proof}

\section{Detailed problem statement}
\label{sec:detail-prob}

In this section, we define the concept class we wish to learn precisely.
We want to learn functions that are a composition of a periodic function and a linear function, as these are classically hard to learn via gradient methods~\cite{shamir2018distribution,shalev2017failures}.
Moreover, previous works have shown that this class is hard to learn classically even for SQ algorithms and efficient classical algorithms learning under small amounts of noise~\cite{song2017complexity,song2021cryptographic}.
We consider a slightly restricted setting, which we show is still hard for classical gradient methods in Appendix~\ref{sec:classical-hardness}.
\cite{song2017complexity,song2021cryptographic} do not directly apply to our parameter regimes, but nevertheless, these works constitute strong evidence that the problem is hard for broader classes of classical algorithms.

Let $d \geq 1$ denote the input dimension, and define the set of vectors with fixed norm $R_w > 0$ satisfying $w_j \geq R_w/d^2$:
\begin{equation}
  \label{eq:sw}
  \mathcal{S}_w \triangleq \left\{w \in R_w \mathbb{S}^{d-1} : w_j \geq \frac{R_w}{d^2},\;\forall j \in [d]\right\}.
\end{equation}
Here, $\mathbb{S}^{d-1}$ denotes the $(d-1)$-dimensional unit sphere, which lives in $\mathbb{R}^d$.
Let $\tilde{\mathcal{S}}_w$ be a $0.51$-packing net of the set $\mathcal{S}_w$, i.e., $\tilde{\mathcal{S}}_w \subseteq \mathcal{S}_w$ such that each point in $\tilde{\mathcal{S}}_w$ is separated by a geodesic angle of at least $0.51$.
Let $w^\star \in \tilde{\mathcal{S}}_w$ be a vector in $\tilde{\mathcal{S}}_w$.
We remark that~\cite{shamir2018distribution} considers $w^\star$ in $R_w \mathbb{S}^{d-1}$ directly, without requiring that $w_j \geq R_w/d^2$ or that $w^\star$ is taken from a packing net over this set.
We extend their proof of classical hardness to our setting in Appendix~\ref{sec:classical-hardness}\footnote{We note that our classical hardness in fact holds when considering $w_j \geq R_w/d^v$ where $v$ is any constant greater than $3/2$, but we choose $v = 2$ for simplicity.}.

Let $\tilde{g}: \mathbb{R} \to [-1,1]$ be a periodic function of period $1$ which has bounded variation on every finite interval.
In particular, we assume that $\tilde{g}$ can be written as
\begin{equation}
  \label{eq:g-tilde}
  \tilde{g}(y) = \sum_{j=1}^D \beta_j^\star \cos(2\pi jy),\quad \norm{\beta^\star}_1 = 1,
\end{equation}
for some constant $D > 0$.
It is clear that a function of this form has period $1$ and has bounded variation on every finite interval\footnote{One could also choose to write $\tilde{g}$ as a linear combination of sines and cosines to resemble a Fourier series with a finite number of nonzero terms, but adding sines makes the analysis more cumbersome than instructive and does not affect the classical hardness.}.
Here, the condition on the norm of the $\beta^\star$ coefficients ensures that the range of $\tilde{g}$ is in $[-1,1]$.
This is an additional assumption to those considered in~\cite{shamir2018distribution,shalev2017failures}, but we do not expect this to affect the classical hardness.
Namely, the hardness stems from $\tilde{g}$ preserving the Fourier sparsity of the input distribution, and this property is still preserved when taking $\tilde{g}$ to have this specific form.
Concretely,~\cite{shamir2018distribution} also considers an example where $\tilde{g}$ takes this form (in particular, where $\tilde{g}$ is simply a cosine, i.e., $D = 1$), and the hardness result still holds.
Our concept class consists of these functions
\begin{equation}
  \mathcal{C} \triangleq \{g_{w^\star}: \mathbb{R}^d \to [-1,1] : g_{w^\star}(x) = \tilde{g}(x^\intercal w^\star), \; w^\star \in \tilde{\mathcal{S}}_w\},
\end{equation}
with $\tilde{g}$ defined in \Cref{eq:g-tilde}.
Hence, to learn a target function $g_{w^\star}$ in the concept class, it would be sufficient, but perhaps not necessary, to identify $w^\star$ and $\beta^\star$.

We devise a quantum learning algorithm given QSQ access (see \Cref{def:qsq}) to functions in this concept class when the distribution $\mathcal{D}$ is fixed to be either uniform or a discrete Gaussian with a diagonal covariance matrix $\Sigma = \mathrm{diag}(\sigma_1^2,\dots,\sigma_d^2)$ for sufficiently large $\sigma_j$.
In particular, for QSQ access with respect to a truncation parameter $R$, we require $\sigma_j = \Omega(R)$.
We specify the discretization and truncation parameters in more detail in later sections.

To learn a target concept $g_{w^\star}$ with respect to a distribution $\mathcal{D}$, we want to find a good predictor $f_\theta(x)$ which minimizes the objective function
\begin{equation}
  \min_{\theta \in \Theta} \mathcal{L}_{w^\star}(\theta)\triangleq \min_{\theta \in \Theta}\mathop{\mathbb{E}}_{x \sim \mathcal{D}}[(f_\theta(x) - g_{w^\star}(x))^2],
\end{equation}
where $\theta$ are some parameters that we want to learn.
Here, we use the squared loss to align with the classical hardness results~\cite{shamir2018distribution,shalev2017failures}.
As in the classical case, we assume that we have access to this loss function and can compute it for a given choice of parameters $\theta$.
Here, SQ access~\cite{kearns1998efficient} is more general than only having access to (gradients of) the loss function, as it allows the learning algorithm to access expectations of arbitrary functions of the data.
Nonetheless, the SQ setting is a natural generalization of the gradient access model due to the similarities of the arguments used to prove hardness in~\cite{shamir2018distribution} with those of~\cite{kearns1998efficient,blum1994weakly}.
This is discussed in~\cite{shamir2018distribution}.
Thus, we find that the most natural quantum analogue for learning is the QSQ model with noise tolerance $\tau \geq 0$.
For a given precision $\epsilon > 0$, our quantum algorithm will find parameters $\theta$ such that $\mathcal{L}_{w^\star}(\theta) \leq \epsilon$.

To quantify the performance of our quantum algorithm, we count any accesses to the unknown function $g_{w^\star}$.
Namely, we consider both the number of QSQs and the number of (classical) queries to the gradient of the objective function $\mathcal{L}_{w^\star}$.
This is the most fair comparison to the classical lower bound from~\cite{shamir2018distribution}, which is also in terms of the number of queries to the gradient of the objective function.

\section{Classical hardness}

\subsection{Classical hardness for gradient-based methods}
\label{sec:classical-hardness}

In this section, we discuss the hardness of the task detailed in Appendix~\ref{sec:detail-prob} for classical gradient-based methods.
This hardness result was already proven in Ref.~\cite{shamir2018distribution} under a different setting.
Notably, the classical hardness results~\cite{shamir2018distribution,shalev2017failures} hold for any distribution whose density is Fourier-concentrated, in the sense of the following definition.
\begin{definition}[Fourier-concentrated~\cite{shamir2018distribution}] \label{def:cont_fourier_conc}
Let $\epsilon(r)$ be some function from $[0,\infty) \to [0,1]$. A density function $\varphi^2:\mathbb{R}^d \to \mathbb{R}$ is \emph{$\epsilon(r)$-Fourier-concentrated} if its square root $\varphi$ belongs to $L^2(\mathbb{R}^d)$ (square integrable) and satisfies
\begin{equation}
\norm{\hat{\varphi} \cdot \mathbf{1}_{\geq r}}_2 \leq \norm{\hat{\varphi}}_2 \epsilon(r),
\end{equation}
where $\mathbf{1}_{\geq r}$ is the indicator function of $\{x : \norm{x}_2 \geq r\}$.
\end{definition}
Several common distributions are Fourier concentrated. For instance, $\epsilon(r)$ will decay subexponentially when $\varphi$ is a member of various classes of smooth functions such as Gaussians.

For the task detailed in Appendix~\ref{sec:detail-prob}, we have two additional assumptions compared to~\cite{shamir2018distribution}, designed to facilitate error analysis under finite precision, which we argue here do not affect the classical hardness.
First, we sample the vector $w^\star$ from a $0.51$-packing net $\tilde{\mathcal{S}}_w$ of the set $\mathcal{S}_w$ defined by
\begin{equation}
  \mathcal{S}_w \triangleq \left\{w \in R_w \mathbb{S}^{d-1} : w_j \geq \frac{R_w}{d^2}, \; \forall j \in [d]\right\},
\end{equation}
where $\mathbb{S}^{d-1} \subseteq \mathbb{R}^d$ is the $(d-1)$-dimensional unit sphere.
Second, we consider the function $\tilde{g}$ to be of a specific form given in \Cref{eq:g-tilde}.

Instead, Ref.~\cite{shamir2018distribution} considers $w^\star$ sampled from $R_w \mathbb{S}^{d-1}$ and $\tilde{g}$ as an arbitrary function with period $1$ and with bounded variation on every finite interval.
Note that the latter should not affect classical hardness, as our choice of $\tilde{g}$ still preserves the crucial property of Fourier-concentration.
Moreover,~\cite{shamir2018distribution} considers an example where $\tilde{g}$ takes this form (namely when $\tilde{g}$ is simply a cosine), and the classical hardness still holds.
Thus, we do not concern ourselves with the form of $\tilde{g}$ and mainly focus on the former case.

The key result in~\cite{shamir2018distribution} that proves classical hardness is their Theorem 3.
Examining the proof, we notice that the only part that relies on $w^\star$ being sampled from $R_w\mathbb{S}^{d-1}$ is Lemma 5 in~\cite{shamir2018distribution}, which we restate below.
Informally, Lemma 5 tells us that for any function $h$, for a random choice of $w^\star$, the Fourier transform of the target function does not correlate well with $h$.
Thus, no matter what our hypothesis function is, obtaining information about $w^\star$ should be difficult.
The crux of the classical hardness says that, in particular, the gradient of the loss function does not contain much information about $w^\star$.

\begin{lemma}[Lemma 5 in~\cite{shamir2018distribution}]
\label{lem:5}
Let $\varphi^2$ be a density function on $\mathbb{R}^d$ that is $\epsilon(r)$-Fourier-concentrated.
For any square integrable function $h : \mathbb{R}^d \to \mathbb{R}$, if $d \geq c'$ (for some universal constant $c'$) and we sample $w^\star$ uniformly at random from $R_w \mathbb{S}^{d-1}$, then 
\begin{equation}
  \E\left[\left(\langle h, \widehat{g_{w^\star} \varphi}\rangle - a_0\langle h, \hat{\varphi} \rangle\right)^2\right] \leq 10 \norm{h}^2 \left(\exp(-cd) + \sum_{n=1}^\infty\epsilon\left(\frac{nR_w}{2}\right)\right),
\end{equation}
where $a_0, c$ are constants and $\widehat{g_{w^\star} \varphi}$ denotes the Fourier transform of the pointwise product of $g_{w^\star}$ and $\varphi$.
\end{lemma}

Here, the inner product is defined as
\begin{equation}
  \langle f, h \rangle = \int_x f(x) \overline{h(x)}\,dx
\end{equation}
and the norm is $\norm{f} = \sqrt{\langle f, f\rangle}$.
Also, the hat denotes the Fourier transform defined via
\begin{equation}
  \hat{f}(y) = \int \exp(-2\pi i x^\intercal y) f(x)\,dx. 
\end{equation}
Instead, we prove the following similar result.

\begin{lemma}
\label{lem:5-new}
Let $\varphi^2$ be a density function on $\mathbb{R}^d$ that is $\epsilon(r)$-Fourier-concentrated.
For any square integrable function $h : \mathbb{R}^d \to \mathbb{R}$, if $d \geq c'$ (for some universal constant $c'$) and we sample $w^\star$ uniformly at random from $\tilde{\mathcal{S}}_w$, then 
\begin{equation}
  \E_{w^\star \sim \tilde{\mathcal{S}}_w}\left[\left(\langle h, \widehat{g_{w^\star} \varphi}\rangle - a_0\langle h, \hat{\varphi} \rangle\right)^2\right] \leq 10 \norm{h}^2 \left(\exp(-cd) + \sum_{n=1}^\infty\epsilon\left(\frac{nR_w}{4}\right)\right),
\end{equation}
where $a_0, c$ are constants.
\end{lemma}

Note that the difference from \Cref{lem:5} resulting from sampling from the packing net instead of the continuous space is that $R_w/2$ is replaced by $R_w/4$.
Before proving \Cref{lem:5-new}, we need to show the following lemma, which says that there exists a large $0.51$-packing net of $\mathcal{S}_w$.
The choice of $0.51$ is made for convenience, and other choices are possible.

\begin{lemma}
\label{lem:packing_card}
For $v > 3/2$ and $d$ sufficiently large, there exists a $0.51$-packing net $\tilde{\mathcal{S}}_w$ of the set $\mathcal{S}_w$ such that $|\tilde{\mathcal{S}}_w| > e^{cd}$, where is $c$ is an absolute constant.
\end{lemma}

\begin{proof}
We first prove a lower bound on the volume of $\mathcal{S}_w$, and then show that this implies that a large packing net exists.
Define the annulus of width $R_w/d^v$ around the equator as
\begin{equation}
  \mathrm{Ann}(d-1, R_w, R_w/d^v) \triangleq \left\{w \in R_w \mathbb{S}^{d-1} : |w_1| \leq R_w/d^v\right\}.
\end{equation}
The complement of this annulus on the hypersphere is the union of two antipodal spherical caps, where a spherical cap is a portion of a sphere cut off by a plane.
Note that spherical caps can be defined via the angle between the rays from the center of the sphere to the pole and to the edge of the base of the cap, called the half angle.
The half angle $\theta$ subtended by each of these antipodal spherical caps satisfies $\cos \theta = 1/d^v$.
Moreover, it is known~\cite{li2010concise} that the volume of a hyperspherical cap with half angle $\theta$ can be computed as
\begin{equation}
  \label{eq:vol-cap}
  \mathrm{Vol}(\mathrm{Cap}(d, R_w, \theta)) = \frac{1}{2}\mathrm{Vol}(R_w \mathbb{S}^{d-1}) I_{\sin^2\theta}\left(\frac{d}{2}, \frac{1}{2}\right),
\end{equation}
where $\mathrm{Cap}(d, R_w, \theta)$ denotes a hyperspherical cap with half angle $\theta$ of a sphere in $\mathbb{R}^d$ with radius $R_w$.
Also, $I_x(a,b)$ denotes the normalized incomplete Beta function
\begin{equation}
  I_x(a,b) \triangleq \frac{B_x(a,b)}{B_1(a,b)},\quad B_x(a,b) \triangleq \int_0^x t^{a-1}(1-t)^{b-1}\,dt.
\end{equation}
Using that the annulus defined is the complement of the union of two antipodal spherical caps, we can compute its volume as
\begin{align}
\mathrm{Vol}\left(\mathrm{Ann}\left(d-1, R_w, R_w/d^v\right)\right) &= \mathrm{Vol}(R_w \mathbb{S}^{d-1}) - 2\mathrm{Vol}\left(\mathrm{Cap}\left(d, R_w, \arccos\left(1/d^v\right)\right)\right)\\
&= \mathrm{Vol}(R_w \mathbb{S}^{d-1})\left(1 - I_{\sin^2(\arccos(1/d^v))}\left(d/2, 1/2\right)\right)\\
&= \mathrm{Vol}(R_w \mathbb{S}^{d-1})\left(1 - I_{1 - 1/d^{2v}}(d/2, 1/2)\right).
\end{align}
We can bound the second term above. First, expanding in terms of the definition, we have:
\begin{align}
1 - I_{1-1/d^{2v}}(d/2, 1/2) &= 1 - \frac{1}{B_1(d/2, 1/2)}\left(\int_0^t t^{d/2 - 1}(1-t)^{-1/2}\,dt - \int_{1-1/d^{2v}}^1 t^{d/2-1}(1-t)^{-1/2}\,dt\right)\\
&= \frac{1}{B_1(d/2,1/2)}\int_{1-1/d^{2v}}^1 t^{d/2-1}(1-t)^{-1/2}\,dt.
\end{align}
We can bound the integral as
\begin{equation}
  \int_{1-1/d^{2v}}^1 t^{d/2-1}(1-t)^{-1/2}\,dt \leq \int_{1-1/d^{2v}}^1 (1-t)^{-1/2}\,dt = \frac{2}{d^v}.
\end{equation}
Moreover, we can lower bound the beta function. Recall that the Beta function can be written in terms of Gamma functions:
\begin{equation}
  B(d/2, 1/2) = \frac{\Gamma(d/2)\Gamma(1/2)}{\Gamma(d/2 + 1/2)}.
\end{equation}
Standard bounds on ratios of Gamma functions~\cite{wendel1948note} give
\begin{equation}
  \label{eq:beta_bound}
  B(d/2,1/2) \geq cd^{-1/2}
\end{equation}
for some absolute constant $c$.
Putting everything together, we see that
\begin{equation}
  \mathrm{Vol}(\mathrm{Ann}(d-1, R_w, R_w/d^v)) \leq C\, \mathrm{Vol}(R_w \mathbb{S}^{d-1}) d^{1/2-s}
\end{equation}
for some absolute constant $C$.
Denote
\begin{equation}
  \mathcal{S}_{w,\pm} \triangleq \left\{w \in R_w\mathbb{S}^{d-1} : |w_j| \geq R_w/d^v,\; \forall j \in [d]\right\}.
\end{equation}
Using our previous work, we can lower bound the volume of this set:
\begin{align}
  \mathrm{Vol}(\mathcal{S}_{w,\pm}) &\geq \mathrm{Vol}(R_w \mathbb{S}^{d-1})\left(1 - d\,\mathrm{Vol}(\mathrm{Ann}(d-1, R_w, R_w/d^v))\right)\\
  &\geq \mathrm{Vol}(R_w \mathbb{S}^{d-1})\left(1 - Cd^{1/2-s}\right)\\
  &\geq \frac{1}{2}\mathrm{Vol}(R_w \mathbb{S}^{d-1}),
\end{align}
where in the last line we used $s > 3/2$ and $d$ sufficiently large.
Thus, it follows that
\begin{equation}
  \label{eq:volW}
  \mathrm{Vol}(\mathcal{S}_w) \geq \frac{\mathrm{Vol}(R_w \mathbb{S}^{d-1})}{2^{d+1}}.
\end{equation}
In order to lower bound $|\tilde{\mathcal{S}}_w|$, we use a lower bound in terms of the ratio of $\mathrm{Vol}(\mathcal{S}_w)$ and the volume of a spherical cap with angle $0.51/2 = 0.255$ (see, e.g., Proposition 4.2.12 of~\cite{vershynin2018high}). This gives
\begin{align}
|\tilde{\mathcal{S}}_w| &\geq \frac{\mathrm{Vol}(\mathcal{S}_w)}{\mathrm{Vol}(\mathrm{Cap}(d, R_w, 0.255))}\\
&\geq \frac{\mathrm{Vol}(R_w\mathbb{S}^{d-1})}{2^{d+1} \mathrm{Vol}(\mathrm{Cap}(d, R_w, 0.255))}\\
&= \frac{1}{2^d I_{\sin^2(0.255)}(d/2, 1/2)}\\
&= \frac{B(d/2, 1/2)}{2^d \int_0^{\sin^2(0.255)}t^{d/2 - 1} (1-t)^{-1/2}\,dt}\\
&\geq \frac{cd^{-1/2}}{2^d \int_0^{\sin^2(0.255)}t^{d/2-1} (1-t)^{-1/2}\,dt}\\
&\geq \frac{c'd^{-1/2}}{2^d \int_0^{\sin^2(0.255)} t^{d/2-1}\,dt}\\
&\geq \frac{c'\sqrt{d}}{2}\frac{1}{(2\sin(0.255))^d}\\
&\geq e^{c''d}.
\end{align}
In the second line, we use \Cref{eq:volW}.
In the third line, we use \Cref{eq:vol-cap}.
In the fourth line, we use the definition of $I_x(a,b)$.
In the fifth line, we use \Cref{eq:beta_bound}.
In the sixth line, we redefine the constant by absorbing a factor of $1/(1-\sin^2(0.255))^{-1/2}$.
Finally, in the last line, we assume that $d$ is sufficiently large in order to absorb the polynomial factor in $d$ and use that $2\sin(0.255) < 1$.
\end{proof}

With this result, we can prove \Cref{lem:5-new}.

\begin{proof}[Proof of~\Cref{lem:5-new}]
We follow the proof of Lemma 5 in~\cite{shamir2018distribution} but make appropriate changes.
Note that Lemma 2 from~\cite{shamir2018distribution} proves that for any $w$,
\begin{equation}
  \label{eq:lemma2}
  \widehat{g_{w^\star} \varphi}(x) = \sum_{z\in \mathbb{Z}} a_z \cdot \hat{\varphi}(x - zw^\star),
\end{equation}
where $a_z$ are complex coefficients corresponding to the Fourier series expansion of $\tilde{g}$.
Using this, we can write
\begin{align}
\E_{w^\star \sim \tilde{\mathcal{S}}_w}\left[\left(\langle h, \widehat{g_{w^\star} \varphi} \rangle - a_0 \langle h, \hat{\varphi}\rangle\right)^2\right] &= \E_{w^\star \sim \tilde{\mathcal{S}}_w}\left[\left(\left\langle h, \sum_{z \in \mathbb{Z}} a_z \hat{\varphi}(\cdot - zw^\star) \right\rangle - a_0 \langle h, \hat{\varphi}\rangle\right)^2\right]\\
&= \E_{w^\star \sim \tilde{\mathcal{S}}_w}\left[\left\langle h, \sum_{z \in \mathbb{Z} \setminus \{0\}} a_z \hat{\varphi}(\cdot - zw^\star) \right\rangle^2\right].
\end{align}
For any $w \in \tilde{\mathcal{S}}_w$, define
\begin{equation}
  A_{w,r} \triangleq \{x \in \mathbb{R}^d : \exists z \in \mathbb{Z} \setminus \{0\} \text{ s.t. } \norm{x - zw}_2 < r\}.
\end{equation}
Let $\indicator_{A_{w,r}}$ denote the indicator function to the set $A_{w,r}$ and $\indicator_{A_{w,r}^C}$ denote the indicator of its complement.
Using that $(a+b)^2 \leq 2(a^2 + b^2)$, we can upper bound our previous expression by
\begin{align}
&\E_{w^\star \sim \tilde{\mathcal{S}}_w}\left[\left(\langle h, \widehat{g_{w^\star} \varphi} \rangle - a_0 \langle h, \hat{\varphi}\rangle\right)^2\right]\\
&\leq 2\E_{w^\star \sim \tilde{\mathcal{S}}_w}\left[\left\langle h, \indicator_{A_{w^\star,R_w/4}} \sum_{z \in \mathbb{Z} \setminus \{0\}} a_z \hat{\varphi}(\cdot - zw^\star) \right\rangle^2 \right] + 2\E_{w^\star \sim \tilde{\mathcal{S}}_w}\left[\left\langle h, \indicator_{A_{w^\star,R_w/4}^C} \sum_{z \in \mathbb{Z} \setminus \{0\}} a_z \hat{\varphi}(\cdot - zw^\star) \right\rangle^2\right].
\label{eq:exp-sum}
\end{align}
Note that this is slightly different from the proof in~\cite{shamir2018distribution}, where we use the set $A_{w,R_w/4}$ instead of $A_{w,R_w/2}$.
This is because, as we show shortly, for $w\in \tilde{\mathcal{S}}_w$, the sets $A_{w,R_w/4}$ are disjoint.
In contrast, for the set $\mathcal{W}$ chosen in~\cite{shamir2018distribution}, $A_{w,R_w/2}$ are disjoint instead.

First, let us show that $A_{w,R_w/4}$ are disjoint for $w \in \tilde{\mathcal{S}}_w$.
Suppose for the sake of contradiction that the $A_{w,R_w/4}$ are not disjoint, i.e., there exists some $x \in \mathbb{R}^d$ such that $\norm{x - zw}_2 < R_w/4$ and $\norm{x - z'w'}_2 < R_w/4$ for $z,z' \in \mathbb{Z} \setminus \{0\}$ and $w,w' \in \tilde{\mathcal{S}}_w$.
By triangle inequality, we have
\begin{equation}
  \label{eq:zwz'w'}
  \norm{zw - z'w'}_2 \leq \norm{x-zw}_2 + \norm{x - z'w'}_2 \leq R_w/2.
\end{equation}
Since $w,w'$ are both in $R_w \mathbb{S}^{d-1} \cap \mathbb{R}_+^d$, if the signs of $z$ and $z'$ are different, then the angle between the segments $wz$ and $wz'$ is greater than $\pi/2$.
This implies that the cosine of this angle $\theta$ is negative.
Since $\norm{zw}_2 \geq R_w$ and $\norm{z'w'}_2 \geq R_w$, this implies
\begin{align}
\norm{zw - z'w'}_2^2 &= \norm{zw}_2^2 + \norm{z'w'}_2^2 - 2\norm{z'w}_2 \norm{z'w'}_2 \cos\theta\\
&\geq \norm{zw}_2^2 + \norm{z'w'}_2^2\\
&\geq R_w^2,
\end{align}
which contradicts \Cref{eq:zwz'w'}.
Thus, we can henceforth assume that the signs of $z,z'$ are the same.

Note that $\norm{x-zw}_2 < R_w/4$ implies that $x$ lies on a spherical shell of width $R_w/2$, centered at radius $z R_w$.
Similarly, $\norm{x - z'w'}_2 < R_w/4$ implies that $x$ lies on a spherical shell of width $R_w/2$ centered at radius $z' R_w$.
Since these do not intersect when $z \neq z'$, there can be no such $x$ in the intersection of these two sets.

Finally, it remains to consider the case of $z = z'$.
Note that
\begin{equation}
  \norm{zw - zw'}_2^2 \geq \norm{w - w'}_2^2.
\end{equation}
From the definition of $\tilde{\mathcal{S}}_w$ as a $0.51$-packing net of $\mathcal{S}_w$, we then have
\begin{equation}
  \sin \frac{0.51}{2} = \frac{\norm{w - w'}_2}{2R_w}.
\end{equation}
This implies that
\begin{equation}
  \norm{w - w'}_2 \geq 2R_w \sin \frac{0.51}{2} > \frac{R_w}{2},
\end{equation}
contradicting \Cref{eq:zwz'w'}.
It follows that no such $x$ can exist, and thus $A_{w,R_w/4}$ are disjoint for all $w \in \tilde{\mathcal{S}}_w$.

Now, using this, we want to bound the expression in~\Cref{eq:exp-sum}.
For the first term in \Cref{eq:exp-sum}, the same argument as in~\cite{shamir2018distribution} holds for our case.
We reproduce the argument here.
\begin{align}
&\E_{w^\star \sim \tilde{\mathcal{S}}_w}\left[\left\langle h, \indicator_{A_{w^\star,R_w/4}} \sum_{z \in \mathbb{Z} \setminus \{0\}} a_z \hat{\varphi}(\cdot - zw^\star) \right\rangle^2\right]\\
&= \E_{w^\star \sim \tilde{\mathcal{S}}_w}\left[\left\langle h, \indicator_{A_{w^\star,R_w/4}}\left(\widehat{g_{w^\star}\varphi} - a_0 \hat{\varphi}\right) \right\rangle^2\right]\\
&= \E_{w^\star \sim \tilde{\mathcal{S}}_w}\left[\left\langle \indicator_{A_{w,R_w/4}} h, \widehat{g_{w^\star}\varphi} - a_0 \hat{\varphi}\right\rangle^2\right]\\
&\leq \E_{w^\star \sim \tilde{\mathcal{S}}_w}\left[\norm{\indicator_{A_{w^\star,R_w/4}} h}_2^2 \norm{\widehat{g_{w^\star}\varphi} - a_0\hat{\varphi}}_2^2\right]\\
&\leq 2\E_{w^\star \sim \tilde{\mathcal{S}}_w}\left[\norm{\indicator_{A_{w^\star,R_w/4} h}}_2^2 \left(\norm{\widehat{g_{w^\star} \varphi}}_2^2 + \norm{a_0 \hat{\varphi}}_2^2\right)\right]\\
&= 2\E_{w^\star \sim \tilde{\mathcal{S}}_w}\left[\norm{\indicator_{A_{w^\star,R_w/4} h}}_2^2 \left(\norm{g_{w^\star} \varphi}_2^2 + |a_0|^2\norm{\hat{\varphi}}_2^2\right)\right]\\
&\leq 4 \E_{w^\star \sim \tilde{\mathcal{S}}_w}\left[\norm{\indicator_{A_{w^\star,R_w/4}} h}_2^2\right]\\
&\leq \frac{4}{|\tilde{\mathcal{S}}_w|} \sum_{w^\star \in \tilde{\mathcal{S}}_w} \int \indicator_{A_{w^\star, R_w/4}}|h(x)|^2 \,dx\\
&= \frac{4}{|\tilde{\mathcal{S}}_w|} \int \left(\sum_{w^\star \in \tilde{\mathcal{S}}_w}\indicator_{A_{w^\star, R_w/4}}\right)|h(x)|^2 \,dx\\
&\leq \frac{4}{|\tilde{\mathcal{S}}_w|} \int |h(x)|^2 \,dx\\
&\leq 4e^{-cd}\norm{h}_2^2,\label{eq:first-exp}
\end{align}
where in the second line, we use \Cref{eq:lemma2}.
In the fourth line, we use the Cauchy-Schwarz inequality.
In the seventh line, we use that $\norm{\hat{\varphi}}_2 = \norm{\varphi}_2 = 1$, $|a_0|^2 \leq \sum_z |a_z|^2 \leq 1$, and
\begin{equation}
  \norm{g_{w^\star} \varphi}_2^2 = \int g_{w^\star}^2(x)\varphi^2(x)\,dx \leq \int \varphi^2(x)\,dx = 1.
\end{equation}
In the second to last line, we use that $A_{w,R_w/4}$ are disjoint sets for $w \in \tilde{\mathcal{S}}_w$, as previously argued, so that $\sum_{w \in \tilde{\mathcal{S}}_w} \indicator_{A_{w,R_w/4}}(x) \leq 1$ for any $x$.
The last line follows by \Cref{lem:packing_card}.

Finally, it remains to bound the second term in \Cref{eq:exp-sum}.
We will upper bound the expression deterministically for any $w^\star$, so we may drop the expectation
By Cauchy-Schwarz,
\begin{align}
\left\langle h, \indicator_{A_{w^\star,R_w/4}^C} \sum_{z \in \mathbb{Z} \setminus \{0\}} a_z \hat{\varphi}(\cdot - zw^\star) \right\rangle^2 &\leq \norm{h}_2^2 \cdot \norm{\indicator_{A_{w^\star, R_w/4}^C} \sum_{z \in \mathbb{Z} \setminus \{0\}} a_z \hat{\varphi}(\cdot - zw^\star)}_2^2\\
&= \norm{h}_2^2\left(\sum_{z_1,z_2 \in \mathbb{Z} \setminus \{0\}} a_{z_1} a_{z_2}^* \langle \indicator_{A_{w^\star, R_w/4}^C} \hat{\varphi}(\cdot - z_1w^\star), \hat{\varphi}(\cdot - z_2w^\star) \rangle \right).
\end{align}
First, consider the terms in the above sum with $z_1 = z_2$.
Then, we have
\begin{align}
  \langle \indicator_{A_{w^\star, R_w/4}^C} \hat{\varphi}(\cdot - z_1w^\star), \hat{\varphi}(\cdot - z_2w^\star) \rangle &= \int \indicator_{A_{w^\star, R_w/4}^C} |\hat{\varphi}(x - z_1 w^\star)|^2\,dx\\
  &= \int \indicator_{A_{w^\star, R_w/4}^C}(x + z_1w^\star)|\hat{\varphi}(x)|^2\,dx\\
  &\leq \int_{x : \norm{x}_2 \geq R_w/4} |\hat{\varphi}(x)|^2\,dx\\
  &\leq \epsilon^2(R_w/4).
\end{align}
Here, the third line follows by definition of $A_{w^\star, R_w/4}^C$ and the assumption that $z_1\neq 0$ so that $\indicator_{A_{w^\star, R_w/4}^C}(x + z_1w^\star) = 1$ only if $\norm{x}_2 \geq R_w/4$.
The last line follows since $\varphi$ is $\epsilon(r)$-Fourier-concentrated.

For terms such that $z_1 \neq z_2$, the exact same argument as in~\cite{shamir2018distribution} holds, so we do not reproduce it here.
This gives a bound of
\begin{equation}
  \sum_{\substack{z_1,z_2 \in \mathbb{Z} \setminus \{0\} \\ z_1 \neq z_2}} a_{z_1} a_{z_2}^* \langle \indicator_{A_{w^\star, R_w/4}^C} \hat{\varphi}(\cdot - z_1w^\star), \hat{\varphi}(\cdot - z_2w^\star) \rangle \leq 4 \sum_{n=1}^\infty \epsilon(nR_w/2).
\end{equation}
Thus, putting everything together, we have
\begin{equation}
  \left\langle h, \indicator_{A_{w^\star,R_w/4}^C} \sum_{z \in \mathbb{Z} \setminus \{0\}} a_z \hat{\varphi}(\cdot - zw^\star) \right\rangle^2 \leq \norm{h}_2^2 \left(\epsilon^2(R_w/4) + 4\sum_{n=1}^\infty \epsilon(nR_w/2)\right) \leq 5\norm{h}_2^2 \sum_{n=1}^\infty \epsilon(nR_w/4),
\end{equation}
where we used the fact that $\epsilon^2(R_w/4) \leq \epsilon^2(R_w/4) \leq \sum_{n=1}^\infty \epsilon(n R_w/4)$ and $\epsilon$ is a non-increasing function for distributions of interest.
Together with \Cref{eq:first-exp}, plugging into \Cref{eq:exp-sum}, we obtain the claim.
\end{proof}

\subsection{Correlational SQ lower bound}
\label{sec:csq-hardness}

We extend the classical hardness argument from the previous section to hold against any classical algorithm utilizing correlational SQs.
Recall from Appendix~\ref{sec:learning} that in the correlational SQ model~\cite{bshouty2002using,bendavid1995learning}, queries are restricted to acting only on the input space, i.e., for a query $\phi$, algorithms receive estimates of $\mathop{\mathbb{E}}_{x\sim\mathcal{D}}[\phi(x)g_{w^\star}(x)]$.
In this section, we focus on Gaussian distributions and prove hardness for the simplest case of when $\tilde{g}$ consists of a single cosine.
This clearly implies hardness for the more general $\tilde{g}$ considered in the remainder of the paper.
We prove the following theorem.

\begin{theorem}[Correlational SQ Hardness]
\label{thm:csq-hardness}
Let $\epsilon \in (0,1)$.
Let $\mathcal{C} = \{x\mapsto \cos(2\pi x^\intercal w^\star) : w^\star \in \tilde{\mathcal{S}}_w \subseteq \mathbb{R}^d\}$ be a concept class, where $\tilde{\mathcal{S}}_w$ is a 0.51-packing net of the set $\mathcal{S}_w$ given in \Cref{eq:sw}.
Consider a Gaussian distribution $\mathcal{D}$ with a diagonal covariance matrix $\Sigma = \sigma^2 I$, where $\sigma \geq \tilde{\Omega}(d^2)$.
Then, any classical algorithm using correlational SQs to learn an unknown $c^\star \in \mathcal{C}$ with respect to the distribution $\mathcal{D}$ requires at least $2^{\Omega(d)}$ queries of tolerance $\mathcal{O}(1/d^4)$ to learn $\mathcal{C}$ to error $\epsilon$.
\end{theorem}

First, we note that while some of the parameters, e.g., the variance and tolerance, appear arbitrary, we have in fact carefully chosen these to align with the parameters of our quantum algorithm.
Namely, $\sigma \geq \tilde{\Omega}(d^2)$ is also satisfied for our efficient quantum algorithm solving this problem via \Cref{coro:gauss-satisfies-assum}.
Also, as seen in \Cref{thm:non-unif-guarantee}, a QSQ tolerance of $\mathcal{O}(1/d^4)$ is also sufficient for our quantum algorithm to learn successfully\footnote{The scaling for the variance can be seen by $\sigma \geq \tilde{\Omega}(\tau M_1^2 d^4/R_w^2)$, $M_1 \geq R_w^2/\epsilon_1 \geq R_w^2d$ and taking $\tau = \mathcal{O}(1/(d^4 R_w^2))$. These conditions are all satisfied by the choices of parameters in \Cref{thm:non-unif-guarantee} and \Cref{coro:gauss-satisfies-assum}. To be precise, as stated, the tolerance of the classical hardness only matches the QSQ tolerance with respect to the $d$ scaling. The proof can be extended in the same way such that the tolerances match precisely, but we focus on $d$ scaling for simplicity.}.

Moreover, this lower bound only holds against classical algorithms using \emph{correlational} SQs, rather than general SQs.
In the case of Boolean functions, these two models are in fact equivalent~\cite{bshouty2002using}.
However, for real functions, there exist separations between correlational SQs and general SQs~\cite{andoni2014learning,andoni2019attribute,chen2022learning}.
Nevertheless, \Cref{thm:csq-hardness} is a strengthening of the classical hardness proven in Appendix~\ref{sec:classical-hardness}, and we view it as an important step towards general SQ hardness.
Additionally, it is interesting to observe that only one type of query made by our quantum algorithm is not a ``correlational QSQ.''
Correlational QSQs have not been studied before in the literature, but a clear natural analogue is that queried observables $O$ can only act on the input register, i.e., $O = O_1 \otimes I$, where the identity acts on the output register.
Then, the only queries that our quantum algorithms in Appendices~\ref{sec:uniform} and \ref{sec:non-unif} make that are not correlational in this sense are of the form given in \Cref{eq:Am}.
As all other QSQs are correlational, considering classical algorithms with access only to correlational SQs is arguably not significantly restrictive.

Previously, classical learning theorists have shown similar correlational SQ lower bounds for learning single-layer neural networks~\cite{goel2020superpolynomial,diakonikolas2020algorithms}.
However, the functions for which they show hardness of learning are not the quite same as those considered here.
In particular, they differ in the activation function (cosine vs. ReLu), and also do not restrict the valid affine functions to $w^\star \in \tilde{\mathcal{S}_w}$.
Thus, the above theorem does not follow immediately from existing results and needs to be analyzed separately.

One way to prove correlational SQ lower bounds is via the \emph{statistical dimension}~\cite{blum1994weakly,yang2005new,feldman2017statistical}, which captures the difficulty of learning a concept class, similarly to the more commonly known VC dimension.
Informally, the statistical dimension quantifies the size of the largest subset of the concept class whose elements have ``low correlation.''
More precisely, we state the following definition, following the presentation of~\cite{goel2020superpolynomial}.

\begin{definition}[Statistical dimension]
\label{def:sd}
Let $\mathcal{C}$ be a concept class, and let $\mathcal{D}$ be a distribution on the same domain as functions in $\mathcal{C}$. Define the \emph{average correlation} of $\mathcal{C}$ as
\begin{equation}
  \rho_{\mathcal{D}}(\mathcal{C}) \triangleq \frac{1}{|\mathcal{C}|^2} \sum_{c,c' \in \mathcal{C}} \left| \mathop{\mathbb{E}}_{x \sim \mathcal{D}}[c(x) c'(x)]\right|.
\end{equation}
Then, the \emph{statistical dimension} of $\mathcal{C}$ at threshold $\gamma$, denoted $\mathrm{SDA}_{\mathcal{D}}(\mathcal{C}, \gamma)$ is the largest $A$ such that for all $\mathcal{C}' \subseteq \mathcal{C}$ with $|\mathcal{C}'| \geq |\mathcal{C}|/A$, then $\rho_{\mathcal{D}}(\mathcal{C}') \leq \gamma$.
\end{definition}

Intuitively, functions in the subset $\mathcal{C}'$ have low correlation and thus should be hard to distinguish and hard to learn.
This fundamental relationship between the statistical dimension and SQ lower bounds has been formalized~\cite{feldman2017statistical,szorenyi2009characterizing,feldman2012complete}.
We state a version of this result, following the presentation of~\cite{goel2020superpolynomial}.

\begin{theorem}[Theorem 4.1 in~\cite{goel2020superpolynomial}]
\label{thm:sda-implies-lower}
Let $\epsilon \in (0,1)$, and let $\gamma > 0$.
Let $\mathcal{C}$ be a concept class, and let $\mathcal{D}$ be a distribution on the same domain as functions in $\mathcal{C}$.
Suppose that $\mathop{\mathbb{E}}_{x\sim\mathcal{D}}[c^2(x)] > \epsilon^2$ for all $c \in \mathcal{C}$.
Let $A = \mathrm{SDA}_{\mathcal{D}}(\mathcal{C}, \gamma)$.
Then, any SQ learning making only correlational SQs to some unknown $c \in \mathcal{C}$ requires at least $\Omega(A)$ queries of tolerance $\sqrt{\gamma}$ to learn $\mathcal{C}$ up to error $\epsilon$.
\end{theorem}

\Cref{thm:csq-hardness} follows as a consequence of this theorem.
In order to apply it, we first need to lower bound the statistical dimension of our concept class.

\begin{lemma}
\label{lem:sda}
Let $\mathcal{C} = \{x \mapsto \cos(2\pi x^\intercal w^\star) : w^\star \in \tilde{\mathcal{S}}_w \subseteq \mathbb{R}^d\}$ be a concept class, where $\tilde{\mathcal{S}}_w$ is a 0.51 packing net of the set $\mathcal{S}_w$ given in \Cref{eq:sw}.
Let $\mathcal{D}$ be a Gaussian distribution with a digaonal covariance matrix $\Sigma = \sigma^2 I$, where $\sigma \geq \tilde{\Omega}(d^2)$.
Then, $\mathrm{SDA}_{\mathcal{D}}(\mathcal{C}, \gamma) \geq 2^{\Omega(d)}$ for $\gamma = \mathcal{O}(1/d^8)$.
\end{lemma}

\begin{proof}
By definition of the statistical dimension, we need to consider the average correlation.
In particular, we need to control expectations of the form
\begin{equation}
  \mathop{\mathbb{E}}_{x\sim\mathcal{D}}[\cos(2\pi x^\intercal w^\star) \cos(2\pi x^\intercal v^\star)].
\end{equation}
Using the sum-product formula for cosines, we have
\begin{equation}
  \mathop{\mathbb{E}}_{x\sim\mathcal{D}}[\cos(2\pi x^\intercal w^\star) \cos(2\pi x^\intercal v^\star)] = \frac{1}{2}\mathop{\mathbb{E}}_{x \sim \mathcal{D}}[\cos(2\pi x^\intercal (w^\star + v^\star))] + \frac{1}{2}\mathop{\mathbb{E}}_{x \sim \mathcal{D}}[\cos(2\pi x^\intercal (w^\star - v^\star))].
\end{equation}
Now, to evaluate these expectations, recall the definition of a characteristic function $\varphi(t) = \mathop{\mathbb{E}}[e^{-it^\intercal X}]$, where $X$ is a random vector.
When the distribution of $X$ is a multivariate normal distribution with mean vector $\mu$ and covariance matrix $\Sigma$, it is well known that the characteristic function is $\varphi(t) = e^{it^\intercal \mu - t^\intercal \Sigma t/2}$.
In our case, this simplifies to $\varphi(t) = e^{-\sigma^2 \norm{t}_2^2/2}$.
Meanwhile, expanding the definiton of $\varphi(t)$, we have
\begin{equation}
  \varphi(t) = \mathbb{E}[e^{-it^\intercal X}] = \mathbb{E}[\cos(t^\intercal X) - i\sin(t^\intercal X)].
\end{equation}
Setting this equal to the known expression for Gaussian distributions, i.e., $\varphi(t) = e^{-\sigma^2 \norm{t}_2^2/2}$, then notice that this expression is only real.
Thus, one can conclude that the imaginary part of the characteristic function is zero and hence
\begin{equation}
  \mathbb{E}[\cos(t^\intercal X)] = e^{-\sigma^2 \norm{t}_2^2/2}
\end{equation}
for a random vector $X$ distributed according to a multivariate Gaussian.
Plugging this into our previous expression, we have
\begin{equation}
  \left|\mathop{\mathbb{E}}_{x\sim\mathcal{D}}[\cos(2\pi x^\intercal w^\star) \cos(2\pi x^\intercal v^\star)]\right| \leq \frac{1}{2}e^{-\sigma^2\norm{2\pi (w^\star + v^\star)}_2^2/2} + \frac{1}{2}e^{-\sigma^2\norm{2\pi (w^\star - v^\star)}_2^2/2}.
\end{equation}
It remains to lower bound these two norms.
This can be done via a simple calculation while using the fact that $w^\star, v^\star \in \tilde{\mathcal{S}}_w$, which is a 0.51-packing net of $\mathcal{S}_w$.
In particular, this means that each vector in $\tilde{\mathcal{S}}_w$ is separated by a geodesic angle of at least $0.51$, which implies that $|(w^\star)^\intercal v^\star| \leq R_w^2 \cos(0.51)$.
Hence, we have
\begin{align}
  \norm{2\pi(w^\star + v^\star)}_2^2 &= (2\pi)^2 (\norm{w^\star}_2^2 + \norm{v^\star}_2^2 + 2(w^\star)^\intercal v^\star)\\ 
  &\geq (2\pi)^2(2R_w^2 - 2R_w^2\cos(0.51))\\
  &= 2(2\pi R_w)^2(1 - \cos(0.51))\\
  &= (4\pi R_w)^2 \sin^2(0.255)
\end{align}
and similarly
\begin{align}
  \norm{2\pi(w^\star - v^\star)}_2^2 &= (2\pi)^2 (\norm{w^\star}_2^2 + \norm{v^\star}_2^2 - 2(w^\star)^\intercal v^\star)\\ 
  &\geq (2\pi)^2(2R_w^2 - 2R_w^2\cos(0.51))\\
  &= (4\pi R_w)^2 \sin^2(0.255).
\end{align}
Thus, we have
\begin{equation}
  \left|\mathop{\mathbb{E}}_{x\sim\mathcal{D}}[\cos(2\pi x^\intercal w^\star) \cos(2\pi x^\intercal v^\star)]\right| \leq e^{-\frac{1}{2}\sigma^2(4\pi R_w)^2 \sin^2(0.255)} \leq e^{-\tilde{\mathcal{O}}(d^4R_w^2)},
\end{equation}
where in the last inequality, we use our condition that $\sigma \geq \tilde{\Omega}(d^2)$.
Using this, we can bound the average correlation.
Here, we write $g_{w^\star}(x) = \cos(2\pi x^\intercal w^\star)$.
\begin{equation}
  \rho_{\mathcal{D}}(\mathcal{C}') = \frac{1}{|\mathcal{C}'|^2} \sum_{g_{w^\star}, g_{v^\star} \in \mathcal{C}'} \left|\mathop{\mathbb{E}}_{x \sim \mathcal{D}}[g_{w^\star}(x)g_{v^\star}(x)]\right| \leq \frac{1}{|\mathcal{C}'|^2}\left(|\mathcal{C}'| + \sum_{\substack{g_{w^\star}, g_{v^\star} \in \mathcal{C}'\\w^\star \neq v^\star}} e^{-\tilde{\mathcal{O}}(d^4R_w^2)}\right) \leq \frac{1}{|\mathcal{C}'|} + e^{-\tilde{\mathcal{O}}(d^4R_w^2)}.
\end{equation}
This is less than $\mathcal{O}(1/d^8)$ if $|\mathcal{C}'| \geq d^8$.
The size of $\mathcal{C}'$ is in turn greater than $|\mathcal{C}|/A = e^{cd}/A$, where $A$ is the statistical dimension when $A \geq e^{cd}/d^8 = 2^{\Omega(d)}$.
This completes the proof of the lemma.
\end{proof}

With this, we can prove \Cref{thm:csq-hardness} by applying \Cref{thm:sda-implies-lower}.

\begin{proof}[Proof of~\Cref{thm:csq-hardness}]
Now that we have a lower bound on the statistical dimension, it remains to check the conditions of \Cref{thm:sda-implies-lower} and apply the theorem.
We only need to check the condition that $\mathop{\mathbb{E}}_{x \sim \mathcal{D}}[\cos^2(2\pi x^\intercal w^\star)] > \epsilon^2$ for all $w^\star \in \tilde{\mathcal{S}}_w$.
This follows easily by the same manipulations as in the proof of~\Cref{lem:sda}.
\begin{equation}
  \mathop{\mathbb{E}}_{x \sim\mathcal{D}}[\cos^2(2\pi x^\intercal w^\star)] = \mathop{\mathbb{E}}_{x \sim\mathcal{D}}[1 + \cos(4\pi x^\intercal w^\star)] = 1 + e^{-\sigma^2 \norm{4\pi w^\star}_2^2/2} = 1 + e^{-(4\pi \sigma)^2R_w^2/2}.
\end{equation}
This is clearly greater than $\epsilon^2$ for any $\epsilon \in (0,1)$.
Thus, by \Cref{thm:sda-implies-lower} and \Cref{lem:sda}, then the theorem clearly follows.
\end{proof}

\section{Uniform data distribution}

\label{sec:uniform}

In this section, we consider learning our concept class defined in Appendix~\ref{sec:detail-prob} when given QSQ access to quantum example states with respect to the uniform distribution.

\begin{theorem}[Guarantee; Uniform Case]
\label{thm:unif}
Let $\epsilon, \delta > 0, \tau \geq 0$.
Let $\varphi^2$ be the uniform distribution.
Let $w^\star \in \mathbb{R}^d$ be unknown with norm $R_w > 0$ and $w_j^\star \geq R_w/d^2$, for all $j \in [d]$.
Let $g_{w^\star}: \mathbb{R}^d \to [-1,1]$ be defined as $g_{w^\star}(x) = \tilde{g}(x^\intercal w^\star)$, where $\tilde{g}: \mathbb{R} \to [-1,1]$ is a function defined in \Cref{eq:g-tilde}.
Consider parameters $M_1 = \max(70\pi dD^3 R_w, R_w^2/\epsilon_1)$, $M_2 = c M_1$, where $c$ is any constant such that $M_2$ is an integer and $c < 1/(8\pi D R_w)$, and
\begin{equation}
    \tilde{R} = \tilde{\Omega}\left(\max\left(\frac{\tau M_1^2d^4}{R_w^2}, \frac{D^2}{\epsilon}, \frac{D^2\sqrt{d}}{R_w \epsilon}, \frac{D^{5/2}}{\sqrt{\epsilon}}, \frac{D^{3/2}\sqrt{d}}{R_w \sqrt{\epsilon}}\right)\right),\quad \epsilon_1 = \tilde{\mathcal{O}}\left(\min\left(\frac{\epsilon^3}{D^6 d}, \frac{\epsilon^{3/2}}{D^{13/2}d}, \frac{R_w}{D\sqrt{d}}\right)\right)
\end{equation}
Suppose we have QSQ access (see \Cref{def:qsq}) with respect to discretization parameters $M_{1,m} \triangleq m M_1$, $M_{2,m} \triangleq m M_2$ and a truncation parameter $R \geq \tilde{R}$, for $m \in \{1,\dots, D\}$.
Then, there exists a quantum algorithm with this QSQ access that can efficiently find parameters $\hat{\beta} \in \mathbb{R}^D$ such that $\mathcal{L}_{w^\star}(\hat{\beta}) \leq \epsilon$ with probability at least $1-\delta$.
Moreover, this algorithm uses
\begin{equation}
   N = \mathcal{O}\left(d D \log\left(\frac{1}{\delta}\right) \log^5\left(\frac{M_1 d^2}{R_w}\right)\right)
\end{equation}
quantum statistical queries with tolerance $\tau \leq \min\left(\frac{1}{M_2^2}\left(\frac{7}{40D} - \frac{1}{M_2}\right), \frac{1}{2D^2M_2^2}\left(\frac{2}{15} - \frac{1}{8}\left(\frac{2\pi R_w}{M_1}\right)^2 + \frac{2D^2}{M_2}\right)\right)$ and
\begin{equation}
  t = \Theta\left(\log\left(\sqrt{\frac{D}{\epsilon}}\right)\right)
\end{equation}
iterations of gradient descent.
\end{theorem}

In particular, our algorithm uses QSQs with different choices of discretization/truncation parameters for the different parts of Hallgren's algorithm (Appendix~\ref{sec:hallgren}).
Recall that Hallgren's algorithm has two subroutines: quantum Fourier sampling and a verification procedure.
In the quantum Fourier sampling part, we use QSQs with respect to discretization parameters $M_1, M_2$ and truncation parameter $R = \tilde{R}$.
For verification, we use discretization parameters $M_{1,m} \triangleq m M_1, M_{2,m} \triangleq mM_2$ and truncation parameter $R = \tilde{R}M_{1,m}$ for $m \in \{1,\dots, D\}$.

Recall that our target functions $g_{w^\star}(x)$ have the nice property that they are periodic with period $e_j/|w_j^\star| = e_j/w_j^\star$ (since $w_j > 0$), where $e_j$ is the unit vector for coordinate $j \in [d]$:
\begin{equation}
  g_{w^\star}\left(x + \frac{e_j}{w_j^\star}\right) = \tilde{g}\left(\left(x + \frac{e_j}{w_j^\star}\right)^\intercal w^\star\right) = \tilde{g}\left(x^\intercal w^\star + 1\right) = \tilde{g}(x^\intercal w^\star) = g_{w^\star}(x).
\end{equation}
The second to last equality holds because $\tilde{g}$ is periodic with period $1$.
In other words, each individual coordinate of $g_{w^\star}$ is periodic with period $1/w_j^\star$.
To quantumly learn $g_{w^\star}$, then we can first perform period finding to find $w^\star$ one component at a time.
Then, given the specific form of $\tilde{g}$ (\Cref{eq:g-tilde}), we can find the parameters $\beta_j^\star$, which can be done via gradient methods, as this is effectively a regression problem.

Despite the simplicity of this algorithm, there are several nontrivial issues that arise.
First, recall that the quantum example states are given by \Cref{eq:example-state}, which we reproduce below for convenience taking the distribution $\mathcal{D}$ to be the uniform distribution
\begin{equation}
  \ket{c^\star} = \frac{1}{\sqrt{|\mathcal{X}|}}\sum_{x \in \mathcal{X}} \ket{x}\ket{c^\star(x)}.
\end{equation}
We are given access to expectations with respect to these states.
In our case, notice that the target function $g_{w^\star}: \mathbb{R}^d \to [-1,1]$ takes inputs and outputs in a continuous and uncountably infinite space.
As we should not have a superposition over this large space, we must truncate and discretize our target function.
However, discretization can cause a loss of information about the period of the function, which is problematic.
Thus, it is important to choose the correct discretization in such a way that information about the period is sufficiently preserved.

In Appendix~\ref{sec:linear-uniform}, we discuss in detail how to apply period finding to our problem and mitigate these discretization issues.
In Appendix~\ref{sec:outer-uniform}, we show how one can use gradient descent to learn the outer function $\tilde{g}$ given knowledge of $w^\star$.

\subsection{Learning the linear function}
\label{sec:linear-uniform}

In this section, we discuss how to use period finding to learn the inner linear function, i.e., how to learn the vector of coefficients $w^\star$.
First, we need to suitably discretize $g_{w^\star}$ such that this discretization satisfies pseudoperiodicity (\Cref{def:pseudoperiod}).
In Appendix~\ref{sec:warmup-uniform}, we consider a simple special case to illustrate the idea behind the discretization and application of period finding for pedagogical purposes.
In Appendix~\ref{sec:general-uniform}, we prove the general case.
Throughout, we will assume uniform discretizations of the intervals in the sense that they will be equal size and not adaptively refined in any way.
The size of the discretization will be defined to be the number of bins in which the function is represented.

\subsubsection{Warmup}
\label{sec:warmup-uniform}

First, let us consider the case of $d = 1$, i.e., the input $x$ to the function and the unknown vector $w^\star$ are both real numbers instead of vectors.
We will later generalize this to the case of general $d \geq 1$.
The main simplifying assumption made in this section is that $1/w^\star \in \mathbb{Z}$ is an integer.
This will allow us to present this step of the algorithm without being hindered by too many approximations in the first instance.
If $1/w^\star \in \mathbb{Z}$, we have the following lemma, which tells us the correct discretization that satisfies pseudoperiodicity.

\begin{lemma}[Discretization; Simple Case]
\label{lem:discrete-simple-unif}
Let $w^\star \in \mathbb{R}$ be unknown with $1/w^\star \in \mathbb{Z}$ and $w^\star > 0$.
Let $g_{w^\star}:\mathbb{R} \to [-1,1]$ be defined as $g_{w^\star}(x) = \tilde{g}(x w^\star)$, where $\tilde{g}: \mathbb{R} \to [-1,1]$ is a function with period $1$ which has bounded variation on every finite interval.
Let $M \in \mathbb{Z}$ be the size of the discretization. Consider the discretized function $h_{w^\star, M} : \mathbb{Z} \to \frac{1}{M}\mathbb{Z}$ defined by
\begin{equation}
  h_{w^\star, M}(k) = \left\lfloor g_{w^\star}\left(\frac{k}{M}\right)\right\rfloor_M,
\end{equation}
where $\lfloor \cdot \rfloor_M$ denotes rounding down to the nearest multiple of $1/M$. Then, $h_{w^\star, M}$ is pseudoperiodic (in fact, periodic) with period $M/w^\star$.
\end{lemma}

Note that for our choice of $\tilde{g}$ as a linear combination of cosines (\Cref{eq:g-tilde}), the conditions in the lemma are clearly satisfied.

\begin{proof}
We need to show that $h_{w^\star, M}(k + [ \ell M/w^\star]) = h_{w^\star, M}(k)$ for all $\ell \in \mathbb{Z}$, where we use $[\cdot ]$ to denote one of either $\lfloor \cdot \rfloor$ or $\lceil \cdot \rceil$.
In fact, because $1/w^\star \in \mathbb{Z}$ and $\ell, N \in \mathbb{Z}$, then $[\ell M /w^\star] = \ell M /w^\star$.
Since $\tilde{g}$ is a bounded variation function, it is equal everywhere to its Fourier series expansion:
\begin{equation}
  \tilde{g}(x) = \sum_{z \in \mathbb{Z}} a_z e^{2\pi i zx}.
\end{equation}
Using this to expand out $h_{w^\star, M}(k + \ell M /w^\star)$, we have
\begin{align}
h_{w^\star, M}(k + \ell M /w^\star) &= \left\lfloor g_{w^\star}\left(\frac{k}{M} + \frac{\ell M/w^\star}{M}\right)\right\rfloor_M\\
&= \left\lfloor \tilde{g}\left(\left(\frac{k}{M} + \frac{\ell}{w^\star}\right)w^\star\right)\right\rfloor_M\\
&= \left\lfloor \sum_{z \in \mathbb{Z}} a_z \exp\left(2\pi i z\left(\frac{kw^\star}{M} + \ell\right)\right) \right\rfloor_M\\
&= \left\lfloor \sum_{z \in \mathbb{Z}} a_z \exp\left(\frac{2\pi i z kw^\star}{M}\right) \exp\left(2\pi iz \ell\right) \right\rfloor_M\\
&= \left\lfloor \sum_{z \in \mathbb{Z}} a_z \exp\left(\frac{2\pi i z kw^\star}{M}\right) \right\rfloor_M\\
&= \left\lfloor \tilde{g}\left(\frac{kw^\star}{M}\right)\right\rfloor_M\\
&= \left\lfloor g_{w^\star}\left(\frac{k}{M}\right)\right\rfloor_M\\
&= h_{w^\star, M}(k).
\end{align}
Here, in the fifth line, we use that $z, \ell \in \mathbb{Z}$ so that $\exp(2\pi i z \ell) = 1$.
This gives the claim.
\end{proof}

Now, we have a suitable discretization, but we also need to truncate the domain of the function so that it is not all of $\mathbb{Z}$.
Let $R$ be this truncation parameter.
Then, as defined in \Cref{def:qsq}, we should have QSQ access to the quantum example state with respect to the truncated and discretized target function, i.e.,
\begin{equation}
  \label{eq:trunc}
  \ket{h_{w^\star}} = \frac{1}{(2R)^{d/2}}\sum_{x_1,\dots, x_d = -R}^{R-1} \ket{x}\ket{h_{w^\star, M}(x)}.
\end{equation}

Before proving our guarantee on learning $w^\star$, first notice that by the definition of $\mathcal{S}_w$ (\Cref{eq:sw}), we have an upper bound on the period of $h_{w^\star, M}$.
In particular, $w_j^\star \geq R_w/d^2$.
Then, the period satisfies
\begin{equation}
  \frac{M e_j}{w^\star_j} \leq \frac{Me_jd^2}{R_w}.
\end{equation}
This is useful for choosing our truncation parameter $R$ in the following result.

\begin{prop}[Linear Function Guarantee; Simple Uniform Case]
  \label{prop:warmup-linear-uniform}
  Let $\varphi^2$ be the uniform distribution.
  Let $\tau \geq 0$. Let $w^\star \in \mathbb{R}^d$ be unknown with $1/w^\star_j \in \mathbb{Z}$ for all $j \in [d]$.
  Also suppose that $w_j^\star \geq R_w/d^2$, for all $j \in [d]$.
  Let $g_{w^\star}: \mathbb{R}^d \to [-1,1]$ be defined as $g_{w^\star}(x) = \tilde{g}(x^\intercal w^\star)$, where $\tilde{g}:\mathbb{R} \to [-1,1]$ is a function with period $1$ which has bounded variation on every finite interval.
  Then, for any choice of discretization parameter $M \geq 1$ and truncation parameter $R \geq (1+2\tau)M^2 d^{4}/R_w^2$, there exists a quantum algorithm that learns $w^\star$ exactly with constant probability using
  \begin{equation}
    N = d
  \end{equation}
  quantum statistical queries with tolerance $\tau$ (with respect to the discretized and truncated example state).
\end{prop}

\begin{proof}
We first consider the case of $d =1$.
Consider the state in \Cref{eq:trunc}, with respect to which we have access to expectations.
By \Cref{lem:discrete-simple-unif}, we know that $h_{w^\star, M}$ is periodic with period $M/w^\star$ via our choice of discretization.
With the truncation, $h_{w^\star, M}$ is still periodic with period $M/w^\star$.
Moreover, since $1/w^\star \in \mathbb{Z}$ in this simple case, then the period is an integer.
Thus, we can simply apply standard period finding~\cite{shor1994algorithms} to solve for $M/w^\star$, i.e., apply the quantum Fourier transform (QFT) and measure.
We can encode this algorithm into an observable $O$ with $\norm{O} \leq 1$ as follows
\begin{equation}
  \label{eq:o}
  O = \left(\mathsf{QFT}_q^{-1} \sum_{\ell \in [M]} \frac{\ell}{M}\ketbra{\ell} \mathsf{QFT}_q\right) \otimes I.
\end{equation}
Here, $\mathsf{QFT}_q$ denotes the QFT in a dimension of size $q = 2R$ (since the input $x$ is between $-R$ and $R$ due to truncation), and $I$ is the identity operator acting on the qubits encoding the output $h_{w^\star, M}(x)$.
$O$ is simply applying a QFT on the first register and measuring these qubits with proper normalization factors to ensure that $\norm{O} \leq 1$.
By a standard analysis of the period finding algorithm (see, e.g.,~\cite{nielsen2010quantum}), if our QSQs were noiseless ($\tau = 0$), this allows us to recover the period $M/w^\star$ exactly with constant success probability using only one (noiseless) QSQ.

However, we consider the case of a general noise tolerance $\tau \geq 0$ for our QSQs.
By the standard analysis of period finding, with constant probability, the output of the QSQ is some number $\alpha$ such that $|\alpha - y| \leq \tau$, where $|y - kRw^\star/M| \leq 1/2$ for some integer $k \geq 0$.
By the reverse triangle inequality, this implies that $|\alpha - kR w^\star/M| \leq \tau + 1/2$.
Then, dividing by $R$, we see that
\begin{equation}
  \label{eq:close}
  \left| \frac{\alpha}{R} - \frac{kw^\star}{M}\right| \leq \frac{\tau + 1/2}{R}.
\end{equation}
Notice that $k w^\star/M$ can be thought of as a fraction with denominator $M/w^\star$ (since $1/w^\star \in \mathbb{Z}$ in this case), which is the period of our target function.
Let $A \triangleq Md^{2}/R_w$ be an upper bound on the period $M/w^\star$.
Then, $kw^\star/M$ is a fraction with denominator less than $A$.
Two distinct rational numbers with denominator less than $A$ must be at least a distance of $1/A^2 \geq (1+2\tau)/R$ apart, where the inequality comes from our choice of $R \geq (1+2\tau)A^2$.
Thus, \Cref{eq:close} implies that there exists a unique fraction $kw^\star/M$ that is determined by $\alpha/R$.
Moreover, by our choice of $R$ again,
\begin{equation}
  \left| \frac{\alpha}{R} - \frac{kw^\star}{M}\right| \leq \frac{\tau + 1/2}{R} \leq \frac{1}{2A^2}.
\end{equation}
Thus, by standard results for the continued fractions expansion~\cite{schrijver1998theory}, we can recover the unique $kw^\star/M$ from $\alpha/R$.
The rest of the analysis follows in the same way as the usual period finding algorithm.
This tells us that we can recover the period $M/w^\star$ exactly with constant probability.
Moreover, from $M/w^\star$, we can recover $w^\star$ exactly as well since $M$ is known.

Thus far in this section, we have only considered the case of $d = 1$.
Our above discussion is easily generalized to arbitrary $d \geq 1$.
In particular, our simplifying assumption is now that $1/w^\star_j \in \mathbb{Z}$ for all $j \in [d]$.
Recall that for general $d$, our target function is $g_{w^\star}(x) = \tilde{g}(x^\intercal w^\star)$, where $\tilde{g}: \mathbb{R} \to [-1,1]$ is again a function with period $1$ and now $w^\star \in \mathbb{R}^d$.
Then, we can define the discretized function as before but this time $h_{w^\star, M}: \mathbb{Z}^d \to \frac{1}{M}\mathbb{Z}$.
By essentially the same proof as \Cref{lem:discrete-simple-unif}, one can show that $h_{w^\star, M}$ is periodic with period $Me_j/w^\star_j$, where $e_j$ is the unit vector for coordinate $j \in [d]$.
Then, we can perform period finding one coordinate at a time, encoding in the QSQ operator
\begin{equation}
  \label{eq:oj}
  O_j = \left(\mathsf{QFT}_{j,q}^{-1} \sum_{\ell \in [M]} \frac{\ell}{M}\ketbra{\ell} \mathsf{QFT}_{j,q}\right)\otimes I_{-j} \otimes I,
\end{equation}
where $\mathsf{QFT}_{j,q}$ denotes the QFT in a dimension $q = 2R$ acting only on the qubits that encode the $j$th coordinate of the input $x_j$.
Also, $I_{-j}$ is the identity operator acting on all qubits that do not encode $x_j$ and $I$ is the identity operator acting on the qubits encoding the output of the function.
By the same argument as before, applying $O_j$ for each coordinate $j \in [d]$, we can recover the whole vector $w^\star$ exactly.
Moreover, this uses $N = d$ QSQs to learn $w^\star$ exactly with constant probability.
\end{proof}

From this learned $w^\star$, we can use classical gradient methods to learn the $\tilde{g}$, assuming it has the form given in \Cref{eq:g-tilde}.
This is discussed in Appendix~\ref{sec:outer-uniform}.

\subsubsection{General Case}
\label{sec:general-uniform}

In the previous section, we proved that $w^\star$ can be recovered exactly in a simple case.
We presented this first to give an overview of the algorithm without becoming overwhelmed by the technicalities involved for the general case.
In this section, we prove the general case, with the formal statement given below.

\begin{theorem}[Linear Function Guarantee; Uniform Case]
\label{thm:linear-uniform}
Let $\varphi^2$ be the uniform distribution.
Let $1 > \epsilon_1 > 0, \delta > 0, \tau \geq 0$.
Let $w^\star \in \mathbb{R}^d$ be unknown with norm $R_w > 0$ and $w_j^\star \geq R_w/d^2$, for all $j \in [d]$.
Let $g_{w^\star}: \mathbb{R}^d \to [-1,1]$ be defined as $g_{w^\star}(x) = \tilde{g}(x^\intercal w^\star)$, where $\tilde{g}: \mathbb{R} \to [-1,1]$ is given in \Cref{eq:g-tilde}.
Consider parameters $M_1 = \lceil\max(70 \pi d D^3 R_w, R_w^2/\epsilon
_1)\rceil$, $M_2 = c M_1$, where $c$ is any constant such that $M_2$ is an integer and $c < 1/(8\pi D R_w)$, and
\begin{equation}
    \tilde{R} = \tilde{\Omega}\left(\max\left(\frac{\tau M_1^2d^4}{R_w^2}, \frac{D^2}{\epsilon}, \frac{D^2\sqrt{d}}{R_w \epsilon}, \frac{D^{5/2}}{\sqrt{\epsilon}}, \frac{D^{3/2}\sqrt{d}}{R_w \sqrt{\epsilon}}\right)\right).
\end{equation}
Suppose we have QSQ access (\Cref{def:qsq}) with respect to discretization parameters $M_{1,m} \triangleq m M_1$, $M_{2,m} \triangleq m M_2$, and a truncation parameter $R \geq \tilde{R}$, for $m \in \{1,\dots, D\}$.
Then, there exists a quantum algorithm with this QSQ access that can learn an approximation $\hat{w}$ of $w^\star$ such that $\norm{\hat{w} - w^\star}_\infty \leq \epsilon_1$ with probability at least $1-\delta$ using 
\begin{equation}
  N = \mathcal{O}\left(d D \log\left(\frac{1}{\delta}\right) \log^5\left(\frac{M_1 d^2}{R_w}\right)\right)
\end{equation}
quantum statistical queries with tolerance $\tau \leq \min\left(\frac{1}{M_2^2}\left(\frac{7}{40D} - \frac{1}{M_2}\right), \frac{1}{2D^2M_2^2}\left(\frac{2}{15} - \frac{1}{8}\left(\frac{2\pi R_w}{M_1}\right)^2 + \frac{2D^2}{M_2}\right)\right)$.
\end{theorem}

As stated before, our algorithm has two subroutines as in Hallgren's algorithm: quantum Fourier sampling and the verification procedure.
For quantum Fourier sampling, we use QSQs with respect to discretization parameters $M_1, M_2$ and truncation parameter $R = \tilde{R}$.
For verification, we use discretization parameters $M_{1,m} \triangleq m M_1, M_{2,m} \triangleq mM_2$ and truncation parameter $R = \tilde{R}M_{1,m}$ for $m \in \{1,\dots, D\}$.

There are two main differences with the presentation in Appendix~\ref{sec:warmup-uniform}.
First, in \Cref{lem:discrete-simple-unif}, it was fortuitous that the pseudoperiodicity property (\Cref{def:pseudoperiod}) required of the discretization turned out to simply be periodicity under our simplifying assumption.
However, in general, this is not the case, so we will need to prove a new version of \Cref{lem:discrete-simple-unif}.
Second, even with the discretization, the period of the discretized function may not be an integer in general.
Thus, the standard period finding algorithm~\cite{shor1994algorithms} does not apply.
Instead, we turn to a subroutine of Hallgren's algorithm~\cite{hallgren2007polynomial} which performs irrational period finding for pseudoperiodic functions.
This has some additional conditions that we must fulfill, as discussed in Appendix~\ref{sec:hallgren}.

To address the first point, we have the following lemma, which is a generalization of \Cref{lem:discrete-simple-unif}.
Again, we first consider $d = 1$ and later generalize to $d \geq 1$.
In this case, we need two discretization parameters: one to control the fineness of the discretization of the input to $g_{w^\star}$ and another to control the outer rounding.
We consider the case when the latter is more coarse than the former to obtain pseudoperiodicity.

\begin{lemma}[Discretization; General Case]
\label{lem:discrete-general-unif}
Let $w^\star \in \mathbb{R}$ be unknown with $|w^\star| \leq R_w$ for some $R_w > 0$.
Let $g_{w^\star}:\mathbb{R} \to [-1,1]$ be defined as $g_{w^\star}(x) = \tilde{g}(x w^\star)$, where $\tilde{g}: \mathbb{R} \to [-1,1]$ is a function with period $1$ which has bounded variation on every finite interval, is given by a trigonometric polynomial of degree at most $D$, and is $\lambda$-Lipschitz.
Let $M \geq 1$ and consider discretization parameters $M_1 = M$, $M_2 = c M$, where $c$ is any constant such that $M_2$ is an integer and $c < 1/(4\lambda R_w)$.
Consider the discretized function $h_{w^\star, M_1, M_2} : \mathbb{Z} \to \frac{1}{M_2}\mathbb{Z}$ defined by
\begin{equation}
  h_{w^\star, M_1, M_2}(k) = \left\lfloor g_{w^\star}\left(\frac{k}{M_1}\right)\right\rfloor_{M_2},
\end{equation}
where $\lfloor \cdot \rfloor_{M_2}$ denotes rounding down to the nearest multiple of $1/M_2$. Then, $h_{w^\star, M_1, M_2}$ is $(1-4DR_w/M)$-pseudoperiodic with period $M_1/w^\star$.

In particular, when $\tilde{g}$ is given by \Cref{eq:g-tilde} and $M \geq 70 \pi d D^3 R_w$, then for discretization parameters $M_1 = M$, $M_2 = c M$ with $c < 1/(8\pi D R_w)$, then $h_{w^\star, M_1, M_2}$ is $(33/35)$-pseudoperiodic with period $M_1/w^\star$.
\end{lemma}

\begin{proof}
We prove the first statement first, so we want to show that $h_{w^\star, M_1, M_2}$ is $(1-4DR_w/M)$-pseudoperiodic.
In other words, we want to show that $h_{w^\star, M_1, M_2}(k + \lfloor \ell M_1 / w^\star\rfloor)$ or $h_{w^\star, M_1, M_2}(k + \lceil \ell M_1 / w^\star \rceil)$ equals $h_{w^\star, M_1, M_2}(k)$ for at least a $(1-4DR_w/M)$-fraction of the inputs $k$, for all $\ell \in \mathbb{Z}$.
Fixing some $\ell \in \mathbb{Z}$, denote
\begin{equation}
  h_+(k) \triangleq h_{w^\star, M_1, M_2}\left(k + \left\lceil \frac{\ell M_1}{w^\star}\right\rceil\right),\quad h_-(k) \triangleq h_{w^\star, M_1, M_2}\left(k + \left\lfloor \frac{\ell M_1}{w^\star}\right\rfloor\right).
\end{equation}
Also fix some input $k$ such that $0 \leq k \leq \lfloor M / w^\star \rfloor$.
Suppose for now that $g_{w^\star}$ is monotonically increasing in the interval $( (k-1)M_1, (k+1)M_1)$.
We will show that in this case, either $h_+(k) = h_{w^\star, M_1, M_2}(k)$ or $h_-(k) = h_{w^\star, M_1, M_2}(k)$.
We have the following upper bound on $h_+(k)$:
\begin{align}
  h_+(k) &= h_{w^\star, M_1, M_2}\left(k + \left\lceil \frac{\ell M_1}{w^\star} \right\rceil\right)\\
  &= \left\lfloor g_{w^\star}\left(\frac{k + \left\lceil \frac{\ell M_1}{w^\star} \right\rceil}{M_1}\right) M_2\right\rfloor/ M_2\\
  &= \left\lfloor g_{w^\star}\left(\frac{k + \frac{\ell M_1}{w^\star} + \Delta}{M_1}\right) c M_1\right\rfloor/ (c M_1)\\
  &= \left\lfloor g_{w^\star}\left(\frac{k + \Delta}{M_1}\right) c M_1\right\rfloor/ (c M_1)\\
  &\leq \left\lfloor g_{w^\star}\left(\frac{k}{M_1}\right) c M_1 + \Delta \lambda w^\star c\right\rfloor/ (c M_1)\\
  &\leq \left\lfloor g_{w^\star}\left(\frac{k}{M_1}\right) c M_1 +\lambda w^\star c\right\rfloor/ (c M_1).
\end{align}
Here, the first line follows by the definition of $h_+(k)$.
The second line follows by the definition of $h_{w^\star, M_1, M_2}$.
In the third line, we define $\Delta$ such that $0 \leq \Delta < 1$ and use $M_2 = c M_1$.
In the fourth line, we use that $g_{w^\star}$ has period $1/w^\star$.
In the fifth line, because $\tilde{g}$ is $\lambda$-Lipschitz, then $g_{w^\star}$ is $(\lambda w^\star)$-Lipschitz.
Finally, in the last line, we use that $\Delta < 1$.
We can also lower bound $h_+(k)$ using the fourth line of the above calculation and our assumption that $g_{w^\star}$ is monotonically increasing.
\begin{equation}
  h_+(k) = \left\lfloor g_{w^\star}\left(\frac{k + \Delta}{M_1}\right) c M_1\right\rfloor/ (c M_1) \geq \left\lfloor g\left(\frac{k}{M_1}\right) c M_1 \right\rfloor/(c M_1) = h_{w^\star, M_1, M_2}(k).
\end{equation}
Similarly, one can show that
\begin{equation}
  \left\lfloor g_{w^\star}\left(\frac{k}{M_1}\right) c M_1 -\lambda w^\star c\right\rfloor/ (c M_1) \leq h_-(k) \leq h_{w^\star, M_1, M_2}(k).
\end{equation}
Using that $c < 1/(4\lambda R_w)$, then
\begin{equation}
  2\lambda w^\star c \leq 2\lambda R_w c \leq \frac{1}{2}.
\end{equation}
Then,
\begin{align}
  |h_+(k) - h_-(k)| &\leq \left| \frac{\lfloor g_{w^\star}(k/M_1)c M_1 + \lambda w^\star c \rfloor}{c M_1} -  \frac{\lfloor g_{w^\star}(k/M_1)c M_1 - \lambda w^\star c \rfloor}{c M_1}\right|\\
  &= \left| \frac{g_{w^\star}(k/M_1)c M_1 + \lambda w^\star c - \Delta_+ -  g_{w^\star}(k/M_1)c M_1 + \lambda w^\star c + \Delta_-}{c M_1}\right|\\
  &= \left|\frac{2\lambda w^\star c + \Delta_- + \Delta_+}{c M_1} \right|\\
  &\leq \frac{3}{2c M_1},
\end{align}
where in the last line we use that $0 \leq \Delta_-, \Delta_+ < 1$.
Because the outputs of $h_+(k)$ and $h_-(k)$ are discretized in steps of $1/(c M_1)$, this implies that
\begin{equation}
  |h_+(k) - h_-(k)| \leq \frac{1}{c M_1} = \frac{1}{M_2}.
\end{equation}
Moreover, by the above work, we know that $h_-(k) \leq h_{w^\star, M_1, M_2}(k) \leq h_+(k)$.
Thus, because all three functions have outputs discretized in units of $1/(c M_2)$, it follows that either $h_-(k) = h_{w^\star, M_1, M_2}(k)$ or $h_+(k) = h_{w^\star, M_1, M_2}(k)$.
A similar argument holds when $g_{w^\star}$ is instead assumed to be monotonically decreasing in the interval $( (k-1)/M_1, (k+1)/M_1)$.

Thus, we have shown that if $g_{w^\star}$ is monotone, then $h_{w^\star, M_1, M_2}$ satisfies the property required for pseudoperiodicity.
It suffices to show that $g_{w^\star}$ is monotone in regions $( (k -1)/M_1, (k+1)/M_1)$ for all except a $4DR_w/M$-fraction of the inputs $k$ within a single period $0 \leq k \leq \lfloor M_1/w^\star \rfloor$.
Note that these intervals are just neighborhoods of size $2/M_1$ centered around some $k/M_1$ with $0 \leq k \leq \lfloor M_1/w^\star \rfloor$.
Thus, we can instead consider neighborhoods of size $2/M_1$ around points $k$ with $0 \leq k \leq 1/w^\star$, i.e., within a single period of $g_{w^\star}$.
Note that $g_{w^\star}$ will be monotone in the interval unless it contains a critical point.
Thus, it remains to consider neighborhoods of the critical points of $g_{w^\star}$.

By assumption, $\tilde{g}$ is a trigonometric polynomial of degree at most $D$.
Moreover, it is known that trigonometric polynomials with degree at most $D$ have at most $2D$ zeroes in a single period (see, e.g., Chapter 13 of~\cite{powell1981approximation}).
The derivative of a trigonometric polynomial with degree at most $D$ is clearly still a trigonometric polynomial of degree at most $D$.
Thus, $\tilde{g}$ must have at most $2D$ critical points in a single period.
The same holds for $g_{w^\star}$ since $w^\star \neq 0$.

Recall that the period of $g_{w^\star}$ is $1/w^\star$, so there are at most $1/w^\star$ integer values to consider within one period of $g_{w^\star}$.
Thus, there are at most $M_1/(2w^\star)$ intervals of size $2/M_1$ around these $1/w^\star$ values.
Now, since there are at most $2D$ critical points, at most $2D$ of these intervals contain a critical point.
Hence, the proportion of intervals (and hence inputs) for which $g_{w^\star}$ will not be monotone is at most
\begin{equation}
  \frac{2D}{M_1/(2w^\star)} = \frac{4Dw^\star}{M_1} \leq \frac{4D R_w}{M_1}.
\end{equation}
Outside of this proportion, we have already shown that $h_{w^\star, M_1, M_2}$ is pseudoperiodic.
Thus, we can conclude that $h_{w^\star, M_1, M_2}$ is $(1-4DR_w/M_1)$-pseudoperiodic.

In the specific case where $\tilde{g}$ is given by \Cref{eq:g-tilde}, $\tilde{g}$ clearly has bounded variation on every finite interval and is a trigonometric polynomial of degree at least $D$.
Moreover, it is $\lambda$-Lipschitz with $\lambda = 2\pi D$:
\begin{equation}
  \tilde{g}'(y) = -\sum_{j=1}^D \beta_j^\star \sin(2\pi j y) \cdot 2\pi j
\end{equation}
\begin{equation}
  |\tilde{g}'(y)| \leq 2\pi \left|\sum_{j=1} j \beta_j^\star \sin(2\pi j y)\right| \leq 2\pi D \sum_{j=1}^D |\beta_j^\star| = 2\pi D,
\end{equation}
where we used that $j \leq D$ and $\norm{\beta^\star}_1 = 1$.
Thus, we can apply the result we just proved for the case of $\lambda = 2\pi D$.
Consider $M \geq 70 \pi d D^3 R_w$, and take $M_1 = M$ and $M_2 = c M$ for $c < 1/(8\pi D R_w)$.
Then,
\begin{equation}
  \frac{4DR_w}{M_1} \leq \frac{4 DR_w}{70 D R_w} = \frac{2}{35}.
\end{equation}
Hence, $h_{w^\star, M_1, M_2}$ is $(33/35)$-pseudoperiodic.
\end{proof}

Thus, we see that this discretization still contains information about the period of the original function.
We also truncate the domain of the function as well with truncation parameter $R$.
Then, for $d =1$, we require QSQ access to
\begin{equation}
    \ket{h_{w^\star}} = \frac{1}{\sqrt{2R}}\sum_{x=-R}^{+R-1}\ket{x}\ket{h_{w^\star, M_1, M_2}(x)}.
\end{equation}
With this truncation, $h_{w^\star, M_1, M_2}$ is still $(33/35)$-pseudoperiodic with period $M_1/w^\star$.
Moreover, \Cref{lem:discrete-general-unif} also implies that for $d \geq 1$, $h_{w^\star, M_1, M_2}$ is $(33/35)$-pseudoperiodic in each coordinate with period $M_1/w^\star_j$.
We discuss this in more detail later.
However, note that because $1/w^\star_j$ is not necessarily an integer, the period $M_1/w^\star_j$ may also not be an integer.
Hence, the standard period finding algorithm~\cite{shor1994algorithms} does not apply.
Instead, we want use an irrational period finding algorithm~\cite{hallgren2007polynomial}, which works even if $M_1/w^\star_j$ is irrational.
We review Hallgren's algorithm in Appendix~\ref{sec:hallgren}.

\Cref{lem:discrete-general-unif} guarantees that $h_{w^\star, M_1, M_2}$ is $\eta$-pseudoperiodic with $\eta = 33/35$.
Moreover, note that \Cref{thm:hallgren} requires an upper bound on the period, which we have because $w_j \geq R_w/d^2$ by the definition of $\mathcal{S}_w$.
The final condition of \Cref{thm:hallgren} that we need is this verification procedure to check if a given $T$ is close to an integer multiple of the true period.
We design such a verification procedure in \Cref{alg:verification} and analyze it in \Cref{thm:verification}.
Note that in \Cref{alg:verification}, we must restrict the noise tolerance of our QSQs to be inverse polynomial in some of our parameters.
Classically, the hardness results have access to gradients that are exponentially accurate, so requiring the tolerance parameter to scale inverse polynomially is not particularly strong.

\begin{algorithm}
   \caption{Verification Procedure; Uniform Case} 
   \label{alg:verification}
   \begin{algorithmic}[1]
   \State Choose parameters $M_1 = \max(70\pi d D^3 R_w, R_w^2/\epsilon_1)$, $M_2 = c M_1$ for some $c$ such that $M_2 \in \mathbb{Z}$ and $c < 1/(8\pi D R_w)$, and $\tilde{R} = \tilde{\Omega}\left(\max\left(\frac{D^2}{\epsilon}, \frac{D^2\sqrt{d}}{R_w \epsilon}, \frac{D^{5/2}}{\sqrt{\epsilon}}, \frac{D^{3/2}\sqrt{d}}{R_w \sqrt{\epsilon}}\right)\right)$.
   \State For $m \in \{1,\dots, D\}$, query the QSQ oracle with observable $O_{k,m}$ (defined in \Cref{eq:ov}), discretization parameters $M_{1,m} \triangleq m M_1$, $M_{2,m} \triangleq mM_2$, truncation parameter $R \triangleq \tilde{R}M_{1,m}$, and tolerance $\tau \leq \min\left(\frac{1}{M_2^2}\left(\frac{7}{40D} - \frac{1}{M_2}\right), \frac{1}{2D^2M_2^2}\left(\frac{2}{15} - \frac{1}{8}\left(\frac{2\pi R_w}{M_1}\right)^2 + \frac{2D^2}{M_2}\right)\right)$ to obtain values $\alpha_m$.
   \State Check if $\alpha_1 \geq \frac{1}{M_2^2}\left(\frac{21}{40D} - \frac{3}{M_2}\right)$.
   \State Check if $\sum_{m=1}^D \alpha_m \leq \frac{1}{M_2^2}\left(\frac{20}{39}D + \frac{1}{2D}\left(\frac{2}{15} - \frac{1}{8}\left(\frac{2\pi R_w}{M_1}\right)^2 +\frac{2D^2}{M_2}\right)\right)$.
   \State \Return ``yes'' iff both conditions in Steps 3 and 4 are satisfied.
   \end{algorithmic}
\end{algorithm}

\begin{theorem}[Verification Procedure; Uniform Case]
    \label{thm:verification}
    Let $\varphi^2$ be the uniform distribution.
    Let $1 > \epsilon_1 > 0$.
    Let $w^\star \in \mathbb{R}^d$ be unknown with norm $R_w > 0$ and $w_j^\star \geq R_w/d^2$ for all $j \in [d]$.
    Let $g_{w^\star} : \mathbb{R}^d \to [-1,1]$ be defined as $g_{w^\star}(x) = \tilde{g}(x^\intercal w^\star)$ for $\tilde{g}$ given in \Cref{eq:g-tilde}.
    Consider parameters $M_1 = \max(70\pi d D^3 R_w, R_w^2/\epsilon_1)$, $M_2 = c M_1$ for some constant $c$ such that $c < 1/(8\pi D R_w)$ and $M_2 \in \mathbb{Z}$, and 
    \begin{equation}
        \tilde{R} = \tilde{\Omega}\left(\max\left(\frac{D^2}{\epsilon}, \frac{D^2\sqrt{d}}{R_w \epsilon}, \frac{D^{5/2}}{\sqrt{\epsilon}}, \frac{D^{3/2}\sqrt{d}}{R_w \sqrt{\epsilon}}\right)\right).
    \end{equation}
    Suppose we have QSQ access (see \Cref{def:qsq}) with respect to discretization parameters $M_{1,m} \triangleq m M_1$, $M_{2,m} \triangleq mM_2$ and truncation parameter $R \triangleq \tilde{R}M_{1,m}$ for $m \in \{1,\dots, D\}$.
    Then, given an integer $T$ and $k \in [d]$, \Cref{alg:verification} can check whether or not $|T - \frac{\ell M_1}{w_k^\star}| \leq 1$ for some integer $\ell$ using $D$ QSQs with tolerance $\tau \leq \min\left(\frac{1}{M_2^2}\left(\frac{7}{40D} - \frac{1}{M_2}\right), \frac{1}{2D^2M_2^2}\left(\frac{2}{15} - \frac{1}{8}\left(\frac{2\pi R_w}{M_1}\right)^2 + \frac{2D^2}{M_2}\right)\right)$.
\end{theorem}

\begin{proof}
Explicitly, the example state for our QSQ access is
\begin{equation}
  \label{eq:example-detail}
  \ket{h_{w^\star, M_{1,m}, M_{2,m}}} = \frac{1}{\sqrt{(2\tilde{R}M_{1,m})^d}} \sum_{x_1,\dots, x_d = -\tilde{R}M_{1,m}}^{\tilde{R}M_{1,m} - 1} \ket{x}\ket{h_{w^\star, M_{1,m}, M_{2,m}}(x)},
\end{equation}
where $h_{w^\star, M_{1,m}, M_{2,m}}$ is a discretization of $g_{w^\star}$ from \Cref{lem:discrete-general-unif}.
We query $D$ QSQs, each with the different parameters indexed by $m$ as specified previously.

The main idea behind our verification procedure is to compute the inner product between $h_{w^\star, M_{1,m}, M_{2,m}}$ and this function with its input shifted by the guess $T$ for the period.
This inner product should be large for a good guess.
The technical work behind this theorem goes into defining an observable to approximate this inner product and finding a suitable threshold for the inner product to surpass such that $T$ is close to the true period.

Consider defining the observable
\begin{equation}
  \label{eq:Am}
  A_m \triangleq I \otimes 2\ketbra{-} \otimes \left(\frac{1}{M_{2,m}^2} \sum_{i,j=0}^{M_{2,m} - 1} ij \ketbra{i}{j}\right),
\end{equation}
where the identity is on the first $\log(\tilde{R}M_{1,m}) + d$ qubits (the extra $d$ qubits are to represent the sign of each entry of $x_1,\dots, x_d$).
Also define an operator $S_{k,a}$ that cyclically shifts the $k$th entry of the input register by $a$.
In particular, this acts as
\begin{equation}
  S_{k,a} : \ket{x}\ket{h_{w^\star, M_{1,m}, M_{2,m}}(x)} \mapsto \ket{x + ae_k}\ket{h_{w^\star, M_{1,m}, M_{2,m}}(x)},
\end{equation}
where we use $e_k$ to denote the unit vector with a one in the $k$th coordinate and zeros elsewhere.
Then, we query the following observable as our QSQ to verify the period of the $k$th coordinate:
\begin{equation}
  \label{eq:ov}
  O_{k,m} \triangleq A_m S_{k, -T}.
\end{equation}
First, we claim that this observable does indeed reflect our idea about computing the inner product between $h_{w^\star, M_{1,m}, M_{2,m}}$ and this function with its input shifted by $T$.

\begin{claim}[Approximating inner product]
\label{claim:inner-prod}
For $m \in \{1,\dots, D\}$, consider parameters $M_{1,m}, M_{2,m}$ as defined above.
Also consider a parameter $\tilde{R}$ and an observable $O_{k,m}$ as defined above.
Then, the expectation value of $O_m$ with respect to the example state in \Cref{eq:example-detail} is given by
\begin{align}
  &\expval{O_{k,m}}{h_{w^\star, M_{1,m}, M_{2,m}}}\\
  &= \frac{1}{(2\tilde{R}M_{1,m})^d M_{2,m}^2} \sum_{x_1,\dots, x_d = -\tilde{R}M_{1,m}}^{\tilde{R}M_{1,m}-1} h_{w^\star, M_{1,m}, M_{2,m}}(x) h_{w^\star, M_{1,m}, M_{2,m}}(x + Te_k),
\end{align}
where $e_k$ denotes the unit vector with a single one in the $k$th coordinate.
\end{claim}

\begin{proof}[Proof of \Cref{claim:inner-prod}]
This follows by a simple calculation.
\begin{align}
  &\expval{O_{k,m}}{h_{w^\star, M_{1,m}, M_{2,m}}}\\
  &=\expval{A_m S_{k,-T}}{h_{w^\star, M_{1,m}, M_{2,m}}}\\
  &= \frac{1}{(2\tilde{R} M_{1,m})^d} \left(\sum_{x_1,\dots, x_d = -\tilde{R}M_{1,m}}^{\tilde{R}M_{1,m}-1} \bra{x}\bra{h_{w^\star, M_{1,m}, M_{2,m}}(x)}\right) A_m \left(\sum_{x_1',\dots, x_d' = -\tilde{R}M_{1,m}}^{\tilde{R}M_{1,m}-1} \ket{x' - Te_k}\ket{h_{w^\star, M_{1,m}, M_{2,m}}(x')}\right)\\
  &= \frac{1}{(2\tilde{R}M_{1,m})^d} \sum_{\substack{x_1,\dots, x_d = -\tilde{R}M_{1,m} \\ x_1',\dots, x_d' = -\tilde{R}M_{1,m}}}^{\tilde{R}M_{1,m}-1}\bra{x}\bra{h_{w^\star, M_{1,m}, M_{2,m}}(x)} A_m \ket{x'} \ket{h_{w^\star, M_{1,m}, M_{2,m}}(x' + Te_k)}\\
  &= \frac{1}{(2\tilde{R}M_{1,m})^d} \sum_{x_1,\dots, x_d =-\tilde{R}M_{1,m}}^{\tilde{R}M_{1,m} - 1} \bra{h_{w^\star, M_{1,m}, M_{2,m}}(x)}\left(2 \ketbra{-} \otimes \frac{1}{M_{2,m}^2} \sum_{i,j=0}^{M_{2,m}-1} ij\ketbra{i}{j}\right)\ket{h_{w^\star, M_{1,m}, M_{2,m}}(x + Te_k)}\\
  &= \frac{1}{(2\tilde{R}M_{1,m})^d} \sum_{x_1,\dots, x_d=-\tilde{R}M_{1,m}}^{\tilde{R}M_{1,m}-1} h_{w^\star, M_{1,m}, M_{2,m}}(x) h_{w^\star, M_{1,m}, M_{2,m}}(x + Te_k).
\end{align}
In the second line, we use the definition of $O_{k,m}$.
In the third line, we use the definition of $S_{k,-T}$.
In the fourth line, we relabel the $x'$ indices in the summation $x' \mapsto x' - Te_k$.
This still results in summing over the same values because $S_{k,-T}$ is defined to be a cyclical shift.
In the fifth line, we use the definition of $A_m$ and collapse the second summation by evaluating $\braket{x}{x'}$.
In the last line, we use the following calculation.
For any two computational basis states $\ket{a}, \ket{b}$, where $a,b \in \{0,\dots, M_{2,m} - 1\}$, it is clear that
\begin{equation}
  \bra{a}\left(\frac{1}{M_{2,m}^2} \sum_{i,j=0}^{M_{2,m} - 1} ij\ketbra{i}{j}\right)\ket{b} = \frac{1}{M_{2,m}^2} ab.
\end{equation}
Similarly, if $\ket{a}, \ket{b}$ are instead representations of numbers in $[-1,1]$ using $\log(M_{2,m}) + 1$ bits, where the first qubit encodes the sign, then 
\begin{equation}
  \bra{a}\left(2\ketbra{-} \otimes \frac{1}{M_{2,m}^2} \sum_{i,j=0}^{M_{2,m} - 1} ij\ketbra{i}{j}\right)\ket{b} = \frac{1}{M_{2,m}^2} ab.
\end{equation}
If the sign qubits are the same for both $\ket{a}$ and $\ket{b}$, then the $2\ketbra{-}$ term does not affect the overall sign.
However, if the sign qubits are different, then the $2\ketbra{-}$ term gives an extra minus sign, as required.
Thus, we have proven the claim.
\end{proof}

Now, we want to show that the conditions checked in Steps 3 and 4 in \Cref{alg:verification} are satisfied if and only if $|T - \ell M_1/w_k^\star| \leq 1$.
To do so, we first simplify our approximate inner product from \Cref{claim:inner-prod} further using the particular form of $h_{w^\star, M_{1,m}, M_{2,m}}$ from \Cref{lem:discrete-general-unif} and $\tilde{g}$ from \Cref{eq:g-tilde}.
\begin{align}
  &\expval{O_{k,m}}{h_{w^\star, M_{1,m}, M_{2,m}}}\\
  &= \frac{1}{(2\tilde{R}M_{1,m})^d M_{2,m}^2} \sum_{x_1,\dots, x_d=-\tilde{R}M_{1,m}}^{\tilde{R}M_{1,m}-1} h_{w^\star, M_{1,m}, M_{2,m}}(x) h_{w^\star, M_{1,m}, M_{2,m}}(x + Te_k)\\
  &= \frac{1}{(2\tilde{R}M_{1,m})^d M_{2,m}^2} \sum_{x_1,\dots, x_d=-\tilde{R}M_{1,m}}^{\tilde{R}M_{1,m}-1} \sum_{j,j'=1}^D \beta_j^\star \beta_{j'}^\star \left\lfloor \cos\left(\frac{2\pi j x^\intercal w^\star}{M_{1,m}}\right)\right\rfloor_{M_{2,m}} \left\lfloor \cos\left(\frac{2\pi j' (x + Te_k)^\intercal w^\star}{M_{1,m}}\right)\right\rfloor_{M_{2,m}}\\
  &= \frac{1}{(2\tilde{R}M_{1,m})^d M_{2,m}^2} \sum_{x_1,\dots, x_d=-\tilde{R}M_{1,m}}^{\tilde{R}M_{1,m}-1} \sum_{j,j'=1}^D \beta_j^\star \beta_{j'}^\star \cos\left(\frac{2\pi j x^\intercal w^\star}{M_{1,m}}\right) \cos\left(\frac{2\pi j' x^\intercal w^\star}{M_{1,m}} + \frac{2\pi j' T w_k^\star}{M_{1,m}}\right) + \epsilon_d\\
  &\begin{aligned}
    =\frac{1}{(2\tilde{R}M_{1,m})^d M_{2,m}^2} \sum_{x_1,\dots, x_d=-\tilde{R}M_{1,m}}^{\tilde{R}M_{1,m}-1} \sum_{j,j'=1}^D \beta_j^\star \beta_{j'}^\star \cos\left(\frac{2\pi j x^\intercal w^\star}{M_{1,m}}\right)&\left(\cos\left(\frac{2\pi j' x^\intercal w^\star}{M_{1,m}}\right) \cos\left(\frac{2\pi j' T w_k^\star}{M_{1,m}}\right)\right.\\
    &\left.- \sin\left(\frac{2\pi j' x^\intercal w^\star}{M_{1,m}}\right)\sin\left(\frac{2\pi j' Tw_k^\star}{M_{1,m}}\right)\right) + \epsilon_d
  \end{aligned}\\
  &\begin{aligned}
    =\frac{1}{(2\tilde{R}M_{1,m})^d M_{2,m}^2} \sum_{x_1,\dots, x_d=-\tilde{R}M_{1,m}}^{\tilde{R}M_{1,m}-1} \sum_{j=1}^D (\beta_j^\star)^2 &\left(\cos^2\left(\frac{2\pi j x^\intercal w^\star}{M_{1,m}}\right)\cos\left(\frac{2\pi j T w_k^\star}{M_{1,m}}\right)\right.\\
    &\left.- \cos\left(\frac{2\pi j x^\intercal w^\star}{M_{1,m}}\right)\sin\left(\frac{2\pi j x^\intercal w^\star}{M_{1,m}}\right)\sin\left(\frac{2\pi j T w_k^\star}{M_{1,m}}\right) \right)
    \end{aligned}\label{eq:h-o-expval}\\
  &\begin{aligned}
    +\frac{1}{(2\tilde{R}M_{1,m})^d M_{2,m}^2} \sum_{x_1,\dots, x_d=-\tilde{R}M_{1,m}}^{\tilde{R}M_{1,m}-1} \sum_{\substack{j,j'=1\\j\neq j'}}^D \beta_j^\star \beta_{j'}^\star&\left(\cos\left(\frac{2\pi j x^\intercal w^\star}{M_{1,m}}\right)\cos\left(\frac{2\pi j' x^\intercal w^\star}{M_{1,m}}\right)\cos\left(\frac{2\pi j' Tw_k^\star}{M_{1,m}}\right)\right.\\
    &\left.- \cos\left(\frac{2\pi j x^\intercal w^\star}{M_{1,m}}\right)\sin\left(\frac{2\pi j' x^\intercal w^\star}{M_{1,m}}\right)\sin\left(\frac{2\pi j' T w_k^\star}{M_{1,m}}\right)\right) + \epsilon_d
    \label{eq:h-o-expval2}
  \end{aligned}
\end{align}
In the second line, we use \Cref{claim:inner-prod}.
In the third line, we use the definition of $h_{w^\star, M_{1,m}, M_{2,m}}$ from \Cref{lem:discrete-general-unif} and \Cref{eq:g-tilde}.
Here, recall that $\lfloor \cdot \rfloor_{M_{2,m}}$ denotes rounding to the nearest integer multiple of $M_{2,m}$.
In the fourth line, we define a discretization error, denoted by $\epsilon_d$, which accounts for the error in getting rid of the rounding.
In the fifth line, we use the sum formula for cosine.
In the last equality, we split up the sum into the cases when $j = j'$ and $j\neq j'$.

We want to upper and lower bound this expression.
To do so, we find it easier to work with integrals over $x$ instead of these discrete sums.
We can then bound the integrals, which we relegate to Appendix~\ref{sec:int-bounds}.
To this end, we first need to bound the error from approximating our summation by an integral.

\begin{claim}[Sum-to-integral error]
\label{claim:sum-to-int}
For $m \in \{1,\dots, D\}$, consider parameters $M_{1,m}, M_{2,m}$ as defined above. Also consider a parameter $\tilde{R}$ defined above. Then, for an integer $1 \leq j \leq D$,
\begin{equation}
  \frac{1}{(2\tilde{R})^d}\left|\int_{[-\tilde{R},\tilde{R}]^d} \cos^2\left(2\pi j x^\intercal w^\star\right)\,dx - \frac{1}{M_{1,m}^d} \sum_{x_1,\dots, x_d = -\tilde{R}M_{1,m}}^{\tilde{R}M_{1,m} - 1} \cos^2\left(\frac{2\pi j x^\intercal w^\star}{M_{1,m}}\right) \right|\leq \frac{1}{35 D^2}.
\end{equation}
\end{claim}

\begin{proof}[Proof of \Cref{claim:sum-to-int}]
We prove this by induction on the dimension $d$.
In particular, denoting $f(x) \triangleq \cos^2(2\pi j x^\intercal w^\star)$, we will prove the following by induction:
\begin{equation}
  \label{eq:induct}
  \frac{1}{(2\tilde{R})^d} \left|\int_{[-\tilde{R},\tilde{R}]^d} f(x_1,\dots, x_d)\,dx - \frac{1}{M_{1,m}^d} \sum_{x_1,\dots, x_d=-\tilde{R}M_{1,m}}^{\tilde{R}M_{1,m}-1} f\left(\frac{x_1}{M_{1,m}},\dots, \frac{x_d}{M_{1,m}}\right) \right| \leq \frac{2\pi d DR_w}{M_{1,m}}.
\end{equation}
Note that this implies our claim by our choice of $M_{1,m} = mM_1 \geq 70m \pi d D^3 R_w \geq 70\pi d D^3 R_w$.
Thus, it suffices to prove \Cref{eq:induct}.
In fact, we will use induction to prove that
\begin{align}
  \label{eq:induct2}
  &\frac{1}{(2\tilde{R})^{d-1}} \left|\int_{[-\tilde{R},\tilde{R}]^{d-1}} f(x_1,\dots, x_{d-1},y)\,dx - \frac{1}{M_{1,m}^{d-1}} \sum_{x_1,\dots, x_{d-1}=-\tilde{R}M_{1,m}}^{\tilde{R}M_{1,m}-1} f\left(\frac{x_1}{M_{1,m}},\dots, \frac{x_{d-1}}{M_{1,m}},y\right) \right|\\
  &\leq \frac{2\pi (d-1) DR_w}{M_{1,m}}
\end{align}
for some fixed $y$.
In the process, we show that \Cref{eq:induct} follows from this.

First, consider the base case.
Then, we want to prove
\begin{equation}
  \label{eq:induct-base}
  \frac{1}{2\tilde{R}} \left|\int_{-\tilde{R}}^{+\tilde{R}} f(x)\,dx - \frac{1}{M_{1,m}} \sum_{x=-\tilde{R}M_{1,m}}^{\tilde{R}M_{1,m}-1} f\left(\frac{x}{M_{1,m}}\right) \right| \leq \frac{2\pi D R_w}{M_{1,m}}
\end{equation}
and 
\begin{equation}
    \label{eq:induct-base2}
    \frac{1}{2\tilde{R}}\left|\int_{-\tilde{R}}^{+\tilde{R}} f(x,y)\,dx - \frac{1}{M_{1,m}}\sum_{x=-\tilde{R}M_{1,m}}^{\tilde{R}M_{1,m}-1} f\left(\frac{x}{M_{1,m}},y\right)\right| \leq \frac{2\pi D R_w}{M_{1,m}},
\end{equation}
for some fixed $y$ and $f(x,y) \triangleq \cos^2(2\pi j(xw_1^\star + yw_2^\star))$.
First, consider \Cref{eq:induct-base}.
Notice that the sum in \Cref{eq:induct-base} is just the lefthand Riemann sum for the integral.
In particular, we approximate the integral by $2\tilde{R}M_{1,m}$ rectangles of width $2\tilde{R}/(2\tilde{R}M_{1,m}) = 1/M_{1,m}$.
Thus, we have
\begin{equation}
  \int_{-\tilde{R}}^{+\tilde{R}} f(x)\,dx \approx \frac{1}{M_{1,m}}\sum_{i=0}^{2\tilde{R}M_{1,m}-1} f\left(-\tilde{R} + \frac{i}{M_{1,m}}\right) = \frac{1}{M_{1,m}}\sum_{x=-\tilde{R}M_{1,m}}^{\tilde{R}M_{1,m} - 1}f\left(\frac{x}{M_{1,m}}\right).
\end{equation}
Moreover, the error in this approximation can be bounded by standard results:
\begin{equation}
  \label{eq:sum-to-int-d1}
  \left|\int_{-\tilde{R}}^{+\tilde{R}} f(x)\,dx - \frac{1}{M_{1,m}}\sum_{x=-\tilde{R}M_{1,m}}^{\tilde{R}M_{1,m}-1} f\left(\frac{x}{M_{1,m}}\right)\right| \leq \frac{L\tilde{R}}{M_{1,m}},
\end{equation}
where $L \triangleq \max_{x \in [-\tilde{R},\tilde{R}]} |f'(x)|$.
For our choice of $f(x) = \cos^2(2\pi j x w^\star)$, then
\begin{equation}
  f'(x) = -4\pi j w^\star\cos(2\pi j x w^\star)\sin(2\pi j x w^\star)
\end{equation}
so that $|f'(x)| \leq 4\pi j R_w \leq 4\pi D R_w$.
Thus, $L \leq 4\pi D R_w$.
Dividing both sides by $2\tilde{R}$, we obtain the claim.
Note that \Cref{eq:induct-base2} also follows by the same argument as above for $\tilde{f}(x) \triangleq f(x,y)$ for a fixed $y$, where $f(x,y) = \cos^2(2\pi j (xw_1^\star + yw_2^\star))$.
Namely, the only part of the above argument that relies on properties of the function $f$ was a bound on the derivative.
For $\tilde{f}$, we have the same bound:
\begin{equation}
    \tilde{f}'(x) = -4\pi j w_1^\star \cos(2\pi j(xw_1^\star + yw_2^\star))\sin(2\pi j(xw_1^\star + yw_2^\star))
\end{equation}
so that $|\tilde{f}'(x)| \leq 4\pi D R_w$.

Now, for the inductive step, suppose for $\ell$ such that $d -1 \geq \ell \geq 1$ that
\begin{equation}
  \frac{1}{(2\tilde{R})^\ell}\left|\int_{[-\tilde{R},\tilde{R}]^\ell} f(x_1,\dots, x_\ell, y)\,dx - \frac{1}{M_{1,m}^\ell} \sum_{x_1,\dots, x_\ell = -\tilde{R}M_{1,m}}^{\tilde{R}M_{1,m}-1} f\left(\frac{x_1}{M_{1,m}},\dots, \frac{x_\ell}{M_{1,m}}, y\right)\right| \leq \frac{4\pi \ell D \tilde{R}_w}{M_{1,m}},
\end{equation}
for some fixed $y$ and where $f(x_1,\dots, x_\ell,y) = \cos^2(2\pi j (x_1 w_1^\star + \cdots + x_\ell w_\ell^\star + yw_{\ell+1}^\star))$.
We first show that \Cref{eq:induct} holds for $\ell + 1$.
\begin{align}
  &\frac{1}{(2\tilde{R})^{\ell+1}} \int_{[-\tilde{R},\tilde{R}]^{\ell+1}} f(x)\,dx\\
  &= \frac{1}{2\tilde{R}}\int_{-\tilde{R}}^{+\tilde{R}} \left(\frac{1}{(2\tilde{R})^\ell} \int_{[-\tilde{R},\tilde{R}]^\ell}f(x_1,\dots, x_{\ell+1})\,dx_1\cdots dx_\ell\right)\,dx_{\ell+1}\\
  &\leq \frac{1}{(2\tilde{R})^{\ell+1} M_{1,m}^\ell} \sum_{x_1,\dots, x_\ell=-\tilde{R}M_{1,m}}^{\tilde{R}M_{1,m} -1} \int_{-\tilde{R}}^{+\tilde{R}} f\left(\frac{x_1}{M_{1,m}},\dots, \frac{x_\ell}{M_{1,m}}, x_{\ell+1}\right)\,dx_{\ell+1} + \frac{1}{2\tilde{R}}\int_{-\tilde{R}}^{+\tilde{R}} \frac{2\pi \ell D R_w}{M_{1,m}}\,dx_{\ell+1}\\
  &\leq \frac{1}{(2\tilde{R})^{\ell+1} M_{1,m}^\ell} \sum_{x_1,\dots, x_\ell=-\tilde{R}M_{1,m}}^{\tilde{R}M_{1,m} -1} \left(\frac{1}{M_{1,m}} \sum_{x_{\ell+1}=-\tilde{R}M_{1,m}}^{\tilde{R}M_{1,m}} f\left(\frac{x_1}{M_{1,m}},\dots, \frac{x_{\ell+1}}{M_{1,m}}\right) + \frac{L'\tilde{R}}{M_{1,m}}\right) + \frac{2\pi \ell D R_w}{M_{1,m}}\\
  &=\frac{1}{(2\tilde{R})^{\ell+1} M_{1,m}^{\ell+1}} \sum_{x_1,\dots, x_{\ell+1}=-\tilde{R}M_{1,m}}^{\tilde{R}M_{1,m} -1} f\left(\frac{x_1}{M_{1,m}},\dots, \frac{x_{\ell+1}}{M_{1,m}}\right) + \frac{L'}{2M_{1,m}} + \frac{2\pi \ell D R_w}{M_{1,m}},
\end{align}
where in the third line, we use the inductive hypothesis.
In the fourth line, we apply \Cref{eq:sum-to-int-d1} for the function $\tilde{f}(y) \triangleq f(x_1/M_{1,m},\dots, x_\ell/M_{1,m}, y)$.
Also, here, $L' \triangleq \max_{y\in[-\tilde{R},\tilde{R}]} |\tilde{f}'(y)|$.
For $f(x_1,\dots, x_{\ell+1}) = \cos^2(2\pi j (x_1w_1^\star + \cdots x_{\ell+1} w_{\ell+1}^\star))$, then
\begin{equation}
  \tilde{f}'(y) = -4\pi j w_{\ell+1}^\star \cos\left(2\pi j \left(yw_{\ell+1}^\star + \sum_{i=1}^{\ell} \frac{x_i}{M_{1,m}}w^\star_i\right)\right) \sin\left(2\pi j\left(yw_{\ell+1}^\star + \sum_{i=1}^{\ell} \frac{x_i}{M_{1,m}} w^\star_i\right)\right).
\end{equation}
Thus, $|\tilde{f}'(y)| \leq 4\pi j R_w \leq 4 \pi D R_w$ so that $L' \leq 4\pi D R_w$.
Plugging this back into the above, we have
\begin{equation}
  \frac{1}{(2\tilde{R})^{\ell+1}} \int_{[-\tilde{R},\tilde{R}]^{\ell+1}} f(x)\,dx \leq \frac{1}{(2\tilde{R})^{\ell+1} M_{1,m}^{\ell+1}} \sum_{x_1,\dots, x_{\ell+1}=-\tilde{R}M_{1,m}}^{\tilde{R}M_{1,m} -1} f\left(\frac{x_1}{M_{1,m}},\dots, \frac{x_{\ell+1}}{M_{1,m}}\right) + \frac{4\pi (k+1) D R_w}{M_{1,m}}.
\end{equation}
One can argue similarly for the lower bound.
Thus, we have shown that \Cref{eq:induct} holds for $\ell + 1$.

Now, to complete our induction, we need to show that \Cref{eq:induct2} holds for $\ell + 1$.
Namely, we want to show
\begin{equation}
  \frac{1}{(2\tilde{R})^{\ell+1}}\left|\int_{[-\tilde{R},\tilde{R}]^{\ell+1}} f(x_1,\dots, x_{\ell+1}, z)\,dx - \frac{1}{M_{1,m}^{\ell+1}} \sum_{x_1,\dots, x_{\ell+1} = -\tilde{R}M_{1,m}}^{\tilde{R}M_{1,m}-1} f\left(\frac{x_1}{M_{1,m}},\dots, \frac{x_{\ell+1}}{M_{1,m}}, z\right)\right| \leq \frac{4\pi (\ell+1) D R_w}{M_{1,m}}
\end{equation}
for some fixed $z$ and where $f(x_1,\dots, x_{\ell+1}, z) = \cos^2(2\pi j (x_1w_1^\star + \cdots + x_{\ell+1}w_{\ell+1}^\star + zw_{\ell+2}^\star))$.
This follows by the same argument as above.
Note that the inductive hypothesis can still be applied by taking $y = x_{\ell+1} + z (w_{\ell+2}^\star/w_{\ell+1}^\star)$, which is fixed when integrating with respect to $x_1,\dots, x_\ell$.
Moreover, when applying \Cref{eq:sum-to-int-d1}, we instead consider the function $\tilde{f}(x_{\ell+1}) \triangleq f(x_1/M_{1,m}, \dots, x_\ell/M_{1,m},x_{\ell+1}, z)$.
The bound on the derivative of this function is clearly the same since $z$ is fixed.
Thus, the same argument as above applies, completing the induction.
\end{proof}

Note that the same result can be shown for the cross terms $\cos(2\pi j x^\intercal w^\star/M_{1,m})\cos(2\pi j' x^\intercal w^\star)$ and $\cos(2\pi j x^\intercal w^\star/M_{1,m}) \sin(2\pi j' x^\intercal w^\star/M_{1,m})$ by the same argument.
This is clear because these terms have the same bound on their gradients.

We can also bound the discretization error $\epsilon_d$.
Note that this discretization error is defined as
\begin{align}
  \label{eq:eps-d}
  \epsilon_d \triangleq\frac{1}{(2\tilde{R}M_{1,m})^d M_{2,m}^2} \sum_{x_1,\dots, x_d=-\tilde{R}M_{1,m}}^{\tilde{R}M_{1,m}-1} &\sum_{j,j'=1}^D \beta_j^\star \beta_{j'}^\star \left(\cos\left(\frac{2\pi j x^\intercal w^\star}{M_{1,m}}\right)\cos\left(\frac{2\pi j'(x + Te_k)^\intercal w^\star}{M_{1,m}}\right)\right.\\
  &- \left.\left\lfloor \cos\left(\frac{2\pi j x^\intercal w^\star}{M_{1,m}}\right) \right\rfloor_{M_{2,m}} \left\lfloor \cos\left(\frac{2\pi j'(x+Te_k)^\intercal w^\star}{M_{1,m}}\right) \right\rfloor_{M_{2,m}}\right).
\end{align}

\begin{claim}[Discretization error]
  \label{claim:eps-d}
  For $m \in \{1,\dots, D\}$, consider parameters $M_{1,m}, M_{2,m}$ as defined above.
  Also, consider a parameter $\tilde{R}$ defined above.
  Then, we can bound the discretization error $\epsilon_d$ defined in \Cref{eq:eps-d} as
  \begin{equation}
    |\epsilon_d| \leq \frac{2}{M_{2,m}^3}.
  \end{equation}
\end{claim}

\begin{proof}[Proof of \Cref{claim:eps-d}]
This follows by a simple calculation.
First, we can add and subtract an intermediate term in which $\cos(2\pi j x^\intercal w^\star/M_{1,m})$ is rounded while $\cos(2\pi j' (x+Te_k)^\intercal w^\star/M_{1,m})$.
\begin{align}
  |\epsilon_d| &\leq \frac{1}{(2\tilde{R}M_{1,m})^d M_{2,m}^2} \sum_{x_1,\dots, x_d=-\tilde{R}M_{1,m}}^{\tilde{R}M_{1,m} - 1} \sum_{j,j'=1}^D |\beta_j^\star| |\beta_{j'}^\star|\\
  &\cdot \left(\left| \cos\left(\frac{2\pi j x^\intercal w^\star}{M_{1,m}}\right)\cos\left(\frac{2\pi j'(x+Te_k)^\intercal w^\star}{M_{1,m}}\right) - \left\lfloor \cos\left(\frac{2\pi j x^\intercal w^\star}{M_{1,m}}\right) \right\rfloor_{M_{2,m}} \cos\left(\frac{2\pi j' (x+ Te_k)^\intercal w^\star}{M_{1,m}}\right) \right|\right.\\
  & \left.+ \left|\left\lfloor \cos\left(\frac{2\pi j x^\intercal w^\star}{M_{1,m}}\right)\right\rfloor_{M_{2,m}} \cos\left(\frac{2\pi j'(x+Te_k)^\intercal w^\star}{M_{1,m}}\right)\right.\right.\\
  &-  \left.\left.\left\lfloor \cos\left(\frac{2\pi j x^\intercal w^\star}{M_{1,m}}\right)\right\rfloor_{M_{2,m}} \left\lfloor \cos\left(\frac{2\pi j'(x+Te_k)^\intercal w^\star}{M_{1,m}}\right) \right\rfloor_{M_{2,m}}\right| \right).
\end{align}
Simplifying, we have
\begin{align}
  &|\epsilon_d|\\
  &\begin{aligned}
    \leq \frac{1}{(2\tilde{R}M_{1,m})^d M_{2,m}^2} \sum_{x_1,\dots, x_d=-\tilde{R}M_{1,m}}^{\tilde{R}M_{1,m} - 1} & \sum_{j,j'=1}^D |\beta_j^\star| |\beta_{j'}^\star|\left(\left|\cos\left(\frac{2\pi j x^\intercal w^\star}{M_{1,m}}\right) - \left\lfloor \cos\left(\frac{2\pi j x^\intercal w^\star}{M_{1,m}}\right) \right\rfloor_{M_{2,m}}\right|\right.\\
    &\left.+ \left|\cos\left(\frac{2\pi j' (x + Te_k)^\intercal w^\star}{M_{1,m}}\right) - \left\lfloor \cos\left(\frac{2\pi j' (x + Te_k)^\intercal w^\star}{M_{1,m}}\right) \right\rfloor_{M_{2,m}} \right| \right)
  \end{aligned}\\
  &\leq \frac{1}{(2\tilde{R}M_{1,m})^d M_{2,m}^2} \sum_{x_1,\dots, x_d=-\tilde{R}M_{1,m}}^{\tilde{R}M_{1,m} - 1} \sum_{j,j'=1}^D |\beta_j^\star| |\beta_{j'}^\star|\frac{2}{M_{2,m}}\\
  &= \frac{2}{M_{2,m}^3}.
\end{align}
In the first inequality, we use that $|\cos(x)| \leq 1$.
In the second inequality, we use that $\lfloor \cdot \rfloor_{M_{2,m}}$ means rounding to the nearest integer multiple of $M_{2,m}$.
Thus, the difference between a rounded and unrounded quantity must be at most $1/M_{2,m}$.
Finally, in the last line, we use that $\norm{\beta^\star}_1 = 1$.
\end{proof}

With \Cref{claim:sum-to-int} and \Cref{claim:eps-d}, in \Cref{eq:h-o-expval,eq:h-o-expval2}, we now have
\begin{align}
  &\expval{O_{k,m}}{h_{w^\star, M_{1,m}, M_{2,m}}}\\
  &= \frac{1}{M_{2,m}^2}\int_{x \sim \varphi^2} \sum_{j=1}^D (\beta_j^\star)^2 \left(\cos^2(2\pi j x^\intercal w^\star) \cos\left(\frac{2\pi j T w_k^\star}{M_1}\right) - \cos(2\pi j x^\intercal w^\star) \sin(2\pi j x^\intercal w^\star) \sin\left(\frac{2\pi j T w_k^\star}{M_{1,m}}\right)\right)\,dx\\
  &\begin{aligned}
    +\frac{1}{M_{2,m}^2} \int_{x \sim\varphi^2} \sum_{\substack{j,j'=1\\j\neq j'}}^D \beta_j^\star \beta_{j'}^\star &\left(\cos(2\pi j x^\intercal w^\star) \cos(2\pi j' x^\intercal w^\star) \cos\left(\frac{2\pi j' T w_k^\star}{M_{1,m}}\right)\right.\\
    &\left.- \cos(2\pi j x^\intercal w^\star)\sin(2\pi j' x^\intercal w^\star) \sin\left(\frac{2\pi j' T w_k^\star}{M_{1,m}}\right)\right)\,dx + \epsilon_d + \frac{4}{M_{2,m}^2}\epsilon_{\mathrm{int}},
  \end{aligned}
\end{align}
where $|\epsilon_d| \leq 2/M_{2,m}^3$ and $|\epsilon_{\mathrm{int}}| \leq 1/(35D^2)$.
Here, the integrals are with respect to the uniform distribution over $[-\tilde{R},\tilde{R}]^d$.
We can simplify this using the fact that an integral of an odd function, e.g., $\sin(x) \cos(x)$, over an even interval is zero:
\begin{equation}
  \label{eq:h-o-expval3}
\begin{split}
  &\expval{O_{k,m}}{h_{w^\star, M_{1,m}, M_{2,m}}}\\
  &= \frac{1}{M_{2,m}^2} \sum_{j=1}^D (\beta_j^\star)^2 \cos\left(\frac{2\pi j T w_k^\star}{M_{1,m}}\right) \int_{x\sim\varphi^2} \cos^2(2\pi j x^\intercal w^\star)\,dx\\
  &+ \frac{1}{M_{2,m}^2} \sum_{\substack{j,j'=1\\j\neq j'}}^D \beta_j^\star \beta_{j'}^\star \left(\cos\left(\frac{2\pi j' T w_k^\star}{M_{1,m}}\right) \int_{x \sim \varphi^2} \cos(2\pi j x^\intercal w^\star)\cos(2\pi j' x^\intercal w^\star)\,dx\right.\\
  &\left.- \sin\left(\frac{2\pi j' T w_k^\star}{M_{1,m}}\right) \int_{x \sim \varphi^2} \cos(2\pi j x^\intercal w^\star) \sin(2\pi j' x^\intercal w^\star)\,dx\right) + \epsilon_d + \frac{4}{M_{2,m}^2}\epsilon_{\mathrm{int}}.
\end{split}
\end{equation}

With this, we can finally move on to show that the conditions checked in Steps 3 and 4 of \Cref{alg:verification} are satisfied if and only if $|T - \ell M_1/w_k^\star| \leq 1$.
To do so, we leverage integral bounds from Appendix~\ref{sec:int-bounds}.
The following two claims show this for each direction of the if and only if.

\begin{claim}[Correctness of Step 3 in \Cref{alg:verification}]
  \label{claim:ver-if}
  Consider parameters $M_1, M_2, \tilde{R}$ defined above and the observable $O_{k,1}$ defined in \Cref{eq:ov}.
  Let $\alpha_1$ denote the result of querying the QSQ oracle with observable $O_{k,1}$ with discretization parameters $M_1, M_2$, truncation parameter $R\triangleq \tilde{R}M_1$, and tolerance $\tau \leq \frac{1}{M_2^2}\left(\frac{7}{40D} - \frac{1}{M_2}\right)$.
  If $|T - \ell M_1/w_k^\star| \leq 1$ for some integer $\ell$, then
  \begin{equation}
    \alpha_1 \geq \frac{1}{M_2^2}\left(\frac{21}{40D} - \frac{3}{M_2}\right).
  \end{equation}
\end{claim}

\begin{claim}[Correctness of Step 4 in \Cref{alg:verification}]
  \label{claim:ver-only-if}
  For $m \in \{1,\dots, D\}$, consider parameters $M_{1,m}, M_{2,m}, \tilde{R}$ defined above and the observables $O_{k,m}$ defined in \Cref{eq:ov}.
  Let $\alpha_m$ denote the result of querying the QSQ oracle with observable $O_{k,m}$ with discretization parameters $M_{1,m}, M_{2,m}$, truncation parameter $R \triangleq \tilde{R}M_{1,m}$, and tolerance $\tau \leq \frac{1}{2D^2M_2^2}\left(\frac{2}{15} - \frac{1}{8}\left(\frac{2\pi R_w}{M_1}\right)^2 + \frac{2D^2}{M_2}\right)$.
  If $|T - \ell M_1 / w_k^\star|$ is not less than $1$ for any integer $\ell$, then
  \begin{equation}
    \sum_{m=1}^D \alpha_m \leq \frac{1}{M_2^2}\left(\frac{20}{39}D + \frac{1}{2D}\left(\frac{2}{15} - \frac{1}{8}\left(\frac{2\pi R_w}{M_1}\right)^2 +\frac{2D^2}{M_2}\right)\right).
  \end{equation}
\end{claim}

It suffices to prove these two claims to finish the proof.
Our starting point for both proofs is \Cref{eq:h-o-expval3}.

\begin{proof}[Proof of \Cref{claim:ver-if}]
We can lower bound $\expval{O_{k,1}}{h_{w^\star, M_1, M_2}}$ using \Cref{coro:integral,coro:integral2-wstar,coro:integral2-wstar-sin} and \Cref{eq:h-o-expval3}:
\begin{align}
  \expval{O_{k,1}}{h_{w^\star, M_1, M_2}} &\geq \frac{1}{M_2^2}\sum_{j=1}^D (\beta_j^\star)^2 \left(\frac{1}{2} - \frac{\sqrt{d}}{8\pi R_w \tilde{R}}\right) \cos\left(\frac{2\pi j T w_k^\star}{M_1}\right)\\
  &- \frac{1}{M_2^2}\sum_{\substack{j,j'=1\\j\neq j'}}^D \beta_j^\star \beta_{j'}^\star \left(\frac{\sqrt{d}}{\pi R_w \tilde{R}}\right) + \epsilon_d + \frac{4}{M_{2}^2}\epsilon_{\mathrm{int}}\\
  &\geq \frac{1}{M_2^2}\left(\sum_{j=1}^D (\beta_j^\star)^2 \left(\frac{1}{2} - \frac{\sqrt{d}}{8\pi R_w \tilde{R}}\right) \cos\left(\frac{2\pi j T w_k^\star}{M_1}\right) - \frac{\sqrt{d}}{\pi R_w \tilde{R}} - \frac{2}{M_2} - \frac{4}{35D^2}\right)\\
  &\geq \frac{1}{M_2^2}\left(\frac{19}{39}\sum_{j=1}^D (\beta_j^\star)^2 \cos\left(\frac{2\pi j T w_k^\star}{M_1}\right) - \frac{1}{54D^2} - \frac{2}{M_2} - \frac{4}{35D^2}\right).\label{eq:h-o-expval-lower}
\end{align}
In the second to last line, we use that $\norm{\beta^\star}_2^2 \leq 1$ since $\norm{\beta^\star}_1 = 1$.
We also used that $|\epsilon_d| \leq 2/M_2^3$ by \Cref{claim:eps-d} and $|\epsilon_{\mathrm{int}}| \leq 1/(35 D^2)$ by \Cref{claim:sum-to-int}.
In the last line, we use that $\tilde{R} \geq \max(39\sqrt{d}/(4\pi R_w),\\54D^2\sqrt{d}/(\pi R_w))$ in our choice of $\tilde{R}$.

Here, the key is that the summation over these cosine terms is peaked aroud multiples of $M_1/w_k^\star$.
Thus, this sum should be bounded away from $0$ when the guess $T$ is close to an integer multiple of the period $M_1/w_k^\star$.
The rest of the terms in this expression are error terms.
Suppose that $T = \ell M_1 /w_k^\star + \epsilon$ for some $|\epsilon| \leq 1$.
Then, we have
\begin{align}
  \sum_{j=1}^D (\beta_j^\star)^2 \cos\left(\frac{2\pi j T w_k^\star}{M_1}\right) &= \sum_{j=1}^D (\beta_j^\star)^2 \cos\left(\frac{2\pi j w_k^\star}{M_1}\left(\frac{\ell M_1}{w_k^\star} + \epsilon\right)\right)\\
  &= \sum_{j=1}^D (\beta_j^\star)^2 \cos\left(\frac{2\pi j w^\star_k}{M_1}\epsilon\right)\\
  &\geq \sum_{j=1}^D (\beta_j^\star)^2 \left(1 - \frac{1}{2}\left(2\pi j \frac{w_k^\star}{M_1}\epsilon\right)^2\right)\\
  &\geq \sum_{j=1}^D (\beta_j^\star)^2 \left(1 - \frac{1}{2}\left(\frac{2\pi j w_k^\star}{M_1}\right)^2\right)\\
  &\geq \frac{1}{D} - \frac{1}{2}\sum_{j=1}^D (\beta_j^\star)^2 \left(\frac{2\pi D R_w}{M_1}\right)^2\\
  &\geq \frac{1}{D} - \frac{1}{2}\left(\frac{2\pi D R_w}{M_1}\right)^2\\
  &\geq \frac{1}{D} - \frac{1}{2 \cdot 35^2 D^4}\\
  &\geq \frac{2449}{2450 D}.\label{eq:cos-T-bound}
\end{align}
In the second line, we use the periodicity of cosine.
In the third line, we use that $\cos(x) \geq 1-x/2$.
In the fourth line, we use that $|\epsilon| \leq 1$.
In the fifth line, we use that $\norm{\beta^\star}_2^2 \geq 1/D$ since $\norm{\beta^\star}_1 = 1$.
In the sixth line, we use that $\norm{\beta^\star}_2^2 \leq 1$.
In the seventh line, we use that $M_1 \geq 70 \pi D^3 R_w$.
Finally, in the last line, we use that $D \geq 1$ so that $D^4 \geq D$.

Plugging this into \Cref{eq:h-o-expval-lower}, we have
\begin{align}
  \expval{O_{k,1}}{h_{w^\star, M_1, M_2}} &\geq \frac{1}{M_2^2}\left(\frac{19}{39}\cdot \frac{2449}{2450D} - \frac{1}{54D^2} - \frac{2}{M_2} - \frac{4}{35D^2}\right)\\
  &\geq \frac{1}{M_2^2}\left(\frac{7}{20D} - \frac{2}{M_2}\right).
\end{align}
In the second line, we use that $D^2 \geq D$ and simplify.

Thus, we see that if $|T - \ell M_1/w_k^\star| \leq 1$, then this lower bound on the expectation value must be satisfied.
Recall that QSQs only approximate the expectation value up to some tolerance $\tau$.
By our choice of $\tau$, we have
\begin{equation}
  |\alpha_1 - \expval{O_{k,1}}{h_{w^\star, M_1, M_2}}| \leq \frac{1}{M_2^2}\left(\frac{7}{40D} - \frac{1}{M_2}\right).
\end{equation}
By choosing the condition
\begin{equation}
  \alpha_1 \geq \frac{1}{M_2^2}\left(\frac{21}{40D} - \frac{3}{M_2}\right),
\end{equation}
we can ensure that 
\begin{equation}
  \expval{O_{k,1}}{h_{w^\star, M_1, M_2}} \geq \alpha_1 - \frac{1}{M_2^2}\left(\frac{7}{40D} - \frac{1}{M_2}\right) \geq \frac{1}{M_2^2}\left(\frac{7}{20D} - \frac{2}{M_2}\right),
\end{equation}
as required.
\end{proof}

\begin{proof}[Proof of \Cref{claim:ver-only-if}]
This time, we can upper bound $\expval{O_{k,m}}{h_{w^\star, M_{1,m}, M_{2,m}}}$ for any $m \in \{1,\dots, D\}$ using \Cref{coro:integral-upper,coro:integral2-wstar,coro:integral2-wstar-sin} and \Cref{eq:h-o-expval3}:
\begin{align}
  &\expval{O_{k,m}}{h_{w^\star, M_{1,m}, M_{2,m}}}\\
  &\leq \frac{1}{M_{2,m}^2}\sum_{j=1}^D (\beta_j^\star)^2\left(\frac{1}{2} + \frac{\sqrt{d}}{8\pi R_w \tilde{R}}\right)\cos\left(\frac{2\pi j T w_k^\star}{M_{1,m}}\right) + \frac{1}{M_{2,m}^2}\sum_{\substack{j,j'=1\\j\neq j'}}^D \beta_j^\star \beta_{j'}^\star \left(\frac{\sqrt{d}}{\pi R_w \tilde{R}}\right) + \epsilon_d + \frac{4}{M_{2,m}^2}\epsilon_{\mathrm{int}}\\
  &\leq \frac{1}{M_{2,m}^2}\left(\left(\frac{1}{2} + \frac{\sqrt{d}}{8\pi R_w \tilde{R}}\right)\left( (\beta_m^\star)^2 \cos\left(\frac{2\pi T w_k^\star}{M_1}\right) + \sum_{\substack{j=1\\j\neq m}}^D (\beta_j^\star)^2\right) + \frac{\sqrt{d}}{\pi R_w \tilde{R}} + \frac{2}{M_{2,m}} + \frac{4}{35D^2}\right)\\
  &\leq \frac{1}{M_{2,m}^2}\left(\frac{20}{39}(\beta_m^\star)^2 \cos\left(\frac{2\pi T w_k^\star}{M_1}\right) + \frac{20}{39}\sum_{\substack{j=1\\j\neq m}}^D (\beta_j^\star)^2 + \frac{1}{54D^2} + \frac{2}{M_{2,m}} + \frac{4}{35D^2}\right).\label{eq:h-o-expval-upper}
\end{align}
In the third line, we split up the sum over $j$ into cases where $j =m$ and $j \neq m$.
In the $j = m$ case, we use that $M_{1,m} = m M_1$ by definition.
In the $j\neq m$ case, we bound $\cos(x) \leq 1$.
We also use that $|\epsilon_d| \leq 2/M_{2,m}^3$ by \Cref{claim:eps-d} and $|\epsilon_{\mathrm{int}}| \leq 1/(35D^2)$ by \Cref{claim:sum-to-int} and $\norm{\beta^\star}_1 = 1$.
In the last line, we use that $\tilde{R} \geq \max(39\sqrt{d}/4\pi R_w, 54D^2\sqrt{d}/\pi R_w)$ by our choice of $\tilde{R}$.

Now, suppose that there does not exist any integer $\ell$ such that $|T - \ell M_1/w_k^\star| \leq 1$.
Then, we can write $T = \ell' M_1/w_k^\star + c$ for some $c$ satisfying $1 < c < M_1/w_k^\star - 1$.
Then,
\begin{equation}
  \cos\left(\frac{2\pi T w_k^\star}{M_1}\right) = \cos\left(\frac{2\pi w_k^\star}{M_1}\left(\frac{\ell' M_1}{w_k^\star} + c\right)\right) = \cos\left(\frac{2\pi w_k^\star}{M_1}c\right).
\end{equation}
Without loss of generality, we can assume that $w_k^\star c/M_1 \leq 1/2$.
Otherwise, we can write
\begin{equation}
  \cos\left(\frac{2\pi w_k^\star}{M_1}c\right) = \cos\left(\frac{2\pi w_k^\star}{M_1}\left(\frac{M_1}{w_k^\star} - c'\right)\right) = \cos\left(\frac{2\pi w_k^\star}{M_1} c'\right)
\end{equation}
for some $c'$ such that $w_k^\star c'/M_1 \leq 1/2$.
Then, we can bound this cosine term:
\begin{align}
  \cos\left(\frac{2\pi T w_k^\star}{M_1}\right) &= \cos\left(\frac{2\pi w_k^\star}{M_1}c\right)\\
  &\leq 1 - \frac{1}{8}\left(\frac{2\pi w_k^\star}{M_1}c\right)^2\\
  &\leq 1 - \frac{1}{8}\left(\frac{2\pi R_w}{M_1}\right)^2.\label{eq:cos-upper}
\end{align}
Here, in the second line, we use that $\cos(x) \leq 1-x/8$ for $x \in [0,\pi]$, which is satisfied because we can assume that $w_k^\star c/M_1 \leq 1/2$ as discussed above.
In the last line, we use that $c > 1$ and $w^\star_k \leq R_w$.

Plugging this into \Cref{eq:h-o-expval-upper}, we have
\begin{align}
  &\expval{O_{k,m}}{h_{w^\star, M_{1,m}, M_{2,m}}}\\
  &\leq \frac{1}{M_{2,m}^2}\left(\frac{20}{39}(\beta_m^\star)^2 \left(1 - \frac{1}{8}\left(\frac{2\pi R_w}{M_1}\right)^2\right) + \frac{20}{39}\sum_{\substack{j=1\\j\neq m}}^D (\beta_j^\star)^2 + \frac{1}{54D^2} + \frac{2}{M_{2,m}} + \frac{4}{35D^2}\right)\\
  &\leq \frac{1}{M_{2,m}^2}\left(\frac{20}{39}\left(1 - \frac{1}{8}\left(\frac{2\pi R_w}{M_1}\right)^2 (\beta_m^\star)^2\right) + \frac{1}{54D^2} + \frac{2}{M_{2,m}} + \frac{4}{35D^2}\right).
\end{align}
In the last line, we use that $\norm{\beta^\star}_2^2 \leq 1$.
Summing over all $m \in \{1,\dots, D\}$, then we have
\begin{align}
  &\sum_{m=1}^D\expval{O_{k,m}}{h_{w^\star, M_{1,m}, M_{2,m}}}\\
  &\leq \sum_{m=1}^D\left(\frac{1}{M_{2,m}^2}\left(\frac{20}{39}\left(1 - \frac{1}{8}\left(\frac{2\pi R_w}{M_1}\right)^2 (\beta_m^\star)^2\right) + \frac{1}{54D^2} + \frac{2}{M_{2,m}} + \frac{4}{35D^2}\right)\right)\\
  &\leq \frac{1}{M_2^2}\sum_{m=1}^D\left(\frac{20}{39}\left(1 - \frac{1}{8}\left(\frac{2\pi R_w}{M_1}\right)^2(\beta_m^\star)^2\right) + \frac{1}{54D^2} + \frac{2}{M_2} + \frac{4}{35D^2}\right)\\
  &\leq \frac{1}{M_2^2}\left(\frac{20}{39}D -\frac{1}{8D}\left(\frac{2\pi R_w}{M_1}\right)^2 + \frac{2}{15D} + \frac{2D}{M_2}\right).
\end{align}
In the third line, we use that $M_{2,m} = m M_2$ by definition and $m \geq 1$.
In the last line, we use that $D^2 \geq D$ and $\norm{\beta^\star}_2^2 \geq 1/D$.

Thus, we see that if $|T - \ell M_1/w_k^\star| \not\leq 1$ for any integer $\ell$, then this upper bound on the sum of expectation values must be satisfied.
Recall that QSQs only approximate the expectation value up to some tolerance $\tau$.
By our choice of $\tau$, we have
\begin{equation}
  |\alpha_m - \expval{O_{k,m}}{h_{w^\star, M_{1,m}, M_{2,m}}}| \leq \frac{1}{2D^2 M_2^2}\left(\frac{2}{15} - \frac{1}{8}\left(\frac{2\pi R_w}{M_1}\right)^2 + \frac{2D^2}{M_2}\right).
\end{equation}
By choosing the condition
\begin{equation}
  \sum_{m=1}^D \alpha_m \leq \frac{1}{M_2^2}\left(\frac{20}{39}D + \frac{1}{2D}\left(\frac{2}{15} - \frac{1}{8}\left(\frac{2\pi R_w}{M_1}\right)^2 +\frac{2D^2}{M_2}\right)\right),
\end{equation}
we can ensure that
\begin{align}
  \sum_{m=1}^D \expval{O_{k,m}}{h_{w^\star, M_{1,m}, M_{2,m}}} &\leq \sum_{m=1}^D \alpha_m + \frac{1}{2D M_2^2}\left(\frac{2}{15} - \frac{1}{8}\left(\frac{2\pi R_w}{M_1}\right)^2 + \frac{2D^2}{M_2}\right)\\
  &\leq \frac{1}{M_2^2}\left(\frac{20}{39}D + \frac{1}{D}\left(\frac{2}{15} - \frac{1}{8}\left(\frac{2\pi R_w}{M_1}\right)^2\right) + \frac{2D}{M_2}\right),
\end{align}
as required.
\end{proof}
\end{proof}

With each of these parts, we can put everything together to prove \Cref{thm:linear-uniform}.

\begin{proof}[Proof of \Cref{thm:linear-uniform}]
We first consider the case of $d = 1$.
Our algorithm is simply to apply Hallgren's algorithm (Appendix~\ref{sec:hallgren}) to our setting using QSQs.
Choose the discretization parameters to be $M_1 = \max(70 \pi d D^3 R_w, R_w^2/\epsilon
_1)$ and $M_2 = c M_1$ for some constant $c$ such that $M_2 \in \mathbb{Z}$ and $c < 1/(8\pi D R_w)$.
By \Cref{lem:discrete-general-unif}, we know that there exists a discretization $h_{w^\star, M_1, M_2}$ of the target function $g_{w^\star}$ such that $h_{w^\star, M_1, M_2}$ is $(33/35)$-pseudoperiodic with period $M_1/w^\star$ by our choice of $M_1, M_2$.

Recall that by definition of $\mathcal{S}_w$ that $w_j \geq R_w/d^2$.
This gives an upper bound on the period, which we denote as $A \triangleq M_1d^2/R_w$.
We carry this $d$ factor through to avoid losing track of it.
Choose the truncation parameter $R \geq 6(1/2 + \tau)A^2$.

Then, we want to apply period finding to $h_{w^\star, M_1, M_2}$ using our QSQ access for discretization/truncation parameters $M_1, M_2, R$ as chosen above.
With \Cref{thm:verification}, we fulfill all of the conditions to apply the irrational period finding subroutine from Hallgren's algorithm~\cite{hallgren2007polynomial} reviewed in Appendix~\ref{sec:hallgren}.
Note that the main quantum part of the algorithm (Step 2 in \Cref{alg:hallgren}) is the same as standard period finding, i.e., simply quantum Fourier sampling.
The classical postprocessing and analysis is mainly what differs.
Thus, we can use the same QSQ operator as from Appendix~\ref{sec:warmup-uniform}, namely $O$ given in \Cref{eq:o}, to apply the quantum part of this algorithm.
In particular, this applies the QFT over $q = 2R$ and measures.

We can repeat the analysis of Hallgren's algorithm (\Cref{alg:hallgren}) Steps 3-5 to account for the noise $\tau \geq 0$ in the QSQs.
From the analysis of Hallgren's algorithm~\cite{hallgren2007polynomial}, the outputs of the QSQs are some numbers $\alpha, \beta$ such that 
\begin{equation}
  \label{eq:hallgren-output1}
  |\alpha - b| \leq \tau, \quad \left| b - \frac{kRw^\star}{M_1}\right| \leq \frac{1}{2}
\end{equation}
\begin{equation}
  |\beta - c| \leq \tau, \quad \left| c - \frac{\ell Rw^\star}{M_1}\right| \leq \frac{1}{2}
\end{equation}
for some integers $k, \ell \geq 1$.
We want to show that $k/\ell$ is a convergent in the continued fraction expansion of $\alpha/\beta$.
We use the fact that if $x$ is any irrational number, $e/f \in \mathbb{Q}$, and $|x - e/f| \leq 1/(2f^2)$, then $e/f$ is a convergent in the continued fraction expansion of $x$~\cite{schrijver1998theory}.
We write
\begin{equation}
  \label{eq:alpha}
  \alpha = b + \tau_k, \quad b = \frac{kR}{S} + \epsilon_k, \quad |\tau_k| \leq \tau, |\epsilon_k| \leq \frac{1}{2}
\end{equation}
\begin{equation}
  \beta = c + \tau_\ell,\quad c = \frac{\ell R}{S} + \epsilon_\ell, \quad |\tau_\ell| \leq \tau, |\epsilon_k| \leq \frac{1}{2},
\end{equation}
where we denote $S \triangleq M_1/w^\star$ as the period of our target function for simplicity.
Without loss of generality, suppose that $1 \leq k \leq \ell \leq S$.
Then, we have
\begin{align}
\left| \frac{\alpha}{\beta} - \frac{k}{\ell}\right| &= \left| \frac{kR + S(\epsilon_k + \tau_k)}{\ell R + S(\epsilon_\ell + \tau_\ell)} - \frac{k}{\ell}\right|\\
&= \left|\frac{S(\ell(\epsilon_k + \tau_k) - k(\epsilon_\ell + \tau_\ell))}{\ell^2 R - S(\epsilon_\ell + \tau_\ell) \ell}\right|\\
&\leq \left|\frac{S(\ell + k)}{\frac{1}{1/2 + \tau}(\ell^2 R - S(1/2 + \tau)\ell)}\right|\\
&\leq \left|\frac{2\ell S}{6\ell^2 S^2 -S\ell}\right|\\
&= \left|\frac{2\ell S}{2\ell(3\ell S^2 - S/2)}\right|\\
&= \left|\frac{1}{3\ell S - 1/2}\right|\\
&\leq \frac{1}{3\ell^2 - 1/2}\\
&\leq \frac{1}{2\ell^2}.
\end{align}
Here, in the third line, we use $|\epsilon_k + \tau_k| \leq 1/2 + \tau$.
In the fourth line, we use our choice of $R \geq 6(1/2 + \tau)A^2 \geq 6(1/2 + \tau)S^2$ since $A \geq S$ and $k \leq \ell$.
In the seventh line, we use $\ell \leq S$.
Finally, in the last line, we use that $3\ell^2 - 1/2 \geq 2\ell^2$ for $\ell \geq 1$.
This shows that $k/\ell$ is a convergent in the continued fraction expansion of $\alpha/\beta$.

Now, by Step 4 of \Cref{alg:hallgren}, when $k/\ell$ is a convergent in the continued fraction expansion of $\alpha/\beta$, we want to show that either $\lfloor k R/ \alpha \rfloor$ or $\lceil kR /\alpha \rceil$ is close to the period $S$ for some $k$.
We denote $\lfloor kR / \alpha \rceil$ to denote rounding to the closest integer.
In particular, we will show that $|S - \lfloor kR/\alpha \rceil| \leq 1$.
Again, we write $\alpha$ as in \Cref{eq:alpha}.
Then,
\begin{equation}
  \frac{kR}{\alpha} = kR\left(\frac{1}{\frac{kR}{S} + \epsilon_k + \tau_k}\right) = \frac{S}{1 + \frac{(\epsilon_k + \tau_k) S}{kR}} = \frac{S}{1 + \gamma},
\end{equation}
where in the last equality, we define
\begin{equation}
  \gamma \triangleq \frac{(\epsilon_k + \tau_k)S}{kR}.
\end{equation}
Notice that
\begin{equation}
  |\gamma| \leq \frac{(1/2 + \tau)S}{kR} \leq \frac{S}{6kS^2} = \frac{1}{6kS} \leq \frac{1}{6S},
\end{equation}
where in the first inequality, we use that $|\epsilon_k + \tau_k| \leq 1/2 + \tau$.
In the second inequality, we use our choice of $R \geq 6(1/2 + \tau)S^2$.
In the last inequality, we use $k \geq 1$.
Now, we can write
\begin{equation}
  \frac{kR}{\alpha} = \frac{S}{1+\gamma} = S - \frac{S\gamma}{1 + \gamma},\quad \left|\frac{S\gamma}{1 + \gamma}\right| < \frac{1}{2}.
\end{equation}
Thus, we see that $|S - \lfloor k R/\alpha \rceil| \leq 1$, as required.
Overall, this shows that Hallgren's algorithm correctly recovers the period $S$ even with noise from QSQs, as long as $R$ is chosen large enough.

Now, we analyze the number of QSQs that the algorithm requires.
Step 2 of \Cref{alg:hallgren} requires two QSQs as we are applying quantum Fourier sampling twice.
The only other part of the algorithm that requires QSQs is the verification subroutine, which uses $D$ QSQs each time it is called.
In Step 4 of \Cref{alg:hallgren}, this verification procedure must be repeated for each convergent in the continued fraction expansion of $\alpha/\beta$, where $\alpha$ and $\beta$ are the outputs from quantum Fourier sampling via the noisy QSQs.
Since we assume that the QSQs output rational numbers\footnote{As discussed in Appendix~\ref{sec:detail-prob}, the rational numbers are dense in $\mathbb{R}$. Then, if a QSQ outputs an irrational number, we can find a rational number close to it. We can then consider the error in this approximation as a part of the tolerance of the QSQ.}, then $\alpha/\beta$ is a rational number, which has a finite continued fraction expansion.
In fact, it is well known that the continued fraction expansion for rational numbers $\alpha/\beta$ can be computed via the steps of Euclid's algorithm on the numerator and denominator (see, e.g., the discussion after Theorem 161 in~\cite{hardy1979introduction}).
Moreover, Euclid's algorithm requires a number of steps scaling logarithmically in the numbers it is run on.
Thus, in our case, then we must run the verification procedure at most $\mathcal{O}(\log(S)) = \mathcal{O}(\log A) = \mathcal{O}\left(\log (M_1 d^2/R_w)\right)$ times, which uses $\mathcal{O}(D \log (M_1d^2 /R_w))$ QSQs in total.

Overall, this shows that we can find an integer $a$ within $1$ of $M_1/w^\star$ with some probability using $\mathcal{O}(D \log (M_1d^2 /R_w))$ QSQs.
In particular, $a$ satisfies
\begin{equation}
  \frac{a}{M_1} \in \left[\frac{1}{w^\star} \pm \frac{1}{M_1}\right]
\end{equation}
with probability $\Omega(\eta^2/\log^4(A))$, where $\eta = 33/35$.
We want to choose $M_1$ such that $M_1/a$ is close to $w^\star$.
For this, we use the fact that the relative error for $z = 1/x$ is the same as the relative error for $x$ (see, e.g.,~\cite{taylor1982introduction}), i.e., $(\Delta z)/z = (\Delta x)/x$, where $\Delta z$ and $\Delta x$ are the uncertainties in $z$ and $x$, respectively.
Thus, taking $z = w^\star, x = 1/w^\star$, we have
\begin{equation}
  \frac{\Delta z}{w^\star} = \frac{1/M_1}{1/w^\star}.
\end{equation}
Solving for $\Delta z$, we clearly see that $\Delta z = (w^\star)^2/M_1$. Hence, using the $a$ output from \Cref{thm:hallgren}, we can compute $\hat{w} = M_1/a$ satisfying
\begin{equation}
  \frac{M_1}{a} \in \left[w^\star \pm \frac{(w^\star)^2}{M_1}\right]
\end{equation}
with probability $\Omega(\eta^2/\log^4(A))$, where $\eta = 33/35$.
Here, in order to guarantee that $|\hat{w} - w^\star| \leq \epsilon_1$ for some $\epsilon_1 > 0$, we should choose the discretization parameter $M_1$ as $M_1 \geq R_w^2/\epsilon_1$, which is satisfied by our choice of $M_1$.
Then, the success probability simplifies to
\begin{equation}
  p = \Omega\left(\frac{1}{\log^4(A)}\right) = \Omega\left(\frac{1}{\log^4(M_1 d^2/R_w)}\right).
\end{equation}
To boost the success probability to at least $1-\delta$ (using the verification procedure to check if the period is correct), for some $\delta > 0$, we can repeat this $\mathcal{O}(\log(1/\delta)/p)$ times.
In total, this is
\begin{equation}
  \mathcal{O}\left(\log\left(\frac{1}{\delta}\right) \log^4\left(\frac{M_1 d^2}{R_w}\right)\right)
\end{equation}
repetitions, where in each repetition, we use $\mathcal{O}(D \log (M_1 d^2 /R_w))$ QSQs from the above analysis.

Finally, the generalization to arbitrary $d \geq 1$ is straightforward, using the observable $O_j$ from \Cref{eq:oj}.
Here, we only perform quantum Fourier sampling one coordinate at a time.
In this case, the function we are Fourier sampling from is effectively
\begin{equation}
    \label{eq:general-d}
    g_{w^\star, j}(x_j; x_{-j}) \triangleq \sum_{k=1}^D \beta_k^\star \cos(2\pi k (x_j w_j^\star + x_{-j}^\intercal w_{-j}^\star)),
\end{equation}
where $x_{-j}$ denotes the vector $x$ with all coordinates except the $j$-th one.
Here, $x_{-j}$ is a fixed vector because the observable $O_j$ collapses the register storing all but the $j$-th coordinate of the input.
Thus, we can consider the function
\begin{equation}
    \tilde{g}_j(z; x_{-j}) \triangleq \sum_{k=1}^D \beta_k^\star \cos(2\pi (z + x_{-j}^\intercal w_{-j}^\star)).
\end{equation}
This function clearly satisfies the conditions of \Cref{lem:discrete-general-unif}.
Thus, the resulting discretized function $h_{w^\star, M_1, M_2}$ with the $x_{-j}$ coordinates fixed is also $(33/35)$-pseudoperiodic with period $M_1/w_j^\star$.
Hence, we can apply Hallgren's algorithm one coordinate at a time, learning $M_1/w_j^\star$.
It is clear that the argument above still holds for this case as well.
For this, we need to repeat the algorithm to learn each entry of the vector $w^\star \in \mathbb{R}^d$ at a time.
Altogether, this gives the bound from \Cref{thm:linear-uniform}.
\end{proof}

\subsection{Learning the outer function via gradient methods}
\label{sec:outer-uniform}

From the previous section (in particular, \Cref{thm:linear-uniform}), we have seen that we can obtain an approximation $\hat{w}$ of $w^\star$ such that $\norm{\hat{w} - w^\star}_\infty \leq \epsilon_1$ with high probability, for some $\epsilon_1 > 0$.
In this section, we complete the algorithm for the uniform case by leveraging this approximation of $w^\star$ to learn the outer periodic function $\tilde{g}: \mathbb{R} \to [-1,1]$ via classical gradient methods.
We emphasize here that this portion of the algorithm is purely classical, where we have classical access to the loss function and its gradients.
Recall that we assume that $\tilde{g}$ takes the specific form given in \Cref{eq:g-tilde}, reproduced here for convenience:
\begin{equation}
  \tilde{g}(y) = \sum_{j=1}^D \beta_j^\star \cos(2\pi j y), \quad \norm{\beta^\star}_1 = 1,
\end{equation}
for some constant $D > 0$.
In this way, then our target function can be written as
\begin{equation}
  g_{w^\star}(x) = \tilde{g}(x^\intercal w^\star) = \sum_{j=1}^D \beta_j^\star \cos(2\pi j x^\intercal w^\star).
\end{equation}
Also recall that our ultimate goal is to find a good predictor $f_\theta(x)$ that minimizes the objective function given by
\begin{equation}
  \mathcal{L}_{w^\star}(\theta) = \mathop{\mathbb{E}}_{x \sim \varphi^2}[(f_\theta(x) - g_{w^\star}(x))^2],
\end{equation}
where $\theta$ are some parameters that we want to learn and $\varphi^2$ in this case is a uniform distribution.
Here, because we assume this simple form of $\tilde{g}$, then the predictors take a similar form
\begin{equation}
  f_\beta(x) = \sum_{j=1}^D \beta_j \cos(2\pi j x^\intercal \hat{w}),
\end{equation}
where $\hat{w}$ is our approximation of $w^\star$ from \Cref{thm:linear-uniform}.
Thus, the parameters that we want to learn here are given by the $\beta \in \mathbb{R}^d$.
Then, our loss function can be written more explicitly as
\begin{equation}
  \label{eq:loss}
  \mathcal{L}_{w^\star}(\beta) = \int\limits_{x\sim \varphi^2}\left(\sum_{j=1}^D \beta_j^\star \cos(2\pi j x^\intercal w^\star) - \sum_{j=1}^D \beta_j \cos(2\pi j x^\intercal \hat{w})\right)^2\,dx.
\end{equation}
As in the classical hardness result~\cite{shamir2018distribution}, our algorithm is given access to this loss function and its gradients.
Using this, we design a classical algorithm that can efficiently find a predictor specified by parameters $\hat{\beta}$ such that $\mathcal{L}_{w^\star}(\hat{\beta}) \leq \epsilon$ for a given precision $\epsilon > 0$.

Recall in the previous section that we needed to discretize and truncate our access to the target function $g_{w^\star}$.
We no longer require discretization since classically we can perform computations up to arbitrary precision, but we still truncate with truncation parameter $R$.
Namely, we consider $\varphi^2$ as the uniform distribution over an $\ell_1$-ball of radius $R$ centered at the origin in $\mathbb{R}^d$.
To show that $\mathcal{L}_{w^\star}(\hat{\beta}) \leq \epsilon$, we appropriately choose $R$ and $\epsilon_1$ sufficiently large/small enough, respectively.

\begin{theorem}[Learning $\tilde{g}$ Guarantee; Uniform Case]
\label{thm:g-tilde-uniform}
Let $\epsilon > 0$.
Let $w^\star \in \mathbb{R}^d$ be unknown with norm $R_w > 0$.
Let $g_{w^\star}:\mathbb{R}^d \to [-1,1]$ be defined as $g_{w^\star}(x) = \tilde{g}(x^\intercal w^\star)$ for $\tilde{g}$ given in \Cref{eq:g-tilde}.
Choose
\begin{equation}
  \label{eq:final-R-bounds}
  R = \tilde{\Omega}\left(\max\left(\frac{D^2}{\epsilon}, \frac{D^2\sqrt{d}}{R_w \epsilon}, \frac{D^{5/2}}{\sqrt{\epsilon}}, \frac{D^{3/2}\sqrt{d}}{R_w \sqrt{\epsilon}}\right)\right),
\end{equation}
\begin{equation}
  \label{eq:final-eps1-bounds}
  \epsilon_1 = \tilde{\mathcal{O}}\left(\min\left(\frac{\epsilon^3}{D^6 d}, \frac{\epsilon^{3/2}}{D^{13/2}d}, \frac{R_w}{D\sqrt{d}}\right)\right).
\end{equation}
Suppose we have an approximation $\hat{w} \in \mathbb{R}^d$ such that $\norm{\hat{w} - w^\star}_\infty \leq \epsilon_1$.
Then, there exists a classical algorithm with access to the loss function from \Cref{eq:loss} and its derivatives that can efficiently find a parameters $\hat{\beta} \in \mathbb{R}^d$ such that $\mathcal{L}_{w^\star}(\hat{\beta}) \leq \epsilon$.
Moreover, this algorithm requires at most
\begin{equation}
t = \Theta\left(\log\left(\sqrt{\frac{D}{\epsilon}}\right)\right)
\end{equation}
iterations of gradient descent.
\end{theorem}

The rest of this section is dedicated to proving this theorem.
The algorithm is simple: just run gradient descent using the loss function to estimate the parameters $\beta^\star$.
We prove this using arguments from convex optimization (see, e.g.,~\cite{nesterov2018lectures}).
Throughout the proof, we require some technical lemmas bounding integrals of exponential functions over our truncated domain, which we relegate to Appendix~\ref{sec:int-bounds}.

\textbf{Proof sketch.} The proof of \Cref{thm:g-tilde-uniform} is fairly technical, but the idea is simple.
First, we show that the gradients are informative, i.e., taking the derivative of our loss function with respect to each of the parameters $\beta_k$ indeed reflects how far $\beta_k$ is from the true parameter $\beta^\star_k$.
Then, we can just apply the standard gradient descent algorithm (see, e.g.,~\cite{nesterov2018lectures}).
Much of the work then goes into choosing the parameters (e.g., number of iterations to run gradient descent, how accurate we need period finding to be, etc.) to guarantee that the value of the loss function is small.
Throughout, we use the following notation: $\epsilon_1$ denotes the error for our estimate of $w^\star$ (in $\ell_\infty$-norm), $\epsilon_2$ quantifies how informative the gradients are, $\epsilon_3$ denotes the error for our estimate of $\beta^\star$ (in $\ell_2$-norm), and $\epsilon$ is the desired value of the loss function.

First, we show that the gradients are informative in the following lemma.
The idea is that we can choose $R$ sufficiently large and $\epsilon_1$ sufficiently small so that $\partial \mathcal{L}_{w^\star}/\partial \beta_k$ is close to $(\beta_k - \beta_k^\star)$.

\begin{lemma}[Informative gradients]
\label{lem:grad-inform}
Let $w^\star \in \mathbb{R}^d$ be unknown with norm $R_w > 0$.
Let $g_{w^\star}:\mathbb{R}^d \to [-1,1]$ be defined as $g_{w^\star}(x) = \tilde{g}(x^\intercal w^\star)$ for $\tilde{g}$ given in \Cref{eq:g-tilde}.
Suppose we have an approximation $\hat{w} \in \mathbb{R}^d$ such that $\norm{\hat{w} - w^\star}_\infty \leq \epsilon_1$, for $0 < \epsilon_1 \leq R_w/(D\sqrt{d})$.
Then for any $k \in [D]$,
\begin{align}
  \left| \frac{\partial \mathcal{L}_{w^\star}}{\partial \beta_k} - (\beta_k - \beta_k^\star) \right| &\leq \left(\frac{\sqrt{d}}{2\pi R(R_w - \sqrt{d}\epsilon_1)} + \frac{10\pi^2 D^2 dR^2 \epsilon_1}{3}\right) \max(|\beta_k^\star|, |\beta_k|)\\
  &+ \frac{\sqrt{d}}{\pi R(R_w - D\sqrt{d}\epsilon_1)} + \sum_{\substack{j=1\\j\neq k}}^D |\beta_j| \frac{\sqrt{d}}{\pi R (R_w - \sqrt{d}\epsilon_1)}.
\end{align}
\end{lemma}

\begin{proof}
Recall that our loss function is
\begin{equation}
  \mathcal{L}_{w^\star}(\beta) = \int\limits_{x\sim \varphi^2}\left(\sum_{j=1}^D \beta_j^\star \cos(2\pi j x^\intercal w^\star) - \sum_{j=1}^D \beta_j \cos(2\pi j x^\intercal \hat{w})\right)^2\,dx.
\end{equation}
Taking the derivative of this with respect to $\beta_k$, we have
\begin{equation}
  \frac{\partial \mathcal{L}_{w^\star}}{\partial \beta_k} = -2\int_{x \sim \varphi^2}\cos(2\pi kx^\intercal \hat{w}) \left(\sum_{j=1}^D \beta_j^\star \cos(2\pi j x^\intercal w^\star) - \sum_{j=1}^D \beta_j \cos(2\pi j x^\intercal \hat{w})\right)\,dx.
\end{equation}
Separating out terms with $k \neq j$, we have
\begin{align}
  \frac{\partial \mathcal{L}_{w^\star}}{\partial \beta_k} &= 2\beta_k \int_{x \sim \varphi^2} \cos^2(2\pi kx^\intercal \hat{w})\,dx -2 \beta_k^\star\int_{x \sim \varphi^2} \cos(2\pi kx^\intercal \hat{w}) \cos(2\pi k x^\intercal w^\star)\,dx\\
  &- 2\int_{x\sim\varphi^2} \sum_{\substack{j=1\\j\neq k}}^D \beta_j^\star \cos(2\pi k x^\intercal \hat{w})\cos(2\pi jx^\intercal w^\star)\,dx\\
  &+ 2\int_{x\sim\varphi^2} \sum_{\substack{j=1\\j\neq k}}^D \beta_j \cos(2\pi kx^\intercal \hat{w})\cos(2\pi j x^\intercal \hat{w})\,dx.
\end{align}
We can upper and lower bound this expression using the integral bounds from Appendix~\ref{sec:int-bounds}.
First, to upper bound, we can use \Cref{coro:integral-upper}, \Cref{lem:integral}, \Cref{coro:integral2-wstar-hat}, and \Cref{lem:integral2}, for each of the terms respectively.
Then, we have
\begin{align}
  \frac{\partial \mathcal{L}_{w^\star}}{\partial \beta_k} &\leq 2\beta_k \int_{x \sim \varphi^2} \cos^2(2\pi kx^\intercal \hat{w})\,dx -2 \beta_k^\star\int_{x \sim \varphi^2} \cos(2\pi kx^\intercal \hat{w}) \cos(2\pi k x^\intercal w^\star)\,dx\\
  &+ 2\sum_{\substack{j=1\\j\neq k}}^D |\beta_k^\star| \left|\int_{x \sim \varphi^2} \cos(2\pi k x^\intercal \hat{w}) \cos(2\pi j x^\intercal w^\star)\,dx \right|+ 2\sum_{\substack{j=1\\j\neq k}}^D |\beta_k| \left|\int_{x \sim \varphi^2} \cos(2\pi k x^\intercal \hat{w}) \cos(2\pi j x^\intercal \hat{w})\,dx \right|\\
  &\leq \beta_k - \beta_k^\star + \frac{\sqrt{d}}{4\pi R(R_w - \sqrt{d}\epsilon_1)}|\beta_k| + \left(\frac{\sqrt{d}}{4\pi R_w R} + \frac{10\pi^2 D^2 dR^2 \epsilon_1}{3}\right) |\beta_k^\star| \\
  &+\sum_{\substack{j=1\\j\neq k}}^D |\beta_j^\star| \frac{\sqrt{d}}{\pi R(R_w - D\sqrt{d}\epsilon_1)} + \sum_{\substack{j=1\\j\neq k}}^D |\beta_j| \frac{\sqrt{d}}{\pi R(R_w - \sqrt{d}\epsilon_1)}\\
  &\leq (\beta_k - \beta_k^\star) + \left(\frac{\sqrt{d}}{2\pi R(R_w - \sqrt{d}\epsilon_1)} + \frac{10\pi^2 D^2 dR^2 \epsilon_1}{3}\right) \max(|\beta_k^\star|, |\beta_k|)\\
  &+ \frac{\sqrt{d}}{\pi R (R_w - D\sqrt{d} \epsilon_1)} + \sum_{\substack{j=1\\j\neq k}}^D |\beta_j| \frac{\sqrt{d}}{\pi R(R_w - \sqrt{d}\epsilon_1)}
\end{align}
where in the second inequality, we use \Cref{coro:integral-upper}, \Cref{lem:integral}, \Cref{coro:integral2-wstar-hat}, and \Cref{lem:integral2} for each term respectively.
In the third inequality, we use that $\max(|\beta_k^\star|, |\beta_k|) \geq |\beta_k^\star|, |\beta_k|$ and $\norm{\beta^\star}_1 = 1$ so that $\sum_{j\neq k} |\beta_j^\star| \leq 1$.
We also use that $R_w \geq R_w - \sqrt{d}\epsilon_1$.

We can also obtain a similar lower bound using \Cref{coro:integral}, \Cref{lem:integral-upper}, \Cref{coro:integral2-wstar-hat}, and \Cref{lem:integral2}.
\begin{align}
  \frac{\partial \mathcal{L}_{w^\star}}{\partial \beta_k} &\geq (\beta_k - \beta_k^\star) - \frac{\sqrt{d}}{4\pi R(R_w - \sqrt{d}\epsilon_1)}|\beta_k| - \left(\frac{\sqrt{d}}{4\pi R_w R} + 2\pi D d\epsilon_1 R\right)|\beta_k^\star|\\
  &- 2\int_{x\sim\varphi^2} \sum_{\substack{j=1\\j\neq k}}^D \beta_j^\star \cos(2\pi k x^\intercal \hat{w})\cos(2\pi jx^\intercal w^\star)\,dx + 2\int_{x\sim\varphi^2} \sum_{\substack{j=1\\j\neq k}}^D \beta_j \cos(2\pi kx^\intercal \hat{w})\cos(2\pi j x^\intercal \hat{w})\,dx\\
  &\geq (\beta_k - \beta_k^\star) - \left(\frac{\sqrt{d}}{2\pi R(R_w - \sqrt{d}\epsilon_1)} + 2\pi D d\epsilon_1 R\right)\max(|\beta_k^\star|, |\beta_k|)\\
  &-\left|-2\int_{x\sim\varphi^2} \sum_{\substack{j=1\\j\neq k}}^D \beta_j^\star \cos(2\pi k x^\intercal \hat{w})\cos(2\pi jx^\intercal w^\star)\,dx\right| - \left|2\int_{x\sim\varphi^2} \sum_{\substack{j=1\\j\neq k}}^D \beta_j \cos(2\pi kx^\intercal \hat{w})\cos(2\pi j x^\intercal \hat{w})\,dx\right|\\
  &\geq (\beta_k - \beta_k^\star) - \left(\frac{\sqrt{d}}{2\pi R(R_w - \sqrt{d}\epsilon_1)} + 2\pi D d\epsilon_1 R\right)\max(|\beta_k^\star|, |\beta_k|)\\
  &-\frac{\sqrt{d}}{\pi R(R_w - D\sqrt{d}\epsilon_1)} - \sum_{\substack{j=1\\j\neq k}}^D |\beta_j| \frac{\sqrt{d}}{\pi R (R_w - \sqrt{d}\epsilon_1)}\\
  &\geq (\beta_k - \beta_k^\star) - \left(\frac{\sqrt{d}}{2\pi R(R_w - \sqrt{d}\epsilon_1)} + \frac{10\pi^2 D^2dR^2\epsilon_1}{3}\right)\max(|\beta_k^\star|, |\beta_k|)\\
  &-\frac{\sqrt{d}}{\pi R(R_w - D\sqrt{d}\epsilon_1)} - \sum_{\substack{j=1\\j\neq k}}^D |\beta_j| \frac{\sqrt{d}}{\pi R (R_w - \sqrt{d}\epsilon_1)}.
\end{align}
In the first inequality, we use \Cref{coro:integral} and \Cref{lem:integral-upper}.
In the second inequality, we use that $\max(|\beta_k^\star|, |\beta_k|) \geq |\beta_k|, |\beta_k^\star|$ and $R_w \geq R_w - \sqrt{d}\epsilon_1$.
In the third inequality, we use \Cref{lem:integral2} and \Cref{coro:integral2-wstar-hat}.
We also use that $\norm{\beta^\star}_1 = 1$ so that $\sum_{j \neq k}|\beta_j^\star| \leq 1$.
In the last inequality, we use that $2\pi D d\epsilon_1 R \leq 10\pi^2 D^2 dR^2 \epsilon_1/3$.

Combining these two inequalities, we have that
\begin{align}
  \left| \frac{\partial \mathcal{L}_{w^\star}}{\partial \beta_k} - (\beta_k - \beta_k^\star) \right| &\leq \left(\frac{\sqrt{d}}{2\pi R(R_w - \sqrt{d}\epsilon_1)} + \frac{10\pi^2 D^2 dR^2 \epsilon_1}{3}\right) \max(|\beta_k^\star|, |\beta_k|)\\
  &+ \frac{\sqrt{d}}{\pi R(R_w - D\sqrt{d}\epsilon_1)} + \sum_{\substack{j=1\\j\neq k}}^D |\beta_j| \frac{\sqrt{d}}{\pi R (R_w - \sqrt{d}\epsilon_1)}.
\end{align}
\end{proof}

Now, we can use standard gradient descent, which converges as follows.

\begin{lemma}[Gradient descent convergence]
\label{lem:grad-converge}
Let $\epsilon_1, \epsilon_2 > 0$.
Let $w^\star \in \mathbb{R}^d$ be unknown with norm $R_w > 0$.
Let $g_{w^\star}:\mathbb{R}^d \to [-1,1]$ be defined as $g_{w^\star}(x) = \tilde{g}(x^\intercal w^\star)$ for $\tilde{g}$ given in \Cref{eq:g-tilde}.
Suppose we have an approximation $\hat{w} \in \mathbb{R}^d$ such that $\norm{\hat{w} - w^\star}_\infty \leq \epsilon_1$.
Also, suppose that
\begin{equation}
\label{eq:grad-inform}
\left|\frac{\partial \mathcal{L}_{w^\star}}{\partial \beta_k}(\beta_k^{(t)}) - (\beta_k^{(t)} - \beta_k^\star)\right| < \epsilon_2
\end{equation}
for all $k \in [D]$. Here, $t$ denotes the step of gradient descent. Then, gradient descent with step size $\eta = \mathcal{O}(1)$ with $0 < \eta < 1$ and initial point $\beta^{(0)} = 0$ converges as follows:
\begin{equation}
|\beta_k^{(t+1)} - \beta_k^\star| \leq (1-\eta)^t + \eta t\epsilon_2.
\end{equation}
\end{lemma}

\begin{proof}
This proof is straightforward following the standard gradient descent rule
\begin{equation}
\beta_k^{(t+1)} = \beta_k^{(t)} - \eta \frac{\partial \mathcal{L}_{w^\star}}{\partial \beta_k}(\beta_k^{(t)}).
\end{equation}
Plugging this in and applying Eq.~\eqref{eq:grad-inform}, we have
\begin{align}
|\beta_k^{(t+1)} - \beta_k^\star| &= \left|\beta_k^{(t)} - \eta\frac{\partial \mathcal{L}_{w^\star}}{\partial \beta_k}(\beta_k^{(t)}) - \beta_k^\star\right|\\
&\leq \left|\beta_k^{(t)} - \beta_k^\star - \eta(\beta_k^{(t)} - \beta_k^\star) + \eta \epsilon_2\right|\\
&\leq (1-\eta)|\beta_k^{(t)} - \beta_k^\star| + \eta\epsilon_2.
\end{align}
Applying this inequality recursively, we have
\begin{align}
|\beta_k^{(t+1)} - \beta_k^\star| &\leq (1-\eta)^t |\beta_k^{(0)} - \beta_k^\star| + \eta\sum_{i=1}^t (1-\eta)^i \epsilon_2\\
&\leq (1-\eta)^t |\beta_k^{(0)} - \beta_k^\star| + \eta t \epsilon_2,
\end{align}
where the last line follows because $0 < \eta < 1$ so that $0 < 1-\eta < 1$.
Now, because we initialize to $\beta^{(0)} = 0$, then
\begin{equation}
|\beta_k^{(0)} - \beta_k^\star| = |\beta_k^\star| \leq 1,
\end{equation}
where $|\beta_k^\star| \leq 1$ because $\norm{\beta_k^\star}_1 =1$.
Thus, we have
\begin{equation}
|\beta_k^{(t+1)} - \beta_k^\star| \leq (1-\eta)^t + \eta t \epsilon_2,
\end{equation}
as claimed.
\end{proof}

To help us choose parameters such as $\epsilon_1, \epsilon_2,$ and $t$ properly, we also need to show that the updated parameters via gradient descent do not become too large.
In particular, recall from \Cref{eq:g-tilde} that the true parameters satisfy $|\beta^\star_k| < 1$ because $\norm{\beta^\star}_1 = 1$.
The following lemma states that the parameters found via gradient descent are not much larger than this.

\begin{lemma}[Parameter bound]
\label{lem:param-bound}
Let $w^\star \in \mathbb{R}^d$ be unknown with norm $R_w > 0$. Let $g_{w^\star}:\mathbb{R}^d \to [-1,1]$ be defined as $g_{w^\star}(x) = \tilde{g}(x^\intercal w^\star)$ for $\tilde{g}$ given in \Cref{eq:g-tilde}.
Suppose we have an approximation $\hat{w} \in \mathbb{R}^d$ such that $\norm{\hat{w} - w^\star}_\infty \leq \epsilon_1$.
Suppose
\begin{equation}
  R \geq \max\left(D^2, \frac{16 D\sqrt{d}}{\pi R_w}\right), \quad \epsilon_1 \leq \min\left(\frac{3}{40 \pi^2 D^6 d}, \frac{R_w}{2D\sqrt{d}}\right).
\end{equation}
Then, 
\begin{equation}
|\beta_k^{(t)}| < 2
\end{equation}
for all $k \in [D]$. Here, $\beta_k^{(t)}$ denotes the parameters at the $t$-th step of gradient descent.
\end{lemma}

\begin{proof}
We prove this by induction on the $t$ steps of gradient descent. For the base case of $t = 0$, this is clearly satisfied by our choice of initialization. Namely, we initialize to $\beta_k^{(0)} = 0$ for all $k$. Thus, we clearly have $|\beta_k^{(0)}| = 0 < 2$ for all $k \in [D]$.

For the inductive step, suppose that for some step $t > 0$ that $|\beta_k^{(t)}| < 2$ for all $k \in [D]$.
We want to prove that $|\beta_k^{(t + 1)}| < 2$ for all $k \in [D]$. Let $k \in [D]$. 
By~\Cref{lem:grad-inform},
\begin{align}
  \left| \frac{\partial \mathcal{L}_{w^\star}}{\partial \beta_k} - (\beta_k - \beta_k^\star) \right| &\leq \left(\frac{\sqrt{d}}{2\pi R(R_w - \sqrt{d}\epsilon_1)} + \frac{10\pi^2 D^2 dR^2 \epsilon_1}{3}\right) \max(|\beta_k^\star|, |\beta_k|)\\
  &+ \frac{\sqrt{d}}{\pi R(R_w - D\sqrt{d}\epsilon_1)} + \sum_{\substack{j=1\\j\neq k}}^D |\beta_j| \frac{\sqrt{d}}{\pi R (R_w - \sqrt{d}\epsilon_1)}.
\end{align}
Note that the condition needed for \Cref{lem:grad-inform} (i.e., $\epsilon_1 \leq R_w/(D\sqrt{d})$) is satisfied for our choice of $\epsilon_1$.
Using that $D\epsilon_1 \geq \epsilon_1$ (since $D \geq 1$), we can simplify this:
\begin{align}
  \left| \frac{\partial \mathcal{L}_{w^\star}}{\partial \beta_k} - (\beta_k - \beta_k^\star) \right| &\leq \left(\frac{\sqrt{d}}{2\pi R(R_w - D\sqrt{d}\epsilon_1)} + \frac{10\pi^2 D^2 dR^2 \epsilon_1}{3}\right) \max(|\beta_k^\star|, |\beta_k|)\\
  &+ \frac{\sqrt{d}}{\pi R(R_w - D\sqrt{d}\epsilon_1)} + \sum_{\substack{j=1\\j\neq k}}^D |\beta_j| \frac{\sqrt{d}}{\pi R (R_w - D\sqrt{d}\epsilon_1)}\\
  &= \left(\frac{\sqrt{d}}{2\pi R(R_w - D\sqrt{d}\epsilon_1)} + \frac{10\pi^2 D^2 dR^2 \epsilon_1}{3}\right) \max(|\beta_k^\star|, |\beta_k|) + \frac{\sqrt{d}\sum_{\substack{j=1,j\neq k}}^D |\beta_j| + \sqrt{d}}{\pi R(R_w - D\sqrt{d}\epsilon_1)}.
\end{align}
Evaluating at $\beta_k = \beta_k^{(t)}$, we have
\begin{align}
  &\left|\frac{\partial \mathcal{L}_{w^\star}}{\partial \beta_k}(\beta_k^{(t)}) - (\beta_k^{(t)} - \beta_k^\star) \right|\\
  &\leq \left(\frac{\sqrt{d}}{2\pi R(R_w - D\sqrt{d}\epsilon_1)} + \frac{10\pi^2 D^2 dR^2 \epsilon_1}{3}\right) \max(|\beta_k^\star|, |\beta_k^{(t)}|) + \frac{\sqrt{d}\sum_{\substack{j=1,j\neq k}}^D |\beta_j^{(t)}| + \sqrt{d}}{\pi R(R_w - D\sqrt{d}\epsilon_1)}\\
  &\leq \frac{\sqrt{d}}{\pi R(R_w - D\sqrt{d}\epsilon_1)} + \frac{20\pi^2 D^2 dR^2 \epsilon_1}{3} + \frac{2D+1}{\pi R(R_w/\sqrt{d} - D\epsilon_1)}\\
  &= \frac{20\pi^2 D^2 dR^2 \epsilon_1}{3} + \frac{2D+2}{\pi R(R_w/\sqrt{d} - D\epsilon_1)},\label{eq:grad-bound}
\end{align}
where in the second to last line we used the inductive hypothesis.
We will use this to bound the parameters after one step of gradient descent. Recall that the update rule for gradient descent is
\begin{equation}
\beta_k^{(t+1)} = \beta_k^{(t)} - \eta \frac{\partial \mathcal{L}_{w^\star}}{\partial \beta_k}(\beta_k^{(t)})
\end{equation}
for a step size $\eta = \mathcal{O}(1)$. Then, using the above inequality, we have
\begin{align}
  |\beta_k^{(t+1)}| &= \left|\beta_k^{(t)} - \eta \frac{\partial \mathcal{L}_{w^\star}}{\partial \beta_k}(\beta_k^{(t)})\right|\\
  &\leq \left|\beta_k^{(t)} - \eta(\beta_k^{(t)} - \beta_k^\star) + \eta\left(\frac{20\pi^2 D^2 dR^2 \epsilon_1}{3} + \frac{2D+2}{\pi R(R_w/\sqrt{d} - D\epsilon_1)}\right)\right|\\
  &< 2(1-\eta) + \eta + \eta\left(\frac{20\pi^2 D^2 dR^2 \epsilon_1}{3} + \frac{2D+2}{\pi R(R_w/\sqrt{d} - D\epsilon_1)}\right)\\
  &= 2 - \eta + \eta\left(\frac{20\pi^2 D^2 dR^2 \epsilon_1}{3} + \frac{2D+2}{\pi R(R_w/\sqrt{d} - D\epsilon_1)}\right)\\
  &\leq 2 - \eta + \eta\left(\frac{20\pi^2 D^2 dR^2 \epsilon_1}{3} + \frac{4D}{\pi R(R_w/\sqrt{d} - D\epsilon_1)}\right),
\end{align}
where in the second line, we used~\Cref{eq:grad-bound}.
In the third line, we used triangle inequality, the inductive hypothesis that $|\beta_k^{(t)}| < 2$, and $|\beta_k^\star| < 1$.
In the last line, we use $D \geq 1$.
In order to achieve the result, we need
\begin{equation}
-\eta + \eta\left(\frac{20\pi^2 D^2 dR^2 \epsilon_1}{3} + \frac{4D}{\pi R(R_w/\sqrt{d} - D\epsilon_1)}\right) \leq 0.
\end{equation}
Rearranging, we need to show that
\begin{equation}
  \label{eq:param-to-bound}
  \frac{20\pi^2 D^2 dR^2 \epsilon_1}{3} + \frac{4D}{\pi R(R_w/\sqrt{d} - D\epsilon_1)} \leq 1.
\end{equation}
Consider taking
\begin{equation}
  \label{eq:param-eps1-R-bounds}
  R \geq \max\left(D^2, \frac{16 D\sqrt{d}}{\pi R_w}\right), \quad \epsilon_1 \leq \min\left(\frac{3}{40 \pi^2 D^6 d}, \frac{R_w}{2D\sqrt{d}}\right).
\end{equation}
We want to show that these choices of $R,\epsilon_1$ allow us to bound each term on the lefthand side by $1/2$ to obtain the required bound.
For the first term, consider taking $\epsilon_1 \leq 3/(40\pi^2 R^3 d)$ and $R$ as in the first element in the max of \Cref{eq:param-eps1-R-bounds}, we have
\begin{equation}
  \frac{20\pi^2 D^2 dR^2 \epsilon_1}{3} \leq \frac{D^2}{2R} \leq \frac{1}{2}.
\end{equation}
Finally, for the last term, using $\epsilon_1 \leq R_w/(2D\sqrt{d})$ and $R \geq 16D\sqrt{d}/(\pi R_w)$, we have
\begin{equation}
  \frac{4D}{\pi R(R_w/\sqrt{d} - D\epsilon_1)} \leq \frac{8D\sqrt{d}}{\pi R R_w} \leq \frac{1}{2}.
\end{equation}
\end{proof}

With the past three lemmas, we can now begin to set the parameters involved to obtain the desired guarantees.
As a corollary of \Cref{lem:param-bound}, we can obtain the number of steps $t$ and accuracy of the gradient $\epsilon_2$ needed to achieve a desired accuracy for gradient descent.

\begin{corollary}[Convergence steps and accuracy]
\label{coro:grad-steps}
Let $w^\star \in \mathbb{R}^d$ be unknown with norm $R_w > 0$ and $|w_i^\star| \geq R_w/d^2$ for all $i \in [d]$.
Let $g_{w^\star}:\mathbb{R}^d \to [-1,1]$ be defined as $g_{w^\star}(x) = \tilde{g}(x^\intercal w^\star)$ for $\tilde{g}$ given in \Cref{eq:g-tilde}.
Suppose that
\begin{equation}
\left|\frac{\partial \mathcal{L}_{w^\star}}{\partial \beta_k}(\beta_k^{(t)}) - (\beta_k^{(t)} - \beta_k^\star)\right| < \epsilon_2
\end{equation}
for $\epsilon_2 > 0$, for all $k \in [D]$. Here, $t$ denotes the $t$-th step of gradient descent.
Let $\epsilon_3 > 0$. Then, gradient descent with step size $\eta = \mathcal{O}(1)$ with $0 < \eta < 1$ and initial point $\beta^{(0)} = 0$ requires
\begin{equation}
t = \Theta\left(\log(\sqrt{D}/\epsilon_3)\right)
\end{equation}
and
\begin{equation}
\epsilon_2 = \mathcal{O}\left(\frac{\epsilon_3}{\sqrt{D}\log(\sqrt{D}/\epsilon_3)}\right)
\end{equation}
to converge such that
\begin{equation}
\norm{\beta^{(t+1)} - \beta^\star}_2 \leq \epsilon_3.
\end{equation}
\end{corollary}

\begin{proof}
By~\Cref{lem:grad-converge}, we have
\begin{equation}
  |\beta_k^{(t+1)} - \beta_k^\star| \leq (1-\eta)^t + \eta t \epsilon_2.
\end{equation}
Then, in order to have $(1-\eta)^t \leq \epsilon_3/(2\sqrt{D})$, we can use
\begin{equation}
  t\log(1-\eta) = \log\left(\frac{\epsilon_3}{2\sqrt{D}}\right).
\end{equation}
Solving for $t$, we obtain
\begin{equation}
  t = \frac{\log(2\sqrt{D}/\epsilon_3)}{\log(1/c)}
\end{equation}
for $c = 1-\eta < 1$. Since $\eta$ is a constant, then we obtain the claim.
It remains to find $\epsilon_2$ such that
\begin{equation}
  \eta t \epsilon_2 < \frac{\epsilon_3}{2\sqrt{D}}.
\end{equation}
Plugging in our previously found $t$, then we arrive at
\begin{equation}
\label{eq:eps2-upper-eps3}
  \epsilon_2 \leq \frac{\log(1/c)\epsilon_3}{2\eta \sqrt{D}\log(2\sqrt{D}/\epsilon_3)},
\end{equation}
where again taking $\eta =\mathcal{O}(1)$ gives the claim.
Putting these two pieces together, we have
\begin{equation}
  |\beta_k^{(t+1)} - \beta_k^\star| \leq (1-\eta)^t+ \eta t \epsilon_2 \leq \frac{\epsilon_3}{\sqrt{D}}.
\end{equation}
Finally, we obtain the $2$-norm bound
\begin{equation}
  \norm{\beta^{(t+1)} - \beta^\star}_2 = \sqrt{\sum_{k=1}^D |\beta^{(t+1)}_k - \beta_k^\star|^2} \leq \epsilon_3.
\end{equation}
\end{proof}

With this, we have set an accuracy $\epsilon_2$, which we need the gradients to satisfy.
Using \Cref{lem:grad-inform} and \Cref{lem:param-bound}, we show that we can achieve this $\epsilon_2$ accuracy from \Cref{coro:grad-steps} by setting the parameters $R, \epsilon_1$ appropriately.

\begin{corollary}[Achieving required gradient accuracy]
\label{coro:grad-accuracy}
Let $1 > \epsilon_2, \epsilon_3 > 0$.
Let $w^\star \in \mathbb{R}^d$ be unknown with norm $R_w > 0$ and $|w_i^\star| \geq R_w/d^2$.
Let $g_{w^\star}:\mathbb{R}^d \to [-1,1]$ be defined as $g_{w^\star}(x) = \tilde{g}(x^\intercal w^\star)$ for $\tilde{g}$ given in \Cref{eq:g-tilde}.
Suppose we have an approximation $\hat{w} \in \mathbb{R}^d$ such that $\norm{\hat{w} - w^\star}_\infty \leq \epsilon_1$.
Suppose that
\begin{equation}
  R \geq \max\left(\frac{D^2}{\epsilon_2}, \frac{16D\sqrt{d}}{\pi R_w \epsilon_2}\right),\quad \epsilon_1 \leq \min\left(\frac{3\epsilon_2^3}{40\pi^2 D^6 d}, \frac{R_w}{2D\sqrt{d}}\right).
\end{equation}
Then, we can achieve
\begin{equation}
  \left|\frac{\partial \mathcal{L}_{w^\star}}{\partial \beta_k}(\beta_k^{(t)}) - (\beta_k^{(t)} - \beta_k^\star)\right| < \epsilon_2
\end{equation}
for all $k \in [D]$, where
\begin{equation}
  \epsilon_2 = \mathcal{O}\left(\frac{\epsilon_3}{\sqrt{D}\log(\sqrt{D}/\epsilon_3)}\right).
\end{equation}
Here, $t$ denotes the $t$-th step of gradient descent. 
Writing the bounds on $R$ and $\epsilon_1$ in terms of $\epsilon_3$, we have
\begin{equation}
  R = \tilde{\Omega}\left(\max\left(\frac{D^{5/2}}{\epsilon_3}, \frac{D^{3/2}\sqrt{d}}{R_w \epsilon_3}\right)\right),\quad \epsilon_1 = \tilde{\mathcal{O}}\left(\min\left(\frac{\epsilon_3^2}{D^{13/2}d}, \frac{R_w}{D\sqrt{d}}\right)\right).
\end{equation}
\end{corollary}

\begin{proof}
We need to show that we can indeed achieve this $\epsilon_2$ error for the gradients.
This introduces some constraints on $R$ and $\epsilon_1$.
By \Cref{eq:grad-bound} (since we already proved this parameter bound in \Cref{lem:param-bound} and this result holds given our choice of $R, \epsilon_1$), we have
\begin{equation}
  \left|\frac{\partial\mathcal{L}_{w^\star}}{\partial \beta_k}(\beta_k^{(t)}) - (\beta_k^{(t)} - \beta_k^\star)\right| \leq \frac{20\pi^2 D^2 dR^2 \epsilon_1}{3} + \frac{2D+2}{\pi R(R_w/\sqrt{d} - D\epsilon_1)}.
\end{equation}
In order for gradient descent to converge well, as shown in \Cref{coro:grad-steps}, we need
\begin{equation}
  \left|\frac{\partial\mathcal{L}_{w^\star}}{\partial \beta_k}(\beta_k^{(t)}) - (\beta_k^{(t)} - \beta_k^\star)\right| \leq \epsilon_2 = \mathcal{O}\left(\frac{\epsilon_3}{\sqrt{D}\log\left(\sqrt{D}/\epsilon_3\right)}\right).
\end{equation}
Thus, we must set $R, \epsilon_1$ such that
\begin{equation}
  \label{eq:grad-to-bound}
  \frac{20\pi^2 D^2 dR^2 \epsilon_1}{3} + \frac{4D}{\pi R(R_w/\sqrt{d} - D\epsilon_1)} \leq \epsilon_2 =  \mathcal{O}\left(\frac{\epsilon_3}{\sqrt{D}\log\left(\sqrt{D}/\epsilon_3\right)}\right).
\end{equation}
This can be satisfied by taking 
\begin{equation}
  \label{eq:grad-R-bounds}
  R \geq \max\left(\frac{D^2}{\epsilon_2}, \frac{16D\sqrt{d}}{\pi R_w \epsilon_2}, D^2, \frac{16 D\sqrt{d}}{\pi R_w}\right) = \max\left(\frac{D^2}{\epsilon_2}, \frac{16D\sqrt{d}}{\pi R_w \epsilon_2}\right).
\end{equation}
\begin{equation}
  \label{eq:grad-eps1-bounds}
  \epsilon_1 \leq \min\left(\frac{3\epsilon_2^3}{40 \pi^2 D^6 d}, \frac{3}{40 \pi^2 D^6 d}, \frac{R_w}{2D\sqrt{d}}\right) = \min\left(\frac{3\epsilon_2^3}{40\pi^2 D^6 d}, \frac{R_w}{2D\sqrt{d}}\right).
\end{equation}
Note that the last two terms in the maximum for $R$ in \Cref{eq:grad-R-bounds} and in the minimum for $\epsilon_1$ in \Cref{eq:grad-eps1-bounds} are from the constraints on $R,\epsilon_1$ in \Cref{lem:param-bound}.
The equalities follow because $0 < \epsilon_2 < 1$.
We can write this in terms of $\epsilon_3$ by using upper bound of $\epsilon_2$ in terms of $\epsilon_3$ (\Cref{eq:eps2-upper-eps3})
\begin{align}
  \label{eq:grad-R-bounds2}
  R &\geq \max\left(\frac{2\eta D^{5/2}\log(2\sqrt{D}/\epsilon_3)}{\log(1/c)\epsilon_3}, \frac{32\eta D^{3/2} \sqrt{d} \log(2\sqrt{D}/\epsilon_3)}{\pi R_w \log(1/c)\epsilon_3}\right) = \tilde{\Omega}\left(\max\left(\frac{D^{5/2}}{\epsilon_3}, \frac{D^{3/2}\sqrt{d}}{R_w \epsilon_3}\right)\right).
\end{align}
\begin{align}
  \label{eq:grad-eps1-bounds2}
  \epsilon_1 \leq \min\left(\frac{3\log^3(1/c) \epsilon_3^3}{80 \pi^2 \eta D^{13/2} d \log(2\sqrt{D}/\epsilon_3)}, \frac{R_w}{2D\sqrt{d}}\right) = \tilde{\mathcal{O}}\left(\min\left(\frac{\epsilon_3^3}{D^{13/2} d}, \frac{R_w}{D\sqrt{d}}\right)\right).
\end{align}
where $\eta = \mathcal{O}(1)$ is the step size of gradient descent and $c = 1-\eta$.
We will prove that \Cref{eq:grad-to-bound} holds for the $\epsilon_2$ dependence.
Writing in terms of $\epsilon_3$ follows simply from the upper bound of $\epsilon_2$ in terms of $\epsilon_3$ in \Cref{eq:eps2-upper-eps3}.
We bound each term on the lefthand side of \Cref{eq:grad-to-bound} by $\epsilon_2/2$ to obtain the required bound.

For the first term, using that $R \geq D^2/\epsilon_2$ and $\epsilon_1 \leq 3/(40 \pi^2 R^3 d) \leq 3\epsilon_2^3/(40\pi^2 D^6 d)$, we have
\begin{equation}
  \frac{20\pi^2D^2 dR^2 \epsilon_1}{3} \leq \frac{D^2}{2R} \leq \frac{\epsilon_2}{2}.
\end{equation}
Finally, for the second term, using that $\epsilon_1 \leq R_w/(2D\sqrt{d})$ and $R \geq 16D\sqrt{d}/(\pi R_w \epsilon_2)$, we have
\begin{equation}
  \frac{4D}{\pi R(R_w/\sqrt{d} - D\epsilon_1)} \leq \frac{8D\sqrt{d}}{\pi R R_w} \leq \frac{\epsilon_2}{2}.
\end{equation}
This completes the proof.
\end{proof}

With these choices of parameters, we can plug them in to determine the value of the loss function.

\begin{lemma}[Loss bound]
\label{lem:loss-bound}
Let $\epsilon_3 > 0$. Let
\begin{equation}
  R = \tilde{\Omega}\left(\max\left(\frac{D^{5/2}}{\epsilon_3}, \frac{D^{3/2}\sqrt{d}}{R_w \epsilon_3}\right)\right),\quad \epsilon_1 = \tilde{\mathcal{O}}\left(\min\left(\frac{\epsilon_3^2}{D^{13/2}d}, \frac{R_w}{D\sqrt{d}}\right)\right).
\end{equation}
as in \Cref{coro:grad-accuracy}.
Let $w^\star \in \mathbb{R}^d$ be unknown with norm $R_w > 0$.
Let $g_{w^\star}:\mathbb{R}^d \to [-1,1]$ be defined as $g_{w^\star}(x) = \tilde{g}(x^\intercal w^\star)$ for $\tilde{g}$ given in \Cref{eq:g-tilde}.
Suppose we have an approximation $\hat{w} \in \mathbb{R}^d$ such that $\norm{\hat{w} - w^\star}_\infty \leq \epsilon_1$.
Then, gradient descent can find a predictor $\hat{\beta}$ such that
\begin{align}
  \mathcal{L}_{w^\star}(\hat{\beta}) &\leq \frac{\epsilon_3^2}{2} + \frac{13\sqrt{d}}{8\pi R_w R} + \frac{32\pi^2 D^2 dR^2 \epsilon_1}{3} + \frac{9D^2\sqrt{d}}{2\pi R(R_w - D\sqrt{d}\epsilon_1)}.
\end{align}
\end{lemma}

\begin{proof}
This proof will be somewhat similar to~\Cref{lem:grad-inform}.
First, let us expand the loss function:
\begin{align}
  \mathcal{L}_{w^\star}(\beta) &= \int\limits_{x\sim \varphi^2}\left(\sum_{j=1}^D \beta_j^\star \cos(2\pi j x^\intercal w^\star) - \sum_{j=1}^D \beta_j \cos(2\pi j x^\intercal \hat{w})\right)^2\,dx\\
  &\begin{aligned}
  = \int\limits_{x \sim \varphi^2} \sum_{j,j'=1}^D &\beta_j^\star \beta_{j'}^\star \cos(2\pi j x^\intercal w^\star) \cos(2\pi j' x^\intercal w^\star) + \beta_j \beta_{j'} \cos(2\pi j x^\intercal \hat{w}) \cos(2\pi j' x^\intercal \hat{w})\\& - 2\beta_j^\star \beta_{j'} \cos(2\pi j x^\intercal w^\star) \cos(2\pi j' x^\intercal \hat{w})\,dx
  \end{aligned}
\end{align}
Separating out terms with $j \neq j'$, we have
\begin{align}
  \mathcal{L}_{w^\star}(\beta) &= \sum_{j=1}^D \left((\beta_j^\star)^2 \int_{x\sim\varphi^2}\cos^2(2\pi j x^\intercal w^\star)\,dx + \beta_j^2 \int_{x \sim \varphi^2}\cos^2(2\pi j x^\intercal \hat{w})\,dx\right.\label{eq:term-squared}\\
  & \left.- 2\beta_j^\star \beta_j \int_{x \sim \varphi^2}\cos(2\pi jx^\intercal w^\star)\cos(2\pi j x^\intercal \hat{w})\,dx\right)\label{eq:termjj'equal}\\
  &+ \sum_{\substack{j,j'=1\\j \neq j'}}^D \beta_j^\star \beta_{j'}^\star \int_{x \sim \varphi^2} \cos(2\pi j x^\intercal w^\star)\cos(2\pi j' x^\intercal w^\star)\,dx\label{eq:term1-l}\\
  &+ \sum_{\substack{j,j'=1\\j \neq j'}}^D \beta_j \beta_{j'} \int_{x \sim \varphi^2}  \cos(2\pi j x^\intercal \hat{w}) \cos(2\pi j' x^\intercal \hat{w})\,dx\label{eq:term2-l}\\
  &-2\sum_{\substack{j,j'=1\\j \neq j'}}^D \beta_j^\star \beta_{j'} \int_{x \sim \varphi^2} \cos(2\pi j x^\intercal w^\star) \cos(2\pi j' x^\intercal \hat{w})\,dx\label{eq:term3-l}
\end{align}
We can upper bound the absolute values of the last three terms.
For the term in \Cref{eq:term1-l}, by \Cref{coro:integral2-wstar}, we have
\begin{align}
  &\left|\sum_{\substack{j,j'=1\\j \neq j'}}^D \beta_j^\star \beta_{j'}^\star \int_{x \sim \varphi^2} \cos(2\pi j x^\intercal w^\star)\cos(2\pi j' x^\intercal w^\star)\,dx\right|\\
  &\leq \sum_{\substack{j,j'=1\\j \neq j'}}^D |\beta_j^\star| |\beta_{j'}^\star| \left|\int_{x \sim \varphi^2} \cos(2\pi j x^\intercal w^\star)\cos(2\pi j' x^\intercal w^\star)\,dx\right|\\
  &\leq \sum_{\substack{j,j'=1\\j \neq j'}}^D |\beta_j^\star| |\beta_{j'}^\star| \frac{\sqrt{d}}{2\pi R R_w}\\
  &\leq \frac{\sqrt{d}}{2\pi R R_w}.\label{eq:term1-l-upper}
\end{align}
In the last line, we use that $\norm{\beta^\star}_1 = 1$.
Similarly, we can upper bound the absolute value of \Cref{eq:term2-l}:
\begin{align}
  \left|\sum_{\substack{j,j'=1\\j \neq j'}}^D \beta_j \beta_{j'} \int_{x \sim \varphi^2}  \cos(2\pi j x^\intercal \hat{w}) \cos(2\pi j' x^\intercal \hat{w})\,dx\right| &\leq \sum_{\substack{j,j'=1\\j \neq j'}}^D |\beta_j| |\beta_{j'}| \frac{\sqrt{d}}{2\pi R(R_w - \sqrt{d}\epsilon_1)}\\
  &\leq \frac{2D^2\sqrt{d}}{\pi R(R_w - \sqrt{d}\epsilon_1)}.\label{eq:term2-l-upper}
\end{align}
In the first inequality, we use \Cref{lem:integral2}, and in the second line we use \Cref{lem:param-bound} and our choice of $R,\epsilon_1$.
We can also upper bound the absolute value of \Cref{eq:term3-l}:
\begin{align}
  \left|2\sum_{\substack{j,j'=1\\j \neq j'}}^D \beta_j^\star \beta_{j'} \int_{x \sim \varphi^2} \cos(2\pi j x^\intercal w^\star) \cos(2\pi j' x^\intercal \hat{w})\,dx\right| &\leq 2\sum_{\substack{j,j'=1\\j \neq j'}}^D |\beta_j^\star| |\beta_{j'}| \frac{\sqrt{d}}{2\pi R (R_w - D\sqrt{d}\epsilon_1)}\\
  &\leq \frac{2D\sqrt{d}}{\pi R (R_w - D\sqrt{d}\epsilon_1)}.\label{eq:term3-l-upper}
\end{align}
In the first inequality, we use \Cref{coro:integral2-wstar-hat}, and in the second line, we use \Cref{lem:param-bound} and our choice of $R,\epsilon_1$ as well as $\norm{\beta^\star}_1 = 1$.
Combining \Cref{eq:term1-l-upper,eq:term2-l-upper,eq:term3-l-upper}, we have
\begin{align}
  \mathcal{L}_{w^\star}(\beta) &\leq \sum_{j=1}^D \left((\beta_j^\star)^2 \int_{x\sim\varphi^2}\cos^2(2\pi j x^\intercal w^\star)\,dx + \beta_j^2 \int_{x \sim \varphi^2}\cos^2(2\pi j x^\intercal \hat{w})\,dx\right.\\
  & \left.- 2\beta_j^\star \beta_j \int_{x \sim \varphi^2}\cos(2\pi jx^\intercal w^\star)\cos(2\pi j x^\intercal \hat{w})\,dx\right)\\
  &+ \frac{\sqrt{d}}{2\pi R R_w} + \frac{2D^2 \sqrt{d}}{\pi R(R_w - \sqrt{d}\epsilon_1)} + \frac{2D\sqrt{d}}{\pi R(R_w - D\sqrt{d}\epsilon_1)}.
\end{align}
It remains to bound the terms involving the integral of cosine.
By the proof of \Cref{lem:integral-upper} and \Cref{coro:integral-upper}, then
\begin{align}
  \mathcal{L}_{w^\star}(\beta) &\leq \left(\frac{1}{2} + \frac{\sqrt{d}}{8\pi R_w R}\right)\norm{\beta^\star}_2^2 + \left(\frac{1}{2} + \frac{\sqrt{d}}{8\pi R (R_w - \sqrt{d}\epsilon_1)}\right)\norm{\beta}_2^2\\
  &- 2\beta_j^\star \beta_j \int_{x \sim \varphi^2}\cos(2\pi jx^\intercal w^\star)\cos(2\pi j x^\intercal \hat{w})\,dx\\
  &+ \frac{\sqrt{d}}{2\pi R R_w} + \frac{2D^2 \sqrt{d}}{\pi R(R_w - \sqrt{d}\epsilon_1)} + \frac{2D\sqrt{d}}{\pi R(R_w - D\sqrt{d}\epsilon_1)}\\
  &\leq \frac{1}{2}\norm{\beta^\star}_2^2 + \frac{1}{2}\norm{\beta}_2^2 - 2\beta_j^\star \beta_j \int_{x \sim \varphi^2}\cos(2\pi jx^\intercal w^\star)\cos(2\pi j x^\intercal \hat{w})\,dx\\
  &+ \frac{5\sqrt{d}}{8\pi R_w R} + \frac{(4D^2 + D)\sqrt{d}}{2\pi R(R_w - \sqrt{d}\epsilon_1)} + \frac{2D\sqrt{d}}{\pi R(R_w - D\sqrt{d}\epsilon_1)}.
\end{align}
In the inequality, we use that $\norm{\beta^\star}_1 = 1$ and \Cref{lem:param-bound}.
For the last remaining integral term, we have the following
\begin{align}
  &2\beta_j^\star \beta_j \int_{x \sim \varphi^2}\cos(2\pi jx^\intercal w^\star)\cos(2\pi j x^\intercal \hat{w})\,dx\\
  &= - 2\sum_{\substack{j=1\\\mathrm{sign}(\beta_j) = \mathrm{sign}(\beta_j^\star)}}^D \beta_j^\star \beta_j \int_{x \sim \varphi^2}\cos(2\pi jx^\intercal w^\star)\cos(2\pi j x^\intercal \hat{w})\,dx\\
  &- 2\sum_{\substack{j=1\\\mathrm{sign}(\beta_j) \neq \mathrm{sign}(\beta_j^\star)}}^D \beta_j^\star \beta_j \int_{x \sim \varphi^2}\cos(2\pi jx^\intercal w^\star)\cos(2\pi j x^\intercal \hat{w})\,dx\\
  &\leq - 2\sum_{\substack{j=1\\\mathrm{sign}(\beta_j) = \mathrm{sign}(\beta_j^\star)}}^D \beta_j^\star \beta_j \left(\frac{1}{2} - \frac{\sqrt{d}}{8\pi R_w R} - \frac{5\pi^2 D^2 dR^2\epsilon_1}{3}\right) \\
  &- 2\sum_{\substack{j=1\\\mathrm{sign}(\beta_j) \neq \mathrm{sign}(\beta_j^\star)}}^D \beta_j^\star \beta_j\left(\frac{1}{2} + \frac{\sqrt{d}}{8\pi R_w R} + \pi D d\epsilon_1 R\right)\\
  &= -\sum_{j=1}^D \beta_j^\star \beta_j + \left(\frac{\sqrt{d}}{4\pi R_w R} + \frac{10\pi^2 D^2 dR^2\epsilon_1}{3}\right)\sum_{\substack{j=1\\\mathrm{sign}(\beta_j) = \mathrm{sign}(\beta_j^\star)}}^D \beta_j^\star \beta_j\\
  &- \left(\frac{\sqrt{d}}{4\pi R_w R} + 2\pi D d\epsilon_1 R\right) \sum_{\substack{j=1\\\mathrm{sign}(\beta_j) \neq \mathrm{sign}(\beta_j^\star)}}^D \beta_j^\star \beta_j.
\end{align}
Here, in the second line, we split the sum depending on if the signs of the $\beta_j, \beta_j^\star$ match.
In the fourth line, since $\mathrm{sign}(\beta_j) \neq \mathrm{sign}(\beta_j^\star)$, then $\beta_j\beta_j^\star \leq 0$ so that the last term has a positive coefficient overall. Thus, we can use an upper bound on the integral, where we use \Cref{lem:integral-upper}.
Also, since $\mathrm{sign}(\beta_j) = \mathrm{sign}(\beta_j^\star)$, then $\beta_j\beta_j^\star \geq 0$ so that the first term has a negative coefficient overall.
Thus, we can use a lower bound on the integral, where we use \Cref{lem:integral}.
In the last equality, we combined the summations over $j$ again.
Plugging this into the expression we had before, we have
\begin{align}
  \mathcal{L}_{w^\star}(\beta) &\leq \frac{1}{2}\norm{\beta - \beta^\star}_2^2 + \left(\frac{\sqrt{d}}{4\pi R_w R} + \frac{10\pi^2 D^2 dR^2\epsilon_1}{3}\right)\sum_{\substack{j=1\\\mathrm{sign}(\beta_j) = \mathrm{sign}(\beta_j^\star)}}^D \beta_j^\star \beta_j\\
  &- \left(\frac{\sqrt{d}}{4\pi R_w R} + 2\pi D d\epsilon_1 R\right) \sum_{\substack{j=1\\\mathrm{sign}(\beta_j) \neq \mathrm{sign}(\beta_j^\star)}}^D \beta_j^\star \beta_j + \frac{5\sqrt{d}}{8\pi R_w R} + \frac{(4D^2 + D)\sqrt{d}}{2\pi R(R_w - \sqrt{d}\epsilon_1)} + \frac{2D\sqrt{d}}{\pi R(R_w - D\sqrt{d}\epsilon_1)}
\end{align}
We can further bound this by taking the absolute value to get
\begin{align}
  \mathcal{L}_{w^\star}(\beta) &\leq \frac{1}{2}\norm{\beta - \beta^\star}_2^2 + \left(\frac{\sqrt{d}}{4\pi R_w R} + \frac{10\pi^2 D^2 dR^2\epsilon_1}{3}\right)\sum_{\substack{j=1\\\mathrm{sign}(\beta_j) = \mathrm{sign}(\beta_j^\star)}}^D |\beta_j^\star| |\beta_j|\\
  &+ \left(\frac{\sqrt{d}}{4\pi R_w R} + 2\pi D d\epsilon_1 R\right) \sum_{\substack{j=1\\\mathrm{sign}(\beta_j) \neq \mathrm{sign}(\beta_j^\star)}}^D |\beta_j^\star| |\beta_j| + \frac{5\sqrt{d}}{8\pi R_w R} + \frac{(4D^2 + D)\sqrt{d}}{2\pi R(R_w - \sqrt{d}\epsilon_1)} + \frac{2D\sqrt{d}}{\pi R(R_w - D\sqrt{d}\epsilon_1)}.
\end{align}
Using \Cref{lem:param-bound} with our choice of $R, \epsilon_1$ and $\norm{\beta^\star}_1 = 1$, then we have
\begin{align}
  \mathcal{L}_{w^\star}(\beta) &\leq \frac{1}{2}\norm{\beta - \beta^\star}_2^2 + \frac{\sqrt{d}}{2\pi R_w R} + \frac{20\pi^2 D^2 d R^2\epsilon_1}{3} + \frac{\sqrt{d}}{2\pi R_w R} + 4\pi Dd \epsilon_1 R\\
  &+ \frac{5\sqrt{d}}{8\pi R_w R} + \frac{(4D^2 + D)\sqrt{d}}{2\pi R(R_w - \sqrt{d}\epsilon_1)} + \frac{2D\sqrt{d}}{\pi R(R_w - D\sqrt{d}\epsilon_1)}\\
  &= \frac{1}{2}\norm{\beta - \beta^\star}_2^2 + \frac{13\sqrt{d}}{8\pi R_w R} + \frac{20\pi^2 D^2 d R^2\epsilon_1}{3} + 4\pi Dd \epsilon_1 R + \frac{(4D^2 + D)\sqrt{d}}{2\pi R(R_w - \sqrt{d}\epsilon_1)} + \frac{2D\sqrt{d}}{\pi R(R_w - D\sqrt{d}\epsilon_1)}\\
  &\leq \frac{1}{2}\norm{\beta - \beta^\star}_2^2 + \frac{13\sqrt{d}}{8\pi R_w R} + \frac{32\pi^2 D^2 dR^2 \epsilon_1}{3} + \frac{(4D^2 + D)\sqrt{d}}{2\pi R(R_w - \sqrt{d}\epsilon_1)} + \frac{2D\sqrt{d}}{\pi R(R_w - D\sqrt{d}\epsilon_1)}\\
  &\leq \frac{1}{2}\norm{\beta - \beta^\star}_2^2 + \frac{13\sqrt{d}}{8\pi R_w R} + \frac{32\pi^2 D^2 dR^2 \epsilon_1}{3} + \frac{9D^2\sqrt{d}}{2\pi R(R_w - D\sqrt{d}\epsilon_1)}.
\end{align}
Here, in the first inequality, we use \Cref{lem:param-bound} and $\norm{\beta^\star}_1 = 1$.
In the second inequality, we use that $D \geq 1$ so that $D^2 \geq D$ and $R^2 \geq R$.
In the last line, we use that $D \epsilon_1 \geq \epsilon_1$ and $D^2 \geq D$.
The claim then follows from~\Cref{lem:grad-converge} and \Cref{coro:grad-accuracy}, which says that using gradient descent, after a sufficient number of steps, we reach $\hat{\beta} = \beta^{(t+1)}$ such that $\norm{\beta^{(t+1)} - \beta^\star}_2 \leq \epsilon_3$.
\end{proof}

Finally, we can choose $\epsilon_3$ and adjust our choices for $R,\epsilon_1$ to show that the loss function is indeed bounded by $\epsilon$ for our predictor $\hat{\beta}$ found via gradient descent.

\begin{proof}[Proof of \Cref{thm:g-tilde-uniform}]
Let $\epsilon > 0$.
By \Cref{lem:loss-bound}, taking $\epsilon_3 = \sqrt{\epsilon}$, we have that
\begin{equation}
  \label{eq:loss-to-bound}
  \mathcal{L}_{w^\star}(\hat{\beta}) \leq \frac{\epsilon}{2} + \frac{13\sqrt{d}}{8\pi R_w R} + \frac{32\pi^2 D^2 dR^2 \epsilon_1}{3} + \frac{9D^2}{2\pi R(R_w/\sqrt{d} - D\epsilon_1)}.
\end{equation}
for our choice of $R,\epsilon_1$.
Here, recall that $\epsilon_1$ is the accuracy with which we can estimate $w^\star$, i.e., $|\hat{w}_i - w_i^\star| \leq \epsilon_1$.
We want to show that $\mathcal{L}_{w^\star}(\hat{\beta}) \leq \epsilon$.
This can be satisfied by taking
\begin{align}
  R &\geq \max\left(\frac{39 \sqrt{d}}{4\pi R_w \epsilon}, \frac{D^2}{\epsilon}, \frac{54D^2\sqrt{d}}{\pi R_w \epsilon}, \frac{2\eta D^{5/2}\log(2\sqrt{D}/\sqrt{\epsilon})}{\log(1/c)\sqrt{\epsilon}}, \frac{32\eta D^{3/2} \sqrt{d} \log(2\sqrt{D}/\sqrt{\epsilon})}{\pi R_w \log(1/c)\sqrt{\epsilon}}\right)\\
  &= \max\left(\frac{D^2}{\epsilon}, \frac{54D^2\sqrt{d}}{\pi R_w \epsilon}, \frac{2\eta D^{5/2}\log(2\sqrt{D}/\sqrt{\epsilon})}{\log(1/c)\sqrt{\epsilon}}, \frac{32\eta D^{3/2} \sqrt{d} \log(2\sqrt{D}/\sqrt{\epsilon})}{\pi R_w \log(1/c)\sqrt{\epsilon}}\right)
\end{align}
\begin{align}
  \epsilon_1 \leq \min\left(\frac{\epsilon^3}{64\pi^2 D^6 d}, \frac{3\log^3(1/c) \epsilon^{3/2}}{80 \pi^2 \eta D^{13/2} d \log(2\sqrt{D}/\sqrt{\epsilon})}, \frac{R_w}{2D\sqrt{d}}\right),
\end{align}
where $\eta = \mathcal{O}(1)$ is the step size of gradient descent and $c = 1-\eta$.
Note that the last two terms in the maximum for $R$ come from \Cref{coro:grad-accuracy} and similarly for the last two terms in the minimum for $\epsilon_1$.

For the second term in \Cref{eq:loss-to-bound}, since $R \geq 39\sqrt{d}/(4\pi R_w\epsilon)$, we have
\begin{equation}
  \frac{13\sqrt{d}}{8\pi R_w R} \leq \frac{\epsilon}{6}.
\end{equation}
For the third term in \Cref{eq:loss-to-bound}, using $R \geq D^2/\epsilon$ and $\epsilon_1 \leq 1/(64 \pi^2 R^3d) \leq \epsilon^3/ ( 64 \pi^2 D^6 d)$, then
\begin{equation}
  \frac{32\pi^2 D^2 dR^2 \epsilon_1 }{3} \leq \frac{D^2}{6R} \leq \frac{\epsilon}{6}.
\end{equation}
Finally, for the last term in \Cref{eq:loss-to-bound}, using $\epsilon_1 \leq R_w/(2D\sqrt{d})$ and $R \geq 54D^2 \sqrt{d}/(\pi R_w \epsilon)$, we have
\begin{equation}
  \frac{9D^2}{2\pi R (R_w/\sqrt{d} - D\epsilon_1)} \leq \frac{9D^2 \sqrt{d}}{\pi R R_w} \leq \frac{\epsilon}{6}.
\end{equation}
Thus, we have shown that 
\begin{equation}
  \mathcal{L}_{w^\star}(\hat{\beta}) \leq \frac{\epsilon}{2} + \frac{\epsilon}{6}+ \frac{\epsilon}{6}+ \frac{\epsilon}{6} = \epsilon,
\end{equation}
proving the claim.
Moreover, we have the following simplified scaling of $R, \epsilon_1$ by hiding the constants and logarithmic factors:
\begin{align}
  R &= \tilde{\Omega}\left(\max\left(\frac{D^2}{\epsilon}, \frac{D^2\sqrt{d}}{R_w \epsilon}, \frac{D^{5/2}}{\sqrt{\epsilon}}, \frac{D^{3/2}\sqrt{d}}{R_w \sqrt{\epsilon}}\right)\right),
\end{align}
\begin{equation}
  \epsilon_1 = \tilde{\mathcal{O}}\left(\min\left(\frac{\epsilon^3}{D^6 d}, \frac{\epsilon^{3/2}}{D^{13/2}d}, \frac{R_w}{D\sqrt{d}}\right)\right).
\end{equation}
The bound on the number of iterations of gradient descent used simply comes from 
\begin{equation}
  t = \Theta\left(\log(\sqrt{D}/\epsilon_3)\right)
\end{equation}
from \Cref{coro:grad-steps} and the choice $\epsilon_3 = \sqrt{\epsilon}$.
\end{proof}

\subsection{Integral bounds}
\label{sec:int-bounds}

The following technical lemmas for bounding integrals will be useful in the proofs of \Cref{thm:verification} and \Cref{thm:g-tilde-uniform}.

First, we have a bound on a complex exponential that will be useful in several of the other lemmas in this section.

\begin{lemma}
\label{lem:complex-exp}
Let $\varphi^2$ be the uniform density over $[-R,R]^d \subseteq \mathbb{R}^d$.
Let $w^\star \in \mathbb{R}^d$ be unknown with norm $R_w > 0$, and let $\hat{w} \in \mathbb{R}^d$ be an approximation of $w^\star$ with $\norm{\hat{w} - w^\star}_\infty \leq \epsilon_1$.
Let $1 \leq j,j'\leq D$ be integers with $j \neq j'$, for $D \in \mathbb{N}$ from \Cref{eq:g-tilde}.
Then,
\begin{equation}
  \left|\int\limits_{x \sim\varphi^2} e^{2\pi i x^\intercal \hat{w}(j-j')}\,dx \right| \leq \frac{1}{2\pi R}\frac{\sqrt{d}}{R_w - \sqrt{d}\epsilon_1}.
\end{equation}
\end{lemma}

\begin{proof}
Using that $\varphi^2$ is the uniform density:
\begin{align}
  \left|\int_{x \sim \varphi^2} e^{2\pi i x^\intercal  \hat{w}(j-j')}\,dx\right| &= \left|\frac{1}{(2R)^d} \int_{x_1=-R}^{+R}\cdots \int_{x_d=-R}^{+R} e^{2\pi i \sum_{k=1}^d x_k \hat{w}_k (j-j')}\,dx_d\cdots dx_1\right|\\
  &= \left|\frac{1}{(2R)^d} \prod_{k=1}^d \int_{x_k=-R}^{+R} e^{2\pi i x_k \hat{w}_k (j-j')}\,dx_k \right|\label{eq:prod-int}.
\end{align}
Here, notice that we can bound each of these integrals by $2R$:
\begin{equation}
  \label{eq:triv-int-bound}
  \left|\int_{x_k=-R}^{+R} e^{2\pi i x_k\hat{w}_k(j-j')}\,dx_k \right| \leq \int_{x_k=-R}^{+R} \left|e^{2\pi i x_k \hat{w}_k (j-j')}\right| \,dx_k \leq \int_{x_k=-R}^{+R}\,dx_k = 2R.
\end{equation}
We also notice that because $\norm{w^\star}_2^2 = \sum_{i=1}^d |w_i^\star|^2 = R_w^2$, then there must exist some $k \in [d]$ such that $|w_k^\star| \geq R_w/\sqrt{d}$. Here, equality is satisfied for the case when $w_i = R_w/\sqrt{d}$ for all $i \in [d]$.
We will bound each integral in the product in \Cref{eq:prod-int} using \Cref{eq:triv-int-bound} except for this $k$ such that $|w_k^\star| \geq R_w/\sqrt{d}$:
\begin{align}
  \left|\int_{x \sim \varphi^2} e^{2\pi i x^\intercal  \hat{w}(j-j')}\,dx\right| &= \left|\frac{1}{(2R)^d} \prod_{k=1}^d \int_{x_k=-R}^{+R} e^{2\pi i x_k \hat{w}_k (j-j')}\,dx_k \right|\\
  &\leq \frac{1}{2R} \left|\int_{x_k=-R}^{+R} e^{2\pi i x_k \hat{w}_k(j-j')}\,dx_k\right|\\
  &= \frac{1}{2R} \left|\int_{x_k=-R}^{+R} \cos(2\pi x_k \hat{w}_k(j-j'))\,dx_k\right|\\
  &= \frac{1}{2R}\left|\frac{\sin(2\pi(j-j')R \hat{w}_k)}{\pi(j-j')\hat{w}_k}\right|\\
  &\leq \frac{1}{2R}\frac{1}{\pi|j-j'||\hat{w}_k|}\label{eq:exp-bound}\\
  &\leq \frac{1}{2R}\frac{1}{\pi |\hat{w}_k|}.
\end{align}
Here, in the second line, we use \Cref{eq:triv-int-bound}.
In the third line, because we are integrating over a symmetric interval, the sine contribution vanishes.
In the fifth line, we use that $|\sin(x)| \leq 1$, and in the last line we used that $j \neq j'$ so that $|j - j'| \geq 1$.
Now, because we chose $k$ such that $|w_k^\star| \geq R_w/\sqrt{d}$ and $|\hat{w}_i - w_i^\star| \leq \epsilon_1$ for all $i$, we have
\begin{equation}
  \epsilon_1 \geq |\hat{w}_k - w_k^\star| \geq ||\hat{w}_k| - |w_k^\star|| \geq \left|\left|\hat{w}_k\right| - \frac{R_w}{\sqrt{d}}\right|
\end{equation}
so that rearranging, we have
\begin{equation}
  |\hat{w}_k| \geq \frac{R_w}{\sqrt{d}} - \epsilon_1.
\end{equation}
Plugging this back into the above, we have
\begin{equation}
  \left|\int_{x \sim \varphi^2} e^{2\pi i x^\intercal  \hat{w}(j-j')}\,dx\right| \leq \frac{1}{2\pi R}\frac{\sqrt{d}}{R_w - \sqrt{d}\epsilon_1}.
\end{equation}
\end{proof}

\begin{corollary}
\label{coro:complex-exp-wstar}
Let $\varphi^2$ be the uniform density over $[-R,R]^d \subseteq \mathbb{R}^d$.
Let $w^\star \in \mathbb{R}^d$ be unknown with norm $R_w > 0$, and let $\hat{w} \in \mathbb{R}^d$ be an approximation of $w^\star$ with $\norm{\hat{w} - w^\star}_\infty \leq \epsilon_1$.
Let $1 \leq j,j'\leq D$ be integers with $j \neq j'$, for $D \in \mathbb{N}$ from \Cref{eq:g-tilde}.
Then,
\begin{equation}
  \left|\int\limits_{x \sim\varphi^2} e^{2\pi i x^\intercal w^\star(j-j')}\,dx \right| \leq \frac{1}{2\pi R}\frac{\sqrt{d}}{R_w}.
\end{equation}
\end{corollary}

\begin{proof}
This is true by the same proof as \Cref{lem:integral2}.
Because this is for $w^\star$ instead of $\hat{w}$, we no longer have the $\epsilon_1$ term.
\end{proof}

Now, we can use this to obtain a lower bound for an integral of a product of cosines.

\begin{lemma}
\label{lem:integral}
Let $\varphi^2$ be the uniform density over $[-R,R]^d \subseteq \mathbb{R}^d$.
Let $w^\star \in \mathbb{R}^d$ be unknown with norm $R_w > 0$, and let $\hat{w} \in \mathbb{R}^d$ be an approximation of $w^\star$ with $\norm{\hat{w} - w^\star}_\infty \leq \epsilon_1$.
Let $1 \leq j\leq D$ be an integer, for $D \in \mathbb{N}$ from \Cref{eq:g-tilde}.
Then,
\begin{equation}
  \int\limits_{x \sim \varphi^2} \cos(2\pi jx^\intercal \hat{w}) \cos(2\pi jx^\intercal w^\star)\,dx \geq \frac{1}{2} - \frac{\sqrt{d}}{8\pi R_w R} - \frac{5\pi^2 D^2 d R^2 \epsilon_1}{3}.
\end{equation}
\end{lemma}

\begin{proof}
Using the sum formulas for cosine, we have
\begin{align}
  &\int_{x \sim \varphi^2} \cos(2\pi j x^\intercal \hat{w})\cos(2\pi j x^\intercal w^\star)\,dx\\
  &= \int_{x \sim \varphi^2} \cos(2\pi j x^\intercal(w^\star + (\hat{w} - w^\star)))\cos(2\pi j x^\intercal w^\star)\,dx\\
  &= \int_{x \sim \varphi^2} \left(\cos(2\pi j x^\intercal w^\star)\cos(2\pi j x^\intercal (\hat{w} - w^\star)) - \sin(2\pi j x^\intercal w^\star) \sin(2\pi j x^\intercal (\hat{w} - w^\star))\right)\cos(2\pi j x^\intercal w^\star)\,dx\\
  &\geq \int_{x\sim\varphi^2} \cos^2(2\pi j x^\intercal w^\star) \left(1 - \frac{1}{2}(2\pi jx^\intercal (\hat{w} - w^\star))^2\right) - \sin(2\pi jx^\intercal w^\star)\sin(2\pi j x^\intercal (\hat{w}-w^\star))\cos(2\pi jx^\intercal w^\star)\,dx\\
  &\geq \int_{x \sim \varphi^2}\cos^2(2\pi j x^\intercal w^\star) \,dx - 2\pi^2 j^2 \int_{x \sim \varphi^2} \left(x^\intercal (\hat{w} - w^\star)\right)^2\,dx - 2\pi j\int_{x \sim \varphi^2} |x^\intercal (\hat{w} - w^\star)|\,dx.\label{eq:three-terms}
\end{align}
In the third line, we use the sum formula for cosines.
In the fourth line, we use that $\cos(y) \geq 1 - y^2/2$.
In the fifth line, we use that $\sin(y), \cos(y) \leq 1$ and $\sin(y) \leq |y|.$
We want to lower bound the first term and upper bound the second two.

First, we will lower bound the first term in~\Cref{eq:three-terms}.
We can expand the first term in terms of complex exponentials:
\begin{align}
  \int_{x \sim \varphi^2} \cos^2(2\pi j x^\intercal w^\star)\,dx &= \frac{1}{4}\int_{x \sim \varphi^2}\left(e^{2\pi i jx^\intercal w^\star} + e^{-2\pi ij x^\intercal w^\star}\right)^2\,dx\\
  &= \frac{1}{2} + \frac{1}{4}\int_{x \sim \varphi^2}e^{4\pi i jx^\intercal w^\star}\,dx + \frac{1}{4}\int_{x \sim \varphi^2}e^{-4\pi ij x^\intercal w^\star}\,dx.
\end{align}
Now, we can bound the absolute value of these complex exponentials via \Cref{coro:complex-exp-wstar}.
Note that \Cref{coro:complex-exp-wstar} applies because we only needed to use that $j \neq j'$ to lower bound $|j-j'| \geq 1$.
This already clearly holds for $j \geq 1$.
Thus, we have
\begin{equation}
  \label{eq:cos-abs}
  \left|\int_{x \sim \varphi^2} \cos^2(2\pi j x^\intercal w^\star)\,dx - \frac{1}{2}\right| \leq \frac{1}{2}\left|\int_{x \sim \varphi^2} e^{4\pi i j x^\intercal w^\star}\,dx \right|\leq \frac{1}{8\pi R} \frac{\sqrt{d}}{R_w}.
\end{equation}
Rearranging, we have
\begin{equation}
  \label{eq:three-terms-1}
  \int_{x \sim \varphi^2} \cos^2(2\pi j x^\intercal w^\star)\,dx \geq \frac{1}{2} - \frac{\sqrt{d}}{8\pi R_w R}.
\end{equation}
This gives a lower bound on the first term in \Cref{eq:three-terms}.
We still need to upper bound the other terms in \Cref{eq:three-terms}.
For the second term, we can first directly evaluate the integral.
\begin{align}
&\int_{x \sim \varphi^2} (x^\intercal  (\hat{w} - w^\star))^2\,dx\\
&= \frac{1}{(2R)^d} \int_{x_1=-R}^{+R} \cdots \int_{x_d=-R}^{+R}  \left(\sum_{i=1}^d x_i \hat{w}_i - x_iw^\star_i\right)^2\,dx_d\cdots dx_1\\
&= \frac{1}{(2R)^d} \int_{x_1=-R}^{+R} \cdots \int_{x_d=-R}^{+R}  \left(\sum_{i,i'=1}^d x_i x_{i'}\hat{w}_i \hat{w}_{i'} + x_i x_{i'}w^\star_i w^\star_{i'} - x_ix_{i'}\hat{w}_i w_{i'}^\star - x_ix_{i'}w_i^\star \hat{w}_{i'}\right)\,dx_d\cdots dx_1.
\end{align}
Here, notice that
\begin{align}
\frac{1}{(2R)^d} \int_{x_1=-R}^{+R} \cdots \int_{x_d=-R}^{+R} x_i x_{i'}\,dx_d\cdots dx_1 &= \frac{1}{(2R)^2} \int_{x_i=-R}^{+R} \int_{x_{i'}=-R}^{+R} x_i x_{i'}\,dx_{i'}\,dx_i\\
&= \frac{\delta_{ii'}}{2R} \int_{x=-R}^{+R}x^2\,dx\\
&= \frac{R^2}{3} \delta_{ii'},
\end{align}
where the second line follows because if $i\neq i'$, we are integrating an odd function $x_{i'}$ over a symmetric interval. Plugging this into our previous expression, we have
\begin{align} 
\int_{x \sim \varphi^2} (x^\intercal(\hat{w} - w^\star))^2\,dx &= \frac{R^2}{3}\sum_{i=1}^d (\hat{w}_i)^2 + \left(w_i^\star\right)^2 - 2\hat{w}_i w_i^\star\\
&= \frac{R^2}{3} \norm{\hat{w} - w^\star}_2^2\\
&\leq \frac{R^2}{3}d\epsilon_1^2.\label{eq:three-terms-2}
\end{align}
In the last line, we used that $|\hat{w}_i - w_i^\star| \leq \epsilon_1$ for all $i \in [d]$.

Finally, we can similarly upper bound the last term in \Cref{eq:three-terms}.
\begin{align}
  \int_{x\sim\varphi^2}|x^\intercal(\hat{w} - w^\star)|\,dx &= \frac{1}{(2R)^d} \int_{x_1=-R}^{+R} \cdots \int_{x_d = -R}^{+R} \left|\sum_{i=1}^d x_i(\hat{w}_i - w_i^\star)\right| \,dx_d \cdots dx_1\\
  &\leq \frac{1}{(2R)^d} \int_{x_1=-R}^{+R} \cdots \int_{x_d = -R}^{+R} \sum_{i=1}^d \left|x_i (\hat{w}_i - w_i^\star)\right| \,dx_d \cdots dx_1\\
  &= \frac{1}{2R}\left(\sum_{i=1}^d |\hat{w}_i - w_i^\star| \int_{x_i=-R}^{+R} |x_i|\,dx_i\right)\\
  &\leq \frac{\epsilon_1}{2R}\sum_{i=1}^d \int_{x_i=-R}^{+R}|x_i|\,dx_i\\
  &= \frac{\epsilon_1d}{2R} R^2\\
  &= \frac{\epsilon_1 d R}{2}.\label{eq:three-terms-3}
\end{align}
In the second line, we use the triangle inequality.
In the fourth line, we use that $|\hat{w}_i - w_i^\star| \leq \epsilon_1$ for all $i \in [d]$.
Now, combining \Cref{eq:three-terms-1,eq:three-terms-2,eq:three-terms-3} in \Cref{eq:three-terms}, we have
\begin{align}
  \int\limits_{x \sim \varphi^2} \cos(2\pi jx^\intercal \hat{w}) \cos(2\pi jx^\intercal w^\star)\,dx &\geq \frac{1}{2} - \frac{\sqrt{d}}{8\pi R_w R} - \frac{2\pi^2j^2 R^2 d\epsilon_1^2}{3} - \pi j \epsilon_1 dR\\
  &\geq \frac{1}{2} - \frac{\sqrt{d}}{8\pi R_w R} - \frac{2\pi^2j^2 R^2 d\epsilon_1}{3} - \pi^2 j^2 \epsilon_1 dR^2\\
  &\geq \frac{1}{2} - \frac{\sqrt{d}}{8\pi R_w R} - \frac{5\pi^2 D^2 R^2 d\epsilon_1}{3},
\end{align}
where in the second line, we use that $j,R \geq 1$ so that $j^2 \geq j$ and $R^2 \geq R$ and $\epsilon_1 < 1$ so that $\epsilon_1^2\leq \epsilon_1$.
In the last line, we use that $j \leq D$.
\end{proof}

\begin{corollary}
\label{coro:integral}
Let $\varphi^2$ be the uniform density over $[-R,R]^d \subseteq \mathbb{R}^d$.
Let $w^\star \in \mathbb{R}^d$ be unknown with norm $R_w > 0$, and let $\hat{w} \in \mathbb{R}^d$ be an approximation of $w^\star$ with $\norm{\hat{w} - w^\star}_\infty \leq \epsilon_1$.
Let $1 \leq j \leq D$ be an integer, for $D \in \mathbb{N}$ from \Cref{eq:g-tilde}.
Then,
\begin{equation}
  \int_{x\sim\varphi^2} \cos^2(2\pi jx^\intercal \hat{w})\,dx \geq \frac{1}{2} - \frac{\sqrt{d}}{8\pi R(R_w - \sqrt{d}\epsilon_1)}. 
\end{equation}
\end{corollary}

\begin{proof}
The proof follows from the lower bound of the first term in \Cref{eq:three-terms} in the proof of \Cref{lem:integral}.
We can expand the first term in terms of complex exponentials:
\begin{align}
  \int_{x \sim \varphi^2} \cos^2(2\pi j x^\intercal \hat{w})\,dx &= \frac{1}{4}\int_{x \sim \varphi^2}\left(e^{2\pi i jx^\intercal\hat{w}} + e^{-2\pi ij x^\intercal \hat{w}}\right)^2\,dx\\
  &= \frac{1}{2} + \frac{1}{4}\int_{x \sim \varphi^2}e^{4\pi i jx^\intercal \hat{w}}\,dx + \frac{1}{4}\int_{x \sim \varphi^2}e^{-4\pi ij x^\intercal \hat{w}}\,dx.
\end{align}
Now, we can bound the absolute value of these complex exponentials via \Cref{lem:complex-exp} (instead of \Cref{coro:complex-exp-wstar}).
Note that \Cref{lem:complex-exp} applies because we only needed to use that $j \neq j'$ to lower bound $|j-j'| \geq 1$.
This already clearly holds for $j \geq 1$.
Thus, we have
\begin{equation}
  \label{eq:cos-abs-what}
  \left|\int_{x \sim \varphi^2} \cos^2(2\pi j x^\intercal \hat{w})\,dx - \frac{1}{2}\right| \leq \frac{1}{2}\left|\int_{x \sim \varphi^2} e^{4\pi i j x^\intercal \hat{w}}\,dx \right|\leq \frac{1}{8\pi R} \frac{\sqrt{d}}{R_w - \sqrt{d}\epsilon_1}.
\end{equation}
Rearranging, we have
\begin{equation}
  \int_{x \sim \varphi^2} \cos^2(2\pi j x^\intercal \hat{w})\,dx \geq \frac{1}{2} - \frac{\sqrt{d}}{8\pi R(R_w - \sqrt{d}\epsilon_1)}.
\end{equation}
\end{proof}

\begin{lemma}
\label{lem:integral-upper}
Let $\varphi^2$ be the uniform density over $[-R,R]^d \subseteq \mathbb{R}^d$.
Let $w^\star \in \mathbb{R}^d$ be unknown with norm $R_w > 0$, and let $\hat{w} \in \mathbb{R}^d$ be an approximation of $w^\star$ with $\norm{\hat{w} - w^\star}_\infty \leq \epsilon_1$.
Let $1 \leq j\leq D$ be an integer, for $D \in \mathbb{N}$ from \Cref{eq:g-tilde}.
Then,
\begin{equation}
  \int\limits_{x \sim \varphi^2} \cos(2\pi jx^\intercal \hat{w}) \cos(2\pi jx^\intercal w^\star)\,dx \leq \frac{1}{2} + \frac{\sqrt{d}}{8\pi R_w R} + \pi D d\epsilon_1 R.
\end{equation}
\end{lemma}

\begin{proof}
The proof of this is similar to that of~\Cref{lem:integral}.
Using the sum formulas for cosine, we have
\begin{align}
  &\int_{x \sim \varphi^2} \cos(2\pi j x^\intercal \hat{w})\cos(2\pi j x^\intercal w^\star)\,dx\\
  &= \int_{x \sim \varphi^2} \cos(2\pi j x^\intercal(w^\star + (\hat{w} - w^\star)))\cos(2\pi j x^\intercal w^\star)\,dx\\
  &= \int_{x \sim \varphi^2} \left(\cos(2\pi j x^\intercal w^\star)\cos(2\pi j x^\intercal (\hat{w} - w^\star)) - \sin(2\pi j x^\intercal w^\star) \sin(2\pi j x^\intercal (\hat{w} - w^\star))\right)\cos(2\pi j x^\intercal w^\star)\,dx\\
  &\leq \int_{x\sim\varphi^2} \cos^2(2\pi j x^\intercal w^\star) -\sin(2\pi jx^\intercal w^\star)\sin(2\pi j x^\intercal (\hat{w}-w^\star))\cos(2\pi jx^\intercal w^\star)\,dx\\
  &\leq \int_{x\sim\varphi^2} \cos^2(2\pi j x^\intercal w^\star) +\sin(2\pi j x^\intercal (\hat{w}-w^\star))\,dx\\
  &\leq \int_{x\sim\varphi^2} \cos^2(2\pi j x^\intercal w^\star)\,dx + 2\pi j\int_{x \sim\varphi^2} |x^\intercal (\hat{w}-w^\star)|\,dx\label{eq:two-terms}.
\end{align}
In the fourth line, we use that $\cos(y) \leq 1$.
In the fifth line, we use that $-\sin(y) \cos(y) \leq 1$.
In the last line, we use that $\sin(y) \leq |y|$.
We want to upper bound both of these terms, which is simple given the proof of \Cref{lem:integral}.

Namely, in \Cref{eq:cos-abs}, we showed that
\begin{equation}
  \left|\int_{x \sim \varphi^2} \cos^2(2\pi j x^\intercal w^\star)\,dx - \frac{1}{2}\right| \leq \frac{1}{8\pi R} \frac{\sqrt{d}}{R_w}.
\end{equation}
Thus, we can upper bound
\begin{equation}
  \label{eq:two-terms-1}
  \int_{x\sim \varphi^2} \cos^2 (2\pi j x^\intercal w^\star)\,dx \leq \frac{1}{2} + \frac{\sqrt{d}}{8\pi R_w R}
\end{equation}
Note that we have already upper bounded the third term in~\Cref{eq:three-terms-3}:
\begin{equation}
  \label{eq:two-terms-2}
  2\pi j\int_{x \sim\varphi^2} |x^\intercal (\hat{w}-w^\star)|\,dx \leq \pi j d \epsilon_1 R \leq \pi D d\epsilon_1 R.
\end{equation}
Combining \Cref{eq:two-terms-1} and \Cref{eq:two-terms-2} in \Cref{eq:two-terms}, we have
\begin{equation}
  \int_{x \sim \varphi^2} \cos(2\pi j x^\intercal \hat{w})\cos(2\pi j x^\intercal w^\star)\,dx \leq \frac{1}{2} + \frac{\sqrt{d}}{8\pi R_w R} + \pi D d\epsilon_1 R.
\end{equation}
\end{proof}

\begin{corollary}
\label{coro:integral-upper}
Let $\varphi^2$ be the uniform density over $[-R,R]^d \subseteq \mathbb{R}^d$.
Let $w^\star \in \mathbb{R}^d$ be unknown with norm $R_w > 0$, and let $\hat{w} \in \mathbb{R}^d$ be an approximation of $w^\star$ with $\norm{\hat{w} - w^\star}_\infty \leq \epsilon_1$.
Let $1 \leq j\leq D$ be an integer, for $D \in \mathbb{N}$ from \Cref{eq:g-tilde}.
Then,
\begin{equation}
  \int\limits_{x \sim \varphi^2} \cos^2(2\pi j x^\intercal \hat{w})\,dx \leq \frac{1}{2} + \frac{\sqrt{d}}{8\pi R (R_w - \sqrt{d}\epsilon_1)}.
\end{equation}
\end{corollary}

\begin{proof}
This follows directly from \Cref{eq:cos-abs-what}.
\end{proof}

\begin{lemma}
\label{lem:integral2}
Let $\varphi^2$ be the uniform density over $[-R,R]^d \subseteq \mathbb{R}^d$.
Let $w^\star \in \mathbb{R}^d$ be unknown with norm $R_w > 0$, and let $\hat{w} \in \mathbb{R}^d$ be an approximation of $w^\star$ with $\norm{\hat{w} - w^\star}_\infty \leq \epsilon_1$.
Let $1 \leq j,j' \leq D$ be integers with $j \neq j'$, for $D \in \mathbb{N}$ from \Cref{eq:g-tilde}.
Then,
\begin{equation}
  \left|\int\limits_{x\sim \varphi^2} \cos(2\pi j x^\intercal \hat{w}) \cos(2\pi j' x^\intercal \hat{w}) \,dx\right|\leq \frac{1}{2\pi R}\frac{\sqrt{d}}{R_w - \sqrt{d}\epsilon_1}.
\end{equation}
\end{lemma}

\begin{proof}
Using the product formulas for cosine, we can write the integral as
\begin{equation}
  \label{eq:sum-prod2}
  \left|\int\limits_{x\sim \varphi^2} \cos(2\pi j x^\intercal \hat{w}) \cos(2\pi j' x^\intercal \hat{w}) \,dx\right| = \left|\frac{1}{2} \int_{x\sim \varphi^2} \cos(2\pi x^\intercal \hat{w}(j -j')) + \cos(2\pi x^\intercal \hat{w} (j + j'))\,dx\right|.
\end{equation}
We can bound each of the integrals on the right hand side similarly.
Starting with the first term, we can write it in terms of complex exponentials
\begin{align}
  \label{eq:complex-sum}
  \left|\int_{x\sim \varphi^2} \cos(2\pi x^\intercal \hat{w}(j -j'))\,dx\right| \leq \frac{1}{2}\left| \int_{x \sim \varphi^2} e^{2\pi i x^\intercal \hat{w} (j-j')}\,dx \right| + \frac{1}{2}\left| \int_{x \sim \varphi^2} e^{2\pi i x^\intercal \hat{w} (j' - j)}\,dx \right|
\end{align}
Both terms in \Cref{eq:exp-sum} can be bounded via \Cref{lem:complex-exp}.
Thus, this bounds the first term in \Cref{eq:sum-prod2} as
\begin{equation}
  \label{eq:sum-prod2-1}
  \left|\int_{x\sim \varphi^2} \cos(2\pi x^\intercal \hat{w}(j -j'))\,dx\right| \leq \frac{1}{2\pi R}\frac{\sqrt{d}}{R_w - \sqrt{d}\epsilon_1}.
\end{equation}
We can similarly bound the second term in \Cref{eq:sum-prod2}.
Namely, the argument is the same as in \Cref{lem:complex-exp}, but in \Cref{eq:exp-bound}, we have
\begin{align}
  \left|\int_{x \sim \varphi^2} e^{2\pi i x^\intercal  \hat{w}(j+j')}\,dx\right| &\leq \frac{1}{2R}\frac{1}{\pi |j + j'||\hat{w}_k|}\\
  &\leq \frac{1}{6R}\frac{1}{\pi |\hat{w}_k|},
\end{align}
where since $j \neq j'$ and $j,j' \geq 1$, then $|j + j'| \geq 3$. The rest of the bound follows the same argument.
Then, we obtain
\begin{equation}
  \left|\int_{x\sim \varphi^2} \cos(2\pi x^\intercal \hat{w}(j +j'))\,dx\right| \leq \frac{1}{6\pi R}\frac{\sqrt{d}}{R_w - \sqrt{d}\epsilon_1} \leq \frac{1}{2\pi R}\frac{\sqrt{d}}{R_w - \sqrt{d}\epsilon_1}.
\end{equation}
Thus, combined with \Cref{eq:sum-prod2-1} in \Cref{eq:sum-prod2}, we have
\begin{equation}
  \left|\int\limits_{x\sim \varphi^2} \cos(2\pi j x^\intercal \hat{w}) \cos(2\pi j' x^\intercal \hat{w}) \,dx\right| \leq \frac{1}{2\pi R}\frac{\sqrt{d}}{R_w - \sqrt{d}\epsilon_1}.
\end{equation}
\end{proof}

By essentially the same proof, we can obtain a similar upper bound replacing $\hat{w}$ with $w^\star$.
This follows by applying \Cref{coro:complex-exp-wstar} instead of \Cref{lem:complex-exp}.

\begin{corollary}
\label{coro:integral2-wstar}
Let $\varphi^2$ be the uniform density over $[-R,R]^d \subseteq \mathbb{R}^d$.
Let $w^\star \in \mathbb{R}^d$ be unknown with norm $R_w > 0$, and let $\hat{w} \in \mathbb{R}^d$ be an approximation of $w^\star$ with $\norm{\hat{w} - w^\star}_\infty \leq \epsilon_1$.
Let $1 \leq j,j' \leq D$ be integers with $j \neq j'$, for $D \in \mathbb{N}$ from \Cref{eq:g-tilde}.
Then,
\begin{equation}
  \left|\int\limits_{x\sim \varphi^2} \cos(2\pi j x^\intercal w^\star) \cos(2\pi j' x^\intercal w^\star) \,dx\right|\leq \frac{1}{2\pi R}\frac{\sqrt{d}}{R_w}.
\end{equation}
\end{corollary}

\begin{corollary}
\label{coro:integral2-wstar-sin}
Let $\varphi^2$ be the uniform density over $[-R,R]^d \subseteq \mathbb{R}^d$.
Let $w^\star \in \mathbb{R}^d$ be unknown with norm $R_w > 0$, and let $\hat{w} \in \mathbb{R}^d$ be an approximation of $w^\star$ with $\norm{\hat{w} - w^\star}_\infty \leq \epsilon_1$.
Let $1 \leq j,j' \leq D$ be integers with $j \neq j'$, for $D \in \mathbb{N}$ from \Cref{eq:g-tilde}.
Then,
\begin{equation}
  \left|\int\limits_{x\sim \varphi^2} \cos(2\pi j x^\intercal w^\star) \sin(2\pi j' x^\intercal w^\star) \,dx\right|\leq \frac{1}{2\pi R}\frac{\sqrt{d}}{R_w}.
\end{equation}
\end{corollary}

\begin{proof}
This follows by the same proof as \Cref{lem:integral2} and \Cref{coro:integral2-wstar}.
In particular, using the sum-product formulas for sine and cosine, we have
\begin{equation}
  \left|\int_{x\sim \varphi^2} \cos(2\pi j x^\intercal w^\star) \sin(2\pi j' x^\intercal w^\star) \,dx\right| = \left|\frac{1}{2}\int_{x\sim\varphi^2} \sin(2\pi (j+j') x^\intercal w^\star) + \sin(2\pi (j'-j)x^\intercal w^\star)\,dx \right|.
\end{equation}
Then, writing in terms of complex exponentials, we have
\begin{equation}
  \left|\int_{x\sim\varphi^2} \sin(2\pi (j'-j)x^\intercal w^\star)\,dx \right| \leq \frac{1}{|2i|}\left|\int_{x\sim\varphi^2} e^{2\pi ix^\intercal w^\star (j'-j)}\,dx \right| + \frac{1}{|2i|}\left|\int_{x\sim\varphi^2} e^{2\pi ix^\intercal w^\star (j-j')}\,dx \right|.
\end{equation}
The rest of the proof is the same as \Cref{lem:integral2}, using \Cref{coro:complex-exp-wstar} instead of \Cref{lem:complex-exp} to bound the complex exponential terms.
\end{proof}

Finally, we need another integral bound similar to \Cref{lem:integral2}.

\begin{corollary}
\label{coro:integral2-wstar-hat}
Let $\varphi^2$ be the uniform density over $[-R,R]^d \subseteq \mathbb{R}^d$.
Let $w^\star \in \mathbb{R}^d$ be unknown with norm $R_w > 0$, and let $\hat{w} \in \mathbb{R}^d$ be an approximation of $w^\star$ with $\norm{\hat{w} - w^\star}_\infty \leq \epsilon_1$, where $\epsilon_1 \leq R_w/(D\sqrt{d})$.
Let $1 \leq j,j' \leq D$ be integers with $j \neq j'$, for $D \in \mathbb{N}$ from \Cref{eq:g-tilde}.
Then,
\begin{equation}
  \left|\int\limits_{x\sim \varphi^2} \cos(2\pi j x^\intercal w^\star) \cos(2\pi j' x^\intercal \hat{w}) \,dx\right|\leq \frac{1}{2\pi R}\frac{\sqrt{d}}{R_w - D\sqrt{d}\epsilon_1}.
\end{equation}
\end{corollary}

\begin{proof}
The proof is similar to that of \Cref{lem:integral2}, but we write it out fully to keep track of the differences.
Using the product formulas for cosine, we can write the integral as
\begin{equation}
  \label{eq:sum-prod2-wstar-hat}
  \left|\int\limits_{x\sim \varphi^2} \cos(2\pi j x^\intercal w^\star) \cos(2\pi j' x^\intercal \hat{w}) \,dx\right| = \left|\frac{1}{2} \int_{x\sim \varphi^2} \cos(2\pi (jx^\intercal w^\star - j'x^\intercal \hat{w})) + \cos(2\pi (jx^\intercal w^\star + j'x^\intercal \hat{w}))\,dx\right|.
\end{equation}
We can bound each of the integrals on the right hand side similarly.
Starting with the first term, we can write it in terms of complex exponentials
\begin{align}
  \label{eq:complex-sum-wstar-hat}
  \left|\int_{x\sim \varphi^2} \cos(2\pi (jx^\intercal w^\star - j'x^\intercal \hat{w}))\,dx\right| \leq \frac{1}{2}\left| \int_{x \sim \varphi^2} e^{2\pi i (jx^\intercal w^\star - j'x^\intercal \hat{w})}\,dx \right| + \frac{1}{2}\left| \int_{x \sim \varphi^2} e^{2\pi i (j'x^\intercal \hat{w} - jx^\intercal w^\star)}\,dx \right|
\end{align}
Each of these complex exponentials can be bounded by an argument similar to \Cref{lem:complex-exp}.
Using that $\varphi^2$ is the uniform density:
\begin{equation}
\left|\int_{x \sim \varphi^2} e^{2\pi i(j'x^\intercal \hat{w}-jx^\intercal w^\star)}\,dx\right| = \left|\frac{1}{(2R)^d} \prod_{k=1}^d \int_{x_k=-R}^{+R} e^{2\pi i (j'x_k\hat{w}_k -jx_kw_k^\star)}\,dx_k \right|.
\end{equation}
Again, we can bound each of these integrals by $2R$ as in \Cref{eq:triv-int-bound}.
Notice that because $\norm{w^\star}_2^2 = \sum_{i=1}^d |w_i^\star|^2 = R_w^2$, then there must exist some $k \in [d]$ such that $|w_k^\star| \geq \sqrt{R_w/d}$.
Then, we will bound each integral in the product above using \Cref{eq:triv-int-bound} except for this $k$ such that $|w_k^\star|\geq R_w/\sqrt{d}$:
\begin{align}
  \left|\int_{x \sim \varphi^2} e^{2\pi i(j'x^\intercal \hat{w}-jx^\intercal w^\star)}\,dx\right| &= \left|\frac{1}{(2R)^d} \prod_{k=1}^d \int_{x_k=-R}^{+R} e^{2\pi i (j'x_k\hat{w}_k -jx_kw_k^\star)}\,dx_k \right|\\
  &\leq \frac{1}{2R} \left|\int_{x_k=-R}^{+R} e^{2\pi i (j'x_k\hat{w}_k -jx_kw_k^\star)}\,dx_k\right|\\
  &= \frac{1}{2R} \left|\int_{x_k=-R}^{+R} \cos(2\pi (j'x_k\hat{w}_k -jx_kw_k^\star))\,dx_k\right|\\
  &= \frac{1}{2R}\left|\frac{\sin(2\pi R(j'\hat{w}_k-jw_k^\star))}{\pi(j'\hat{w}_k-jw_k^\star)}\right|\\
  &\leq \frac{1}{2\pi R}\frac{1}{|j'\hat{w}_k - jw_k^\star|}.\label{eq:exp-bound-wstar-hat}
\end{align}
Here, in the second line, we used \Cref{eq:triv-int-bound}.
In the third line, because we are integrating over a symmetric interval, the sine contribution vanishes.
In the fifth line, we use that $|\sin(x)| \leq 1$.
We wish to lower bound $|j'\hat{w}_k - jw_k^\star|$.
Recall that we chose $k$ such that $|w_k^\star| \geq R_w/\sqrt{d}$ and $|\hat{w}_i - w_i^\star| \leq \epsilon_1$ for all $i \in [d]$ with $\epsilon_1 \leq R_w/(D\sqrt{d})$.
Without loss of generality, for $D \geq j' > j \geq 1$, then note that
\begin{equation}
  j'\hat{w}_k - jw_k^\star \geq (j' - j)w_k^\star - j'\epsilon_1 \geq \frac{R_w}{\sqrt{d}} - D\epsilon_1 \geq 0.
\end{equation}
Here, in the first inequality, we use that $\hat{w}_k \geq w_k^\star -\epsilon_1$.
In the second inequality, we use that $w_k^\star \geq R_w/\sqrt{d}$ for our choice of $k$, $j' - j \geq 1$ (since $j' > j$ in this case), and $j' \leq D$.
The last inequality holds due to our choice of $\epsilon \leq R_w/(D\sqrt{d})$.
Thus, since these terms are nonnegative, then taking the absolute value, we have
\begin{equation}
  |j'\hat{w}_k - jw_k^\star| \geq |(j'-j)w_k^\star - j'\epsilon_1|.
\end{equation}
We can further lower bound this using the reverse triangle inequality:
\begin{equation}
  |j'\hat{w}_k - jw_k^\star| \geq ||j' - j| |w_k^\star| - j'\epsilon_1|.
\end{equation}
One can arrive at the same inequality for $j' \leq j$ as well by a similar argument.
Thus, from here, we can simply consider any $j' \neq j$.
Since $j\neq j'$, then $|j' - j| \geq 1$, and we also know that $j' \leq D$.
Thus, we have
\begin{equation}
  |j'-j||w_k^\star| - j'\epsilon_1 \geq |w_k^\star| - D\epsilon_1 \geq 0,
\end{equation}
where the last inequality follows again by our choice of $k$ with $|w_k^\star| \geq R_w/\sqrt{d}$ and the choice of $\epsilon_1$.
Thus, taking the absolute value, we have
\begin{equation}
  |j'\hat{w}_k - jw_k^\star| \geq ||w_k^\star| - D\epsilon_1|.
\end{equation}
Finally, using that $|w_k^\star| \geq R_w/\sqrt{d}$, we have
\begin{equation}
  |j'\hat{w}_k - jw_k^\star| \geq \left|\frac{R_w}{\sqrt{d}} - D\epsilon_1 \right| \geq \frac{R_w}{\sqrt{d}} - D\epsilon_1.
\end{equation}
Plugging this back into \Cref{eq:exp-bound-wstar-hat},  then we have
\begin{equation}
  \left|\int_{x \sim \varphi^2} e^{2\pi i(j'x^\intercal \hat{w}-jx^\intercal w^\star)}\,dx\right| \leq \frac{1}{2\pi R}\frac{\sqrt{d}}{R_w - D\sqrt{d}\epsilon_1}.
\end{equation}
Putting this together with \Cref{eq:complex-sum-wstar-hat}, we can bound the first term in \Cref{eq:sum-prod2-wstar-hat}
\begin{equation}
  \left|\int_{x\sim \varphi^2} \cos(2\pi (jx^\intercal w^\star - j'x^\intercal \hat{w}))\,dx\right| \leq \frac{1}{2\pi R}\frac{\sqrt{d}}{R_w - D\sqrt{d}\epsilon_1}.
\end{equation}
We can similarly bound the second term in \Cref{eq:sum-prod2-wstar-hat} using the same approach.
Thus, with \Cref{eq:sum-prod2-wstar-hat}, we have the desired bound.
\end{proof}

\section{Non-uniform distributions}
\label{sec:non-unif}

In this section, we repeat the steps of Appendix~\ref{sec:uniform} when instead given QSQ access to quantum example states with respect to a non-uniform distribution satisfying some technical assumptions.

Recall that we want to learn the target function $g_{w^\star}(x) = \tilde{g}(x^\intercal w^\star)$ for some unknown $w^\star \in \mathbb{R}^d$ and $\tilde{g}$ a function given in \Cref{eq:g-tilde}.
We refer to the definitions in Appendix~\ref{sec:detail-prob} for the precise problem statement.
Again, the overall idea of the algorithm is to apply period finding to find the unknown vector $w^\star$ one component at a time.
Then, given the form of $\tilde{g}$ from \Cref{eq:g-tilde}, we can find the unknown parameters $\beta^\star_j$ via gradient methods.

As in Appendix~\ref{sec:uniform}, we need to suitable discretize and truncate our target function.
In addition, we also need to introduce a suitable discretization of our distribution.
For the discretization of the target function, the results from Appendix~\ref{sec:linear-uniform} carry over.
Thus, we refer to the discretized function as $h_{w^\star, M_1, M_2}$, which has discretization parameters $M_1, M_2 \in \mathbb{Z}$.

Now, we define our discretized distribution and state our assumptions.
Consider a nonnegative function $p: \mathbb{R}^d \to [0,1]$ that can be written as
\begin{equation}
  p(x) \triangleq \prod_{j=1}^d p_j(x_j)
\end{equation}
for some nonnegative function $p_j: \mathbb{R} \to [0,1]$.
Let $\varphi^2$ denote the probability distribution defined by $p^2$, suitably normalized.
In particular, we consider the quantum example state
\begin{equation}
  \ket{h_{w^\star, M_1, M_2}} = \frac{1}{\tilde{G}}\sum_{x_1,\dots, x_d =-R}^{R-1} p(x)\ket{x}\ket{h_{w^\star, M_1, M_2}(x)},\quad \tilde{G} \triangleq \sum_{x_1,\dots, x_d=-R}^{R-1} p^2(x),
\end{equation}
where $\tilde{G}$ is a normalization constant and $x = x_1\cdots x_d$.
Throughout the rest of this section, we suppose that we are given access to quantum statistical queries with respect to this example state and discretization/truncation parameters $M_1, M_2, R$.
We note that one can consider preparing example states by first preparing $\sum_x p_\Sigma(x) \ket{x}$ and then evaluating the function $h_{w^\star, M_1, M_2}$ coherently.
to Algorithms to prepare this superposition over all inputs $x$ for, e.g., discrete Gaussian distributions, has been well-studied~\cite{chen2021quantum,chen2024quantum,grover2002creating,rattew2021efficient,mcardle2022quantum,holmes2020efficient,iaconis2024quantum}.

We consider the following additional assumptions on the probability distributions:
\begin{enumerate}
  \item\label{assum:fourier-conc} $\varphi^2$ is Fourier concentrated (\Cref{def:cont_fourier_conc}).
  \item\label{assum:pointwise-close} For a chosen truncation parameter $R$, for all $x \in [-R, R]$ and $j \in [d]$, then $|1 - p_j^2(x)| \leq 1/10$.
  \item\label{assum:bounded-1} The functions $p_j:\mathbb{R} \to [0,1]$ are nonnegative and bounded by $1$.
  \item\label{assum:even} $p^2$ is an even function.
  \item\label{assum:deriv} Let $M_1$ be a chosen discretization parameter, and let the truncation parameter be $R \triangleq \tilde{R}M_1$ for some suitably chosen $\tilde{R} \geq 1$. Then, the derivative of $p_j'$ is bounded: $|p_j'(M_1 x)| \leq \frac{\pi D R_w}{2M_1}$ and $|p_j'(M_1x + T)| \leq \frac{\pi D R_w}{2M_1}$ for all $x \in [-\tilde{R}, \tilde{R}]$ and $j \in [d]$, where $T$ is a guess for the period from Hallgren's algorithm (Appendix~\ref{sec:hallgren}).
  \item\label{assum:crit-points} $p^2$ has a constant number of critical points.
\end{enumerate}

Note that Assumption~\ref{assum:fourier-conc} is necessary in order for classical hardness to hold~\cite{shamir2018distribution}.
Also, one can think of Assumption~\ref{assum:deriv} as just needing this bound on the absolute value of the derivative for all inputs.
We state it more specifically in the form we require for the proofs.
While these assumptions may seem restrictive at first, we show later in this section that they are satisfied by several natural distributions when taking the scale parameter large enough, such as Gaussians, generalized Gaussians~\cite{subbotin1923law}, and logistic distributions.
With these assumptions on the input distribution, we can efficiently learn the target functions $g_{w^\star}$ using QSQs.

\begin{theorem}[Guarantee; Non-Uniform Case]
\label{thm:non-unif-guarantee}
Let $\epsilon, \delta > 0, \tau \geq 0$.
Let $\varphi^2 \propto \prod_{j=1}^d p_j^2$ be a probability distribution over $[-R, R]^d$ satisfying Assumptions~\ref{assum:fourier-conc}-\ref{assum:crit-points} for the parameters specified shortly.
Let $w^\star \in \mathbb{R}^d$ be unknown with norm $R_w > 0$ and $w_j^\star \geq R_w/d^2$, for all $j \in [d]$.
Let $g_{w^\star}: \mathbb{R}^d \to [-1,1]$ be defined as $g_{w^\star}(x) = \tilde{g}(x^\intercal w^\star)$, where $\tilde{g}: \mathbb{R} \to [-1,1]$ is a function defined in \Cref{eq:g-tilde}.
Consider parameters $M_1 = \max(70\pi d^2 D^3 R_w, R_w^2/\epsilon_1)$, $M_2 = c M_1$, where $c$ is any constant such that $M_2$ is an integer and $c < 1/(8\pi DR_w)$, and
\begin{equation}
    \tilde{R} = \tilde{\Omega}\left(\max\left(\frac{\tau M_1^2d^4}{R_w^2}, \frac{D^2}{\epsilon}, \frac{D^2\sqrt{d}}{R_w \epsilon}, \frac{D^{5/2}}{\sqrt{\epsilon}}, \frac{D^{3/2}\sqrt{d}}{R_w \sqrt{\epsilon}}, \frac{d^2D}{R_w^2}\right)\right),\quad \epsilon_1 = \tilde{\mathcal{O}}\left(\min\left(\frac{\epsilon^3}{D^6 d}, \frac{\epsilon^{3/2}}{D^{13/2}d}, \frac{R_w}{D\sqrt{d}}\right)\right).
\end{equation}
Suppose we have QSQ access (see \Cref{def:qsq}) with respect to discretization parameters $M_{1,m} \triangleq m M_1$, $M_{2,m} \triangleq m M_2$ and a truncation parameter $R \geq \tilde{R}$, for $m \in \{1,\dots, D\}$.
Then, there exists a quantum algorithm with this QSQ access that can efficiently find parameters $\hat{\beta} \in \mathbb{R}^D$ such that $\mathcal{L}_{w^\star}(\hat{\beta}) \leq \epsilon$ with probability at least $1-\delta$.
Moreover, this algorithm uses
\begin{equation}
   N = \mathcal{O}\left(d D \log\left(\frac{1}{\delta}\right) \log^5\left(\frac{M_1 d^2}{R_w}\right)\right)
\end{equation}
quantum statistical queries with tolerance $\tau \leq \min\left(\frac{1}{M_2^2}\left(\frac{5}{42} - \frac{3}{2M_2}\right), \frac{1}{2D^2M_2^2}\left(\frac{2}{9} - \frac{1}{8}\left(\frac{2\pi R_w}{M_1}\right)^2 + \frac{3D^2}{M_2}\right)\right)$ and
\begin{equation}
  t = \Theta\left(\log\left(\sqrt{\frac{D}{\epsilon}}\right)\right)
\end{equation}
iterations of gradient descent.
\end{theorem}

As in the uniform case, our algorithm uses QSQs with different choices of discretization/truncation parameters for the two subroutines of Hallgren's algorithm (Appendix~\ref{sec:hallgren}): quantum Fourier sampling and the verification procedure.
In the quantum Fourier sampling part, we use QSQs with respect to discretization parameters $M_1, M_2$ and truncation parameter $\tilde{R} = R$.
For verification, we use discretization parameters $M_{1,m} \triangleq m M_1, M_{2,m} \triangleq mM_2$ and truncation parameter $\tilde{R} = RM_{1,m}$ for $m \in \{1,\dots, D\}$.

Note that for non-uniform distributions, Hallgren's algorithm does not immediately apply.
To remedy this, we give a new analysis for Hallgren's algorithm for non-uniform distributions satisfying our assumptions.
In particular, the quantum Fourier sampling part of this algorithm only requires Assumptions~\ref{assum:pointwise-close} and \ref{assum:bounded-1}.

While it may seem like we have many (potentially restrictive) assumptions, the next few results show that it still captures natural classes of distributions.
In particular, the next proposition shows that generalized Gaussian distributions~\cite{subbotin1923law} with a large enough scale parameter satisfy all of them.
As a corollary, Gaussians with large enough variance also satisfy the assumptions.

\begin{prop}[Generalized Gaussians satisfy assumptions]
    \label{prop:gen-gauss-satisfies-assum}
    Let $\tau \geq 0$.
    Let $\alpha_j \geq 2$ be even shape parameters, and let $s_j > 0$ be scale parameters specified later, for $j \in [d]$.
    Let $p^2(x) = \prod_{j=1}^d \exp(-(x_j/s_j)^{\alpha_j})$ and $\varphi^2 \propto p^2$ suitably normalized so that $\varphi^2$ is a generalized Gaussian distribution.
    Let $M_1, M_2$ be discretization parameters with $M_1 \geq R_w$.
    Let
    \begin{equation}
        \tilde{R} = \tilde{\Omega}\left(\max\left(\frac{\tau M_1^2 d^4}{R_w^2}, \frac{D^2\sqrt{d}}{R_w}, \max_j\left(\frac{d^2}{R_w}\right)^{a_j - 1} \frac{D}{R_w}\right)\right)
    \end{equation}
    and let $R \geq \tilde{R}$ be the truncation parameter.
    Then, if $s_j \geq 2R\sqrt{\pi}$ for all $j \in [d]$, Assumptions~\ref{assum:fourier-conc}-\ref{assum:crit-points} are satisfied for $\varphi^2$ for truncation parameter $R$.
    In particular, $\varphi^2$ is $\epsilon(r)$-Fourier-concentrated with $\epsilon(r)$ decaying superpolynomially in $r$.
\end{prop}

The corollary follows easily because the Gaussian distribution is a special case of the generalized Gaussian for shape parameter $\alpha_j = 2$.
Fourier concentration follows by standard Gaussian concentration arguments.

\begin{corollary}[Gaussians satisfy assumptions]
  \label{coro:gauss-satisfies-assum}
  Let $\tau \geq 0$.
  Let $\sigma_j > 0$ be standard deviations to be specified later.
  Let $p(x) \triangleq \exp(- x^\intercal \Sigma^{-1} x/2) = \prod_{j=1}^d \exp(-x_j^2/(2\sigma_j^2))$.
  Let $\varphi^2 \propto p^2$ suitably normalized so that $\varphi^2$ is a Gaussian distribution with a diagonal covariance matrix $\Sigma = \mathrm{diag}(\sigma_1^2, \dots, \sigma_d^2)$.
  Let $M_1, M_2$ be discretization parameters, with $M_1 \geq R_w$.
  Let
  \begin{equation}
    \tilde{R} = \tilde{\Omega}\left(\max\left(\frac{\tau M_1^2 d^4}{R_w^2}, \frac{D^2\sqrt{d}}{R_w}, \frac{d^2D}{R_w^2}\right)\right)
  \end{equation}
  and let $R \geq \tilde{R}$ be the truncation parameter.
  Then, if $\sigma_j \geq 2R\sqrt{\pi}$ for all $j \in [d]$, Assumptions~\ref{assum:fourier-conc}-\ref{assum:crit-points} are satisfied for $\varphi^2$ for truncation parameter $R$.
  In particular, $\varphi^2$ is $\epsilon(r)$-Fourier-concentrated with $\epsilon(r) = \exp(-\Omega(r^2))$.
\end{corollary}
Note here that we distinguish between $R$ and $\tilde{R}$.
This is because our algorithm, as explained above, uses different discretization/truncation parameters for different subroutines, and some assumptions are only relevant for particular subroutines.
Assumption~\ref{assum:deriv}, in particular, is used in the verification subroutine from Hallgren's algorithm (Appendix~\ref{sec:hallgren}), which is why we consider the truncation parameter $R = \tilde{R}M_1$.
Note that for these different choices of truncation parameters, the required lower bound on the scale parameter $s_j$ also changes.
To satisfy all conditions simultaneously, one may take $s_j \geq R D M_1 \sqrt{5\pi \sqrt{18}}$.

\begin{proof}[Proof of Proposition~\ref{prop:gen-gauss-satisfies-assum}]
We consider the functions $p_j^2(x) = \exp(-(x/s_j)^{\alpha_j})$.
For even $\alpha_j \geq 2$, then $p_j$ is a Schwartz function.
Moreover, since the product of Schwartz functions is still a Schwartz function, then $\varphi$ is also a Schwartz function.
The Fourier transform of a Schwartz function is also a Schwartz function (see, e.g., Proposition 11.25 of~\cite{hunter2001applied}).
Thus, $\varphi^2$ is $\epsilon(r)$-Fourier-concentrated with superpolynomially decaying $\epsilon(r)$, satisfying Assumption~\ref{assum:fourier-conc}.

Moreover, Assumptions~\ref{assum:bounded-1}, \ref{assum:even}, and \ref{assum:crit-points} are clearly satisfied.
For Assumption~\ref{assum:pointwise-close}, we have
\begin{equation}
    |1 - p_j^2(x)| = \left|1 - e^{-\left(\frac{x}{s_j}\right)^{\alpha_j}} \right| \leq \left|\left(\frac{x}{s_j}\right)^{\alpha_j}\right|\leq \frac{R^{\alpha_j}}{s_j^{\alpha_j}} \leq \frac{1}{(4\pi)^{\alpha_j/2}} \leq \frac{1}{4\pi}\leq \frac{1}{10},
\end{equation}
where in the first inequality, we use that $|e^{-x^\alpha} - 1| \leq x^{\alpha}$ for all $x$.
In the second inequality, we use that $x \in [-R,R]$.
In the third inequality, we use that $s_j \geq 2R\sqrt{\pi}$.
In the last inequality, we use that $\alpha_j \geq 2$.
Thus, Assumption~\ref{assum:pointwise-close} is satisfied.

Finally, we need to show that Assumption~\ref{assum:deriv} is satisfied as well.
Consider truncation parameter $R = \tilde{R}M_1$ for $\tilde{R}$ and $M_1$ defined in the proposition statement.
First, we show that $|p_j'(M_1x)| \leq \pi DR_w/(2M_1)$ for all $x \in [-\tilde{R}, \tilde{R}]$.
Taking the derivative, we have
\begin{equation}
  p_j'(x) = -\frac{\alpha_j}{2}\frac{x^{\alpha_j - 1}}{s_j^{\alpha_j}} e^{-\frac{1}{2}\left(\frac{x}{s_j}\right)^{\alpha_j}}.
\end{equation}
Plugging in $M_1x$, we have
\begin{equation}
  \label{eq:deriv-bound}
  |p_j'(M_1x)| \leq \frac{\alpha_j}{2 s_j^{\alpha_j}}|M_1 x|^{\alpha_j - 1} \leq \frac{\alpha_j}{2 s_j^{\alpha_j}}(M_1\tilde{R} )^{\alpha_j - 1} \leq \frac{\alpha_j}{2 (4\pi)^{\alpha_j/2}}\frac{1}{M_1 \tilde{R}} \leq \frac{\pi DR_w}{360 M_1},
\end{equation}
where in the first inequality, we use that $e^{-\pi z}\leq 1$.
In the second inequality, we use that $x \in [-\tilde{R}, \tilde{R}]$.
In the third inequality, we use that $s_j \geq 2R \sqrt{\pi} = 2\tilde{R}M_1 \sqrt{\pi}$.
In the last inequality, we use $\tilde{R} \geq 54D^2 \sqrt{d}/(\pi R_w) \geq 54/(\pi DR_w)$ by our choice of $\tilde{R}$ and the maximum of the function $x/(2(4\pi)^{x/2})$.
Thus, $|p_j'(M_1x)|$ satisfies the required bound.
For $|p_j'(M_1x + T)|$, we have
\begin{equation}
    \label{eq:deriv-bound2}
    |p_j'(M_1x + T)| \leq \left|\frac{\alpha_j}{2}\frac{(M_1 x + T)^{\alpha_j - 1}}{s_j^{\alpha_j}}\right| \leq \frac{\alpha_j}{2 s_j^{\alpha_j}}(M_1 \tilde{R} + |T|)^{\alpha_j - 1}.
\end{equation}
In the last inequality, we use triangle inequality and $x \in [-\tilde{R}, \tilde{R}]$.
To bound $T$, we need to appeal to the specifics of the problem, namely how the guess $T$ is produced from Hallgren's algorithm (\Cref{alg:hallgren}).
Note we will later show that Hallgren's algorithm works for non-uniform distributions, but the following analysis is the same regardless so it suffices to recall the analysis for the uniform case from \Cref{thm:linear-uniform}.
In particular, consider Step 4 of \Cref{alg:hallgren}.
Here, $T$ is either $\lfloor \alpha_i \tilde{R}/\alpha \rfloor$ or $\lceil \alpha_i \tilde{R}/\alpha \rceil$, where $\alpha, \beta$ are the outputs of running quantum Fourier sampling using QSQs and $\alpha_i/\beta_i$ are the convergents of the continued fraction expansion of $\alpha/\beta$.
In the proof of \Cref{thm:linear-uniform}, we show that $e/f$ for $1\leq e,f\leq M_1/w^\star_k$ are convergents of the continued fraction expansion of $\alpha/\beta$.
We also showed that
\begin{equation}
  |\alpha - b| \leq \tau, \quad \left| b - \frac{e\tilde{R}w^\star_k}{M_1}\right| \leq \frac{1}{2},
\end{equation}
for an integer $e \geq 1$, see, e.g., \Cref{eq:hallgren-output1}.
Then, we have
\begin{equation}
  \alpha \geq b - \tau \geq \frac{\tilde{R}w^\star_k}{M_1}- \tau - \frac{1}{2} \geq \frac{R_w \tilde{R}}{M_1 d^2} - \left(\tau - \frac{1}{2}\right) \geq 6\left(\frac{1}{2} + \tau\right)\frac{M_1d^2}{R_w} - \left(\tau + \frac{1}{2}\right) \geq 5\left(\frac{1}{2} + \tau\right) \geq 1.
\end{equation}
Here, in the third inequality, we use that $w_k^\star \geq R_w/d^2$.
In the fourth inequality, we use that $\tilde{R} \geq 6(1/2 + \tau)M_1^2d^4/R_w^2$.
In the fifth inequality, we use that $M_1/w^\star_k \geq R_w/w^\star_k \geq 1$ by our choice of $M_1$ so that $M_1d^2/R_w \geq M_1/w^\star_k \geq 1$ as well.
Finally, in the last inequality, we use that the QSQ tolerance is $\tau \geq 0$.
We can use this to bound $|T|$. Let $\lfloor \cdot \rceil$ denote either $\lfloor \cdot \rfloor$ or $\lceil \cdot \rceil$.
\begin{equation}
  |T| = \left| \left\lfloor \frac{e\tilde{R}}{\alpha} \right\rceil \right| \leq \frac{e \tilde{R}}{\alpha} + 1 \leq \frac{M_1 \tilde{R}}{w^\star_k \alpha} + 1 \leq \frac{M_1 d^2\tilde{R}}{R_w} + 1 \leq \frac{2M_1 d^2 \tilde{R}}{R_w}.
\end{equation}
In the second inequality, we use that $e \leq M_1/w^\star_k$.
In the third inequality, we use that $w_k^\star \geq R_w/d^2$ and $\alpha \geq 1$.
Putting everything together with \Cref{eq:deriv-bound2}, then we have
\begin{align}
    |p_j'(M_1x + T)| &\leq \frac{\alpha_j}{2 s_j^{\alpha_j}}\left(M_1 \tilde{R} + \frac{2M_1 d^2 \tilde{R}}{R_w} \right)^{\alpha_j - 1}\\
    &= \frac{\alpha_j}{2 s_j^{\alpha_j}}\left(1 + \frac{2d^2}{R_w}\right)^{\alpha_j - 1}(M_1\tilde{R})^{\alpha_j - 1}\\
    &\leq \frac{\alpha_j}{2 (4\pi)^{\alpha_j/2}}\left(1 + \frac{2d^2}{R_w}\right)^{\alpha_j - 1} \frac{1}{M_1 \tilde{R}}\\
    &\leq \frac{3}{20}\frac{\pi DR_w}{M_1},
\end{align}
where in the third line, we use $s_j \geq 2 R \sqrt{\pi} = 2\tilde{R}M_1 \sqrt{\pi}$.
In the last line, we use that $\tilde{R} \geq (d^2/R_w)^{\alpha_j - 1} D/(\pi R_w) \geq (d^2/R_w)^{\alpha_j - 1} /(\pi D R_w)$.
This gives the desired bound on $|p_j'(M_1x + T)|$ as well, completing the proof.
\end{proof}

As another example, our assumptions are also satisfied by logistic distributions.

\begin{prop}[Logistic distributions satisfy assumptions]
    \label{prop:logistic-satisfies-assum}
    Let $s_j > 0$ be scale parameters specified later for $j \in [d]$.
    Let $p^2(x) = \prod_{j=1}^d \sech^2(x_j/(2s_j))$ and $\varphi^2 \propto p^2$ suitably normalized so that $\varphi^2$ is a logistic distribution.
    Let $M_1, M_2$ be discretization parameters.
    Let $\tilde{R} = \tilde{\Omega}(\tau M_1^2 d^4/R_w^2, \sqrt{d})$ and let $R \geq \tilde{R}$ be the truncation parameter.
    Then, if $s_j \geq \max(4\pi R, M_1/(\pi D R_w))$ for all $j \in [d]$, Assumptions~\ref{assum:fourier-conc}-\ref{assum:crit-points} are satisfied for $\varphi^2$ for truncation parameter $R$.
    In particular, $\varphi^2$ is $\epsilon(r)$-Fourier-concentrated with $\epsilon(r) = \exp(-\Omega(rd))$.
\end{prop}

\begin{proof}
    We consider the functions $p_j(x) = \sech(x/(2s_j))$.
    Assumptions~\ref{assum:bounded-1}, \ref{assum:even}, and \ref{assum:crit-points} are clearly satisfied.

    Consider Assumption~\ref{assum:fourier-conc}.
    Properly normalized, we have $\varphi_j^2(x) = \sech^2(x/(2s_j))/(4s_j)$, where $\varphi^2(x) = \prod_{j=1}^d \varphi_j^2(x_j)$.
    Thus, $\varphi_j(x) = \sech(x/(2s_j))/(2\sqrt{s_j})$.
    Then, the Fourier transform of $\varphi_j$ is
    \begin{equation}
        \hat{\varphi}_j(y) = \frac{1}{2\sqrt{s_j}}\int_{-\infty}^{+\infty} e^{-2\pi i x y} \sech\left(\frac{x}{2s_j}\right)\,dx = \pi \sqrt{s_j}\sech(2s_j \pi^2 y).
    \end{equation}
    We can use this to compute the Fourier transform of $\varphi$:
    \begin{equation}
        \hat{\varphi}(y) = \int_{\mathbb{R}^d} e^{-2\pi i x^\intercal y} \varphi(x)\,dx = \prod_{j=1}^d \int_{-\infty}^{+\infty} e^{-2\pi i x_j y_j} \varphi_j(x_j) = \prod_{j=1}^d \hat{\varphi}_j(y_j) = \pi^d \prod_{j=1}^d \sqrt{s_j} \sech(2s_j \pi^2 y_j)
    \end{equation}
    Because $\norm{\varphi}_2 = \norm{\hat{\varphi}}_2$, then $\norm{\hat{\varphi}}_2 = 1$.
    Then, to show Fourier concentration (\Cref{def:cont_fourier_conc}), we want to show that $\norm{\hat{\varphi} \cdot \mathbbm{1}_{\geq r}}_2 \leq \epsilon(r)$ for some function $\epsilon(r)$ and $\mathbbm{1}_{\geq r}$ is the indicator function for $\{x : \norm{x}_2 \geq r\}$.
    We have
    \begin{equation}
        \norm{\hat{\varphi} \cdot \mathbbm{1}_{\geq r}}_2^2 = \int_{\norm{y}_2 \geq r} \hat{\varphi}^2(y)\,dy = \pi^{2d} \int_{\norm{y}_2 \geq r} \prod_{j=1}^d s_j \sech^2(2s_j \pi^2 y_j)\,dy
    \end{equation}
    Consider the hypercube inscribed in the hypersphere $\norm{y}_2 \leq r$:
    \begin{equation}
        S \triangleq \left\{y \in \mathbb{R}^d : -\frac{r}{\sqrt{d}} \leq y_1,\dots, y_d \leq \frac{r}{\sqrt{d}} \right\} \subseteq \{y \in \mathbb{R}^d : \norm{y}_2 \leq r\}.
    \end{equation}
    Thus, $\{y : \norm{y}_2 \geq r\} \subseteq \mathbb{R}^d \setminus S$ so that we can bound the integral over this domain
    \begin{equation}
        \norm{\hat{\varphi} \cdot \mathbbm{1}_{\geq r}}_2^2 \leq \pi^{2d} \int_{\mathbb{R}^d \setminus S} \prod_{j=1}^d s_j \sech^2(2s_j \pi^2 y_j)\,dy = \prod_{j=1}^d \left(2\pi^2 s_j \int_{r/\sqrt{d}}^{+\infty} \sech^2(2s_j \pi^2 y_j)\,dy_j\right),
    \end{equation}
    where we also used that $\sech^2$ is an even function.
    Evaluating the integral, we have
    \begin{align}
        \norm{\hat{\varphi} \cdot \mathbbm{1}_{\geq r}}_2^2 &\leq \prod_{j=1}^d \left(\int_{2s_j \pi^2 r/\sqrt{d}}^{+\infty} \sech^2(u_j)\,du_j \right)\\
        &= \prod_{j=1}^d \left(1 - \tanh\left(\frac{2s_j \pi^2 r}{\sqrt{d}}\right)\right)\\
        &= \prod_{j=1}^d \left(1 - \frac{e^{4s_j \pi^2 r/\sqrt{d}}-1}{e^{4s_j \pi^2 r/\sqrt{d}} + 1}\right)\\
        &= \prod_{j=1}^d \left(\frac{2}{e^{4s_j \pi^2 r/\sqrt{d}} + 1}\right)\\
        &\leq \prod_{j=1}^d \left(2e^{-4s_j \pi^2 r/\sqrt{d}}\right)\\
        &\leq 2^d e^{-4\pi^2 r d}\\
        &= e^{-\Omega(rd)}.
    \end{align}
    In the first line, we use the change of variables $u_j = 2s_j \pi^2 y_j$.
    In the second to last line, we use that $s_j \geq 4\pi R \geq 4\pi \sqrt{d} \geq \sqrt{d}$.
    Thus, Assumption~\ref{assum:fourier-conc} is satisfied with $\epsilon(r) = e^{-\Omega(rd)}$.

    For Assumption~\ref{assum:pointwise-close}, we have
    \begin{equation}
        |1 - p_j^2(x)| = \left|1 - \sech^2\left(\frac{x}{2s_j}\right)\right| = 1 - \frac{4e^{x/s_j}}{(e^{x/s_j} + 1)^2}.
    \end{equation}
    For $x \in [0, R]$, $e^x \leq e^R$ and $e^x \geq 1$ so that
    \begin{equation}
        |1 - p^2_j(x)| \leq 1 - \frac{4}{(e^{R/s_j} + 1)^2} \leq 1 - \frac{4}{(e^{1/(4\pi)} + 1)^2} \leq \frac{1}{10},
    \end{equation}
    where in the second inequality, we used that $s_j \geq 4\pi R$.
    Since $p_j^2$ is symmetric, the same holds for $x \in [-R, 0]$.

    Finally, we need to check Assumption~\ref{assum:deriv}.
    Taking the derivative, we have
    \begin{equation}
        p_j'(x) = -\frac{1}{2s_j}\tanh\left(\frac{x}{2s_j}\right)\sech\left(\frac{x}{2s_j}\right).
    \end{equation}
    Plugging in $M_1x$ for $x \in [-\tilde{R}, \tilde{R}]$, then
    \begin{equation}
        |p_j'(M_1 x)| = \left|\frac{1}{2s_j}\tanh\left(\frac{M_1x}{2s_j}\right)\sech\left(\frac{M_1x}{2s_j}\right) \right| \leq \frac{1}{2s_j} \leq \frac{\pi DR_w}{2M_1},
    \end{equation}
    where the last inequality comes from $s_j \geq M_1/(\pi D R_w)$.
    A similar calculation holds for $p_j'(M_1 x + T)$ so that Assumption~\ref{assum:deriv} holds.
\end{proof}

As a corollary of \Cref{thm:non-unif-guarantee}, we obtain the same complexity for learning over generalized Gaussians, Gaussians, and logistic distributions with large enough scale parameters.
Note that in the generalized Gaussian case, one may need to take a larger truncation parameter (as specified in Proposition~\ref{prop:gen-gauss-satisfies-assum}), but the sample complexity remains the same.
Thus, using the Fourier-concentration in Propositions~\ref{prop:gen-gauss-satisfies-assum} and \ref{prop:logistic-satisfies-assum} and Corollary~\ref{coro:gauss-satisfies-assum}, we see that Gaussian, generalized Gaussian, and logistic distributions achieve an exponential sample complexity quantum advantage.

The next sections are dedicated to proving \Cref{thm:non-unif-guarantee}.
In Appendix~\ref{sec:linear-non-unif}, we discuss the non-uniform period finding algorithm.
In Appendix~\ref{sec:outer-non-unif}, similarly to Appendix~\ref{sec:outer-uniform}, we show how one can use gradient descent to learn the outer function $\tilde{g}$ given knowledge of $w^\star$.
In Appendix~\ref{sec:int-bounds-non-unif}, we prove some integral bounds which are useful for both Appendices~\ref{sec:linear-non-unif} and \ref{sec:outer-non-unif}.

\subsection{Learning the linear function}
\label{sec:linear-non-unif}

In this section, we discuss how to use period finding to learn the inner linear function, i.e., how to find the vector of coefficients $w^\star \in \mathbb{R}^d$, when given QSQ access to an example state with non-uniform amplitudes, when the non-uniform distribution satisfies Assumptions~\ref{assum:fourier-conc}-\ref{assum:crit-points}.
We also utilize the results regarding pseudoperiodicity from Appendix~\ref{sec:linear-uniform}.

In Appendix~\ref{sec:warmup-non-unif}, we consider a simple special case for pedagogical purposes to demonstrate how our non-uniform period finding algorithm works.
In Appendix~\ref{sec:general-non-unif}, we prove the general case.

\subsubsection{Warmup}
\label{sec:warmup-non-unif}

As a warmup, let us consider the simple case where $1/w^\star_j \in \mathbb{Z}$ for all $j \in [d]$.
In fact, we prove a general guarantee on period finding for states with non-uniform amplitudes.
Then, the result in our setting, i.e., for learning $w^\star$ from access to $g_{w^\star}$, is a special case.
Note that in this simple case, we only need Assumptions~\ref{assum:fourier-conc}-\ref{assum:bounded-1} to hold for our non-uniform distributions.
Moreover, Assumption~\ref{assum:fourier-conc} is only needed for classical hardness and is not required for the correctness/complexity of our quantum algorithm.

\begin{prop}[Non-Uniform Period Finding; Simple Case]
\label{prop:warmup-non-unif}
Let $\varphi^2 \propto \prod_{j=1}^d p_j^2$ be a probability distribution over $[-R,R]^d$ satisfying Assumptions~\ref{assum:fourier-conc}-\ref{assum:bounded-1} for the truncation parameter $R$ specified shortly and discretization parameter $1$.
Let $\tau \geq 0$. Let $f: \mathbb{Z}^d \to \mathbb{Z}$ be a periodic function with period $S_j$ in each component.
Suppose that we know an upper bound $A_j$ on the period $S_j$.
Let $A = \max_j A_j$.
Let $R\geq (1 + 2\tau)A^2$ be the truncation parameter.
Then, there exists an algorithm that can learn each $S_1,\dots, S_d$ exactly with constant probability using
\begin{equation}
  N = d
\end{equation}
quantum statistical queries with tolerance $\tau$ (with respect to the truncated example state).
\end{prop}

In our case, by \Cref{lem:discrete-simple-unif}, our target function can be suitably discretized to be periodic with period $S = M/w^\star$.
Moreover, note that we know an upper bound on the period $A = Md^2/R_w$ due to \Cref{eq:sw}.
Thus, the previous proposition readily applies, giving us the following corollary.

\begin{corollary}[Linear Function Guarantee; Simple Non-Uniform Case]
Let $\varphi^2 \propto \prod_{j=1}^d p_j^2$ be a probability distribution over $[-R,R]^d$ satisfying Assumptions~\ref{assum:fourier-conc}-\ref{assum:bounded-1} for the parameters specified shortly.
Let $\tau \geq 0$. Let $w^\star \in \mathbb{R}^d$ be unknown with $1/w^\star_j \in \mathbb{Z}$ for all $j \in [d]$.
Suppose also that $w_j^\star \geq R_w/d^2$, for all $j \in [d]$.
Let $g_{w^\star}:\mathbb{R}^d \to [-1,1]$ be defined as $g_{w^\star}(x) = \tilde{g}(x^\intercal w^\star)$, where $\tilde{g}: \mathbb{R} \to [-1,1]$ is a function with period $1$ which has bounded variation on every finite interval.
Let $M \geq 1$ be any choice of discretization parameter and let $R \geq (1+2\tau)M^2d^{4}/R_w^2$ be the truncation parameter.
Then, there exists an algorithm that can learn $w^\star$ exactly with constant probability using
\begin{equation}
  N = d
\end{equation}
quantum statistical queries with tolerance $\tau$ (with respect to the discretized and truncated example state).
\end{corollary}

Again, because a discrete Gaussian distribution with large enough variance satisfies the assumptions on our non-uniform distribution (\Cref{coro:gauss-satisfies-assum}), then this corollary holds for discrete Gaussian distributions as a special case.
It is instructive to note that one could obtain a similar guarantee for discrete Gaussian distributions by leveraging a discrete Gaussian phase estimation subroutine from~\cite{chen2024quantum} (Theorem 4.2).
However, one drawback of this approach is that it does not generalize as easily to the pseudoperiodic case when the period is irrational.
This seems to stem from the issue that the Fourier transform of a pseudoperiodic function does not have a simple exact form as is the case for periodic functions.
Our result in Proposition~\ref{prop:warmup-non-unif} holds for more general distributions than just discrete Gaussians and generalizes to the pseudoperiodic case, as we will see in the next section.
First, we prove Proposition~\ref{prop:warmup-non-unif}.

\begin{proof}[Proof of Proposition~\ref{prop:warmup-non-unif}]
Because $\varphi^2$ is a product distribution, we can write the quantum example state as
\begin{equation}
  \ket{f} = \bigotimes_{j=1}^d \left(\frac{1}{\tilde{G}_j} \sum_{x_j=-R}^{R-1} p_j(x_j)\ket{x_j}\right)\ket{f(x)},\quad \tilde{G}_j \triangleq \sum_{x_j=-R}^{R-1} p_j^2(x_j),
\end{equation}
where $\tilde{G}_j$ are normalization constants, and we define $\tilde{G} \triangleq \prod_{j=1}^d \tilde{G}_j$.
Because of this factorization, by the same argument as in Appendix~\ref{sec:warmup-uniform}, it suffices to consider $d = 1$, as we can perform period finding one component at a time to find each $S_j$.
Thus, from now on, we consider the case of $d = 1$, where we are given QSQ access to the example state
\begin{equation}
  \ket{f} = \frac{1}{\sqrt{\tilde{G}}}\sum_{x=-R}^{R-1}p(x) \ket{x}\ket{f(x)},
\end{equation}
where $R$ is chosen such that $R \geq (1+2\tau)A^2$ and $p$ satisfies Assumptions~\ref{assum:fourier-conc}-\ref{assum:bounded-1}.

Our algorithm is as before: apply the QFT over $q = 2R$ and measure.
This can be encoded into a QSQ by querying the observable
\begin{equation}
  O = \left(\mathsf{QFT}_q \sum_{\ell \in [M]}\frac{\ell}{M}\ketbra{\ell} \mathsf{QFT}_q^{-1}\right) \otimes I,
\end{equation}
where $\mathsf{QFT}_q$ denotes the QFT in dimension $q = 2R$, and $I$ is the identity operator acting on the qubits encoding the output $f(x)$.
This is exactly the same observable as in Appendix~\ref{sec:warmup-uniform}.
However, because of the non-uniform amplitudes, a standard analysis of this algorithm does not apply, so we analyze it in the following.

First, notice that by periodicity we can rewrite our example state as
\begin{equation}
  \ket{f} = \frac{1}{\sqrt{\tilde{G}}}\sum_{x=-R}^{R-1} p(x)\ket{x}\ket{f(x)} = \frac{1}{\sqrt{\tilde{G}}} \sum_{x=0}^{S-1} \sum_{k=-B}^{B-1} p(x + kS) \ket{x + kS}\ket{f(x)},
\end{equation}
where we denote the period of $f$ as $S$ and we write $B \triangleq \lfloor R/ S \rfloor$.
Applying the QFT over $q = 2R$ to this state, we have
\begin{equation}
  \mathsf{QFT}_q \ket{f} = \frac{1}{\sqrt{2R\tilde{G}}} \sum_{x=0}^{S-1}\sum_{y=0}^{2R-1}\sum_{k=-B}^{B-1} e^{2\pi i(x + kS)y/(2R)} p(x + kS) \ket{y}\ket{f(x)}.
\end{equation}
The probability of measuring some outcome $y$ is then
\begin{align}
  \Pr(\text{measure } y) &= \frac{1}{2R\tilde{G}}\norm{\sum_{x=0}^{S-1}\sum_{k=-B}^{B-1} e^{2\pi i(x + kS)y/(2R)} p(x + kS) \ket{f(x)}}^2\\
  &= \frac{1}{2R\tilde{G}}\sum_{x,z=0}^{S-1}\sum_{k,\ell=-B}^{B-1} e^{2\pi i(x + kS)y/(2R)} e^{-2\pi i(z + \ell S)y/(2R)} p(x + k S) p(z + \ell S) \braket{f(z)}{f(x)}\\
  &= \frac{1}{2R\tilde{G}}\sum_{x =0}^{S-1}\sum_{k,\ell=-B}^{B-1} e^{2\pi i(x + kS)y/(2R)} e^{-2\pi i(x + \ell S)y/(2R)} p(x + kS) p(x + \ell S)\\
  &= \frac{1}{2R\tilde{G}}\sum_{x=0}^{S-1}\left|\sum_{k=-B}^{B-1} e^{2\pi i(x + kS)y/(2R)} p(x + kS)\right|^2\\
  &= \frac{1}{2R\tilde{G}}\sum_{x=0}^{S-1}\left|\sum_{k=-B}^{B-1} e^{2\pi i kS y/(2R)} p(x + kS)\right|^2.
\end{align}
We want to lower bound this probability for $y$ satisfying
\begin{equation}
  \left|y - \frac{a R}{S}\right| \leq \frac{1}{2}
\end{equation}
for $a \in \mathbb{Z}$, i.e., $y = \lfloor a R/S \rceil$, where $\lfloor x \rceil$ denotes the nearest integer above or below $x$.
We start by lower bounding the term in absolute value via the reverse triangle inequality:
\begin{align}
  &\left|\sum_{k=-B}^{B-1} e^{2\pi i k S y/(2R)} p(x+kS)\right|\\
  &= \left|\sum_{k=-B}^{B-1} e^{2\pi i k S y/(2R)} + \sum_{k=-B}^{B-1} e^{2\pi i k S y/(2R)} p(x+kS) - \sum_{k=-B}^{B-1} e^{2\pi i k S y/(2R)}\right|\\
  &\geq \left|\left|\sum_{k=-B}^{B-1} e^{2\pi i k S y/(2R)} \right| - \left|\sum_{k=-B}^{B-1} e^{2\pi i k S y/(2R)} p(x+kS) - \sum_{k=-B}^{B-1} e^{2\pi i k S y/(2R)}\right|\right|.
\end{align}
To lower bound this further, we lower bound the first term and upper bound the second.
First, to lower bound the first term, we can change the index of summation to see that
\begin{equation}
  \left|\sum_{k=-B}^{B-1} e^{2\pi i k S y/(2R)} \right| = \left|\sum_{\ell=0}^{2B-1} e^{2\pi i (\ell-B) S y/(2R)} \right| = \left|\sum_{\ell=0}^{2B-1} e^{2\pi i \ell S y/(2R)} \right|,
\end{equation}
where we set $\ell = k + B$. Then, for $y = aR/S + \epsilon$, where $|\epsilon| \leq 1/2$, then this is equal to
\begin{equation}
  \left|\sum_{\ell=0}^{2B-1} e^{2\pi i \ell S \epsilon/(2R)} \right| = \left|\sum_{\ell=0}^{2B-1} e^{2\pi i C \ell /(2B)} \right|.
\end{equation}
Here, we define $C \triangleq B S \epsilon/R$.
Because $B = \lfloor R/S\rfloor$, then $S B \leq R$ and hence $|C| \leq |\epsilon| \leq 1/2$.
Then, by Lemma 3 in~\cite{jozsa2003notes} (or Claim 3.1 in~\cite{hallgren2007polynomial}), we obtain the desired lower bound:
\begin{equation}
  \left|\sum_{k=-B}^{B-1} e^{2\pi i k S y/(2R)} \right| \geq \frac{2}{\sqrt{18}}B.
\end{equation}
Now, we consider the other term.
\begin{align}
  \left|\sum_{k=-B}^{B-1} e^{2\pi i k S y/2R} p(x + kS) - \sum_{k=-B}^{B-1} e^{2\pi i k S y/2R}\right| &\leq \sum_{k=-B}^{B-1}\left|p(x+kS)-1\right|\\
  &= \sum_{k=-B}^{B-1} \left(1-p(x+kS)\right)\\
  &\leq \frac{B}{5}\\
  &\leq \frac{B}{\sqrt{18}}.
\end{align}
Here, in the second line, we use Assumption~\ref{assum:bounded-1}.
In the third line, we use Assumption~\ref{assum:pointwise-close}.
Namely, because $p(x) \leq 1$ by Assumption~\ref{assum:bounded-1}, then $1-p(x+kS) \leq 1 - p^2(x+kS)$, which is in turn less than $1/10$ by Assumption~\ref{assum:pointwise-close}.
Note that Assumption~\ref{assum:pointwise-close} applies because for the range of $x,k$ considered, then $x+kS \in [-R, R]$.
Putting everything together, we thus see that
\begin{equation}
  \left|\sum_{k=-B}^{B-1} e^{2\pi ik S y/2R} p(x + kS)\right| \geq \frac{2}{\sqrt{18}}B - \frac{1}{\sqrt{18}}B = \frac{1}{\sqrt{18}}B.
\end{equation}
Then, plugging this back into our original expression, the probability that we obtain some output $y = \lfloor aR/S\rceil$ is
\begin{align}
  \Pr\left(y = \left\lfloor \frac{a R}{S} \right\rceil\right) &= \frac{1}{2R\tilde{G}}\sum_{x=0}^{S-1}\left|\sum_{k=-B}^{B-1} e^{2\pi ik S y/2R}p(x +kS)\right|\\
  &\geq \frac{1}{2R\tilde{G}}\sum_{x=0}^{S-1}\frac{1}{18}B^2\\
  &\geq \frac{1}{72}\frac{1}{R^2}SB^2\\
  &= \Omega\left(\frac{SB^2}{R^2}\right)\\
  &= \Omega\left(\frac{1}{S}\right).
\end{align}
The third line follows because
\begin{equation}
  \tilde{G} = \sum_{x=-R}^{R-1} p(x)^2 \leq 2R,
\end{equation}
using Assumption~\ref{assum:bounded-1}.
The last line follows because $B = \lfloor R/S\rfloor = \Theta(R/S)$.
Instead of measuring $y = \lfloor a R/S\rfloor$ exactly, we obtain some estimate due to the noisy QSQs.
From here, the analysis is the same as that of Proposition~\ref{prop:warmup-linear-uniform}.
In particular, by the choice of $R \geq (1 + 2\tau)A^2$, so by the same analysis as Proposition~\ref{prop:warmup-linear-uniform}, we can recover $S$ with constant probability.
This only required one QSQ.
To learn each period $S_1,\dots, S_d$, it thus requires $d$ QSQs.
\end{proof}

\subsubsection{General case}
\label{sec:general-non-unif}

In the previous section, we showed that $w^\star$ can be recovered exactly in a simple case when $1/w^\star \in \mathbb{Z}$ and our example state has amplitudes distributed according to a non-uniform distribution satisfying Assumptions~\ref{assum:fourier-conc}-\ref{assum:bounded-1}.
In general, $1/w^\star$ may not be an integer, but nevertheless we can again prove a general guarantee on period finding given an example state with non-uniform amplitudes.
This is in contrast to standard period finding guarantees which only hold for uniform amplitudes.
This then implies that we can find the period of $g_{w^\star}$ as a simple corollary.

We require that our non-uniform distributions satisfy Assumptions~\ref{assum:fourier-conc}-\ref{assum:crit-points}.
As in Appendix~\ref{sec:warmup-non-unif}, Assumption~\ref{assum:fourier-conc} is only needed to ensure classical hardness and is not required for the correctness/complexity of our quantum algorithm.
Our algorithm is the same as Hallgren's algorithm~\cite{hallgren2007polynomial} (see Appendix~\ref{sec:hallgren}) but requires a new analysis due to the non-uniform amplitudes.
This analysis is similar to that of Appendix~\ref{sec:warmup-non-unif}.

\begin{theorem}[Non-Uniform Period Finding]
\label{thm:non-unif-period}
Let $\varphi^2 \propto \prod_{k=1}^d p_k^2$ be a probability distribution over $[-R, R]^d$ satisfying Assumptions~\ref{assum:fourier-conc}-\ref{assum:bounded-1} for a truncation parameter $R$ specified shortly and discretization parameter $1$.
Let $\tau, \eta \geq 0$. Let $f : \mathbb{Z}^d \to \mathbb{Z}$ be an $\eta$-pseudoperiodic function with period $S_j \geq 1$ in each component.
Suppose that, given an integer $T$, we can efficiently check (in time $\mathrm{polylog}(S)$) whether or not $|kS - T| < 1$ for some $k \in \mathbb{Z}$.
Suppose that we know an upper bound $A_j$ on the period $S_j$.
Let $A = \max_j A_j$.
Let $R \geq 6(1/2 + \tau)A^2$ be the truncation parameter.
Then, there exists an algorithm that outputs integers $a_j$ such that $|S_j - a_j| \leq 1$ with probability $\Omega(\eta^2/\log^4 A)$ for all $j \in [d]$ using
\begin{equation}
  N = 2d
\end{equation}
quantum statistical queries with tolerance $\tau$ (with respect to the truncated example state).
\end{theorem}

In the case of $d = 1$, our algorithm is the same as Hallgren's algorithm, which we present in \Cref{alg:hallgren-non-unif}.

\begin{algorithm}
   \caption{Non-Uniform Period Finding} 
   \label{alg:hallgren-non-unif}
   \begin{algorithmic}[1]
   \State Choose a truncation parameter $R \geq 6(1/2 + \tau)A^2$.
   \State Apply quantum Fourier sampling to the function $f$ over $\mathbb{Z}_{q}$, $q = 2R$ twice. Let $\alpha,\beta$ be the outputs.
   \State Compute the continued fraction expansion of $\alpha/\beta$.
   \State For each convergent $\alpha_i/\beta_i$ in the continued fractions expansion, use the verification procedure to check whether $\lfloor \alpha_i R/\alpha \rfloor$ or $\lceil \alpha_i R / \alpha \rceil$ is an integer multiple of the period $S$.
   \State \Return the smallest value that passed the test from the previous step.
   \end{algorithmic}
\end{algorithm}

As a corollary, we can apply \Cref{thm:non-unif-period} to our particular setting to obtain guarantees.

\begin{corollary}[Linear Function Guarantee; Non-Uniform Case]
\label{coro:linear-non-unif}
Let $1 > \epsilon_1 > 0, \delta > 0, \tau \geq 0$.
Let $\varphi^2 \propto \prod_{k=1}^d p_k^2$ be a probability distribution over $[-R, R]^d$ satisfying Assumptions~\ref{assum:fourier-conc}-\ref{assum:crit-points} for the parameters specified shortly.
Let $w^\star \in \mathbb{R}^d$ be unknown with norm $R_w > 0$ and $w_j^\star \geq R_w/d^2$, for all $j \in [d]$.
Let $g_{w^\star}: \mathbb{R}^d \to [-1,1]$ be defined as $g_{w^\star}(x) = \tilde{g}(x^\intercal w^\star)$, where $\tilde{g}: \mathbb{R} \to [-1,1]$ is given in \Cref{eq:g-tilde}.
Consider parameters $M_1 = \max(70\pi d^2D^3 R_w, R_w^2/\epsilon_1), M_2 = c M_1$, where $c$ is any constant such that $M_2$ is an integer and $c < 1/(8\pi D R_w)$, and
\begin{equation}
    \tilde{R} = \tilde{\Omega}\left(\max\left(\frac{\tau M_1^2d^4}{R_w^2}, \frac{D^2}{\epsilon}, \frac{D^2\sqrt{d}}{R_w \epsilon}, \frac{D^{5/2}}{\sqrt{\epsilon}}, \frac{D^{3/2}\sqrt{d}}{R_w \sqrt{\epsilon}}, \frac{d^2D}{R_w^2}\right)\right).
\end{equation}
Suppose we have QSQ access (see \Cref{def:qsq}) with respect to discretization parameters $M_{1,m} \triangleq m M_1$, $M_{2,m} \triangleq m M_2$ and a truncation parameter $R \geq \tilde{R}$, for $m \in \{1,\dots, D\}$.
Then, there exists a quantum algorithm with this QSQ access that can learn an approximation $\hat{w}$ of $w^\star$ such that $\norm{\hat{w} - w^\star}_\infty \leq \epsilon_1$ with probability at least $1-\delta$ using 
\begin{equation}
   N = \mathcal{O}\left(d D \log\left(\frac{1}{\delta}\right) \log^5\left(\frac{M_1 d^2}{R_w}\right)\right)
\end{equation}
quantum statistical queries with tolerance $\tau \leq \min\left(\frac{1}{M_2^2}\left(\frac{5}{42} - \frac{3}{2M_2}\right), \frac{1}{2D^2M_2^2}\left(\frac{2}{9} - \frac{1}{8}\left(\frac{2\pi R_w}{M_1}\right)^2 + \frac{3D^2}{M_2}\right)\right)$ (with respect to the discretized and truncated state).
\end{corollary}

As stated before, our algorithm has two subroutines.
In one, we use QSQs with respect to discretization parameters $M_1, M_2$ and truncation parameter $R = \tilde{R}$.
In the other, we use discretization parameters $M_{1,m} \triangleq m M_1, M_{2,m} \triangleq mM_2$ and truncation parameter $R = \tilde{R}M_{1,m}$ for $m \in \{1,\dots, D\}$.

Also, notice that \Cref{thm:non-unif-period} requires the distribution to satisfy only Assumptions~\ref{assum:fourier-conc}-\ref{assum:bounded-1} while \Cref{coro:linear-non-unif} requires Assumptions~\ref{assum:fourier-conc}-\ref{assum:crit-points}.
This is because Assumptions~\ref{assum:even}-\ref{assum:crit-points} are only needed to instantiate the verification procedure for checking if a given integer is close to an integer multiple of the true period.
This is assumed in \Cref{thm:non-unif-period} whereas we need to construct this algorithm in order to obtain \Cref{coro:linear-non-unif}.
Again, we remark that Assumption~\ref{assum:fourier-conc} is only needed for classical hardness and does not affect the correctness/complexity of our quantum algorithm.

We first prove the correctness of \Cref{alg:hallgren-non-unif} as stated in \Cref{thm:non-unif-period}.
Then, we will prove \Cref{coro:linear-non-unif} using this.

\begin{proof}[Proof of~\Cref{thm:non-unif-period}]
Because $\varphi^2$ is a product distribution, we can write the quantum example state as
\begin{equation}
  \ket{f} = \bigotimes_{j=1}^d \left(\frac{1}{\sqrt{\tilde{G}_j}} \sum_{x_j=-R}^{R-1} p_j(x_j) \ket{x_j}\right) \ket{f(x)},\quad \tilde{G}_j \triangleq \sum_{x_j=-R}^{R - 1}p_j^2(x_j),
\end{equation}
where $\tilde{G}_j$ are normalization constants with $\tilde{G} = \prod_j \tilde{G}_j$.
Because of this factorization, by the same argument as in Appendix~\ref{sec:warmup-uniform}, it suffices to consider $d = 1$, as we can perform period finding one component at a time to find each $S_j$.
Thus, from now on, we consider the case of $d = 1$, where we are given QSQ access to the example state
\begin{equation}
  \ket{f} = \frac{1}{\sqrt{\tilde{G}}}\sum_{x=-R}^{R-1} p(x)\ket{x}\ket{f(x)},
\end{equation}
where $R$ is chosen such that $R \geq 6(1/2 + \tau)A^2$ and $p$ satisfies Assumptions~\ref{assum:fourier-conc}-\ref{assum:bounded-1}.
To apply quantum Fourier sampling in Step 2 of \Cref{alg:hallgren-non-unif}, we can encode this into a QSQ by querying the observable 
\begin{equation}
  O = \left(\mathsf{QFT}_q \sum_{\ell \in [M]}\frac{\ell}{M}\ketbra{\ell} \mathsf{QFT}_q^{-1}\right) \otimes I,
\end{equation}
where $\mathsf{QFT}_q$ denotes the QFT in dimension $q = 2R$, and $I$ is the identity operator acting on the qubits encoding the output $f(x)$.
This is exactly the same observable as in Appendix~\ref{sec:warmup-uniform}.
However, because of the non-uniform amplitudes, a standard analysis of this algorithm does not apply, so we analyze it in the following.
The analysis is similar to Proposition~\ref{prop:warmup-non-unif} but is a bit more complicated due to pseudoperiodicity.

For simplicity, suppose that $f$ is pseudoperiodic on the whole domain rather than an $\eta$-fraction.
This only affects the probability of success of the algorithm, which we will reintroduce in at the end.
Our argument is still valid for only an $\eta$-fraction but would lead to unnecessary complications.
This is also how the proofs of~\cite{jozsa2003notes,hallgren2007polynomial} proceed.

First, notice that by pseudoperiodicity, we can rewrite our example state as
\begin{equation}
  \ket{f} = \frac{1}{\sqrt{\tilde{G}}}\sum_{x=-R}^{R-1} p(x)\ket{x}\ket{f(x)} = \frac{1}{\sqrt{\tilde{G}}}\sum_{x=0}^{S-1}\sum_{k=-B}^{B-1} p(x + [kS])\ket{x + [kS]}\ket{f(x)},
\end{equation}
where we denote the period of $f$ as $S$ and we write $B \triangleq \lfloor R/ S \rfloor$.
We also use $[x]$ to denote a chosen one of the two values $\lfloor x \rfloor$ or $\lceil x \rceil$.
Note that this is different from $\lfloor x \rceil$, which denotes rounding to the nearest integer above or below $x$.
Applying the QFT over $q = 2R$ to this state as in Step 2 of the algorithm, we have
\begin{equation}
  \mathsf{QFT}_q \ket{f} = \frac{1}{\sqrt{2R\tilde{G}}} \sum_{x=0}^{S-1}\sum_{y=0}^{2R-1}\sum_{k=-B}^{B-1} e^{2\pi i(x + [kS])y/(2R)} p(x + [kS]) \ket{y}\ket{f(x)}.
\end{equation}
The probability of measuring some outcome $y$ is then
\begin{align}
  \Pr(\text{measure } y) &= \frac{1}{2R\tilde{G}}\norm{\sum_{x=0}^{S-1}\sum_{k=-B}^{B-1} e^{2\pi i(x + [kS])y/(2R)} p(x + [kS]) \ket{f(x)}}^2\\
  &= \frac{1}{2R\tilde{G}}\sum_{x,z=0}^{S-1}\sum_{k,\ell=-B}^{B-1} e^{2\pi i(x + [kS])y/(2R)} e^{-2\pi i(z + [\ell S])y/(2R)} p(x + [k S]) p(z + [\ell S]) \braket{f(z)}{f(x)}\\
  &= \frac{1}{2R\tilde{G}}\sum_{x =0}^{S-1}\sum_{k,\ell=-B}^{B-1} e^{2\pi i(x +[ kS])y/(2R)} e^{-2\pi i(x + [\ell S])y/(2R)} p(x + [kS]) p(x + [\ell S])\\
  &= \frac{1}{2R\tilde{G}}\sum_{x=0}^{S-1}\left|\sum_{k=-B}^{B-1} e^{2\pi i(x + [kS])y/(2R)} p(x + [kS])\right|^2\\
  &= \frac{1}{2R\tilde{G}}\sum_{x=0}^{S-1}\left|\sum_{k=-B}^{B-1} e^{2\pi i [kS] y/(2R)} p(x + [kS])\right|^2.
\end{align}
We want to lower bound this probability for $y = \lfloor a R/S \rceil$ for $a \in \mathbb{Z}$ and $y < R / \log A$.
We start by lower bounding the term in absolute value via the reverse triangle inequality:
\begin{align}
  &\left|\sum_{k=-B}^{B-1} e^{2\pi i [k S] y/(2R)} p(x + [kS])\right|\\
  &= \left|\sum_{k=-B}^{B-1} e^{2\pi i [k S] y/(2R)} + \sum_{k=-B}^{B-1} e^{2\pi i [k S] y/(2R)} p(x+[kS]) - \sum_{k=-B}^{B-1} e^{2\pi i [k S] y/(2R)}\right|\\
  &\geq \left|\left|\sum_{k=-B}^{B-1} e^{2\pi i [k S] y/(2R)} \right| - \left|\sum_{k=-B}^{B-1} e^{2\pi i [k S] y/(2R)} p(x+[kS]) - \sum_{k=-B}^{B-1} e^{2\pi i [k S] y/(2R)}\right|\right|.
\end{align}
To lower bound this further, we lower bound the first term and upper bound the second.
First, to lower bound the first term, we can change the index of summation to see that
\begin{equation}
  \left|\sum_{k=-B}^{B-1} e^{2\pi i [k S] y/(2R)} \right| = \left|\sum_{k=-B}^{B-1} e^{2\pi i (kS + \delta_k) y/(2R)} \right| = \left|\sum_{\ell=0}^{2B-1} e^{\frac{2\pi i(\ell - B) Sy + 2\pi i\delta_{\ell - B}y}{2R}} \right| = \left|\sum_{\ell=0}^{2B-1} e^{2\pi i(\ell S + \delta_{\ell - B})y/(2R)} \right|.
\end{equation}
Here, we wrote $[kS] = kS = \delta_k$, where $|\delta_k| < 1$.
Then, for $y = aR/S + \epsilon$, where $|\epsilon| \leq 1/2$, then this is equal to
\begin{equation}
  \left|\sum_{\ell=0}^{2B-1} e^{2\pi i(\ell S + \delta_{\ell - B})(aR/S + \epsilon)/2R} \right| = \left|\sum_{\ell=0}^{2B-1} e^{2\pi i\left(\frac{\epsilon \ell S}{2R} + \frac{a \delta_{\ell - B}}{S} + \frac{\epsilon \delta_{\ell - B}}{2R}\right)} \right|.
\end{equation}
Now, define $C \triangleq BS \epsilon/R$.
Because $B = \lfloor R/S\rfloor$, then $S B \leq R$ and hence $|C| \leq |\epsilon| \leq 1/2$.
Also, note that because we are considering $y = a R/S + \epsilon < R/\log S$ and $|\epsilon| \leq 1/2$, then $a/S < 1/\log S + 1/(2R)$.
Since $|\delta_k| < 1$ and $R \geq 6(1/2 + \tau)A^2 \geq 3S^2$, then
\begin{equation}
  \left|\frac{a \delta_{\ell - B}}{S} + \frac{\epsilon\delta_{\ell - B}}{2R}\right| < \frac{1}{\log S} + \frac{1}{2R} + \frac{1}{4R} \leq \frac{2}{\log S}.
\end{equation}
Thus, we can write our summation as
\begin{equation}
  \left|\sum_{\ell=0}^{2B-1} e^{2\pi i(C\ell/(2B) + \xi(\ell))}\right|,
\end{equation}
where $|\xi(\ell)| \leq 2/\log S$.
By Lemma 3 in~\cite{jozsa2003notes} (or Claim 3.1 in~\cite{hallgren2007polynomial}), we obtain the desired lower bound:
\begin{equation}
   \left|\sum_{k=-B}^{B-1} e^{2\pi i [k S] y/(2R)} \right| \geq \frac{2}{\sqrt{18}}B
\end{equation} 
if $y = \lfloor a R/S \rceil$ and $y < R/\log S$.
Now, we consider the other term.
\begin{align}
\left|\sum_{k=-B}^{B-1} e^{2\pi i [k S] y/(2R)} p(x+[kS]) - \sum_{k=-B}^{B-1} e^{2\pi i [k S] y/(2R)}\right| &\leq \sum_{k=-B}^{B-1} \left|p(x + [kS]) - 1 \right|\\
&= \sum_{k=-B}^{B-1} (1 - p(x+[kS]))\\
&\leq \frac{B}{5}\\
&\leq \frac{B}{\sqrt{18}}.
\end{align}
Here, in the second line, we use Assumption~\ref{assum:bounded-1}.
In the third line, we use Assumption~\ref{assum:pointwise-close}.
Namely, because $p(x) \leq 1$ by Assumption~\ref{assum:bounded-1}, then $1 - p(x+[kS]) \leq 1 - p^2(x + [kS])$, which is in turn less than $1/10$ by Assumption~\ref{assum:pointwise-close}.
Note that Assumption~\ref{assum:pointwise-close} applies because for the range of $x,k$ considered, then $x + [kS]\in[-R,R]$.

Putting everything together, we thus see that
\begin{equation}
  \left|\sum_{k=-B}^{B-1} e^{2\pi i[k S] y/(2R)}p(x + [kS])\right| \geq \frac{2}{\sqrt{18}}B - \frac{1}{\sqrt{18}}B = \frac{1}{\sqrt{18}}B
\end{equation}
for $y = \lfloor a R/S \rceil$ and $y < R/\log S$.
Then, plugging this back into our original expression, the probability that we obtain some output $y = \lfloor aR/S\rceil$ and $y < R/\log S$ is
\begin{align}
  \Pr\left(y = \left\lfloor \frac{a R}{S} \text{ and } y < R/\log S\right\rceil\right) &= \frac{1}{2R\tilde{G}}\sum_{x=0}^{S-1}\left|\sum_{k=-B}^{B-1} e^{2\pi i[k S] y/2R}p(x + [kS])\right|^2\\
  &\geq \frac{1}{2R\tilde{G}}\sum_{x=0}^{S-1}\frac{1}{18}B^2\\
  &\geq \frac{1}{72}\frac{1}{R^2}SB^2\\
  &= \Omega\left(\frac{SB^2}{R^2}\right)\\
  &= \Omega\left(\frac{1}{S}\right).
\end{align}
Here, the second line follows from our above argument.
The third line follows because
\begin{equation}
  \tilde{G} = \sum_{x=-R}^{R-1} p^2(x) \leq 2R,
\end{equation}
where we used Assumption~\ref{assum:bounded-1}.
The last line follows because $B = \lfloor R/S\rfloor = \Theta(R/S)$.

There are $S / \log A$ integer multiples of $R/S$ less than $R/\log A$ (and hence less than $R / \log S$).
Thus, the probability of measuring two values less than $R/ \log A$ (as in Step 2 of \Cref{alg:hallgren-non-unif}) is $\Omega(1/\log^2 A)$.
Furthermore, the probability that the two values are relatively prime is at least $\Omega(1/\log(S/\log A))^2$ by the prime number theorem.
The probability of measuring two such values satisfying all the conditions (including pseudoperiodicity) is $\Omega(\eta^2/\log^4 A)$.

Steps 3-5 of \Cref{alg:hallgren-non-unif} are analyzed in the case that we obtain a noisy estimate with tolerance $\tau$ in the same way as \Cref{thm:linear-uniform}.
Thus, we obtain the claim.
\end{proof}

Now, we can prove \Cref{coro:linear-non-unif} using \Cref{thm:non-unif-period}.
We need to show that the condition about checking whether a guess for the period is close or not is satisfied.
We design such a verification procedure in \Cref{alg:verification-non-unif} and analyze it in \Cref{thm:verification-non-unif}.
This is analogous to \Cref{thm:verification}.
As before, in \Cref{alg:verification-non-unif}, we must restrict the noise tolerance of our QSQs to be inverse polynomial in some of our parameters.
Classically, the hardness results have access to gradients that are exponentially accurate, so requiring the tolerance parameter to scale inverse polynomially is not particularly strong.

\begin{algorithm}
   \caption{Verification Procedure; Non-Uniform Case} 
   \label{alg:verification-non-unif}
   \begin{algorithmic}[1]
   \State Choose parameters $M_1 = \max(70\pi d^2 D^3 R_w, R_w^2/\epsilon_1)$, $M_2 = cM_1$ for $c$ any constant such that $M_2 \in \mathbb{Z}$ and $c < 1/(8\pi DR_w)$, and $\tilde{R} = \tilde{\Omega}\left(\max\left(\frac{\tau M_1^2d^4}{R_w^2},\frac{D^2}{\epsilon}, \frac{D^2\sqrt{d}}{R_w \epsilon}, \frac{D^{5/2}}{\sqrt{\epsilon}}, \frac{D^{3/2}\sqrt{d}}{R_w \sqrt{\epsilon}}, \frac{d^2 D}{R_w^2}\right)\right)$.
   \State For $m \in \{1,\dots, D\}$, query the QSQ oracle with observable $O_{k,m}$ (defined in \Cref{eq:ov}), discretization parameters $M_{1,m} \triangleq m M_1$, $M_{2,m} \triangleq mM_2$, truncation parameter $R \triangleq \tilde{R}M_{1,m}$, and tolerance $\tau \leq \min\left(\frac{1}{M_2^2}\left(\frac{5}{42} - \frac{3}{2M_2}\right), \frac{1}{2D^2M_2^2}\left(\frac{2}{9} - \frac{1}{8}\left(\frac{2\pi R_w}{M_1}\right)^2 + \frac{3D^2}{M_2}\right)\right)$ to obtain values $\alpha_m$.
   \State Check if $\alpha_1 \geq \frac{1}{M_2^2}\left(\frac{5}{14} - \frac{9}{2M_2}\right)$.
   \State Check if $\sum_{m=1}^D \alpha_m \leq \frac{1}{M_2^2}\left(\frac{13}{25}D + \frac{1}{2D}\left(\frac{2}{9} - \frac{1}{8}\left(\frac{2\pi R_w}{M_1}\right)^2 +\frac{3D^2}{M_2}\right)\right)$.
   \State \Return ``yes'' iff both conditions in Steps 3 and 4 are satisfied.
   \end{algorithmic}
\end{algorithm}

\begin{theorem}[Verification Procedure; Non-Uniform Case]
\label{thm:verification-non-unif}
Let $\varphi^2 \propto \prod_{k=1}^d p_k^2$ be a probability distribution over $[-R, R]^d$ satisfying Assumptions~\ref{assum:fourier-conc}-\ref{assum:crit-points} for a truncation parameter $R$ specified later.
Let $1 > \epsilon_1 > 0$.
Let $w^\star \in \mathbb{R}^d$ be unknown with norm $R_w > 0$ and $w_j^\star \geq R_w/d^2$ for all $j \in [d]$.
Let $g_{w^\star} : \mathbb{R}^d \to [-1,1]$ be defined as $g_{w^\star}(x) = \tilde{g}(x^\intercal w^\star)$ for $\tilde{g}$ given in \Cref{eq:g-tilde}.
Consider parameters $M_1 = \max(70\pi d^2 D^3 R_w, R_w^2/\epsilon_1)$, $M_2 = cM_1$, where $c$ is any constant such that $M_2 \in \mathbb{Z}$ and $c < 1/(8\pi DR_w)$, and 
\begin{equation}
    \tilde{R} = \tilde{\Omega}\left(\max\left(\frac{\tau M_1^2d^4}{R_w^2},\frac{D^2}{\epsilon}, \frac{D^2\sqrt{d}}{R_w \epsilon}, \frac{D^{5/2}}{\sqrt{\epsilon}}, \frac{D^{3/2}\sqrt{d}}{R_w \sqrt{\epsilon}}, \frac{d^2 D}{R_w^2}\right)\right).
\end{equation}
Suppose we have QSQ access (see \Cref{def:qsq}) with respect to discretization parameters $M_{1,m} \triangleq m M_1$, $M_{2,m} \triangleq m M_2$ and a truncation parameter $R \triangleq \tilde{R}M_{1,m}$, for $m \in \{1,\dots, D\}$.
Then, given an integer $T$ found as in \Cref{alg:hallgren-non-unif} and $k \in [d]$, \Cref{alg:verification-non-unif} can check whether or not $|T - \frac{\ell M_1}{w_k^\star}| \leq 1$ for some integer $\ell$ using $D$ QSQs with tolerance $\tau \leq \min\left(\frac{1}{M_2^2}\left(\frac{5}{42} - \frac{3}{2M_2}\right), \frac{1}{2D^2M_2^2}\left(\frac{2}{9} - \frac{1}{8}\left(\frac{2\pi R_w}{M_1}\right)^2 + \frac{3D^2}{M_2}\right)\right)$.
\end{theorem}

\begin{proof}
This proof is similar to \Cref{thm:verification}, so we omit some details when they follow straightforwardly from \Cref{thm:verification}.
Explicitly, the example state for our QSQ access is
\begin{equation}
  \label{eq:example-detail-non-unif}
  \ket{h_{w^\star, M_{1,m}, M_{2,m}}} = \frac{1}{\sqrt{\tilde{G}_d}} \sum_{x_1,\dots, x_d = -\tilde{R}M_{1,m}}^{\tilde{R}M_{1,m} - 1} p_1(x_1) \cdots p_d(x_d) \ket{x}\ket{h_{w^\star, M_{1,m}, M_{2,m}}(x)},
\end{equation}
where
\begin{equation}
  \tilde{G}_d \triangleq \sum_{x_1,\dots, x_d=-\tilde{R}M_{1,m}}^{\tilde{R}M_{1,m}-1} p_1^2(x_1) \cdots p_d^2(x_d)
\end{equation}
is a normalization constant.
Also, $h_{w^\star, M_{1,m}, M_{2,m}}$ is a discretization of $g_{w^\star}$ from \Cref{lem:discrete-general-unif}.
Note that $p_j$ satisfies Assumptions~\ref{assum:fourier-conc}-\ref{assum:crit-points} for the truncation parameter $R = \tilde{R}M_{1,m}$, not $\tilde{R}$.
We query $D$ QSQs, each with the different parameters indexed by $m$ as specified previously.

As in \Cref{thm:verification}, the main idea behind our verification procedure is to compute the inner product between $h_{w^\star, M_{1,m}, M_{2,m}}$ and this function with its input shifted by the guess $T$ for the period.
This inner product should be large for a good guess.
We again consider the observable $O_{k,m}$ defined in \Cref{eq:ov} which computes the inner product between $h_{w^\star, M_{1,m}, M_{2,m}}$ and this function with its input shifted by $T$.

\begin{claim}[Approximating inner product; Non-uniform case]
\label{claim:inner-prod-non-unif}
For $m \in \{1,\dots, D\}$, consider parameters $M_{1,m}, M_{2,m}$ as defined above.
Also consider a parameter $\tilde{R}$ and an observable $O_{k,m}$ as defined above.
Then, the expectation value of $O_m$ with respect to the example state in \Cref{eq:example-detail-non-unif} is given by
\begin{align}
  &\expval{O_{k,m}}{h_{w^\star, M_{1,m}, M_{2,m}}}\\
  &= \frac{1}{\tilde{G}_d M_{2,m}^2} \sum_{x_1,\dots, x_d = -\tilde{R}M_{1,m}}^{\tilde{R}M_{1,m}-1} p_1^2(x_1)\cdots p_k(x_k)p_k(x_k + T) \cdots p_d^2(x_d) h_{w^\star, M_{1,m}, M_{2,m}}(x) h_{w^\star, M_{1,m}, M_{2,m}}(x + Te_k),
\end{align}
where $e_k$ denotes the unit vector with a single one in the $k$th coordinate.
\end{claim}

We omit the proof of this claim, as it follows in the same way as \Cref{claim:inner-prod}.
Now, we want to show that the conditions checked in Steps 3 and 4 in \Cref{alg:verification-non-unif} are satisfied if and only if $|T - \ell M_1/w_k^\star| \leq 1$.
To do so, we first simplify our approximate inner product from \Cref{claim:inner-prod-non-unif} further using the particular form of $h_{w^\star, M_{1,m}, M_{2,m}}$ from \Cref{lem:discrete-general-unif} and $\tilde{g}$ from \Cref{eq:g-tilde}.
This follows in the same way as \Cref{thm:verification}, just carrying along the extra discrete Gaussian terms.
For simplicity, denote
\begin{equation}
  p^2_T(x) \triangleq p_1^2(x_1)\cdots p_k(x_k)p_k(x_k + T) \cdots p_d^2(x_d).
\end{equation}
Then, we have
\begin{align}
  &\expval{O_{k,m}}{h_{w^\star, M_{1,m}, M_{2,m}}}\\
  &\begin{aligned}
    =\frac{1}{\tilde{G}_d M_{2,m}^2} \sum_{x_1,\dots, x_d=-\tilde{R}M_{1,m}}^{\tilde{R}M_{1,m}-1} \sum_{j=1}^D p^2_T\left(x\right) (\beta_j^\star)^2 &\left(\cos^2\left(\frac{2\pi j x^\intercal w^\star}{M_{1,m}}\right)\cos\left(\frac{2\pi j T w_k^\star}{M_{1,m}}\right)\right.\\
    &\left.- \cos\left(\frac{2\pi j x^\intercal w^\star}{M_{1,m}}\right)\sin\left(\frac{2\pi j x^\intercal w^\star}{M_{1,m}}\right)\sin\left(\frac{2\pi j T w_k^\star}{M_{1,m}}\right) \right)
    \end{aligned}\label{eq:h-o-expval-non-unif}\\
  &\begin{aligned}
    +\frac{1}{\tilde{G}_d M_{2,m}^2} \sum_{x_1,\dots, x_d=-\tilde{R}M_{1,m}}^{\tilde{R}M_{1,m}-1} \sum_{\substack{j,j'=1\\j\neq j'}}^D p^2_T\left(x\right) &\beta_j^\star \beta_{j'}^\star\left(\cos\left(\frac{2\pi j x^\intercal w^\star}{M_{1,m}}\right)\cos\left(\frac{2\pi j' x^\intercal w^\star}{M_{1,m}}\right)\cos\left(\frac{2\pi j' Tw_k^\star}{M_{1,m}}\right)\right.\\
    &\left.- \cos\left(\frac{2\pi j x^\intercal w^\star}{M_{1,m}}\right)\sin\left(\frac{2\pi j' x^\intercal w^\star}{M_{1,m}}\right)\sin\left(\frac{2\pi j' T w_k^\star}{M_{1,m}}\right)\right) + \epsilon_d
    \label{eq:h-o-expval2-non-unif}
  \end{aligned}
\end{align}

We want to upper and lower bound this expression.
To do so, we find it easier to work with integrals over $x$ instead of these discrete sums.
We can then bound the integrals, which we relegate to Appendix~\ref{sec:int-bounds-non-unif}.
To this end, we first need to bound the error from approximating our summation by an integral.

\begin{claim}[Sum-to-integral error; Non-uniform case]
\label{claim:sum-to-int-non-unif}
For $m \in \{1,\dots, D\}$, consider parameters $M_{1,m}, M_{2,m}$ as defined above. Also consider a parameter $R$ defined above. Then, for an integer $1 \leq j \leq D$,
\begin{align}
  &\frac{1}{G_d}\left|\int_{[-\tilde{R},\tilde{R}]^d} p^2_T(M_{1,m}x) \cos^2\left(2\pi j x^\intercal w^\star\right)\,dx - \frac{1}{M_{1,m}^d} \sum_{x_1,\dots, x_d = -\tilde{R}M_{1,m}}^{\tilde{R}M_{1,m} - 1} p^2_T\left(x\right)\cos^2\left(\frac{2\pi j x^\intercal w^\star}{M_{1,m}}\right) \right|\\
  &\leq \frac{6805}{6804}\frac{1}{21 D^2},
\end{align}
where
\begin{equation}
  G_d \triangleq \prod_{i=1}^d \left(\int_{-\tilde{R}}^{+\tilde{R}} p^2_j(M_{1,m} x_j)\,dx_j\right).
\end{equation}
\end{claim}

\begin{proof}[Proof of \Cref{claim:sum-to-int-non-unif}]
The proof is similar to that of \Cref{claim:sum-to-int}, so we omit some details.
As in \Cref{claim:sum-to-int}, we prove this by induction on the dimension $d$.
Denoting $f(x) \triangleq \cos^2(2\pi j x^\intercal w^\star)$, we will prove
\begin{equation}
  \label{eq:induct-non-unif}
  \frac{1}{G_d}\left|\int_{[-\tilde{R},\tilde{R}]^d} p^2_T(M_{1,m}x) f(x)\,dx - \frac{1}{M_{1,m}^d} \sum_{x_1,\dots, x_d = -\tilde{R}M_{1,m}}^{\tilde{R}M_{1,m} - 1} p^2_T\left(x\right)f\left(\frac{x}{M_{1,m}}\right) \right| \leq \frac{6805}{6804}\frac{10\pi d DR_w}{3 M_{1,m}}.
\end{equation}
Note that this implies our claim by our choice of $M_{1,m} = m M_1 \geq 70m \pi d^2 D^3 R_w \geq 70 \pi d D^3 R_w$.
Thus, it suffices to prove \Cref{eq:induct-non-unif}.
In fact, we will use induction to prove that
\begin{align}
  \label{eq:induct2-non-unif}
  &\frac{1}{G_{d-1}} \left|\int_{[-\tilde{R},\tilde{R}]^{d-1}} p^2_T(M_{1,m}x) f(x,y)\,dx - \frac{1}{M_{1,m}^{d-1}} \sum_{x_1,\dots, x_{d-1}=-\tilde{R}M_{1,m}}^{\tilde{R}M_{1,m}-1} p^2_T\left(x\right) f\left(\frac{x}{M_{1,m}}, y\right) \right|\\
  &\leq \frac{6805}{6804}\frac{10\pi (d-1) DR_w}{3M_{1,m}}
\end{align}
for some fixed $y$.
In the process, we show that \Cref{eq:induct-non-unif} follows from this.
First, consider the base case. We want to prove
\begin{equation}
  \label{eq:induct-base-non-unif}
  \frac{1}{G_1}\left|\int_{-\tilde{R}}^{+\tilde{R}} p(M_{1,m}x) p(M_{1,m} x + T)f(x)\,dx - \frac{1}{M_{1,m}} \sum_{x=-\tilde{R}M_{1,m}}^{\tilde{R}M_{1,m}-1} p(x) p(x + T) f\left(\frac{x}{M_{1,m}}\right)\,dx\right| \leq \frac{6805}{6804} \frac{10\pi D R_w}{3M_{1,m}}
\end{equation}
and
\begin{align}
  \label{eq:induct-base2-non-unif}
  &\frac{1}{G_1}\left|\int_{-\tilde{R}}^{+\tilde{R}} p(M_{1,m}x) p(M_{1,m}x + T) f(x,y)\,dx - \frac{1}{M_{1,m}} \sum_{x=-\tilde{R}M_{1,m}}^{\tilde{R}M_{1,m}-1} p(x)p(x+T) f\left(\frac{x}{M_{1,m}}, y\right)\,dx\right|\\
  &\leq \frac{6805}{6804} \frac{10\pi D R_w}{3M_{1,m}}
\end{align}
for some fixed $y$.
First, for \Cref{eq:induct-base-non-unif}, the error can be bounded by standard results in approximating integrals by Riemann sums:
\begin{equation}
  \label{eq:riemann-approx-non-unif}
  \left|\int_{-\tilde{R}}^{+\tilde{R}} p(M_{1,m}x) p(M_{1,m} x +T)f(x)\,dx - \frac{1}{M_{1,m}} \sum_{x=-\tilde{R}M_{1,m}}^{\tilde{R}M_{1,m}-1} p(x)p(x+T) f\left(\frac{x}{M_{1,m}}\right)\,dx\right| \leq \frac{L\tilde{R}}{M_{1,m}},
\end{equation}
where $L \triangleq \max_{x \in [-\tilde{R}, \tilde{R}]} |\tilde{f}'(x)|$ and $\tilde{f}(x) \triangleq p(M_{1,m}x) p(M_{1,m}x +T)f(x)$.
By definition, $f(x) = \cos^2(2\pi j xw^\star)$, so
\begin{align}
  \tilde{f}'(x) &= M_{1,m}p'(M_{1,m} x) p(M_{1,m}x + T)\cos^2(2\pi j x w^\star) + M_{1,m}p(M_{1,m} x)p'(M_{1,m}x + T)\cos^2(2\pi j xw^\star) \\
  &- 2p(M_{1,m}x) p(M_{1,m}x + T)\cos(2\pi j xw^\star)\sin(2\pi j xw^\star) \cdot 2\pi j w^\star.
\end{align}
Then,
\begin{align}
  \label{eq:grad-bound-non-unif}
  |\tilde{f}'(x)| \leq M_{1,m}|p'(M_{1,m}x)| + M_{1,m}|p'(M_{1,m}x + T)| + 4\pi D R_w \leq 5\pi DR_w.
\end{align}
In the first inequality, we used Assumption~\ref{assum:bounded-1} so that $p(x) \leq 1$.
In the second inequality, we used Assumption~\ref{assum:deriv}, which bounds the derivative of $p$ by $\pi DR_w/(2M_{1,m})$ since we chose our discretization parameter as $M_{1,m}$.
Thus, we can conclude that
\begin{equation}
  \left|\int_{-\tilde{R}}^{+\tilde{R}} p(M_{1,m}x) p(M_{1,m}x+T)f(x)\,dx - \frac{1}{M_{1,m}} \sum_{x=-\tilde{R}M_{1,m}}^{\tilde{R}M_{1,m}-1} p(x)p(x+T) f\left(\frac{x}{M_{1,m}}\right)\,dx\right| \leq \frac{5\pi DR_w \tilde{R}}{M_{1,m}}.
\end{equation}
When dividing both sides by $G_1$, note that
\begin{equation}
  \label{eq:g1-lower-non-unif}
  G_1 = \int_{-\tilde{R}}^{+\tilde{R}}p^2(M_{1,m}z)\,dz \geq \frac{9\tilde{R}}{5} \geq \frac{5\tilde{R}}{3},
\end{equation}
where we used Assumption~\ref{assum:pointwise-close}.
In particular, by Assumption~\ref{assum:pointwise-close}, we have
\begin{equation}
  \int_{-\tilde{R}}^{+\tilde{R}} p^2(M_{1,m} z)\,dz - 2\tilde{R} = \int_{-\tilde{R}}^{+\tilde{R}} (p^2(M_{1,m} z) - 1)\,dz \geq -\frac{\tilde{R}}{5}.
\end{equation}
This implies that $\int_{-\tilde{R}}^{+\tilde{R}}p^2(M_{1,m} z)\,dz \geq 2\tilde{R} - \tilde{R}/5 \geq 9\tilde{R}/5 \geq 5\tilde{R}/3$.
Note that Assumption~\ref{assum:pointwise-close} applies here because the truncation parameter is $R = \tilde{R}M_{1,m}$ so that $M_{1,m}z \in [-\tilde{R}M_{1,m}, \tilde{R}M_{1,m}] = [-R,R]$.
Thus, we have that $\tilde{R}/G_1 \leq 3/5$ so that
\begin{align}
  &\frac{1}{G_1}\left|\int_{-\tilde{R}}^{+\tilde{R}} p(M_{1,m}x) p(M_{1,m}x + T)f(x)\,dx - \frac{1}{M_{1,m}} \sum_{x=-\tilde{R}M_{1,m}}^{\tilde{R}M_{1,m}-1} p(x) p(x +T) f\left(\frac{x}{M_{1,m}}\right)\,dx\right|\\
  &\leq \frac{3\pi DR_w}{M_{1,m}}\\
  &\leq \frac{6805}{6804}\frac{10\pi DR_w}{3M_{1,m}}.
\end{align}
The proof of \Cref{eq:induct-base2-non-unif} follows similarly.

Now, for the inductive step, suppose for $\ell$ such that $d - 1 \geq \ell \geq 1$ that
\begin{equation}
  \frac{1}{G_\ell}\left|\int_{[-\tilde{R},\tilde{R}]^\ell} p^2_T(M_{1,m}x)f(x, y)\,dx - \frac{1}{M_{1,m}^\ell} \sum_{x_1,\dots, x_\ell=-\tilde{R}M_{1,m}}^{\tilde{R}M_{1,m}-1} p_T^2(x) f\left(\frac{x}{M_{1,m}}, y\right)\right| \leq \frac{6805}{6804}\frac{10\pi \ell DR_w}{3M_{1,m}}
\end{equation}
for some fixed $y$ and where $f(x_1,\dots, x_\ell,y) = \cos^2(2\pi j (x_1 w_1^\star + \cdots + x_\ell w_\ell^\star + yw_{\ell+1}^\star))$.
We first show that \Cref{eq:induct-non-unif} holds for $\ell + 1$.
Suppose that $\ell + 1 > k$ for now.
\begin{align}
  &\frac{1}{G_{\ell + 1}} \int_{[-\tilde{R}, \tilde{R}]^{\ell + 1}} p^2_T(M_{1,m}x) f(x)\,dx\\
  &\begin{aligned}
    = \frac{1}{\int_{-\tilde{R}}^{+\tilde{R}} p_{\ell+1}^2(M_{1,m} z)\,dz} \int_{-\tilde{R}}^{+\tilde{R}} &\left(\frac{1}{G_\ell} \int_{[-\tilde{R}, \tilde{R}]^\ell} p^2_T(M_{1,m}x_1,\dots, M_{1,m}x_\ell) f(x_1,\dots, x_{\ell+1}) \,dx_1 \cdots dx_\ell\right)\\
    &\cdot p_{\ell+1}^2(M_{1,m} x_{\ell+1})\,dx_{\ell+1}
  \end{aligned}\\
  &\leq \frac{1}{M_{1,m}^\ell G_{\ell+1}} \sum_{x_1,\dots, x_\ell=-\tilde{R}M_{1,m}}^{\tilde{R}M_{1,m} -1} \int_{-\tilde{R}}^{+\tilde{R}} p^2_T\left(x_1,\dots, x_\ell, M_{1,m}x_{\ell+1}\right) f\left(\frac{x_1}{M_{1,m}},\dots, \frac{x_\ell}{M_{1,m}}, x_{\ell+1}\right)\,dx_{\ell+1}\\
  &+ \frac{1}{\int_{-\tilde{R}}^{+\tilde{R}} p^2_{\ell+1}(M_{1,m} z)\,dz} \int_{-\tilde{R}}^{+\tilde{R}} p_{\ell+1}^2(M_{1,m}x_{\ell+1}) \frac{6805}{6804}\frac{10\pi \ell D R_w}{3M_{1,m}}\,dx_{\ell+1}\\
  &= \frac{1}{M_{1,m}^\ell G_{\ell+1}} \sum_{x_1,\dots, x_\ell=-\tilde{R}M_{1,m}}^{\tilde{R}M_{1,m} -1} p^2_T\left(x_1,\dots, x_{\ell}\right) \int_{-\tilde{R}}^{+\tilde{R}} p_{\ell+1}^2(M_{1,m} x_{\ell+1}) f\left(\frac{x_1}{M_{1,m}},\dots, \frac{x_\ell}{M_{1,m}}, x_{\ell+1}\right)\,dx_{\ell+1}\\
  &+ \frac{6805}{6804}\frac{10\pi \ell D R_w}{3M_{1,m}}.
\end{align}
In the inequality, we use the inductive hypothesis.
In the last equality, we rearrange and simplify.
Now, we can approximate this last integral by a Riemann sum for the function $\tilde{f}(y) \triangleq p_{\ell+1}^2(M_{1,m}y) f(x_1/M_{1,m},\dots, x_\ell/M_{1,m}, y)$, with error bounded similarly to \Cref{eq:riemann-approx-non-unif}:
\begin{align}
  &\frac{1}{G_{\ell + 1}} \int_{[-\tilde{R}, \tilde{R}]^{\ell + 1}} p^2_T(M_{1,m}x) f(x)\,dx\\
  &\leq \frac{1}{M_{1,m}^\ell G_{\ell+1}} \sum_{x_1,\dots, x_\ell=-\tilde{R}M_{1,m}}^{\tilde{R}M_{1,m} -1} p^2_T\left(x_1,\dots, x_\ell\right) \left(\frac{1}{M_{1,m}} \sum_{x_{\ell+1} = -\tilde{R}M_{1,m}}^{\tilde{R}M_{1,m} - 1} p_{\ell+1}^2(x_{\ell+1}) f\left(\frac{x}{M_{1,m}}\right) + \frac{L'\tilde{R}}{M_{1,m}}\right)\\
  &+ \frac{6805}{6804}\frac{10\pi \ell D R_w}{3M_{1,m}}.
\end{align}
Here, $L' \triangleq \max_{y \in [-\tilde{R},\tilde{R}]} |\tilde{f}'(y)|$.
Since $f(x) = \cos^2(2\pi j x^\intercal w^\star)$, then
\begin{align}
  \tilde{f}'(y) &= 2M_{1,m}p_{\ell+1}(M_{1,m}y)p'_{\ell+1}(M_{1,m}y) \cos^2\left(2\pi j\left(yw_{\ell+1}^\star + \sum_{i=1}^\ell \frac{x_i w_i^\star}{M_{1,m}}\right)\right)\\
  &- 2p_{\ell+1}^2(M_{1,m}y)\cos\left(2\pi j \left(yw_{\ell+1}^\star + \sum_{i=1}^\ell \frac{x_iw_i^\star}{M_{1,m}}\right)\right)\sin\left(2\pi j \left(yw_{\ell+1}^\star + \sum_{i=1}^\ell \frac{x_iw_i^\star}{M_{1,m}}\right)\right)\cdot 2\pi j w_{\ell+1}^\star.
\end{align}
Thus, for $y \in [-\tilde{R},\tilde{R}]$, then
\begin{equation}
  |\tilde{f}'(y)| \leq 2M_{1,m}|p_{\ell+1}'(M_{1,m}y)| + 4\pi D R_w \leq 5\pi DR_w.
\end{equation}
In the first inequality, we use Assumption~\ref{assum:bounded-1} so that $p(z)\leq 1$ and $j \leq D$.
In the second inequality, we use Assumption~\ref{assum:deriv} so that $|p_{\ell+1}'(M_{1,m}y)| \leq \pi DR_w/(2M_{1,m})$ since we used discretization parameter $M_{1,m}$.
Also, note that this applies because we chose our truncation parameter as $R = \tilde{R}M_{1,m}$ so that $M_{1,m}y \in [-\tilde{R}M_{1,m}, \tilde{R}M_{1,m}] = [-R,R]$.
Then, $L' \leq 5\pi DR_w$.
Plugging this back in,
\begin{align}
  \frac{1}{G_{\ell + 1}} \int_{[-\tilde{R}, \tilde{R}]^{\ell + 1}} p^2_T(M_{1,m} x) f(x)\,dx &\leq \frac{1}{M_{1,m}^{\ell+1} G_{\ell+1}} \sum_{x_1,\dots, x_{\ell+1}=-\tilde{R}M_{1,m}}^{\tilde{R}M_{1,m} -1} p^2_T(x) f\left(\frac{x}{M_{1,m}}\right)\\
  &+ \frac{5\pi DR_w}{M_{1,m}}\frac{\tilde{R}}{\int_{-\tilde{R}}^{+\tilde{R}} p_{\ell+1}^2(M_{1,m}z)\,dz} \frac{1}{M_{1,m}^\ell G_\ell} \sum_{x_1,\dots, x_\ell=-\tilde{R}M_{1,m}}^{\tilde{R}M_{1,m} -1} p^2_T(x)\\
  &+ \frac{6805}{6804}\frac{10\pi \ell DR_w}{3M_{1,m}}.
\end{align}
We previously showed that $\tilde{R}/G_1 \leq 3/5$ (see around \Cref{eq:g1-lower-non-unif}).
By the same argument here, then we can bound
\begin{align}
  \frac{1}{G_{\ell + 1}} \int_{[-\tilde{R}, \tilde{R}]^{\ell + 1}} p^2_T(M_{1,m}x) f(x)\,dx &\leq \frac{1}{M_{1,m}^{\ell+1} G_{\ell+1}} \sum_{x_1,\dots, x_{\ell+1}=-\tilde{R}M_{1,m}}^{\tilde{R}M_{1,m} -1} p^2_T(x) f\left(\frac{x}{M_{1,m}}\right)\\
  &+ \frac{3\pi DR_w}{M_{1,m}} \frac{1}{M_{1,m}^\ell G_\ell} \sum_{x_1,\dots, x_\ell=-\tilde{R}M_{1,m}}^{\tilde{R}M_{1,m} -1} p^2_T(x) + \frac{6805}{6804}\frac{10\pi \ell DR_w}{3M_{1,m}}.
\end{align}
Thus, it is clear to that to complete our argument, we need to show that
\begin{equation}
  \frac{3\pi DR_w}{M_{1,m}} \frac{1}{M_{1,m}^\ell G_\ell} \sum_{x_1,\dots, x_\ell=-\tilde{R}M_{1,m}}^{\tilde{R}M_{1,m} -1} p^2_T\left(x\right) \leq \frac{6805}{6804} \frac{10\pi DR_w}{3M_{1,m}}.
\end{equation}
To see this, first note that
\begin{equation}
  \label{eq:6879-bound-non-unif}
  M_{1,m}^\ell G_\ell \geq M_{1,m}^\ell \prod_{i=1}^\ell \left(\frac{9\tilde{R}}{5}\right) = \left(\frac{9}{5} M_{1,m} \tilde{R}\right)^\ell \geq \left(\frac{9}{5} \cdot 54 \cdot 70\right)^\ell \geq 6804.
\end{equation}
In the first inequality, we use the same argument as \Cref{eq:g1-lower-non-unif}, which relies on Assumption~\ref{assum:pointwise-close}.
In the second inequality, we use that $M_{1,m} \geq M_1 \geq 70\pi d^2 D^3 R_w \geq 70\pi R_w$ and $\tilde{R} \geq 54D^2\sqrt{d}/(\pi R_w) \geq 54/(\pi R_w)$.
In the last inequality, we use that $\ell \geq 1$ and simplify.

Suppose for now that $p^2$ has at most one critical point at $a \in (-\tilde{R}M_{1,m}, \tilde{R}M_{1,m})^\ell$.
Without loss of generality, since $p^2$ is even by Assumption~\ref{assum:even}, then we can assume that the critical point occurs at $a = 0$.
Also suppose without loss of generality that $p^2$ is nondecreasing for $x \leq 0$ and nonincreasing for $x \geq 0$.
The argument is the same for other cases.
By the above argument, we have
\begin{equation}
  \frac{1}{M_{1,m}^\ell G_\ell} = \frac{M_{1,m}^\ell G_\ell + 1}{M_{1,m}^\ell G_\ell} \cdot \frac{1}{M_{1,m}^\ell G_\ell + 1} \leq \frac{6805}{6804}\frac{1}{M_{1,m}^\ell G_\ell + 1}.
\end{equation}
Moreover, by standard results bounding sums in terms of integrals for monotone functions,
\begin{align}
  \tilde{G}_\ell &= \sum_{x_1,\dots, x_\ell =-\tilde{R}M_{1,m}}^{\tilde{R}M_{1,m}-1} p^2_1(x_1)\cdots p_\ell^2(x_\ell)\\
  &\leq \sum_{x_1,\dots, x_\ell =-\tilde{R}M_{1,m}}^{-1} p_1^2(x_1)\cdots p_\ell^2(x_\ell) + \sum_{x_1,\dots, x_\ell =1}^{\tilde{R}M_{1,m}} p_1^2(x_1)\cdots p_\ell^2(x_\ell) + 1\\
  &\leq \int_{[-\tilde{R}M_{1,m}, 0]^\ell} p_1^2(x_1)\cdots p_\ell^2(x_\ell)\,dx + \int_{[0,\tilde{R}M_{1,m}]^d} p_1^2(x_1)\cdots p_\ell^2(x_\ell)\,dx + 1\\
  &= \prod_{i=1}^\ell \left(\int_{-\tilde{R}M_{1,m}}^{+\tilde{R}M_{1,m}} p_i^2(x_i)\,dx_i\right) + 1\\
  &= M_{1,m}^\ell G_{\ell} + 1,
\end{align}
where in the second line, we use Assumption~\ref{assum:bounded-1} that $p_j \leq 1$.
In the last line, we use a change of variables.
Combining this with the above, we have
\begin{equation}
  \label{eq:G-tilde-upper-non-unif}
  \frac{1}{M_{1,m}^\ell G_\ell} \leq \frac{6805}{6804}\frac{1}{M_{1,m}^\ell G_\ell + 1} \leq \frac{6805}{6804}\frac{1}{\tilde{G}_\ell}.
\end{equation}
Earlier, we considered the case when $p^2$ has at most one critical point.
If we instead consider $p^2$ with a constant number of critical points, as in Assumption~\ref{assum:crit-points}, the above argument only changes the constant factor $6805/6804$.
We carry the factor of $6805/6804$ through the analysis, but changing this only affects some of the constants in the overall verification procedure and not the sample complexity.

Putting everything together,
\begin{align}
  \frac{3\pi DR_w}{M_{1,m}} \frac{1}{M_{1,m}^\ell G_\ell} \sum_{x_1,\dots, x_\ell=-\tilde{R}M_{1,m}}^{\tilde{R}M_{1,m} -1} p^2_T\left(x\right)&\leq \frac{6805}{6804}\frac{3\pi DR_w}{M_{1,m}} \frac{1}{\tilde{G}_\ell} \sum_{x_1,\dots, x_\ell=-\tilde{R}M_{1,m}}^{\tilde{R}M_{1,m} -1} p^2_T\left(x\right)\\
  &= \frac{6805}{6804}\frac{3\pi DR_w}{M_{1,m}} \frac{1}{\sum_{x_k=-\tilde{R}M_{1,m}}^{\tilde{R}M_{1,m}-1} p_k^2(x_k)}\sum_{x_k=-\tilde{R}M_{1,m}}^{\tilde{R}M_{1,m}-1} p_k(x_k)p_k(x_k + T)\\
  &\leq \frac{6805}{6804}\frac{3\pi DR_w}{M_{1,m}}\frac{1}{\sum_{x_k=-\tilde{R}M_{1,m}}^{\tilde{R}M_{1,m}-1} p_k^2(x_k)} (2\tilde{R}M_{1,m})\\
  &\leq \frac{6805}{6804}\frac{3\pi DR_w}{M_{1,m}}\frac{5}{9\tilde{R}M_{1,m}}\cdot 2\tilde{R}M_{1,m}\\
  &= \frac{6805}{6804}\frac{10\pi DR_w}{3M_{1,m}},
\end{align}
as required.
In the first line, we use \Cref{eq:G-tilde-upper-non-unif}.
In the third line, we use Assumption~\ref{assum:bounded-1} that $p_k \leq 1$.
In the fourth line, we use Assumption~\ref{assum:pointwise-close}.
In particular, by Assumption~\ref{assum:pointwise-close}, we have
\begin{equation}
  \sum_{x_k=-\tilde{R}M_{1,m}}^{\tilde{R}M_{1,m}-1}p_k^2(x_k) - 2\tilde{R}M_{1,m} = \sum_{x_k=-\tilde{R}M_{1,m}}^{\tilde{R}M_{1,m}-1}(p_k^2(x_k) - 1) \geq -\frac{\tilde{R}}{5}.
\end{equation}
Thus, this implies that $\sum_{x_k=-\tilde{R}M_{1,m}}^{\tilde{R}M_{1,m}-1} p_k^2(x_k) \geq 2\tilde{R} - \tilde{R}/5 = 9\tilde{R}/5$.
Note that we assumed throughout this analysis that $\ell + 1 > k$.
If $\ell + 1 = k$, the only part affected is when we bound $L'$, which would instead be a bound on the derivative of $\tilde{f}(y) = p_{\ell+1}(M_{1,m}y) p_{\ell+1}(M_{1,m}y + T) f(x_1/M_{1,m},\dots, x_\ell /M_{1,m}, y)$.
The derivative of $\tilde{f}(y)$ now has a term depending on $|T|$, which can be bounded again using Assumption~\ref{assum:deriv}, as we did in the base case, resulting in the same bound $L' \leq 5\pi D R_w$.
One can do the same argument for the lower bound, so this concludes the proof that \Cref{eq:induct-non-unif} holds for $\ell + 1$.

To complete the induction, one should also show that \Cref{eq:induct2-non-unif} holds for $\ell + 1$.
This follows by the same argument as above, and we refer to \Cref{claim:sum-to-int} for a sketch of how the argument is modified.
This completes the proof.
\end{proof}

The same result can be shown for the cross terms $\cos(2\pi j x^\intercal w^\star/M_{1,m})\cos(2\pi j' x^\intercal w^\star)$ and\\$\cos(2\pi j x^\intercal w^\star/M_{1,m}) \sin(2\pi j' x^\intercal w^\star/M_{1,m})$ by the same argument.
This is clear because these terms have the same bound on their gradients.

We can also bound the discretization error $\epsilon_d$.
Note that this discretization error is defined as
\begin{align}
  \label{eq:eps-d-non-unif}
  \epsilon_d \triangleq\frac{1}{\tilde{G}_d M_{2,m}^2} \sum_{x_1,\dots, x_d=-\tilde{R}M_{1,m}}^{\tilde{R}M_{1,m}-1} &\sum_{j,j'=1}^D p_T^2\left(x\right)\beta_j^\star \beta_{j'}^\star \left(\cos\left(\frac{2\pi j x^\intercal w^\star}{M_{1,m}}\right)\cos\left(\frac{2\pi j'(x + Te_k)^\intercal w^\star}{M_{1,m}}\right)\right.\\
  &- \left.\left\lfloor \cos\left(\frac{2\pi j x^\intercal w^\star}{M_{1,m}}\right) \right\rfloor_{M_{2,m}} \left\lfloor \cos\left(\frac{2\pi j'(x+Te_k)^\intercal w^\star}{M_{1,m}}\right) \right\rfloor_{M_{2,m}}\right).
\end{align}

\begin{claim}[Discretization error; Non-uniform case]
  \label{claim:eps-d-non-unif}
  For $m \in \{1,\dots, D\}$, consider parameters $M_{1,m}, M_{2,m}$ as defined above.
  Also, consider a parameter $\tilde{R}$ defined above.
  Then, we can bound the discretization error $\epsilon_d$ defined in \Cref{eq:eps-d-non-unif} as
  \begin{equation}
    |\epsilon_d| \leq \frac{3}{M_{2,m}^3}.
  \end{equation}
\end{claim}

\begin{proof}[Proof of \Cref{claim:eps-d-non-unif}]
This follows by a simple calculation and is similar to \Cref{claim:eps-d}.
Following the same steps as the proof of \Cref{claim:eps-d}, we can arrive at
\begin{align}
  |\epsilon_d| &\leq \frac{1}{\tilde{G}_d M_{2,m}^2} \sum_{x_1,\dots, x_d=-\tilde{R}M_{1,m}}^{\tilde{R}M_{1,m} - 1}\sum_{j,j'=1}^D p_T^2\left(x\right) |\beta_j^\star| |\beta_{j'}^\star|\frac{2}{M_{2,m}}\\
  &= \frac{2}{\tilde{G}_d M_{2,m}^3} \sum_{x_1,\dots, x_d=-\tilde{R}M_{1,m}}^{\tilde{R}M_{1,m} - 1} p_T^2\left(x\right),
\end{align}
where in the second line, we use $\norm{\beta^\star}_1 = 1$.
We can simplify this further using the definition of $p^2_T(x)$:
\begin{align}
  |\epsilon_d| &\leq \frac{2}{M_{2,m}^3} \frac{1}{\tilde{G}_d} \left(\sum_{x_k = -\tilde{R}M_{1,m}}^{\tilde{R}M_{1,m}-1} p_k(x_k) p_k(x_k + T)\right) \prod_{\substack{i=1\\i\neq k}}^d \left(\sum_{x_i = -\tilde{R}M_{1,m}}^{\tilde{R}M_{1,m}-1} p_i^2(x_i)\right)\\
  &= \frac{2}{M_{2,m}^3} \frac{1}{\sum_{x_k=-\tilde{R}M_{1,m}}^{\tilde{R}M_{1,m} - 1} p_k^2(x_k)} \sum_{x_k = -\tilde{R}M_{1,m}}^{\tilde{R}M_{1,m}-1} p_k(x_k) p_k(x_k+T)\\
  &\leq \frac{2}{M_{2,m}^3}\frac{1}{\sum_{x_k=-\tilde{R}M_{1,m}}^{\tilde{R}M_{1,m} - 1} p_k^2(x_k)} (2\tilde{R}M_{1,m})\\
  &\leq \frac{2}{M_{2,m}^3} \frac{5}{9\tilde{R}M_{1,m}}\cdot 2\tilde{R}M_{1,m}\\
  &\leq \frac{3}{M_{2,m}^3}.
\end{align}
In the third line, we use Assumption~\ref{assum:bounded-1} that $p_k \leq 1$.
In the fourth line, we use Assumption~\ref{assum:pointwise-close}.
In particular, by Assumption~\ref{assum:pointwise-close}, we have
\begin{equation}
  \sum_{x_k=-\tilde{R}M_{1,m}}^{\tilde{R}M_{1,m}-1}p_k^2(x_k) - 2\tilde{R}M_{1,m} = \sum_{x_k=-\tilde{R}M_{1,m}}^{\tilde{R}M_{1,m}-1}(p_k^2(x_k) - 1) \geq -\frac{\tilde{R}}{5}.
\end{equation}
Thus, this implies that $\sum_{x_k=-\tilde{R}M_{1,m}}^{\tilde{R}M_{1,m}-1} p_k^2(x_k) \geq 2\tilde{R} - \tilde{R}/5 = 9\tilde{R}/5$.
\end{proof}

With this, we can finally move on to show that the conditions checked in Steps 3 and 4 of \Cref{alg:verification-non-unif} are satisfied if and only if $|T - \ell M_1/w_k^\star| \leq 1$.
To do so, we use \Cref{claim:sum-to-int-non-unif} and \Cref{claim:eps-d-non-unif} in \Cref{eq:h-o-expval-non-unif,eq:h-o-expval2-non-unif} and leverage integral bounds from Appendix~\ref{sec:int-bounds-non-unif}.
The following two claims show this for each direction of the if and only if.

\begin{claim}[Correctness of Step 3 in \Cref{alg:verification-non-unif}]
  \label{claim:ver-if-non-unif}
  Consider parameters $M_1, M_2, R$ defined above and the observable $O_{k,1}$ defined in \Cref{eq:ov}.
  Let $\alpha_1$ denote the result of querying the QSQ oracle with observable $O_{k,1}$ with discretization parameters $M_1, M_2$, truncation parameter $R = \tilde{R}M_1$, and tolerance $\tau \leq \frac{1}{M_2^2}\left(\frac{5}{42D} - \frac{3}{2M_2}\right)$.
  If $|T - \ell M_1/w_k^\star| \leq 1$ for some integer $\ell$, then
  \begin{equation}
    \alpha_1 \geq \frac{1}{M_2^2}\left(\frac{5}{14D} - \frac{9}{2M_2}\right).
  \end{equation}
\end{claim}

\begin{claim}[Correctness of Step 4 in \Cref{alg:verification-non-unif}]
  \label{claim:ver-only-if-non-unif}
  For $m \in \{1,\dots, D\}$, consider parameters $M_{1,m}, M_{2,m}, R$ defined above and the observables $O_{k,m}$ defined in \Cref{eq:ov}.
  Let $\alpha_m$ denote the result of querying the QSQ oracle with observable $O_{k,m}$ with discretization parameters $M_{1,m}, M_{2,m}$, truncation parameter $R = \tilde{R}M_{1,m}$, and tolerance $\tau \leq \frac{1}{2D^2M_2^2}\left(\frac{2}{9} - \frac{1}{8}\left(\frac{2\pi R_w}{M_1}\right)^2 + \frac{3D^2}{M_2}\right)$.
  If $|T - \ell M_1 / w_k^\star|$ is not less than $1$ for any integer $\ell$, then
  \begin{equation}
    \sum_{m=1}^D \alpha_m \leq \frac{1}{M_2^2}\left(\frac{13}{25}D + \frac{1}{2D}\left(\frac{2}{9} - \frac{1}{8}\left(\frac{2\pi R_w}{M_1}\right)^2 +\frac{3D^2}{M_2}\right)\right).
  \end{equation}
\end{claim}

It suffices to prove these two claims to finish the proof.
Our starting point for both proofs is \Cref{eq:h-o-expval-non-unif,eq:h-o-expval2-non-unif}.

\begin{proof}[Proof of \Cref{claim:ver-if-non-unif}]
We want to lower bound $\expval{O_{k,1}}{h_{w^\star, M_1, M_2}}$.
As in \Cref{eq:G-tilde-upper-non-unif}, one can show that $\tilde{G}_d \leq M_1^d G_d + 1$.
Recall that this uses Assumption~\ref{assum:crit-points}.
Using this along with \Cref{eq:6879-bound-non-unif}, we have
\begin{equation}
  \frac{1}{\tilde{G}_d} \geq \frac{1}{M_1^d G_d + 1} = \frac{M_1^dG_d}{M_1^d G_d + 1} \frac{1}{M_1^d G_d} \geq \frac{6804}{6805}\frac{1}{M_1^d G_d}.
\end{equation}
Plugging this into \Cref{eq:h-o-expval-non-unif,eq:h-o-expval2-non-unif}, we have
\begin{align}
  &\expval{O_{k,1}}{h_{w^\star, M_{1}, M_{2}}}\\
  &\begin{aligned}
    \geq \frac{6804}{6805}\frac{1}{M_{2}^2} \frac{1}{M_1^d G_d} \sum_{x_1,\dots, x_d=-\tilde{R}M_{1}}^{\tilde{R}M_{1}-1} \sum_{j=1}^D p^2_T\left(x\right) &(\beta_j^\star)^2 \left(\cos^2\left(\frac{2\pi j x^\intercal w^\star}{M_{1}}\right)\cos\left(\frac{2\pi j T w_k^\star}{M_{1}}\right)\right.\\
    &\left.- \cos\left(\frac{2\pi j x^\intercal w^\star}{M_{1}}\right)\sin\left(\frac{2\pi j x^\intercal w^\star}{M_{1}}\right)\sin\left(\frac{2\pi j T w_k^\star}{M_{1}}\right) \right)
    \end{aligned}\\
  &\begin{aligned}
    +\frac{6804}{6805}\frac{1}{M_{2}^2} \frac{1}{M_1^d G_d} \sum_{x_1,\dots, x_d=-\tilde{R}M_{1}}^{\tilde{R}M_{1}-1} \sum_{\substack{j,j'=1\\j\neq j'}}^D p^2_T&\left(x\right) \beta_j^\star \beta_{j'}^\star\left(\cos\left(\frac{2\pi j x^\intercal w^\star}{M_{1}}\right)\cos\left(\frac{2\pi j' x^\intercal w^\star}{M_{1}}\right)\cos\left(\frac{2\pi j' Tw_k^\star}{M_{1}}\right)\right.\\
    &\left.- \cos\left(\frac{2\pi j x^\intercal w^\star}{M_{1}}\right)\sin\left(\frac{2\pi j' x^\intercal w^\star}{M_{1}}\right)\sin\left(\frac{2\pi j' T w_k^\star}{M_{1}}\right)\right) + \epsilon_d.
  \end{aligned}
\end{align}
Applying \Cref{claim:sum-to-int-non-unif}, then
\begin{align}
  &\expval{O_{k,1}}{h_{w^\star, M_{1}, M_{2}}}\\
  &\begin{aligned}
    \geq \frac{6804}{6805}\frac{1}{M_{2}^2} \frac{1}{G_d} \sum_{j=1}^D (\beta_j^\star)^2 \int_{[-\tilde{R},\tilde{R}]^d}p^2_T(M_{1,m}x) &\left(\cos^2\left(2\pi j x^\intercal w^\star\right)\cos\left(\frac{2\pi j T w_k^\star}{M_1}\right)\right. \\
    &\left.- \cos\left(2\pi j x^\intercal w^\star\right)\sin\left(2\pi j x^\intercal w^\star\right)\sin\left(\frac{2\pi j T w_k^\star}{M_{1}}\right) \right)\,dx
    \end{aligned}\\
  &\begin{aligned}
    +\frac{6804}{6805}\frac{1}{M_{2}^2} \frac{1}{G_d} \sum_{\substack{j,j'=1\\j\neq j'}}^D \beta_j^\star \beta_{j'}^\star \int_{[-\tilde{R},\tilde{R}]^d}&p^2_T\left(M_{1,m}x\right) \left(\cos\left(2\pi j x^\intercal w^\star\right)\cos\left(2\pi j' x^\intercal w^\star\right)\cos\left(\frac{2\pi j' Tw_k^\star}{M_{1}}\right)\right.\\
    &\left.- \cos\left(2\pi j x^\intercal w^\star\right)\sin\left(2\pi j' x^\intercal w^\star\right)\sin\left(\frac{2\pi j' T w_k^\star}{M_{1}}\right)\right) + \epsilon_d + \frac{6804}{6805}\frac{4}{M_2^2}\epsilon_{\mathrm{int}}.
  \end{aligned}
\end{align}
We can simplify this using the fact that an integral of an odd function, e.g., $\sin(x) \cos(x)$, over an even interval is zero.
This also uses Assumption~\ref{assum:even} that $p^2$ is an even function.
\begin{align}
  &\expval{O_{k,1}}{h_{w^\star, M_{1}, M_{2}}}\\
  &\begin{aligned}
    \geq \frac{6804}{6805}\frac{1}{M_{2}^2} \sum_{j=1}^D (\beta_j^\star)^2 \cos\left(\frac{2\pi j T w_k^\star}{M_1}\right) \frac{1}{G_d}\int_{[-\tilde{R},\tilde{R}]^d}p^2_T(M_{1,m}x) &\left(\cos^2\left(2\pi j x^\intercal w^\star\right)\right)\,dx + \epsilon_d + \frac{6804}{6805}\frac{4}{M_2^2}\epsilon_{\mathrm{int}}
    \end{aligned}\\
  &\begin{aligned}
    +\frac{6804}{6805}\frac{1}{M_{2}^2} \sum_{\substack{j,j'=1\\j\neq j'}}^D \beta_j^\star \beta_{j'}^\star &\left(\cos\left(\frac{2\pi j' T w_k^\star}{M_1}\right) \frac{1}{G_d} \int_{[-\tilde{R},\tilde{R}]^d} p_T^2(M_{1,m}x)\cos(2\pi j x^\intercal w^\star)\cos(2\pi j' x^\intercal w^\star)\,dx\right.\\
    &\left. - \sin\left(\frac{2\pi j' T w_k^\star}{M_1}\right)\frac{1}{G_d} \int_{[-\tilde{R},\tilde{R}]^d} p_T^2(M_{1,m}x)\cos(2\pi j x^\intercal w^\star)\sin(2\pi j' x^\intercal w^\star)\,dx\right).
  \end{aligned}
\end{align}
Using \Cref{coro:integral-non-unif,coro:integral2-wstar-non-unif,coro:integral2-wstar-sin-non-unif},
\begin{align}
  \expval{O_{k,1}}{h_{w^\star, M_{1}, M_{2}}} &\geq \frac{6804}{6805}\frac{1}{M_2^2} \sum_{j=1}^D (\beta_j^\star)^2 \left(\frac{1}{2} - \frac{3\sqrt{d}}{16\pi R_w \tilde{R}}\right)\cos\left(\frac{2\pi j T w_k^\star}{M_1}\right)\\
  &- \frac{6804}{6805}\frac{1}{M_2^2}\sum_{\substack{j,j'=1\\j\neq j'}}^D \beta_j^\star \beta_{j'}^\star \left(\frac{3\sqrt{d}}{2\pi R_w \tilde{R}}\right) + \epsilon_d +\frac{6804}{6805}\frac{4}{M_2^2}\epsilon_{\mathrm{int}}.
\end{align}
Note that \Cref{coro:integral-non-unif,coro:integral2-wstar-non-unif,coro:integral2-wstar-sin-non-unif} apply when integrating with respect to the non-uniform density, which we don't quite have here.
However, using \Cref{coro:complex-exp-wstar-verif-non-unif} instead of \Cref{coro:complex-exp-wstar-non-unif} in their proofs, we see that the results still hold for integrating with respect to $p^2_T$.
Using our choice of $\tilde{R} \geq \max(39\sqrt{d}/(4\pi R_w),54D^2\sqrt{d}/(\pi R_w))$, we have
\begin{equation}
  \expval{O_{k,1}}{h_{w^\star, M_{1}, M_{2}}} \geq \frac{1}{M_2^2}\left(\frac{6804}{6805}\frac{25}{52} \sum_{j=1}^D (\beta_j^\star)^2 \cos\left(\frac{2\pi j T w_k^\star}{M_1}\right) - \frac{6804}{6805} \frac{1}{36D^2} - \frac{3}{M_2} - \frac{4}{21D^2}\right)
\end{equation}
We also use that $\norm{\beta^\star}_2^2 \leq 1$ since $\norm{\beta^\star}_1 = 1$, $|\epsilon_d| \leq 3/M_2^3$ by \Cref{claim:eps-d-non-unif}, and $|\epsilon_{\mathrm{int}}| \leq 6805/(21 \cdot 6804 D^2)$ by \Cref{claim:sum-to-int-non-unif}.
We can lower bound the summation term by \Cref{eq:cos-T-bound} to obtain
\begin{align}
  \expval{O_{k,1}}{h_{w^\star, M_{1}, M_{2}}} &\geq \frac{1}{M_2^2}\left(\frac{6804}{6805} \frac{25}{52} \frac{2449}{2550D} - \frac{6879}{6880} \frac{1}{36D^2} - \frac{3}{M_2} - \frac{4}{21D^2}\right)\\
  &\geq \frac{1}{M_2^2}\left(\frac{5}{21D} - \frac{3}{M_2}\right).
\end{align}
Thus, we see that if $|T - \ell M_1/w_k^\star| \leq 1$, then this lower bound on the expectation value must be satisfied.
Finally, our choice of $\tau$ and the condition on $\alpha_1$ guarantees that this the lower bound on the expectation value also holds, as required.
\end{proof}

\begin{proof}[Proof of \Cref{claim:ver-only-if-non-unif}]
This time, we want to upper bound $\expval{O_{k,m}}{h_{w^\star, M_{1,m}, M_{2,m}}}$ for any $m \in \{1,\dots, D\}$.
Similarly to \Cref{eq:G-tilde-upper-non-unif}, we can show that $\tilde{G}_d \geq M_{1,m}^d G_d$.

Suppose for now that $p^2$ has at most one critical point at $a \in (-\tilde{R}M_{1,m}, \tilde{R}M_{1,m})^d$.
Without loss of generality, since $p^2$ is even by Assumption~\ref{assum:even}, then we can assume that the critical point occurs at $a = 0$.
Also, suppose without loss of generality that $p^2$ is nondecreasing for $x \leq 0$ and nonincreasing for $x \geq 0$.
The argument is the same for other cases.
By \Cref{eq:6879-bound-non-unif}, we have
\begin{equation}
  \frac{1}{M_{1,m}^d G_d} = \frac{M_{1,m}^d G_d -1}{M_{1,m}^d G_d} \cdot \frac{1}{M_{1,m}^d G_d - 1} \geq \frac{6803}{6804} \frac{1}{M_{1,m}^d G_d - 1}.
\end{equation}
Moreover, by standard results bounding sums in terms of integrals for monotone functions, we have
\begin{align}
  \tilde{G}_d &= \sum_{x_1,\dots, x_d = -\tilde{R}M_{1,m}}^{\tilde{R}M_{1,m} - 1} p_1^2(x_1)\cdots p_d^2(x_d)\\
  &= \sum_{x_1,\dots, x_d=-\tilde{R}M_{1,m}}^{0} p_1^2(x_1)\cdots p_d^2(x_d) + \sum_{x_1,\dots, x_d = 0}^{\tilde{R}M_{1,m}-1} p_1^2(x_1)\cdots p_d^2(x_d) - p^2(0)\\
  &\geq \int_{[-\tilde{R}M_{1,m} - 1,0]^d} p^2(x)\,dx + \int_{[0, \tilde{R}M_{1,m}]^d} p^2(x)\,dx - 1\\
  &\geq \int_{[-\tilde{R}M_{1,m}, \tilde{R}M_{1,m}]^d} p^2(x)\,dx - 1\\
  &= M_{1,m}^d G_d - 1.
\end{align}
In the third line, we use Assumption~\ref{assum:bounded-1} tahat $p_j^2 \leq 1$.
In the last line, we use a change of variables.
Combining this with the above, we have
\begin{equation}
  \label{eq:G-tilde-lower-non-unif}
  \frac{1}{M_{1,m}^d G_d} \geq \frac{6803}{6804}\frac{1}{M_{1,m}^d G_d - 1} \geq \frac{6803}{6804}\frac{1}{\tilde{G}_d}.
\end{equation}
Earlier, we considered the case when $p^2$ has at most one critical point.
If we instead consider $p^2$ with a constant number of critical points, as in Assumption~\ref{assum:crit-points}, the above argument only changes the constant factor $6803/6804$.
We carry the factor of $6803/6804$ through the analysis, but changing this only affects some of the constants in the overall verification procedure and not the sample complexity.

Using this along with \Cref{claim:sum-to-int-non-unif}, plugging into \Cref{eq:h-o-expval-non-unif,eq:h-o-expval2-non-unif}, we have
\begin{align}
  &\expval{O_{k,m}}{h_{w^\star, M_{1,m}, M_{2,m}}}\\
  &\begin{aligned}
    \leq \frac{6804}{6803}\frac{1}{M_{2,m}^2} \frac{1}{G_d} \sum_{j=1}^D (\beta_j^\star)^2 \int_{[-\tilde{R},\tilde{R}]^d}p^2_T(M_{1,m}x) &\left(\cos^2\left(2\pi j x^\intercal w^\star\right)\cos\left(\frac{2\pi j T w_k^\star}{M_{1,m}}\right)\right. \\
    &\left.- \cos\left(2\pi j x^\intercal w^\star\right)\sin\left(2\pi j x^\intercal w^\star\right)\sin\left(\frac{2\pi j T w_k^\star}{M_{1,m}}\right) \right)\,dx
    \end{aligned}\\
  &\begin{aligned}
    +\frac{6804}{6803}\frac{1}{M_{2,m}^2} \frac{1}{G_d} \sum_{\substack{j,j'=1\\j\neq j'}}^D \beta_j^\star \beta_{j'}^\star&\int_{[-\tilde{R},\tilde{R}]^d}p^2_T\left(M_{1,m}x\right) \left(\cos\left(2\pi j x^\intercal w^\star\right)\cos\left(2\pi j' x^\intercal w^\star\right)\cos\left(\frac{2\pi j' Tw_k^\star}{M_{1,m}}\right)\right.\\
    &\left.- \cos\left(2\pi j x^\intercal w^\star\right)\sin\left(2\pi j' x^\intercal w^\star\right)\sin\left(\frac{2\pi j' T w_k^\star}{M_{1,m}}\right)\right) + \epsilon_d + \frac{6804}{6803}\frac{4}{M_{2,m}^2}\epsilon_{\mathrm{int}}.
  \end{aligned}
\end{align}
Now, we use that an integral of an odd function, e.g., $\sin(x)\cos(x)$, over an even interval is zero (also using Assumption~\ref{assum:even} that $p^2$ is even).
We also use \Cref{coro:integral-upper-non-unif,coro:integral2-wstar-non-unif,coro:integral2-wstar-sin-non-unif} so that we have
\begin{align}
  \expval{O_{k,m}}{h_{w^\star, M_{1,m}, M_{2,m}}} &\leq \frac{6804}{6803}\frac{1}{M_{2,m}^2}\sum_{j=1}^D (\beta_j^\star)^2 \left(\frac{1}{2} + \frac{3\sqrt{d}}{16\pi R_w \tilde{R}}\right)\cos\left(\frac{2\pi j Tw_k^\star}{M_{1,m}}\right)\\
  &+ \frac{6804}{6803}\frac{1}{M_{2,m}^2}\sum_{\substack{j,j'=1\\j\neq j'}}^D \beta_j^\star \beta_{j'}^\star \left(\frac{3\sqrt{d}}{2\pi R_w \tilde{R}}\right) + \epsilon_d + \frac{6804}{6803}\frac{4}{M_{2,m}^2}\epsilon_{\mathrm{int}}.
\end{align}
Note that \Cref{coro:integral-upper-non-unif,coro:integral2-wstar-non-unif,coro:integral2-wstar-sin-non-unif} apply when integrating with respect to the Gaussian density.
Using \Cref{coro:complex-exp-wstar-verif-non-unif} instead of \Cref{coro:complex-exp-wstar-non-unif} in their proofs, we see the results still hold when integrating with respect to $p^2_T$.
Using our choice of $\tilde{R} \geq \max(39\sqrt{d}/(4\pi R_w), 54D^2\sqrt{d}/(\pi R_w))$, we have
\begin{align}
  &\expval{O_{k,m}}{h_{w^\star, M_{1,m}, M_{2,m}}}\\
  &\leq \frac{1}{M_{2,m}^2}\left(\frac{6804}{6803}\frac{27}{52}\sum_{j=1}^D (\beta_j^\star)^2 \cos\left(\frac{2\pi j T w_k^\star}{M_{1,m}}\right) + \frac{6804}{6803}\frac{1}{36D^2} + \frac{3}{M_{2,m}} + \frac{6805}{6803}\frac{4}{21D^2}\right)\\
  &\leq \frac{1}{M_{2,m}^2}\left(\frac{6804}{6803}\frac{27}{52}(\beta_m^\star)^2 \cos\left(\frac{2\pi T w_k^\star}{M_1}\right) + \frac{6804}{6803}\frac{27}{52}\sum_{\substack{j=1\\j\neq m}}^D (\beta_j^\star)^2 + \frac{6804}{6803}\frac{1}{36D^2} +\frac{3}{M_{2,m}} + \frac{6805}{6803}\frac{4}{21D^2}\right)
\end{align}
In the first line, we also use $\norm{\beta^\star}_2^2 \leq 1$ since $\norm{\beta^\star}_1 = 1$.
In addition, we use $|\epsilon_d| \leq 3/M_{2,m}^3$ by \Cref{claim:eps-d-non-unif}, and $|\epsilon_{\mathrm{int}}| \leq 6805/(21 \cdot 6804 D^2)$ by \Cref{claim:sum-to-int-non-unif}.
In the second line, we use $M_{1,m} = mM_1$.
We can further bound the cosine term using \Cref{eq:cos-upper}:
\begin{align}
  &\expval{O_{k,m}}{h_{w^\star, M_{1,m}, M_{2,m}}}\\
  &\leq \frac{1}{M_{2,m}^2}\left(\frac{6804}{6803}\frac{27}{52}(\beta_m^\star)^2 \left(1 - \frac{1}{8}\left(\frac{2\pi R_w}{M_1}\right)^2\right) + \frac{6804}{6803}\frac{27}{52}\sum_{\substack{j=1\\j\neq m}}^D (\beta_j^\star)^2 + \frac{6804}{6803}\frac{1}{36D^2} + \frac{3}{M_{2,m}} + \frac{6805}{6803}\frac{4}{21D^2}\right)\\
  &\leq \frac{1}{M_{2,m}^2}\left(\frac{13}{25}\left(1 - \frac{1}{8}\left(\frac{2\pi R_w}{M_1}\right)^2 (\beta_m^\star)^2\right) + \frac{2}{9D^2} + \frac{3}{M_{2,m}}\right).
\end{align}
In the last line, we use that $\norm{\beta^\star}_2^2 \leq 1$ since $\norm{\beta^\star}_1 = 1$.
Summing over all $m \in \{1,\dots, D\}$, then we have
\begin{align}
  \sum_{m=1}^D\expval{O_{k,m}}{h_{w^\star, M_{1,m}, M_{2,m}}} &\leq \sum_{m=1}^D \frac{1}{M_{2,m}^2}\left(\frac{13}{25}\left(1 - \frac{1}{8}\left(\frac{2\pi R_w}{M_1}\right)^2 (\beta_m^\star)^2\right) + \frac{2}{9D^2} + \frac{3}{M_{2,m}}\right)\\
  &\leq \frac{1}{M_2^2}\sum_{m=1}^D \left(\frac{13}{25}\left(1 - \frac{1}{8}\left(\frac{2\pi R_w}{M_1}\right)^2 (\beta_m^\star)^2\right) + \frac{2}{9D^2} + \frac{3}{M_2}\right)\\
  &\leq \frac{1}{M_2^2}\left(\frac{13}{25}D - \frac{1}{8D}\left(\frac{2\pi R_w}{M_1}\right)^2 + \frac{2}{9D} + \frac{3D}{M_2}\right).
\end{align}
In the second line, we use $M_{2,m} = m M_2$ by definition and $m \geq 1$.
In the last line, we use $\norm{\beta^\star}_2^2 \geq 1/D$.
Thus, we see that if $|T - \ell M_1/w_k^\star| \not\leq 1$ for any integer $\ell$, then this upper bound on the sum of expectation values must be satisfied.
Finally, our choice of $\tau$ and the condition on $\sum_{m=1}^D \alpha_m$ guarantees that this upper bound on also holds, as required.
\end{proof}
\end{proof}

Finally, using \Cref{thm:verification-non-unif} and \Cref{thm:non-unif-period}, we can prove \Cref{coro:linear-non-unif}.

\begin{proof}[Proof of \Cref{coro:linear-non-unif}]
Choose the discretization parameter to be $M_1 = \max(70 \pi dD^3 R_w, R_w^2/\epsilon_1)$.
By \Cref{lem:discrete-general-unif}, we know that there exists a discretization $h_{w^\star, M_1, M_2}$ of the target function $g_{w^\star}$ such that $h_{w^\star, M_1, M_2}$ is $(33/35)$-pseudoperiodic with period $S_j = M_1/w^\star_j$ in each component.
Note that $S_j \geq 1$ by our choice of discretization parameter.
Moreover, we know an upper bound on the period $A = M_1d^2/R_w$ by \Cref{eq:sw}.
Finally, we have an efficient verification procedure by \Cref{thm:verification-non-unif}.
Thus, we satisfy all of the conditions of \Cref{thm:non-unif-period}, so applying its result, we can find integers $a_j$ such that $|a_j - S_j| \leq 1$ with probability $\Omega(1/\log^4(M_1d^2/R_w))$.
The rest of the proof then follows in the same way as \Cref{thm:linear-uniform} by our choice of $M_1 \geq R_w^2/\epsilon_1$.
\end{proof}

\subsection{Learning the outer function via gradient methods}
\label{sec:outer-non-unif}

As in Appendix~\ref{sec:outer-uniform}, now that we have an approximation $\hat{w}$ of $w^\star$ such that $\norm{\hat{w} - w^\star}_\infty \leq \epsilon_1$, we want to learn the outer periodic function $\tilde{g}: \mathbb{R} \to [-1,1]$ via classical gradient methods.
Again, this portion of the algorithm is purely classical.
The difference with Appendix~\ref{sec:outer-uniform} is that the density $\varphi^2$ is now not a uniform density.
In particular, we consider a probability distribution $\varphi^2 \propto \prod_{k=1}^d p_k^2$ over $[-R,R]^d$, where $R$ is the truncation parameter.
We also consider that $\varphi^2$ satisfies Assumptions~\ref{assum:fourier-conc}-\ref{assum:crit-points}.
In particular for this part of the algorithm, we only need $\varphi^2$ to satisfy Assumptions~\ref{assum:pointwise-close}-\ref{assum:even}.

Explicitly, we consider a density function
\begin{equation}
  \varphi^2(x) = \frac{1}{\prod_{j=1}^d \left(\int_{-R}^{+R} p_j^2(z)\,dz\right)} \prod_{j=1}^d p_j^2(x_j).
\end{equation}
Recall that our target function is
\begin{equation}
  g_{w^\star}(x) = \tilde{g}(x^\intercal w^\star) = \sum_{j=1}^D \beta_j^\star \cos(2\pi j x^\intercal w^\star),
\end{equation}
and we want to find a good predictor 
\begin{equation}
  f_\beta(x) = \sum_{j=1}^D \beta_j \cos(2\pi j x^\intercal \hat{w}),
\end{equation}
that minimizes the objective function
\begin{align}
  \mathcal{L}_{w^\star}(\beta) &= \mathop{\mathbb{E}}_{x \sim \varphi^2}[(f_\beta(x) - g_{w^\star}(x))^2] = \int\limits_{x\sim \varphi^2}\left(\sum_{j=1}^D \beta_j^\star \cos(2\pi j x^\intercal w^\star) - \sum_{j=1}^D \beta_j \cos(2\pi j x^\intercal \hat{w})\right)^2\,dx,
\end{align}
where $\hat{w}$ is our approximation of $w^\star$ from \Cref{coro:linear-non-unif}.
As in the classical hardness result~\cite{shamir2018distribution}, our algorithm is given access to this loss function and its gradients.
Using this, we design a classical algorithm that can effiicently find a predictor specified by parameters $\hat{\beta}$ such that $\mathcal{L}_{w^\star}(\hat{\beta}) \leq \epsilon$ for a given precision $\epsilon > 0$.

In Appendix~\ref{sec:outer-uniform}, we proved that for an appropriate choice of truncation parameter $R$ and accuracy $\epsilon_1$ such that $\norm{\hat{w} - w^\star}_\infty \leq \epsilon_1$, then we can achieve this small loss (\Cref{thm:g-tilde-uniform}).
In fact, we can achieve the same guarantee for non-uniform distributions.

\begin{theorem}[Learning $\tilde{g}$ Guarantee; Non-Uniform Case]
\label{thm:g-tilde-non-unif}
Let $\varphi^2 \propto \prod_{k=1}^d p_k^2$ be a probability distribution over $[-R, R]^d$ satisfying Assumptions~\ref{assum:pointwise-close}-\ref{assum:even} for $R$ defined shortly
Let $\epsilon > 0$.
Let $w^\star \in \mathbb{R}^d$ be unknown with norm $R_w > 0$, and let $g_{w^\star}:\mathbb{R}^d \to [-1,1]$ be defined as $g_{w^\star}(x) = \tilde{g}(x^\intercal w^\star)$ for $\tilde{g}$ given in \Cref{eq:g-tilde}.
Choose
\begin{equation}
  \label{eq:final-R-bounds-non-unif}
  R = \tilde{\Omega}\left(\max\left(\frac{D^2}{\epsilon}, \frac{D^2\sqrt{d}}{R_w \epsilon}, \frac{D^{5/2}}{\sqrt{\epsilon}}, \frac{D^{3/2}\sqrt{d}}{R_w \sqrt{\epsilon}}\right)\right),
\end{equation}
\begin{equation}
  \label{eq:final-eps1-bounds-non-unif}
  \epsilon_1 = \tilde{\mathcal{O}}\left(\min\left(\frac{\epsilon^3}{D^6 d}, \frac{\epsilon^{3/2}}{D^{13/2}d}, \frac{R_w}{D\sqrt{d}}\right)\right).
\end{equation}
Suppose we have an approximation $\hat{w} \in \mathbb{R}^d$ such that $\norm{\hat{w} - w^\star}_\infty \leq \epsilon_1$.
Then, there exists a classical algorithm with access to the loss function from \Cref{eq:loss} and its derivatives that can efficiently find a parameters $\hat{\beta} \in \mathbb{R}^d$ such that $\mathcal{L}_{w^\star}(\hat{\beta}) \leq \epsilon$.
Moreover, this algorithm requires at most
\begin{equation}
t = \Theta\left(\log\left(\sqrt{\frac{D}{\epsilon}}\right)\right)
\end{equation}
iterations of gradient descent.
\end{theorem}

The proof of this theorem is simple given what we have already proven in Appendix~\ref{sec:outer-uniform}.
There, notice that the proof only depends on the distribution $\varphi^2$ through \Cref{lem:grad-inform,lem:loss-bound}.
In fact, notice that in these lemmas, their proofs only depend on $\varphi^2$ via the integral bounds in Appendix~\ref{sec:int-bounds}.
Thus, to prove \Cref{thm:g-tilde-non-unif}, we only need to obtain similar integral bounds when $\varphi^2$ is a non-uniform distribution satisfying Assumptions~\ref{assum:pointwise-close}-\ref{assum:even}.
We achieve this in Appendix~\ref{sec:int-bounds-non-unif}.
These integral bounds differ from the uniform case only in constant factors, thus immediately giving the result.

\subsection{Integral bounds}
\label{sec:int-bounds-non-unif}

Similarly to Appendix~\ref{sec:int-bounds}, we need the following technical lemmas for bounding integrals when the integral is taken with respect to a non-uniform distribution instead.
We require that the distribution satisfies Assumptions~\ref{assum:pointwise-close}-\ref{assum:even} in order for all of the bounds to hold.
Some bounds only require Assumptions~\ref{assum:pointwise-close} and \ref{assum:bounded-1}.

\begin{lemma}
\label{lem:complex-exp-non-unif}
Let $\varphi^2 \propto \prod_{k=1}^d p_k^2$ be a probability distribution over $[-R, R]^d$ satisfying Assumptions~\ref{assum:pointwise-close} and \ref{assum:bounded-1} for a truncation parameter $R$.
Let $w^\star \in \mathbb{R}^d$ be unknown with norm $R_w > 0$, and let $\hat{w} \in \mathbb{R}^d$ be an approximation of $w^\star$ with $\norm{\hat{w} - w^\star}_\infty \leq \epsilon_1$.
Let $1 \leq j,j' \leq D$ be integers with $j \neq j'$, for $D \in \mathbb{N}$ from \Cref{eq:g-tilde}.
Then,
\begin{equation}
  \left|\int\limits_{x\sim \varphi^2} e^{2\pi i x^\intercal \hat{w}(j-j')}\,dx\right|\leq \frac{3}{4\pi R}\frac{\sqrt{d}}{R_w - \sqrt{d}\epsilon_1}.
\end{equation}
\end{lemma}

\begin{proof}
The proof follows similarly to that of \Cref{lem:complex-exp}.
Denote the normalization constant by
\begin{equation}
  G \triangleq \prod_{k=1}^d \left(\int_{-R}^{+R} p_k^2(z)\,dz\right).
\end{equation}
We can bound this integral using
\begin{align}
  &\left|\int_{x \sim \varphi^2} e^{2\pi i x^\intercal  \hat{w}(j-j')}\,dx\right|\\
  &= \left|\frac{1}{G} \int_{x_1=-R}^{+R}\cdots \int_{x_d=-R}^{+R} e^{2\pi i \sum_{k=1}^d x_k \hat{w}_k (j-j')} p_1^2(x_1) \cdots p_d^2(x_d)\,dx_d\cdots dx_1 \right|\\
  &= \left|\frac{1}{G} \prod_{k=1}^d \int_{x_k=-R}^{+R} e^{2\pi i x_k \hat{w}_k(j-j')} p^2_k(x_k) \,dx_k \right|\label{eq:prod-int-non-unif}.
\end{align}
Here, notice that we can bound each of these integrals trivially
\begin{equation}
  \label{eq:triv-int-bound-non-unif}
  \left|\int_{x_k=-R}^{+R} e^{2\pi i x_k\hat{w}_k(j-j')} p^2_k(x_k)\,dx_k \right| \leq \int_{x_k=-R}^{+R} p_k^2(x_k) \,dx_k,
\end{equation}
where we use that $|e^{2\pi i z}| \leq 1$.
We also notice that because $\norm{w^\star}_2^2 = \sum_{i=1}^d |w_i^\star|^2 = R_w^2$, then there must exist some $k \in [d]$ such that $|w_k^\star| \geq R_w/\sqrt{d}$.
Here, equality is satisfied for the case when $w_i = R_w/\sqrt{d}$ for all $i \in [d]$.
We will bound each integral in the product in \Cref{eq:prod-int-non-unif} using \Cref{eq:triv-int-bound-non-unif} except for this $k$ such that $|w_k^\star| \geq R_w/\sqrt{d}$:
\begin{align}
  \left|\int_{x \sim \varphi^2} e^{2\pi i x^\intercal  \hat{w}(j-j')}\,dx\right| &= \left|\frac{1}{G} \prod_{k=1}^d \int_{x_k=-R}^{+R} e^{2\pi i x_k \hat{w}_k(j-j')} p^2_k(x_k) \,dx_k \right|\\
  &\leq \frac{1}{\left(\int_{-R}^{+R} p_k^2(z)\,dz\right)} \left|\int_{x_k=-R}^{+R} e^{2\pi i x_k \hat{w}_k(j-j')} p_k^2(x_k)\,dx_k\right|\\
  &\leq \frac{1}{\left(\int_{-R}^{+R} p_k^2(z)\,dz\right)} \left|\int_{x_k=-R}^{+R} e^{2\pi i x_k \hat{w}_k(j-j')} \,dx_k\right|\\
  &\leq \frac{3}{4R}\left|\int_{x_k=-R}^{+R} e^{2\pi i x_k \hat{w}_k(j-j')} \,dx_k\right|,
\end{align}
where in the second line, we use \Cref{eq:triv-int-bound-non-unif}.
In the third line, we use that $p^2_k(x) \leq 1$ by Assumption~\ref{assum:bounded-1}.
In the fourth line, we use Assumption~\ref{assum:pointwise-close}.
In particular, by Assumption~\ref{assum:pointwise-close}, we have
\begin{equation}
  \int_{-R}^{+R} p_k^2(z)\,dz - 2R = \int_{-R}^{+R} (p_k^2(z) - 1)\,dz \geq -\frac{R}{5}.
\end{equation}
Thus, this implies that $\int_{-R}^{+R}p_k^2(z)\,dz \geq 2R - R/5 \geq 4R/3$.
From here, the proof is the same as that of \Cref{lem:complex-exp}, just carrying through a constant factor of $3/4$ instead of $1/2$.
\end{proof}

By essentially the same proof, we can obtain a similar upper bound replacing $\hat{w}$ with $w^\star$.

\begin{corollary}
\label{coro:complex-exp-wstar-non-unif}
Let $\varphi^2 \propto \prod_{k=1}^d p_k^2$ be a probability distribution over $[-R, R]^d$ satisfying Assumptions~\ref{assum:pointwise-close} and \ref{assum:bounded-1} for a truncation parameter $R$.
Let $w^\star \in \mathbb{R}^d$ be unknown with norm $R_w > 0$, and let $\hat{w} \in \mathbb{R}^d$ be an approximation of $w^\star$ with $\norm{\hat{w} - w^\star}_\infty \leq \epsilon_1$.
Let $1 \leq j,j' \leq D$ be integers with $j \neq j'$, for $D \in \mathbb{N}$ from \Cref{eq:g-tilde}.
Then,
\begin{equation}
\left|\int_{x\sim \varphi^2} e^{2\pi i x^\intercal w^\star(j-j')}\,dx\right|\leq \frac{3}{4\pi R}\frac{\sqrt{d}}{R_w}.
\end{equation}
\end{corollary}

We also have a similar corollary, where the integral is taken over a slightly different distribution.
This is useful in Appendix~\ref{sec:general-non-unif}.

\begin{corollary}
\label{coro:complex-exp-wstar-verif-non-unif}
Consider the space $[-R, R]^d \subseteq \mathbb{R}^d$ and nonnegative functions $p_i:\mathbb{R} \to [0,1]$ satisfying Assumptions~\ref{assum:pointwise-close} and \ref{assum:bounded-1} for $i \in [d]$.
Let $w^\star \in \mathbb{R}^d$ be unknown with norm $R_w > 0$, and let $\hat{w} \in \mathbb{R}^d$ be an approximation of $w^\star$ with $\norm{\hat{w} - w^\star}_\infty \leq \epsilon_1$.
Let $1 \leq j,j'\leq D$ be integers with $j \neq j'$, for $D \in \mathbb{N}$ from \Cref{eq:g-tilde}.
Let $M_1, T$ and $k \in [d]$ be integers.
Then,
\begin{equation}
  \left|\frac{1}{G} \int_{[-R,R]^d} p_1^2(x_1) \cdots p_k(x_k)p_k(M_1 x_k + T) \cdots p_d^2(x_d) e^{2\pi i x^\intercal w^\star(j-j')}\,dx\right| \leq \frac{3}{4\pi R}\frac{\sqrt{d}}{R_w},
\end{equation}
where
\begin{equation}
  G \triangleq \prod_{i=1}^d \left(\int_{-R}^{+R} p^2(x_j)\,dx_j\right).
\end{equation}
\end{corollary}

\begin{proof}
We can rewrite the integral as
\begin{align}
  &\left|\frac{1}{G} \int_{[-R,R]^d} p_1^2(x_1)\cdots p_k(x_k)p_k(M_1x_k + T)\cdots p_d^2(x_d) e^{2\pi i x^\intercal w^\star(j-j')}\,dx\right|\\
  &= \left|\frac{1}{G}\left(\prod_{\substack{\ell=1\\\ell \neq k}}^d \int_{-R}^{+R} p_\ell^2(x_\ell) e^{2\pi i x_\ell w_\ell^\star(j-j')}\,dx_\ell\right)\left(\int_{-R}^{+R} p_k(x_k) p_k(M_1x_k + T) e^{2\pi i x_kw_k^\star (j-j')}\,dx_k\right)\right|.\label{eq:prod-int-non-unif2}
\end{align}
Notice that we can bound each of these integrals trivially as in \Cref{eq:triv-int-bound-non-unif}.
Also, notice that because $\norm{w^\star}_2^2 = R_w^2$, then there must exist some $k' \in [d]$ such that $|w_k^\star| \geq R_w/\sqrt{d}$.
We will bound each integral in the product in \Cref{eq:prod-int-non-unif2} using \Cref{eq:triv-int-bound-non-unif} except for this $k'$ such that $|w_{k'}^\star| \geq R_w/\sqrt{d}$.
If $k = k'$, then
\begin{align}
  &\left|\frac{1}{G} \int_{[-R,R]^d} p_1^2(x_1)\cdots p_k(x_k)p_k(M_1 x_k + T)\cdots p_d^2(x_d) e^{2\pi i x^\intercal w^\star(j-j')}\,dx\right|\\
  &\leq \frac{1}{\int_{-R}^{+R} p_k^2(z)\,dz}\left|\int_{x_k=-R}^{+R} p_k(x_k)p_k(M_1x_k + T) e^{2\pi x_kw_k^\star(j-j')}\,dx_k \right|\\
  &\leq \frac{1}{\int_{-R}^{+R} p_k^2(z)\,dz}\left|\int_{x_k=-R}^{+R} e^{2\pi i x_k w_k^\star(j-j')}\,dx_k \right|.
\end{align}
Here, in the last line, we used Assumption~\ref{assum:bounded-1} that $p_k(z) \leq 1$.
From here, the proof is the same as \Cref{lem:complex-exp-non-unif} and \Cref{coro:complex-exp-wstar-non-unif}.
If $k \neq k'$, then
\begin{align}
  &\left|\frac{1}{G} \int_{[-R,R]^d} p_1^2(x_1)\cdots p_k(x_k)p_k(M_1 x_k + T)\cdots p_d^2(x_d) e^{2\pi i x^\intercal w^\star(j-j')}\,dx\right|\\
  &\leq \frac{1}{\left(\int_{-R}^{+R} p_k^2(x_k)\,dx_k\right)\left(\int_{-R}^{+R} p_{k'}^2(x_{k'})\,dx_{k'}\right)} \left|\int_{-R}^{+R} p_k(x_k)p_k(M_1x_k + T) e^{2\pi i x_k w_k^\star (j-j')}\,dx_k \right|\\
  &\cdot\left|\int_{-R}^{+R} p_{k'}^2(x_{k'}) e^{2\pi i x_{k'} w_{k'}^\star (j-j')}\,dx_{k'} \right|\\
  &\leq \frac{2R}{\left(\int_{-R}^{+R} p_k^2(x_k)\,dx_k\right)\left(\int_{-R}^{+R} p_{k'}^2(x_{k'})\,dx_{k'}\right)}\left|\int_{-R}^{+R} e^{2\pi i x_{k'} w_{k'}^\star (j-j')}\,dx_{k'} \right|\\
  &\leq \frac{50}{81R}\left|\int_{-R}^{+R} e^{2\pi i x_{k'} w_{k'}^\star (j-j')}\,dx_{k'} \right|.\\
  &\leq \frac{3}{4R}\left|\int_{-R}^{+R} e^{2\pi i x_{k'} w_{k'}^\star (j-j')}\,dx_{k'} \right|
\end{align}
In the second inequality, we use Assumption~\ref{assum:bounded-1} that $p_k(z) \leq 1$.
In the next to last inequality, we use Assumption~\ref{assum:pointwise-close}.
In particular, by Assumption~\ref{assum:pointwise-close}, we have
\begin{equation}
  \int_{-R}^{+R} p_k^2(z)\,dz - 2R = \int_{-R}^{+R} (p_k^2(z) - 1)\,dz \geq -\frac{R}{5}.
\end{equation}
Thus, this implies that $\int_{-R}^{+R}p_k^2(z)\,dz \geq 2R - R/5 = 9R/5$.
From here, again, the proof is the same as \Cref{lem:complex-exp-non-unif} and \Cref{coro:complex-exp-wstar-non-unif}.
\end{proof}

Now, we can use this to obtain a lower bound for an integral of a product of cosines, as in \Cref{lem:integral}.
In this next integral bound, we also require Assumption~\ref{assum:even}.

\begin{lemma}
\label{lem:integral-non-unif}
Let $\varphi^2 \propto \prod_{k=1}^d p_k^2$ be a probability distribution over $[-R, R]^d$ satisfying Assumptions~\ref{assum:pointwise-close}-\ref{assum:even} for a truncation parameter $R$.
Let $w^\star \in \mathbb{R}^d$ be unknown with norm $R_w > 0$, and let $\hat{w} \in \mathbb{R}^d$ be an approximation of $w^\star$ with $\norm{\hat{w} - w^\star}_\infty \leq \epsilon_1$.
Let $1 \leq j \leq D$ be an integer, for $D \in \mathbb{N}$ from \Cref{eq:g-tilde}.
Then,
\begin{equation}
  \int\limits_{x \sim \varphi^2} \cos(2\pi j x^\intercal \hat{w})\cos(2\pi j x^\intercal w^\star)\,dx \geq \frac{1}{2} - \frac{3\sqrt{d}}{16\pi R_w R} - \frac{5\pi^2 D^2 R^2 d \epsilon_1}{2}.
\end{equation}
\end{lemma}

\begin{proof}
The proof follows similarly to that of \Cref{lem:integral}.
Using the sum formulas for cosine, we have
\begin{align}
  &\int_{x \sim \varphi^2} \cos(2\pi j x^\intercal \hat{w})\cos(2\pi j x^\intercal w^\star)\,dx\\
  &= \int_{x \sim \varphi^2} \cos(2\pi j x^\intercal(w^\star + (\hat{w} - w^\star)))\cos(2\pi j x^\intercal w^\star)\,dx\\
  &= \int_{x \sim \varphi^2} \left(\cos(2\pi j x^\intercal w^\star)\cos(2\pi j x^\intercal (\hat{w} - w^\star)) - \sin(2\pi j x^\intercal w^\star) \sin(2\pi j x^\intercal (\hat{w} - w^\star))\right)\cos(2\pi j x^\intercal w^\star)\,dx\\
  &\geq \int_{x\sim\varphi^2} \cos^2(2\pi j x^\intercal w^\star) \left(1 - \frac{1}{2}(2\pi jx^\intercal (\hat{w} - w^\star))^2\right) - \sin(2\pi jx^\intercal w^\star)\sin(2\pi j x^\intercal (\hat{w}-w^\star))\cos(2\pi jx^\intercal w^\star)\,dx\\
  &\geq \int_{x \sim \varphi^2}\cos^2(2\pi j x^\intercal w^\star) \,dx - 2\pi^2 j^2 \int_{x \sim \varphi^2} \left(x^\intercal (\hat{w} - w^\star)\right)^2\,dx - 2\pi j\int_{x \sim \varphi^2} |x^\intercal (\hat{w} - w^\star)|\,dx.\label{eq:three-terms-non-unif}
\end{align}
In the third line, we use the sum formula for cosines.
In the fourth line, we use that $\cos(y) \geq 1 - y^2/2$.
In the fifth line, we use that $\sin(y), \cos(y) \leq 1$ and $\sin(y) \leq |y|.$
We want to lower bound the first term and upper bound the second two.

First, we will lower bound the first term in~\Cref{eq:three-terms-non-unif}.
We can expand the first term in terms of complex exponentials:
\begin{align}
  \int_{x \sim \varphi^2} \cos^2(2\pi j x^\intercal w^\star)\,dx &= \frac{1}{4}\int_{x \sim \varphi^2}\left(e^{2\pi i jx^\intercal w^\star} + e^{-2\pi ij x^\intercal w^\star}\right)^2\,dx\\
  &= \frac{1}{2} + \frac{1}{4}\int_{x \sim \varphi^2}e^{4\pi i jx^\intercal w^\star}\,dx + \frac{1}{4}\int_{x \sim \varphi^2}e^{-4\pi ij x^\intercal w^\star}\,dx.
\end{align}
Now, we can bound the absolute value of these complex exponentials via \Cref{coro:complex-exp-wstar-non-unif}.
Note that \Cref{coro:complex-exp-wstar-non-unif} applies because we only needed to use that $j \neq j'$ to lower bound $|j-j'| \geq 1$.
This already clearly holds for $j \geq 1$.
Thus, we have
\begin{equation}
  \label{eq:cos-abs-non-unif}
  \left|\int_{x \sim \varphi^2} \cos^2(2\pi j x^\intercal w^\star)\,dx - \frac{1}{2}\right| \leq \frac{1}{2}\left|\int_{x \sim \varphi^2} e^{4\pi i j x^\intercal w^\star}\,dx \right|\leq \frac{3}{16\pi R} \frac{\sqrt{d}}{R_w}.
\end{equation}
Rearranging, we have
\begin{equation}
  \label{eq:three-terms-1-non-unif}
  \int_{x \sim \varphi^2} \cos^2(2\pi j x^\intercal w^\star)\,dx \geq \frac{1}{2} - \frac{3\sqrt{d}}{16\pi R_w R}.
\end{equation}
This gives a lower bound on the first term in \Cref{eq:three-terms-non-unif}.
We still need to upper bound the other terms in \Cref{eq:three-terms-non-unif}.
For the second term, we can first directly evaluate the integral.

For $\varphi^2 \propto \prod_{k=1}^d p_k^2$ over $[-R,R]^d$, we have
\begin{equation}
\int_{x \sim \varphi^2} \,dx = \frac{1}{\prod_{k=1}^d \left(\int_{-R}^{+R} p_k^2(z)\,dz\right)}\int_{x_1=-R}^{+R}\cdots \int_{x_d=-R}^{+R} p_1^2(x_1) \cdots p_d^2(x_d) \,dx_d\cdots dx_1 = 1.
\end{equation}
For simplicity, from here on, we denote the normalizing factor by $G$.
Then,
\begin{align}
  &\int_{x \sim \varphi^2} (x^\intercal (\hat{w} - w^\star))^2 \,dx\\
  &= \frac{1}{G}\int_{x_1=-R}^{+R}\cdots \int_{x_d=-R}^{+R}\left(\sum_{i=1}^d x_i \hat{w}_i - x_iw^\star_i\right)^2 p_1^2(x_1)\cdots p_d^2(x_d)\,dx_d \cdots \,dx_1\\
  &= \frac{1}{G}\int_{x_1=-R}^{+R}\cdots \int_{x_d=-R}^{+R}\left(\sum_{i,i'=1}^d x_i x_{i'}\hat{w}_i \hat{w}_{i'} + x_i x_{i'}w^\star_i w^\star_{i'} - x_ix_{i'}\hat{w}_i w_{i'}^\star - x_ix_{i'}w_i^\star \hat{w}_{i'}\right)\\
  & \cdot p_1^2(x_1)\cdots p_d^2(x_d)\,dx_d \cdots \,dx_1.
\end{align}
Here, notice that
\begin{align}
  &\frac{1}{G}\int_{x_1=-R}^{+R}\cdots \int_{x_d=-R}^{+R} x_i x_{i'} p_1^2(x_1)\cdots p_d^2(x_d) \, dx_d \cdots\, dx_1\\
  &= \frac{1}{\left(\int_{-R}^{+R} p_i^2(z)\,dz\right)\left(\int_{-R}^{+R} p_{i'}^2(z)\,dz\right)} \int_{x_i = -R}^{+R} \int_{x_{i'} =-R}^{+R} x_i x_{i'} p_i^2(x_i) p_{i'}^2(x_{i'}) \,dx_{i'}\,dx_i\\
  &= \frac{\delta_{ii'}}{\left(\int_{-R}^{+R} p_i^2(z)\,dz\right)} \int_{x=-R}^{+R} x^2 p_i^4(x)\,dx\\
  &\leq \frac{\delta_{ii'}}{\left(\int_{-R}^{+R} p_i^2(z)\,dz\right)} \int_{x=-R}^{+R}x^2\,dx\\
  &= \frac{\delta_{ii'}}{\left(\int_{-R}^{+R} p_i^2(z)\,dz\right)} \frac{2R^3}{3}\\
  &\leq \frac{R^2}{2} \delta_{ii'},
\end{align}
where the third line follows because if $i \neq i'$, we are integrating an odd function over a symmetric interval since $p^2$ is even by Assumption~\ref{assum:even}.
The fourth line follows by Assumption~\ref{assum:bounded-1} that $p_i(z) \leq 1$.
The last line follows by Assumption~\ref{assum:pointwise-close}.
In particular, by Assumption~\ref{assum:pointwise-close}, we have
\begin{equation}
  \int_{-R}^{+R} p_k^2(z)\,dz - 2R = \int_{-R}^{+R} (p_k^2(z) - 1)\,dz \geq -\frac{R}{5}.
\end{equation}
Thus, this implies that $\int_{-R}^{+R}p_k^2(z)\,dz \geq 2R - R/5 \geq 4R/3$.
Plugging this into our previous expression, we have
\begin{align}
  \int_{x \sim \varphi^2} (x^\intercal (\hat{w} - w^\star))^2 \,dx &\leq \frac{R^2}{2} \left(\sum_{i=1}^d (\hat{w}_i)^2 + \left(w_i^\star\right)^2 - 2\hat{w}_i w_i^\star\right)\\
  &= \frac{R^2}{2}\norm{\hat{w} - w^\star}_2^2\\
  &\leq \frac{R^2}{2}d\epsilon_1^2,\label{eq:three-terms-2-non-unif}
\end{align}
where in the last line, we used $|\hat{w}_i - w_i^\star| \leq \epsilon_1$ for all $i \in [d]$.
Finally, we can similarly upper bound the last term in \Cref{eq:three-terms-non-unif}.
\begin{align}
  &\int_{x\sim \varphi^2} |x^\intercal (\hat{w} - w^\star)| \,dx\\
  &= \frac{1}{G}\int_{x_1=-R}^{+R} \cdots \int_{x_d =-R}^{+R} \left|\sum_{i=1}^d x_i(\hat{w}_i - w_i^\star)\right| p_1^2(x_1)\cdots p_d^2(x_d) \,dx_d\cdots dx_1\\
  &\leq \frac{1}{G}\int_{x_1=-R}^{+R} \cdots \int_{x_d =-R}^{+R}\sum_{i=1}^d |x_i(\hat{w}_i - w_i^\star)| p_1^2(x_1)\cdots p_d^2(x_d)\,dx_d \cdots dx_1\\
  &= \sum_{i=1}^d \frac{1}{\int_{-R}^{+R} p_i^2(z)\,dz} |\hat{w}_i - w^\star_i| \int_{x_i=-R}^{+R}|x_i| p_i^2(x_i)\,dx_i\\
  &\leq \sum_{i=1}^d \frac{1}{\int_{-R}^{+R} p_i^2(z)\,dz} |\hat{w}_i - w^\star_i| \int_{x_i=-R}^{+R}|x_i|\,dx_i\\
  &\leq R^2 \epsilon_1 \sum_{i=1}^d \frac{1}{\int_{-R}^{+R} p_i^2(z)\,dz}\\
  &\leq \frac{3\epsilon_1 dR}{4}.\label{eq:three-terms-3-non-unif}
\end{align}
In the third line, we use triangle inequality.
In the fifth line, we use Assumption~\ref{assum:bounded-1} that $p_k(z) \leq 1$.
In the sixth line, we use that $|\hat{w}_i - w_i^\star| \leq \epsilon_1$ for all $i \in [d]$ and evaluate the integral.
In the last line, we use Assumption~\ref{assum:pointwise-close}.

Combining \Cref{eq:three-terms-1-non-unif,eq:three-terms-2-non-unif,eq:three-terms-3-non-unif} in \Cref{eq:three-terms-non-unif}, we have
\begin{align}
  \int_{x \sim \varphi^2}\cos(2\pi j x^\intercal \hat{w}) \cos(2\pi j x^\intercal w^\star)\,dx &\geq \frac{1}{2} - \frac{3\sqrt{d}}{16\pi R_w R} - \pi^2 j^2 R^2 d\epsilon_1^2 - \frac{3\pi j \epsilon_1 d R}{2}\\
  &\geq \frac{1}{2} - \frac{3\sqrt{d}}{16\pi R_w R} - \pi^2 j^2 R^2 d\epsilon_1 - \frac{3\pi^2 j^2 \epsilon_1 d R^2}{2}\\
  &\geq \frac{1}{2} - \frac{3\sqrt{d}}{16\pi R_w R} - \frac{5\pi^2 D^2 R^2 d \epsilon_1}{2},
\end{align}
where in the second line we use that $j, R \geq 1$ so that $j^2 \geq j$ and $R^2 \geq R$ and $\epsilon_1 < 1$ so that $\epsilon_1^2 \leq \epsilon_1$.
In the last line, we use that $j \leq D$.
\end{proof}

\begin{corollary}
\label{coro:integral-non-unif}
Let $\varphi^2 \propto \prod_{k=1}^d p_k^2$ be a probability distribution over $[-R, R]^d$ satisfying Assumptions~\ref{assum:pointwise-close} and \ref{assum:bounded-1} for a truncation parameter $R$.
Let $w^\star \in \mathbb{R}^d$ be unknown with norm $R_w > 0$, and let $\hat{w} \in \mathbb{R}^d$ be an approximation of $w^\star$ with $\norm{\hat{w} - w^\star}_\infty \leq \epsilon_1$.
Let $1 \leq j\leq D$ be an integer, for $D \in \mathbb{N}$ from \Cref{eq:g-tilde}.
Then,
\begin{equation}
  \int_{x \sim \varphi^2}\cos^2(2 \pi j x^\intercal \hat{w})\,dx \geq \frac{1}{2} - \frac{3\sqrt{d}}{16\pi R(R_w - \sqrt{d}\epsilon_1)}.
\end{equation}
\end{corollary}

\begin{proof}
The proof follows from the lower bound of the first term in \Cref{eq:three-terms-non-unif} in the proof of \Cref{lem:integral-non-unif}.
We can expand the first term in terms of complex exponentials:
\begin{align}
  \int_{x \sim \varphi^2} \cos^2(2\pi j x^\intercal \hat{w})\,dx &= \frac{1}{4}\int_{x \sim \varphi^2}\left(e^{2\pi i jx^\intercal\hat{w}} + e^{-2\pi ij x^\intercal \hat{w}}\right)^2\,dx\\
  &= \frac{1}{2} + \frac{1}{4}\int_{x \sim \varphi^2}e^{4\pi i jx^\intercal \hat{w}}\,dx + \frac{1}{4}\int_{x \sim \varphi^2}e^{-4\pi ij x^\intercal \hat{w}}\,dx.
\end{align}
Now, we can bound the absolute value of these complex exponentials via \Cref{lem:complex-exp-non-unif} (instead of \Cref{coro:complex-exp-wstar-non-unif}).
Note that \Cref{lem:complex-exp-non-unif} applies because we only needed to use that $j \neq j'$ to lower bound $|j-j'| \geq 1$.
This already clearly holds for $j \geq 1$.
Thus, we have
\begin{equation}
  \label{eq:cos-abs-what-non-unif}
  \left|\int_{x \sim \varphi^2} \cos^2(2\pi j x^\intercal \hat{w})\,dx - \frac{1}{2}\right| \leq \frac{1}{2}\left|\int_{x \sim \varphi^2} e^{4\pi i j x^\intercal \hat{w}}\,dx \right|\leq \frac{3}{16\pi R} \frac{\sqrt{d}}{R_w - \sqrt{d}\epsilon_1}.
\end{equation}
Rearranging, we have
\begin{equation}
  \int_{x \sim \varphi^2} \cos^2(2\pi j x^\intercal \hat{w})\,dx \geq \frac{1}{2} - \frac{3\sqrt{d}}{16\pi R(R_w - \sqrt{d}\epsilon_1)}.
\end{equation}
\end{proof}

\begin{lemma}
\label{lem:integral-upper-non-unif}
Let $\varphi^2 \propto \prod_{k=1}^d p_k^2$ be a probability distribution over $[-R, R]^d$ satisfying Assumptions~\ref{assum:pointwise-close} and \ref{assum:bounded-1} for a truncation parameter $R$.
Let $w^\star \in \mathbb{R}^d$ be unknown with norm $R_w > 0$, and let $\hat{w} \in \mathbb{R}^d$ be an approximation of $w^\star$ with $\norm{\hat{w} - w^\star}_\infty \leq \epsilon_1$.
Let $1 \leq j\leq D$ be an integer, for $D \in \mathbb{N}$ from \Cref{eq:g-tilde}.
Then,
\begin{equation}
  \int\limits_{x \sim \varphi^2} \cos(2\pi jx^\intercal \hat{w}) \cos(2\pi jx^\intercal w^\star)\,dx \leq \frac{1}{2} + \frac{3\sqrt{d}}{16\pi R_w R} + \frac{3\pi D d\epsilon_1 R}{2}.
\end{equation}
\end{lemma}

\begin{proof}
The proof of this is similar to that of~\Cref{lem:integral,lem:integral-non-unif}.
Using the sum formulas for cosine, we have
\begin{align}
  &\int_{x \sim \varphi^2} \cos(2\pi j x^\intercal \hat{w})\cos(2\pi j x^\intercal w^\star)\,dx\\
  &= \int_{x \sim \varphi^2} \cos(2\pi j x^\intercal(w^\star + (\hat{w} - w^\star)))\cos(2\pi j x^\intercal w^\star)\,dx\\
  &= \int_{x \sim \varphi^2} \left(\cos(2\pi j x^\intercal w^\star)\cos(2\pi j x^\intercal (\hat{w} - w^\star)) - \sin(2\pi j x^\intercal w^\star) \sin(2\pi j x^\intercal (\hat{w} - w^\star))\right)\cos(2\pi j x^\intercal w^\star)\,dx\\
  &\leq \int_{x\sim\varphi^2} \cos^2(2\pi j x^\intercal w^\star) -\sin(2\pi jx^\intercal w^\star)\sin(2\pi j x^\intercal (\hat{w}-w^\star))\cos(2\pi jx^\intercal w^\star)\,dx\\
  &\leq \int_{x\sim\varphi^2} \cos^2(2\pi j x^\intercal w^\star) +\sin(2\pi j x^\intercal (\hat{w}-w^\star))\,dx\\
  &\leq \int_{x\sim\varphi^2} \cos^2(2\pi j x^\intercal w^\star)\,dx + 2\pi j\int_{x \sim\varphi^2} |x^\intercal (\hat{w}-w^\star)|\,dx\label{eq:two-terms-non-unif}.
\end{align}
In the fourth line, we use that $\cos(y) \leq 1$.
In the fifth line, we use that $-\sin(y) \cos(y) \leq 1$.
In the last line, we use that $\sin(y) \leq |y|$.
We want to upper bound both of these terms, which is simple given the proof of \Cref{lem:integral-non-unif}.

Namely, in \Cref{eq:cos-abs-non-unif}, we showed that
\begin{equation}
  \left|\int_{x \sim \varphi^2} \cos^2(2\pi j x^\intercal w^\star)\,dx - \frac{1}{2}\right| \leq \frac{3}{16\pi R} \frac{\sqrt{d}}{R_w}.
\end{equation}
Thus, we can upper bound
\begin{equation}
  \label{eq:two-terms-1-non-unif}
  \int_{x\sim \varphi^2} \cos^2 (2\pi j x^\intercal w^\star)\,dx \leq \frac{1}{2} + \frac{3\sqrt{d}}{16\pi R_w R}
\end{equation}
Note that we have already upper bounded the third term in~\Cref{eq:three-terms-3-non-unif}:
\begin{equation}
  \label{eq:two-terms-2-non-unif}
  2\pi j\int_{x \sim\varphi^2} |x^\intercal (\hat{w}-w^\star)|\,dx \leq \frac{3 \pi j d \epsilon_1 R}{2}\leq \frac{3\pi D d\epsilon_1 R}{2}.
\end{equation}
Note that this part of the proof did not require Assumption~\ref{assum:even}.
Combining \Cref{eq:two-terms-1-non-unif} and \Cref{eq:two-terms-2-non-unif} in \Cref{eq:two-terms-non-unif}, we have
\begin{equation}
  \int_{x \sim \varphi^2} \cos(2\pi j x^\intercal \hat{w})\cos(2\pi j x^\intercal w^\star)\,dx \leq \frac{1}{2} + \frac{3\sqrt{d}}{16\pi R_w R} + \frac{3\pi D d\epsilon_1 R}{2}.
\end{equation}
\end{proof}

\begin{corollary}
\label{coro:integral-upper-non-unif}
Let $\varphi^2 \propto \prod_{k=1}^d p_k^2$ be a probability distribution over $[-R, R]^d$ satisfying Assumptions~\ref{assum:pointwise-close} and \ref{assum:bounded-1} for a truncation parameter $R$.
Let $w^\star \in \mathbb{R}^d$ be unknown with norm $R_w > 0$, and let $\hat{w} \in \mathbb{R}^d$ be an approximation of $w^\star$ with $\norm{\hat{w} - w^\star}_\infty \leq \epsilon_1$.
Let $1 \leq j\leq D$ be an integer, for $D \in \mathbb{N}$ from \Cref{eq:g-tilde}.
Then,
\begin{equation}
  \int\limits_{x \sim \varphi^2} \cos^2(2\pi j x^\intercal \hat{w})\,dx \leq \frac{1}{2} + \frac{3\sqrt{d}}{16\pi R (R_w - \sqrt{d}\epsilon_1)}.
\end{equation}
\end{corollary}

\begin{proof}
This follows directly from \Cref{eq:cos-abs-what-non-unif}.
\end{proof}

We also have a non-uniform analogue of \Cref{lem:integral2}.
This is similar in spirit to the previous lemmas.

\begin{lemma}
\label{lem:integral2-non-unif}
Let $\varphi^2 \propto \prod_{k=1}^d p_k^2$ be a probability distribution over $[-R, R]^d$ satisfying Assumptions~\ref{assum:pointwise-close} and \ref{assum:bounded-1} for a truncation parameter $R$.
Let $w^\star \in \mathbb{R}^d$ be unknown with norm $R_w > 0$, and let $\hat{w} \in \mathbb{R}^d$ be an approximation of $w^\star$ with $\norm{\hat{w} - w^\star}_\infty \leq \epsilon_1$.
Let $1 \leq j,j' \leq D$ be integers with $j \neq j'$, for $D \in \mathbb{N}$ from \Cref{eq:g-tilde}.
Then,
\begin{equation}
  \left|\int\limits_{x\sim \varphi^2} \cos(2\pi j x^\intercal \hat{w}) \cos(2\pi j' x^\intercal \hat{w}) \,dx\right|\leq \frac{3}{4\pi R}\frac{\sqrt{d}}{R_w - \sqrt{d}\epsilon_1}.
\end{equation}
\end{lemma}

\begin{proof}
The proof follows similarly to that of \Cref{lem:integral2}.
Using the product formulas for cosine, we can write the integral as
\begin{equation}
  \label{eq:sum-prod2-non-unif}
  \left|\int\limits_{x\sim \varphi^2} \cos(2\pi j x^\intercal \hat{w}) \cos(2\pi j' x^\intercal \hat{w}) \,dx\right| = \left|\frac{1}{2} \int_{x\sim \varphi^2} \cos(2\pi x^\intercal \hat{w}(j -j')) + \cos(2\pi x^\intercal \hat{w} (j + j'))\,dx\right|.
\end{equation}
We can bound each of the integrals on the right hand side similarly.
Starting with the first term, we can write it in terms of complex exponentials
\begin{align}
  \label{eq:complex-sum-non-unif}
  \left|\int_{x\sim \varphi^2} \cos(2\pi x^\intercal \hat{w}(j -j'))\,dx\right| \leq \frac{1}{2}\left| \int_{x \sim \varphi^2} e^{2\pi i x^\intercal \hat{w} (j-j')}\,dx \right| + \frac{1}{2}\left| \int_{x \sim \varphi^2} e^{2\pi i x^\intercal \hat{w} (j' - j)}\,dx \right|
\end{align}
Both terms in \Cref{eq:complex-sum-non-unif} can be bounded via \Cref{lem:complex-exp-non-unif}.
Thus, this bounds the first term in \Cref{eq:sum-prod2-non-unif} as
\begin{equation}
  \label{eq:sum-prod2-1-non-unif}
  \left|\int_{x\sim \varphi^2} \cos(2\pi x^\intercal \hat{w}(j -j'))\,dx\right| \leq \frac{3}{4\pi R}\frac{\sqrt{d}}{R_w - \sqrt{d}\epsilon_1}.
\end{equation}
We can similarly bound the second term in \Cref{eq:sum-prod2-non-unif}.
Namely, the argument is the same as the above and \Cref{lem:integral2} so that we have
\begin{align}
  \left|\int_{x \sim \varphi^2} e^{2\pi i x^\intercal  \hat{w}(j+j')}\,dx\right| &\leq \frac{3}{4R}\frac{1}{\pi |j + j'||\hat{w}_k|}\\
  &\leq \frac{1}{4R}\frac{1}{\pi |\hat{w}_k|},
\end{align}
where since $j \neq j'$ and $j,j' \geq 1$, then $|j + j'| \geq 3$. The rest of the bound follows the same argument.
Then, we obtain
\begin{equation}
  \left|\int_{x\sim \varphi^2} \cos(2\pi x^\intercal \hat{w}(j +j'))\,dx\right| \leq \frac{1}{4\pi R}\frac{\sqrt{d}}{R_w - \sqrt{d}\epsilon_1} \leq \frac{3}{4\pi R}\frac{\sqrt{d}}{R_w - \sqrt{d}\epsilon_1}.
\end{equation}
Thus, combined with \Cref{eq:sum-prod2-1-non-unif} in \Cref{eq:sum-prod2-non-unif}, we have
\begin{equation}
  \left|\int\limits_{x\sim \varphi^2} \cos(2\pi j x^\intercal \hat{w}) \cos(2\pi j' x^\intercal \hat{w}) \,dx\right| \leq \frac{3}{4\pi R}\frac{\sqrt{d}}{R_w - \sqrt{d}\epsilon_1}.
\end{equation}
\end{proof}

By essentially the same proof, we can obtain a similar upper bound replacing $\hat{w}$ with $w^\star$.

\begin{corollary}
\label{coro:integral2-wstar-non-unif}
Let $\varphi^2 \propto \prod_{k=1}^d p_k^2$ be a probability distribution over $[-R, R]^d$ satisfying Assumptions~\ref{assum:pointwise-close} and \ref{assum:bounded-1} for a truncation parameter $R$.
Let $w^\star \in \mathbb{R}^d$ be unknown with norm $R_w > 0$, and let $\hat{w} \in \mathbb{R}^d$ be an approximation of $w^\star$ with $\norm{\hat{w} - w^\star}_\infty \leq \epsilon_1$.
Let $1 \leq j,j' \leq D$ be integers with $j \neq j'$, for $D \in \mathbb{N}$ from \Cref{eq:g-tilde}.
Then,
\begin{equation}
  \left|\int\limits_{x\sim \varphi^2} \cos(2\pi j x^\intercal w^\star) \cos(2\pi j' x^\intercal w^\star) \,dx\right|\leq \frac{3}{4\pi R}\frac{\sqrt{d}}{R_w}.
\end{equation}
\end{corollary}

\begin{corollary}
\label{coro:integral2-wstar-sin-non-unif}
Let $\varphi^2 \propto \prod_{k=1}^d p_k^2$ be a probability distribution over $[-R, R]^d$ satisfying Assumptions~\ref{assum:pointwise-close} and \ref{assum:bounded-1} for a truncation parameter $R$.
Let $w^\star \in \mathbb{R}^d$ be unknown with norm $R_w > 0$, and let $\hat{w} \in \mathbb{R}^d$ be an approximation of $w^\star$ with $\norm{\hat{w} - w^\star}_\infty \leq \epsilon_1$.
Let $1 \leq j,j' \leq D$ be integers with $j \neq j'$, for $D \in \mathbb{N}$ from \Cref{eq:g-tilde}.
Then,
\begin{equation}
  \left|\int\limits_{x\sim \varphi^2} \cos(2\pi j x^\intercal w^\star) \sin(2\pi j' x^\intercal w^\star) \,dx\right|\leq \frac{3}{4 \pi R}\frac{\sqrt{d}}{R_w}.
\end{equation}
\end{corollary}

\begin{proof}
This follows by the same proof as \Cref{lem:integral2-non-unif} and \Cref{coro:integral2-wstar-non-unif}.
In particular, using the sum-product formulas for sine and cosine, we have
\begin{equation}
  \left|\int_{x\sim \varphi^2} \cos(2\pi j x^\intercal w^\star) \sin(2\pi j' x^\intercal w^\star) \,dx\right| = \left|\frac{1}{2}\int_{x\sim\varphi^2} \sin(2\pi (j+j') x^\intercal w^\star) + \sin(2\pi (j'-j)x^\intercal w^\star)\,dx \right|.
\end{equation}
Then, writing in terms of complex exponentials, we have
\begin{equation}
  \left|\int_{x\sim\varphi^2} \sin(2\pi (j'-j)x^\intercal w^\star)\,dx \right| \leq \frac{1}{|2i|}\left|\int_{x\sim\varphi^2} e^{2\pi ix^\intercal w^\star (j'-j)}\,dx \right| + \frac{1}{|2i|}\left|\int_{x\sim\varphi^2} e^{2\pi ix^\intercal w^\star (j-j')}\,dx \right|.
\end{equation}
The rest of the proof is the same as \Cref{lem:integral2-non-unif}, using \Cref{coro:complex-exp-wstar-non-unif} instead of \Cref{lem:complex-exp-non-unif} to bound the complex exponential terms.
\end{proof}

Finally, we need another integral bound that is also similar to \Cref{lem:integral2-non-unif}.
This is the non-uniform analogue of \Cref{coro:integral2-wstar-hat}.
The proof of this result follows easily following the steps of \Cref{coro:integral2-wstar-non-unif} and \Cref{coro:integral2-wstar-hat}.

\begin{corollary}
\label{coro:integral2-wstar-hat-non-unif}
Let $\varphi^2 \propto \prod_{k=1}^d p_k^2$ be a probability distribution over $[-R, R]^d$ satisfying Assumptions~\ref{assum:pointwise-close} and \ref{assum:bounded-1} for a truncation parameter $R$.
Let $w^\star \in \mathbb{R}^d$ be unknown with norm $R_w > 0$, and let $\hat{w} \in \mathbb{R}^d$ be an approximation of $w^\star$ with $\norm{\hat{w} - w^\star}_\infty \leq \epsilon_1$, where $\epsilon_1 \leq R_w/(D\sqrt{d})$.
Let $1 \leq j,j' \leq D$ be integers with $j \neq j'$, for $D \in \mathbb{N}$ from \Cref{eq:g-tilde}.
Then,
\begin{equation}
  \left|\int\limits_{x\sim \varphi^2} \cos(2\pi j x^\intercal w^\star) \cos(2\pi j' x^\intercal \hat{w}) \,dx\right|\leq \frac{3}{4\pi R}\frac{\sqrt{d}}{R_w - D\sqrt{d}\epsilon_1}.
\end{equation}
\end{corollary}

\end{document}